\documentclass[11pt]{article}
\usepackage[margin=1.15in]{geometry}
\usepackage{amsmath,amssymb,amsthm}
\usepackage{tikz}
\usetikzlibrary{calc}
\usepackage{subcaption}
\usepackage{xcolor}
\usepackage{microtype}
\usepackage{needspace,etoolbox}
\AtBeginEnvironment{theorem}{\needspace{3\baselineskip}}
\AtBeginEnvironment{lemma}{\needspace{3\baselineskip}}
\AtBeginEnvironment{corollary}{\needspace{3\baselineskip}}
\AtBeginEnvironment{proposition}{\needspace{3\baselineskip}}
\AtBeginEnvironment{definition}{\needspace{3\baselineskip}}
\AtBeginEnvironment{remark}{\needspace{3\baselineskip}}
\pretocmd{\section}{\needspace{10\baselineskip}}{}{}
\usepackage[numbers,sort&compress]{natbib}
\usepackage[colorlinks=true,linkcolor=blue,citecolor=blue,urlcolor=blue]{hyperref}
\newcommand{\doi}[1]{\url{https://doi.org/#1}}

\newtheorem{theorem}{Theorem}[section]
\newtheorem{lemma}[theorem]{Lemma}
\newtheorem{corollary}[theorem]{Corollary}

\newtheorem{definition}[theorem]{Definition}
\newtheorem{remark}[theorem]{Remark}

\newcommand{\md}{\operatorname{md}}
\newcommand{\Int}{\operatorname{int}}
\newcommand{\hb}{\bar h}
\newcommand{\T}{\mathsf T}
\newcommand{\cell}{\operatorname{cell}}
\tikzset{vx/.style={circle,fill,inner sep=1.1pt},
         hv/.style={circle,draw,fill=white,inner sep=1.1pt},
         dg/.style={circle,draw,very thick,fill=white,inner sep=1.8pt},
         lm/.style={rectangle,fill,inner sep=2.1pt}}

\title{Metric Dimension is NP-complete\\ on Graphs of Treewidth Two}
\author{Tristan Benoit\\[2pt]
\small Normandie Univ, UNICAEN, ENSICAEN, CNRS, GREYC, 14000, Caen, France\\[2pt]
\small ORCID \href{https://orcid.org/0009-0008-5203-6124}{0009-0008-5203-6124}}
\date{}

\begin{document}
\maketitle

\begin{abstract}
The \emph{metric dimension} of a connected graph is the least number of
vertices whose distance vectors distinguish all vertices.  We show that
deciding whether the metric dimension of a graph is at most a given
integer is NP-complete on graphs of treewidth at most two, and even on
two-connected two-terminal series--parallel graphs of pathwidth at most five.  As the
problem is polynomial on trees, the graphs of treewidth one, this
settles the treewidth threshold of its hardness, which was known to lie
between two and twenty-four.  The reduction is from satisfiability of
formulas in which every literal occurs at most twice.  To such a formula
with $n$ clauses and $r$ variables it associates a two-terminal
series--parallel graph $G_\varphi$ with $O((n+r)^2)$ vertices and pathwidth at
most five, whose metric dimension
is $2(n+r)+7$ if the formula is satisfiable and $2(n+r)+8$ otherwise.

\end{abstract}


\section{Introduction}\label{sec:N00}

\subsection{The problem and the known results}

For a connected graph $G$, a vertex $w$ \emph{separates} two vertices
$x,y$ when $d(w,x)\ne d(w,y)$, and a set $W\subseteq V(G)$
\emph{resolves} $G$ when every two distinct vertices are separated by a
vertex of $W$; the \emph{metric dimension} $\md(G)$ is the least size of
a resolving set (Definition~\ref{def:N01-resolve}).  The decision problem
\textsc{Metric Dimension} asks, for a connected graph $G$ and an integer
$k\ge0$, whether $\md(G)\le k$.  Resolving sets were introduced by Slater
\cite{Slater75} and by Harary and Melter \cite{Harary76}.  The problem is
NP-complete in general; Garey and Johnson list it as problem GT61 of
their book \cite{Garey79}, and Khuller, Raghavachari and Rosenfeld
\cite{Khuller96} give a reduction from 3-SAT.  On trees it is polynomial:
the metric dimension of a tree that is not a path is the number of its
leaves minus the number of its exterior major vertices
\cite{Slater75,Harary76,Chartrand00}, and Khuller et al.\
\cite{Khuller96} compute a minimum resolving set of a tree in linear
time.

Between these two poles, the problem has been placed on many classes.
It stays NP-complete on planar graphs of maximum degree six
\cite{Diaz17}, on split graphs, bipartite graphs, co-bipartite graphs and
line graphs of bipartite graphs \cite{Epstein15}, on interval graphs of
diameter two and on permutation graphs of diameter two \cite{Foucaud17},
and on a subclass of unit disk graphs \cite{Hoffmann12}.  It is polynomial on outerplanar
graphs \cite{Diaz17}, on cographs and, for every fixed $k$, on trees
with $k$ additional edges \cite{Epstein15}, on chain graphs
\cite{Fernau15}, and, in linear time, on cactus block graphs
\cite{Hoffmann16}.  In the parameterized setting the problem is
W[2]-hard with respect to the solution size $k$, already on graphs of
maximum degree three, and it cannot be approximated within a factor
$o(\log n)$ unless $\mathrm{P}=\mathrm{NP}$ \cite{Hartung13}.  It is
fixed-parameter tractable with respect to the vertex cover number
\cite{Hartung13}, the max leaf number \cite{Eppstein15}, the tree-length
plus the maximum degree, and the modular width \cite{Belmonte17}, and the
distance to cluster or to co-cluster \cite{Galby23}.  For the vertex
cover number $\mathrm{vc}$, Foucaud et al.\ \cite{Foucaud27} give an
algorithm running in time $2^{O(\mathrm{vc}^2)}\cdot n^{O(1)}$ and show
that, unless the Exponential Time Hypothesis fails, no algorithm runs in
time $2^{o(\mathrm{vc}^2)}\cdot n^{O(1)}$, even on graphs of bounded
diameter.

The parameter treewidth resisted longest.  Díaz et al.\ \cite{Diaz17},
Belmonte et al.\ \cite{Belmonte17} and Eppstein \cite{Eppstein15} asked
for the complexity of the problem on graphs of bounded treewidth, and
Hoffmann et al.\ \cite{Hoffmann16} observed that the decision problem,
for each fixed $k\ge2$, cannot be expressed in monadic second-order
logic, so that Courcelle's theorem does not apply.  Bonnet and Purohit \cite{Bonnet21} proved that
the problem is W[1]-hard parameterized by treewidth, even by pathwidth;
Galby et al.\ \cite{Galby23} strengthened this to the feedback vertex set
number plus the pathwidth.  Li and Pilipczuk \cite{Li22} then proved that
the problem is NP-hard on graphs of treewidth, and even pathwidth, $24$,
and asked for the
exact treewidth at which it becomes hard, the case of treewidth two
being open.  Bonnet and Purohit \cite{Bonnet21} and Bousquet, Deschamps
and Parreau \cite{Bousquet23} ask the same question.  On the positive
side, Bousquet, Deschamps and Parreau prove that the problem is
fixed-parameter tractable with respect to treewidth on chordal graphs;
since $k$-trees are chordal, the problem is polynomial on $k$-trees for
every fixed $k$, and in particular on 2-trees, the maximal graphs of
treewidth two.  Foucaud et al.\ \cite{Foucaud24} give an algorithm
running in time $2^{\mathrm{diam}^{O(\mathrm{tw})}}\cdot n^{O(1)}$, so the
problem is polynomial when both the treewidth and the diameter are
bounded, and they show that this running time is essentially optimal
under the Exponential Time Hypothesis.

\subsection{The main theorem}

Our result answers the question of the complexity on graphs of
treewidth two.

\begin{theorem}[main theorem]\label{thm:N00-main}
The following problem is NP-complete: given a two-connected
two-terminal series--parallel graph $(G;s,t)$ of pathwidth at most five
and an integer $k\ge0$, decide whether $\md(G)\le k$.  Consequently the
problem is NP-complete on connected graphs of treewidth at most two and
pathwidth at most five.
\end{theorem}

\begin{proof}
This is Theorem~\ref{thm:N08-np} with Corollary~\ref{cor:N08-tw2}.
\end{proof}

The theorem rests on the following value theorem, with the restricted
formulas of Definition~\ref{def:N03-formula} (formulas in conjunctive
normal form with at most three literals per clause and every literal in
at most two clauses) and the graph $G_\varphi$ of
Definition~\ref{def:N03-graph}.

\begin{theorem}[the value of $\md(G_\varphi)$]\label{thm:N00-value}
Let $\varphi$ be a restricted formula with $n$ clauses and $r$ variables,
and $m:=n+r+1$.  Then $(G_\varphi;T_1,T_2)$ is two-terminal
series--parallel, so $G_\varphi$ has treewidth at most two; $G_\varphi$
has pathwidth at most five; and
\[
 \md(G_\varphi)=\begin{cases}2m+5&\text{if $\varphi$ is satisfiable,}\\
 2m+6&\text{otherwise.}\end{cases}
\]
\end{theorem}

\begin{proof}
This is Theorem~\ref{thm:N08-value} with Theorem~\ref{thm:N07-tw}(i),(ii)
and Theorem~\ref{thm:N10-pw}.
\end{proof}

With the results recalled above, the main theorem has the following
consequences.
\begin{itemize}
\item \emph{The treewidth threshold is two.}  Graphs of treewidth at
  most one are those without a $K_3$ minor \cite[Theorem~17]{Bodlaender98},
  that is, the forests, on which the problem is polynomial
  \cite{Khuller96}; by the main theorem it is NP-complete at treewidth
  two.  Thus, unless $\mathrm{P}=\mathrm{NP}$, the gap left by Li and
  Pilipczuk \cite{Li22}, between two and twenty-four, is closed, and in
  the terms of Bonnet and Purohit
  \cite{Bonnet21}, \textsc{Metric Dimension} parameterized by treewidth
  is para-NP-hard at the first value where it can be.
\item \emph{The pathwidth threshold lies between two and five.}  Li and
  Pilipczuk \cite{Li22} proved the hardness at pathwidth twenty-four.
  The graph $G_\varphi$ has pathwidth at most five
  (Theorem~\ref{thm:N10-pw}), and the problem is polynomial at pathwidth
  one, where the graphs are forests, as the pathwidth of a graph is at
  least its treewidth (Definition~\ref{def:N10-pd}).  So, unless
  $\mathrm{P}=\mathrm{NP}$, the least pathwidth at which the problem is
  NP-complete lies between two and five.
\item \emph{Outerplanar graphs are the boundary.}  Outerplanar graphs
  have treewidth at most two \cite{Bodlaender98}, and the problem is
  polynomial on them \cite{Diaz17}.  The graphs $G_\varphi$ are not
  outerplanar: they contain a subdivision of $K_{2,3}$
  (Lemma~\ref{lem:N09-placement}(i)).  So, within treewidth two, the
  boundary between easy and hard runs between outerplanar graphs and
  graphs with a $K_{2,3}$ minor.
\item \emph{Two-connected graphs.}  Two-connected graphs of treewidth
  two are the two-connected series--parallel graphs
  \cite[Theorem~42]{Bodlaender98}.  The main theorem is stated on
  two-connected two-terminal series--parallel graphs of pathwidth at
  most five, which have treewidth at most two
  (Lemma~\ref{lem:N07-sp}(iii)); so the problem is
  NP-complete on two-connected graphs of treewidth two.
\item \emph{The diameter must grow.}  By the algorithm of Foucaud et al.\
  \cite{Foucaud24}, a hard family of bounded treewidth must have
  unbounded diameter unless $\mathrm{P}=\mathrm{NP}$; the diameter of
  $G_\varphi$ is at least $4m+15$ (Lemma~\ref{lem:N09-placement}(ii)).
\item \emph{Trees with $k$ additional edges.}  Epstein, Levin and
  Woeginger \cite{Epstein15} solve the problem in polynomial time on
  trees with $k$ additional edges for every fixed $k$, with a running
  time whose exponent grows with $k$.  The graph $G_\varphi$ is a tree
  with at most $12m+9$ additional edges
  (Lemma~\ref{lem:N09-placement}(iii)) and has $O(m^2)$ vertices
  (Lemma~\ref{lem:N03-structure}(iv)); so no algorithm for this class
  runs in time polynomial in both the number of vertices and $k$,
  unless $\mathrm{P}=\mathrm{NP}$.
\end{itemize}

\subsection{Overview of the reduction}

The construction and the proof are described in full in
Sections~\ref{sec:N02} to~\ref{sec:N08}, with Sections~\ref{sec:N10}
and~\ref{sec:N09} aside; this overview gives their shape,
and every claim in it points to the statement that proves it.

\emph{The template.}  The building block is the template
$\T(L,\mu;E)$ of Definition~\ref{def:N02-template}: a path, the
\emph{spine}, from a terminal $s$ to a terminal $t$, with one
\emph{$t$-ear} attached near $t$ and up to two \emph{crossings}, ears
(Definition~\ref{def:N01-ear}, Figure~\ref{fig:N01-ear}) attached around
the middle of the spine whose lengths are set by the list $E$
(Figure~\ref{fig:N02-shape}).  Every vertex $u$ of a template has
a profile $(h(u),\hb(u))$, its distances to $s$ and to $t$, and a skew
$\kappa(u)=\hb(u)-h(u)$.  The vertices whose skew equals the number
$\sigma=L-2\mu$ are the \emph{diagonal} vertices: one on the spine and one
on each crossing, at positions set by $E$
(Lemma~\ref{lem:N02-skew}(iii); Figure~\ref{fig:N02-template} marks them
on two examples).  The first vertex of the $t$-ear is the
\emph{core}, a vertex that is closer to every other member of the
template than any path through $s$ or $t$ would make it
(Lemma~\ref{lem:N04-isometric}(iv)).

\emph{The graph $G_\varphi$.}  A restricted formula $\varphi$ with $n$
clauses and $r$ variables has $m=n+r+1$ \emph{levels}, one per clause,
one per variable, and a last level, the idle level, each with three
\emph{options}: at the level of a
clause, an option selects one of its literals; at the level of a
variable, the option $1$ sets it true and the options $2,3$ set it false;
at the idle level the three options are alike
(Definition~\ref{def:N03-options}).  The graph $G_\varphi$
(Definition~\ref{def:N03-graph}; Figure~\ref{fig:N03-shape}(a)) has a
\emph{frame}, two paths
$T_1,s_1,\dots,s_m$ and $t_m,\dots,t_1,z_1,\dots,z_7,T_2$; at each level
$k$, three \emph{copies} $U_{k,1},U_{k,2},U_{k,3}$ of a template between
$s_k$ and $t_k$, one per option, nested inside one another along the
levels (Figure~\ref{fig:N03-shape}(b)); a \emph{pair} $q,q'$ of vertices, each
adjacent to $t_m$ and to $t_{m-1}$ (Figure~\ref{fig:N03-shape}(d)); and four \emph{arms}, paths of
length $D=4m+15$ from $T_1$ to $T_2$, each carrying one \emph{twin}
$\tilde y^j$ of one of its vertices $w_j$ (Figure~\ref{fig:N03-shape}(c)).  The copies of one
level differ only in their crossings: the crossing lengths of $U_{k,c}$
encode the \emph{cells} of the option $c$, integers attached to the
occurrences of literals in $\varphi$, so that the diagonal vertices of
$U_{k,c}$ lie at the distances $2m+4-k$ and the cells of $c$ from $T_1$
(Lemma~\ref{lem:N08-cells}).  The cells are arranged so that two levels
have a common cell exactly when one is a clause level, the other is the
level of the variable of a literal of that clause, and the options chosen
at the two levels select that literal and falsify it
(Lemma~\ref{lem:N03-options}(iv)).

\emph{Patterns.}  A \emph{choice} $f$ selects one option $f_k$ at every
level, and its \emph{pattern} $W_f$ consists of the cores of the two
copies not selected at each level (the black squares of
Figure~\ref{fig:N03-shape}(b)), together with $q$ and $w_1,\dots,w_4$: a
set of $2m+5$ vertices (Definition~\ref{def:N03-pattern}).  The lower
bound and its rigidity are proved in Section~\ref{sec:N05}: every
resolving set contains a vertex of each of the five twin pairs
(Lemma~\ref{lem:N01-twins}) and a vertex inside at least two of the
three copies of every level, so it has at least $2m+5$ vertices
(Lemma~\ref{lem:N05-lower}), and a resolving set of exactly $2m+5$
vertices is a pattern up to an automorphism of $G_\varphi$
(Corollary~\ref{cor:N05-pattern}).

\emph{What a pattern separates.}  Section~\ref{sec:N06} shows that a
pattern $W_f$ separates every pair of vertices except the pairs of
diagonal vertices of two selected copies, at two different levels, at
the same distance from $T_1$ (Theorem~\ref{thm:N06-validity}).  Here each
landmark does one job: the rows $w_j$ read the distance to $T_1$; the
cores of the non-selected copies read the skew of a vertex of the
selected copy of their level and single out the diagonal vertices; $q$
compares the skews of two levels.  Adding $T_2$ to a pattern separates
the remaining pairs, so $\md(G_\varphi)\le2m+6$ always.

\emph{The value.}  By the arrangement of the cells, $W_f$ resolves
$G_\varphi$ exactly when no chosen literal is falsified by the chosen
value of its variable (Lemma~\ref{lem:N08-choice};
Figure~\ref{fig:N08-cells}), that is, exactly when the assignment read
off $f$ satisfies $\varphi$ (Corollary~\ref{cor:N08-sat}).  With the
rigidity of the lower bound this gives the value theorem
(Theorem~\ref{thm:N08-value}), and a linear-time normalisation of 3-CNF
formulas into restricted formulas gives the NP-completeness
(Theorem~\ref{thm:N08-np}).

\emph{Structure of $G_\varphi$.}  Every template, every arm and the
graph $G_\varphi$ itself, between $T_1$ and $T_2$, are two-terminal
series--parallel graphs (Lemma~\ref{lem:N07-template},
Theorem~\ref{thm:N07-tw}(i)); so $G_\varphi$ has treewidth at most two
(Theorem~\ref{thm:N07-tw}(ii)).  Moreover $G_\varphi$ is two-connected
(Lemma~\ref{lem:N09-placement}(iv)) and has pathwidth at most five
(Theorem~\ref{thm:N10-pw}).

\subsection{Open questions}

The main theorem settles the treewidth threshold.  The following
questions remain.
\begin{itemize}
\item \emph{Pathwidth.}  Unless $\mathrm{P}=\mathrm{NP}$, the least
  pathwidth at which the problem is NP-complete lies between two and
  five; its exact value is open.
\item \emph{The weighted problem.}  Epstein, Levin and Woeginger
  \cite{Epstein15} solve the vertex-weighted problem in polynomial time
  on paths, trees, cycles, cographs, wheels and, for every fixed $k$,
  trees with $k$ additional edges, and Díaz et al.\ \cite{Diaz17} ask
  whether the weighted problem is polynomial on outerplanar graphs.
\item \emph{$k$-outerplanar graphs and grids.}  Díaz et al.\
  \cite{Diaz17} ask for the complexity on $k$-outerplanar graphs for
  fixed $k\ge2$, the case $k=1$ being the outerplanar graphs, and on
  irregular grids with or without holes.  The main theorem would answer the first question
  if $G_\varphi$ were $k$-outerplanar for a fixed $k$; we have not
  determined whether it is.
\end{itemize}

\subsection{Organisation}

The paper is written as a sequence of self-contained sections: every
object is a displayed formula, every statement carries all its
hypotheses, and every proof is a numbered argument citing definitions
and earlier statements.  Section~\ref{sec:N01} gives the definitions and
the distance tools.  Section~\ref{sec:N02} defines the template and
computes its profiles.  Section~\ref{sec:N03} defines restricted
formulas, the graph $G_\varphi$ and the patterns.
Section~\ref{sec:N04} computes the distances of $G_\varphi$ that the
proofs use.  Section~\ref{sec:N05} proves the lower bound $2m+5$ and
that every resolving set of that size is a pattern up to an
automorphism.  Section~\ref{sec:N06} determines the pairs that a pattern
leaves unseparated.  Section~\ref{sec:N07} proves that $G_\varphi$ has
treewidth at most two, Section~\ref{sec:N10} that it has pathwidth at
most five, Section~\ref{sec:N09} the four properties
of $G_\varphi$ used above.  Section~\ref{sec:N08} proves the value theorem and the main theorem.
Section~\ref{sec:ack} describes how the research and the paper were
produced with a large language model, and the formalisation of the
main results in Lean~4.

\tableofcontents


\section{Definitions and distance tools}\label{sec:N01}

\begin{definition}[graphs, distance]\label{def:N01-graph}
A \emph{graph} $G=(V,E)$ is finite and simple; $N_G(x)$ is the set of
neighbours of $x$.  A \emph{path} in $G$ is a sequence $(u_0,\dots,u_\ell)$
of distinct vertices with $u_iu_{i+1}\in E$ for $0\le i<\ell$; its
\emph{length} is $\ell$, and its \emph{graph} is the subgraph of $G$ with
the vertices $u_0,\dots,u_\ell$ and the edges $u_iu_{i+1}$
($0\le i<\ell$).  For $x,y\in V$, $d_G(x,y)$ is the least length
(number of edges) of an $x$--$y$ path in $G$, and $d_G(x,y):=\infty$ when
there is none; $G$ is \emph{connected} when $d_G$ is finite on
$V\times V$.  A \emph{tree} is a connected graph in which every two
vertices are joined by exactly one path.  For $S\subseteq V$, $G-S$ is the subgraph induced on
$V\setminus S$, and $G-v:=G-\{v\}$ for $v\in V$; $G$ is
\emph{two-connected} when it has at least three vertices and $G-v$ is
connected for every $v\in V$.  An \emph{isomorphism} from $G$ onto a graph $G'$ is a
bijection $\psi$ from $V$ onto $V(G')$ such that $xy\in E$ if and only if
$\psi(x)\psi(y)\in E(G')$; an \emph{automorphism} of $G$ is an isomorphism
from $G$ onto $G$.  When the graph is fixed by the enclosing statement we
write $d$ for $d_G$ and $N$ for $N_G$.  The triangle inequality
\[
 d(w,x)\le d(w,y)+d(y,x)\qquad(w,x,y\in V)
\]
holds, and in particular $|d(w,x)-d(w,y)|\le d(x,y)$.
\end{definition}

\begin{definition}[separation, resolving sets, metric dimension]\label{def:N01-resolve}
Let $G$ be a connected graph.  A vertex $w$ \emph{separates} two
vertices $x,y$ when $d(w,x)\ne d(w,y)$.  A set $W\subseteq V(G)$
\emph{separates} $x,y$ when some $w\in W$ does, and \emph{resolves} $G$
when it separates every two distinct vertices of $G$; its elements are
its \emph{landmarks}.  The \emph{metric dimension} of $G$ is
\[
 \md(G):=\min\{|W|:W\subseteq V(G)\text{ resolves }G\}.
\]
A landmark $w$ separates $w$ from every vertex $y\ne w$, since
$d(w,w)=0<d(w,y)$.
\end{definition}

The next lemma proves that a function $f$, written down in advance, is
the distance from a vertex $o$.  Its condition (b) gives the lower bound,
the harder half: along any path $f$ grows by at most one per edge, so no
path from $o$ to $x$ has fewer than $f(x)$ edges.  Its condition (c)
gives a path with exactly $f(x)$ edges.  Lemma~\ref{lem:N04-potentials}
uses it to prove the distances from $T_1$ and $T_2$ to every vertex of
$G_\varphi$, one edge at a time.

\begin{lemma}[distance certificates]\label{lem:N01-potential}
Let $G$ be a graph, $o\in V(G)$ and $f:V(G)\to\mathbb Z$ such that
\begin{enumerate}
\item[\rm(a)] $f(o)=0$;
\item[\rm(b)] $|f(x)-f(y)|\le1$ for every edge $xy$ of $G$;
\item[\rm(c)] every vertex $x\ne o$ has a neighbour $y$ with
  $f(y)=f(x)-1$.
\end{enumerate}
Then $G$ is connected and $f(x)=d(o,x)$ for every $x\in V(G)$.
\end{lemma}

\begin{proof}
\emph{Lower bound.}  Let $P$ be a path from $o$ to a vertex $x$, and let
$k$ be its number of edges.  By (b), $f$ increases by at most one along
each edge of $P$, so $f(x)\le f(o)+k=k$ by (a).  Hence every $o$--$x$
path has at least $f(x)$ edges.

\emph{Upper bound.}  We first note that $f\ge0$, with equality only at
$o$.  Indeed, let $z$ be a vertex at which $f$ attains its minimum.  If
$z\ne o$, then by (c) $z$ has a neighbour with a smaller value, a
contradiction.  Hence $o$ is the only vertex at which $f$ is minimal, and
$f(x)>f(o)=0$ for every $x\ne o$.

We now prove, by induction on $n\ge0$, the following claim: every
vertex $x$ with $f(x)=n$ is joined to $o$ by a path
$x=y_0,y_1,\dots,y_n=o$ such that $f(y_i)=n-i$ for $0\le i\le n$.

For $n=0$, a vertex $x$ with $f(x)=0$ is $o$ by the previous paragraph,
and the path consisting of $o$ alone satisfies the claim.

Let $n\ge1$ and assume the claim for $n-1$.  Let $x$ be a vertex with
$f(x)=n$.  As $f(x)\ne0$, $x\ne o$, so by (c) $x$ has a neighbour $y$
with $f(y)=n-1$.  By the induction hypothesis applied to $y$, there is a
path from $y$ to $o$ along which $f$ decreases by one at each step;
numbering its vertices from $1$, it is $y=y_1,y_2,\dots,y_n=o$ with
$f(y_i)=n-i$ for $1\le i\le n$.  Put $y_0:=x$.  Since $y_0$ is adjacent
to $y_1=y$, the sequence $y_0,y_1,\dots,y_n$ is a walk from $x$ to $o$
with $n$ edges, and since $f(y_0)=n$, we have $f(y_i)=n-i$ for
$0\le i\le n$.  The values $n,n-1,\dots,0$ are pairwise distinct, so the
vertices $y_0,\dots,y_n$ are pairwise distinct, and the walk is a path.
Thus $x$ satisfies the claim for $n$, which completes the induction
step.  By induction, the claim holds for every $n\ge0$.

\emph{Conclusion.}  Let $x$ be any vertex.  As $f(x)\ge0$, the claim
applies with $n=f(x)$: $x$ is joined to $o$ by a path with exactly $f(x)$
edges.  Hence $G$ is connected and $d(o,x)\le f(x)$, and by the lower
bound $d(o,x)\ge f(x)$.
\end{proof}

\begin{definition}[attached subgraphs]\label{def:N01-attach}
Let $G$ be a graph, $X$ a subgraph of $G$, and $s\ne t$ two vertices of
$X$.  The \emph{interior} of $X$ is
\[
 \Int(X):=V(X)\setminus\{s,t\},
\]
and $X$ is \emph{attached at} $\{s,t\}$ when every edge of $G$ with an
end in $\Int(X)$ is an edge of $X$.
\end{definition}

\begin{lemma}[first exit]\label{lem:N01-exit}
Let $G$ be a connected graph and $X$ a subgraph attached at $\{s,t\}$
(Definition~\ref{def:N01-attach}).  For every $u\in\Int(X)$ and every
$w\in V(G)\setminus\Int(X)$,
\[
 d(u,w)=\min\bigl(d(u,s)+d(s,w),\;d(u,t)+d(t,w)\bigr).
\]
\end{lemma}

\begin{proof}
\emph{Upper bound.}  For $z\in\{s,t\}$ the triangle inequality
(Definition~\ref{def:N01-graph}) gives $d(u,w)\le d(u,z)+d(z,w)$.  Taking
$z=s$ and $z=t$,
\[
 d(u,w)\le\min\bigl(d(u,s)+d(s,w),\ d(u,t)+d(t,w)\bigr).
\]

\emph{Lower bound.}  Let $u=x_0,x_1,\dots,x_\ell=w$ be a shortest $u$--$w$
path, so $\ell=d(u,w)$.  As $x_0\in\Int(X)$ and $x_\ell\notin\Int(X)$,
there is a least index $i\ge1$ with $x_i\notin\Int(X)$.  Then
$x_{i-1}\in\Int(X)$, so the edge $x_{i-1}x_i$ has an end in $\Int(X)$ and
is an edge of $X$ (Definition~\ref{def:N01-attach}); hence
$x_i\in V(X)\setminus\Int(X)=\{s,t\}$.  Put $z:=x_i$.  The path splits at
$z$ into $x_0,\dots,x_i$, from $u$ to $z$ with $i$ edges, and
$x_i,\dots,x_\ell$, from $z$ to $w$ with $\ell-i$ edges, so
\[
 d(u,z)\le i,\qquad d(z,w)\le\ell-i,\qquad\text{hence}\qquad
 d(u,z)+d(z,w)\le\ell=d(u,w).
\]
As $z\in\{s,t\}$,
\[
 \min\bigl(d(u,s)+d(s,w),\ d(u,t)+d(t,w)\bigr)\ \le\ d(u,z)+d(z,w)\ \le\ d(u,w).
\]
With the upper bound, this ends the proof.
\end{proof}

\begin{lemma}[edges between vertices of an attached subgraph]\label{lem:N01-induced}
Let $G$ be a graph and $X$ a subgraph attached at $\{s,t\}$
(Definition~\ref{def:N01-attach}) such that
\[
 d_{G-\Int(X)}(s,t)\ \ge\ d_X(s,t).
\]
Then every edge of $G$ with both ends in $V(X)$ is an edge of $X$.
\end{lemma}

\begin{proof}
Let $e$ be an edge of $G$ with both ends in $V(X)$.  If an end of $e$
lies in $\Int(X)$, then $e$ is an edge of $X$ by
Definition~\ref{def:N01-attach}.  Otherwise both ends of $e$ lie in
$V(X)\setminus\Int(X)=\{s,t\}$, so $e=st$.  Then $e$ is an edge of
$G-\Int(X)$, so $d_{G-\Int(X)}(s,t)=1$, and the hypothesis gives
$d_X(s,t)\le1$.  As $s\ne t$, $d_X(s,t)=1$: $s$ and $t$ are adjacent in
$X$, that is, $e$ is an edge of $X$.
\end{proof}

\begin{lemma}[isometric attachment]\label{lem:N01-isometric}
Let $G$ be a connected graph, $X$ a connected subgraph attached at
$\{s,t\}$ (Definition~\ref{def:N01-attach}), and suppose
\[
 d_{G-\Int(X)}(s,t)\ \ge\ d_X(s,t).
\]
Then $d_G(u,u')=d_X(u,u')$ for all $u,u'\in V(X)$.
\end{lemma}

\begin{proof}
Put $\delta:=d_{G-\Int(X)}(s,t)$.  Since $X$ is a subgraph of $G$,
$d_G(u,u')\le d_X(u,u')$ for all $u,u'\in V(X)$; it remains to prove the
reverse inequality.

\emph{The claim.}  We prove, by induction on $\ell\ge0$, the following
claim: for all $u,u'\in V(X)$ with $d_G(u,u')=\ell$, we have
$d_X(u,u')\le\ell$.

\emph{Base case.}  If $\ell=0$, then $u=u'$ and $d_X(u,u')=0$.

\emph{Induction step.}  Let $\ell\ge1$ and assume the claim for every
value smaller than $\ell$.  Let $u,u'\in V(X)$ with $d_G(u,u')=\ell$, and
let $P$ be a shortest $u$--$u'$ path of $G$; it has $\ell$ edges.  There
are two cases.

\emph{Case 1: every vertex of $P$ is in $V(X)$.}  By
Lemma~\ref{lem:N01-induced}, every edge of $P$ is an edge of $X$, so $P$
is a path of $X$ and $d_X(u,u')\le\ell$.

\emph{Case 2: some vertex of $P$ is not in $V(X)$}
.  Going along $P$ from
$u$ to $u'$, let $y$ be the first vertex not in $V(X)$, let $z$ be the
vertex just before $y$, and let $z'$ be the first vertex of $V(X)$ after
$y$; $z$ and $z'$ exist because $u,u'\in V(X)$.  If $z$ were in
$\Int(X)$, the edge $zy$ would be an edge of $X$ by
Definition~\ref{def:N01-attach}, and $y$ would be a vertex of $X$, which
is a contradiction.  So $z\in V(X)\setminus\Int(X)=\{s,t\}$.  Likewise
$z'\in\{s,t\}$, taking the vertex just before $z'$ in place of $y$.  As $P$ is a path, $z\ne z'$, so
$\{z,z'\}=\{s,t\}$.  Split $P$ at $z$ and $z'$ into $P_1$ from $u$ to $z$,
$P_2$ from $z$ to $z'$ and $P_3$ from $z'$ to $u'$, with $\ell_1$, $\ell_2$
and $\ell_3$ edges, so that $\ell_1+\ell_2+\ell_3=\ell$.
\begin{itemize}
\item The vertices of $P_2$ other than $z,z'$ lie outside $V(X)$, so
  $P_2$ is a path of $G-\Int(X)$ joining $s$ and $t$, and
  $\ell_2\ge\delta\ge d_X(s,t)=d_X(z,z')$.
\item $P_2$ joins $z$ to $z'\ne z$, so $\ell_2\ge1$, and
  $\ell_1=\ell-\ell_2-\ell_3<\ell$, $\ell_3<\ell$.  As $P_1$ is a $u$--$z$
  path with $\ell_1$ edges, $d_G(u,z)\le\ell_1<\ell$; likewise, $P_3$
  gives $d_G(z',u')\le\ell_3<\ell$.  By the induction hypothesis, applied to the pairs $u,z$ and $z',u'$,
  $d_X(u,z)\le d_G(u,z)\le\ell_1$ and $d_X(z',u')\le d_G(z',u')\le\ell_3$.
\end{itemize}
By the triangle inequality in $X$,
\[
 d_X(u,u')\ \le\ d_X(u,z)+d_X(z,z')+d_X(z',u')\ \le\ \ell_1+\ell_2+\ell_3=\ell .
\]

In both cases $d_X(u,u')\le\ell$, so the claim holds for $\ell$, which
completes the induction step.  By induction, the claim holds for every
$\ell\ge0$.

\emph{Conclusion.}  For $u,u'\in V(X)$, the claim with
$\ell=d_G(u,u')$ gives $d_X(u,u')\le d_G(u,u')$, and with the first
inequality, $d_G(u,u')=d_X(u,u')$.
\end{proof}

\begin{definition}[profiles]\label{def:N01-profile}
Let $X$ be a graph and $s\ne t$ two of its vertices, its
\emph{terminals}.  The \emph{profile} of a vertex $u$ of $X$ is
$(h_X(u),\hb_X(u))$ and its \emph{skew} is $\kappa_X(u)$, where
\[
 h_X(u):=d_X(s,u),\qquad \hb_X(u):=d_X(t,u),\qquad
 \kappa_X(u):=\hb_X(u)-h_X(u).
\]
When the graph is fixed by the enclosing statement we write $h$, $\hb$,
$\kappa$ for $h_X$, $\hb_X$, $\kappa_X$.
\end{definition}

\begin{definition}[ears]\label{def:N01-ear}
An \emph{ear} of a graph $H$ (Figure~\ref{fig:N01-ear}) is a path
$E=(e_0,e_1,\dots,e_\ell)$ of $H$
with $\ell\ge2$ whose vertices $e_1,\dots,e_{\ell-1}$ have degree two in
$H$.  These vertices form its \emph{interior}, $e_0$ and $e_\ell$ are its
\emph{ends}, and $\ell$ is its \emph{length}.  So the only edges of $H$
at an interior vertex $e_i$ are the two edges $e_{i-1}e_i$ and
$e_ie_{i+1}$ of $E$.
\end{definition}

\begin{figure}[ht]
\centering
\begin{tikzpicture}[font=\small,scale=1.05]
\draw[fill=gray!12] (2.5,-1.0) ellipse (3.4 and 0.75);
\node at (2.5,-1.1) {the rest of $H$};
\foreach \k in {1,...,4} {
  \coordinate (e\k) at ({2.5+2.5*cos(180-36*\k)},{1.6*sin(180-36*\k)});
}
\coordinate (e0) at (0,0); \coordinate (e5) at (5,0);
\draw (e0) -- (-0.35,-0.55); \draw (e0) -- (0.4,-0.6);
\draw (e5) -- (5.35,-0.55); \draw (e5) -- (4.6,-0.6);
\draw[very thick] (-1.1,-0.8) -- (e0) -- (e1) -- (e2) -- (e3) -- (e4) -- (e5) -- (6.1,-0.8);
\node[vx,label={[xshift=-4pt]left:{$e_0=x_{a-1}$}}] at (e0) {};
\node[vx,label={[xshift=4pt]right:{$x_{b+1}=e_5$}}] at (e5) {};
\node[vx,label=above left:{$e_1=x_a$}] at (e1) {};
\node[vx,label=above:{$e_2$}] at (e2) {};
\node[vx,label=above:{$e_3$}] at (e3) {};
\node[vx,label=above right:{$x_b=e_4$}] at (e4) {};
\end{tikzpicture}
\caption{An ear $E=(e_0,\dots,e_5)$ of a graph $H$
(Definition~\ref{def:N01-ear}): the interior vertices $e_1,\dots,e_4$ have
no edge other than the two edges of $E$ at them, while the ends $e_0,e_5$
may have other edges.  A path $P$ (thick) with both ends outside the
interior of $E$ crosses $E$ from one end to the other
(Lemma~\ref{lem:N01-earcross}): here its maximal run in the interior of
$E$ (Definition~\ref{def:N01-run}) is $x_a,\dots,x_b=e_1,\dots,e_4$, with
$x_{a-1}=e_0$ and $x_{b+1}=e_5$.}
\label{fig:N01-ear}
\end{figure}
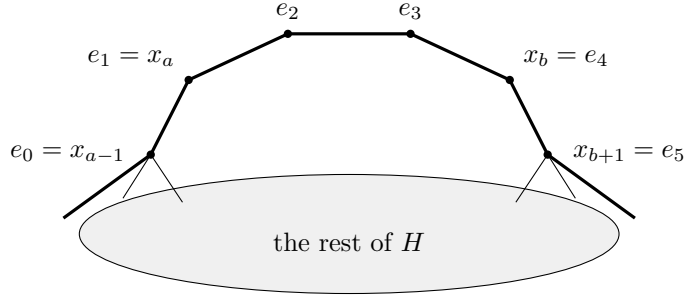

\begin{lemma}[moving along an ear]\label{lem:N01-earmove}
Let $E=(e_0,\dots,e_\ell)$ be an ear of a graph $H$
(Definition~\ref{def:N01-ear}), and let $x_p,x_{p+1},\dots,x_q$, with
$p<q$, be consecutive vertices of a path of $H$ such that $x_h$ is in
the interior of $E$ for every $p<h<q$; if $q=p+1$, assume moreover that
$x_p$ or $x_q$ is in the interior of $E$.  Then every $x_h$
($p\le h\le q$) is a vertex of $E$.  Write $c_h$ for its index, that is,
$x_h=e_{c_h}$.  Then either $c_{h+1}=c_h+1$ for all $p\le h<q$, or
$c_{h+1}=c_h-1$ for all $p\le h<q$.
\end{lemma}

\begin{proof}
By Definition~\ref{def:N01-ear}, the neighbours of an interior vertex
$e_g$ ($0<g<\ell$) are exactly $e_{g-1}$ and $e_{g+1}$.

\emph{Each step changes the index by one.}  Let $p\le h<q$.  At least
one of $x_h,x_{h+1}$ is in the interior of $E$: if $h+1<q$, this is
$x_{h+1}$; if $h>p$, this is $x_h$; and otherwise $h=p$ and $h+1=q$, and
the hypothesis for $q=p+1$ applies.  So the step from $x_h$ to $x_{h+1}$
joins an interior vertex $e_g$ to one of its neighbours $e_{g-1}$,
$e_{g+1}$.  Hence $x_h$ and $x_{h+1}$ are vertices of $E$, and their
indices differ by one: $c_{h+1}=c_h+1$ or $c_{h+1}=c_h-1$.

\emph{The direction does not change.}  Let $p<h<q$, and put $g:=c_h$.
Then $x_h=e_g$ is in the interior of $E$, and $x_{h-1}$, $x_{h+1}$ are two
of its neighbours, distinct because a path does not repeat a vertex.  So
$\{x_{h-1},x_{h+1}\}=\{e_{g-1},e_{g+1}\}$, and $(c_{h-1},c_h,c_{h+1})$ is
$(g-1,g,g+1)$ or $(g+1,g,g-1)$: the step before $x_h$ and the step after
it change the index in the same direction.  As this holds at every $x_h$
with $p<h<q$, all the steps of the stretch change the index in the same
direction.
\end{proof}

\begin{lemma}[distances from an ear]\label{lem:N01-eardist}
Let $H$ be a connected graph and $E=(e_0,\dots,e_\ell)$ an ear of $H$
(Definition~\ref{def:N01-ear}).  For $0<i<\ell$ and every vertex $y$ of
$H$ outside the interior of $E$,
\[
 d_H(e_i,y)=\min\bigl(i+d_H(e_0,y),\;(\ell-i)+d_H(e_\ell,y)\bigr).
\]
\end{lemma}

\begin{proof}
\emph{Upper bound.}  The path $e_i,e_{i-1},\dots,e_0$ has $i$ edges and
the path $e_i,e_{i+1},\dots,e_\ell$ has $\ell-i$ edges, so
\[
 d_H(e_i,y)\le i+d_H(e_0,y),\qquad d_H(e_i,y)\le(\ell-i)+d_H(e_\ell,y).
\]

\emph{Lower bound.}  Let $e_i=x_0,x_1,\dots,x_k=y$ be a shortest path,
so $k=d_H(e_i,y)$, and let $j\ge1$ be the least index with $x_j$ outside
the interior of $E$; it exists because $y$ is outside.  The vertices
$x_0,\dots,x_{j-1}$ are interior; in particular $x_0=e_i$ is, which the
hypothesis of Lemma~\ref{lem:N01-earmove} asks for when $j=1$.  By that
lemma, applied to $x_0,\dots,x_j$, these vertices lie on $E$, and since
$x_0=e_i$, either $x_h=e_{i-h}$ for all $0\le h\le j$, or $x_h=e_{i+h}$
for all $0\le h\le j$.  Moreover $x_j$ lies on $E$ but not in its
interior, so $x_j=e_0$ or $x_j=e_\ell$.
\begin{itemize}
\item If $x_h=e_{i-h}$, then $x_j=e_{i-j}$ and $i-j<\ell$, so $x_j=e_0$
  and $j=i$.
\item If $x_h=e_{i+h}$, then $x_j=e_{i+j}$ and $i+j>0$, so $x_j=e_\ell$
  and $j=\ell-i$.
\end{itemize}
The rest of the path, $x_j,\dots,x_k$, joins $x_j$ to $y$ with $k-j$
edges, so $d_H(x_j,y)\le k-j$, that is,
\[
 d_H(e_i,y)=k\ \ge\ j+d_H(x_j,y).
\]
By the two cases above, $j+d_H(x_j,y)$ equals $i+d_H(e_0,y)$ if
$x_j=e_0$, and $(\ell-i)+d_H(e_\ell,y)$ if $x_j=e_\ell$; in both cases it
is at least the smaller of these two numbers.  Hence
\[
 d_H(e_i,y)\ \ge\ \min\bigl(i+d_H(e_0,y),\ (\ell-i)+d_H(e_\ell,y)\bigr).
\]
With the upper bound, this ends the proof.
\end{proof}

\begin{definition}[maximal runs]\label{def:N01-run}
Let $E$ be an ear of a graph $H$ (Definition~\ref{def:N01-ear}) and
$P=(x_0,\dots,x_k)$ a path of $H$.  A \emph{maximal run of $P$ in the
interior of $E$} is a sequence $x_a,x_{a+1},\dots,x_b$ of consecutive
vertices of $P$, with $0\le a\le b\le k$, all in the interior of $E$, such
that $x_{a-1}$ is not in the interior of $E$ if $a>0$, and $x_{b+1}$ is
not in the interior of $E$ if $b<k$ (Figure~\ref{fig:N01-ear}).
\end{definition}

\begin{lemma}[paths crossing an ear]\label{lem:N01-earcross}
Let $E=(e_0,\dots,e_\ell)$ be an ear of a graph $H$
(Definition~\ref{def:N01-ear}), and $P=(x_0,\dots,x_k)$ a path of $H$
with $x_0$ and $x_k$ outside the interior of $E$.  For every maximal run
$x_a,\dots,x_b$ of $P$ in the interior of $E$
(Definition~\ref{def:N01-run}), $0<a\le b<k$, and $x_{a-1},\dots,x_{b+1}$
is $E$ or $E$ reversed.
\end{lemma}

\begin{proof}
Since $x_0$ and $x_k$ are not in the interior of $E$ but $x_a$ and $x_b$
are, we have $0<a\le b<k$.  Thus $x_{a-1}$ and $x_{b+1}$ exist, and by
Definition~\ref{def:N01-run} they are not in the interior of $E$.

The vertices strictly between $x_{a-1}$ and $x_{b+1}$ are
$x_a,\dots,x_b$, which are in the interior of $E$.  So
Lemma~\ref{lem:N01-earmove} applies to $x_{a-1},\dots,x_{b+1}$, and these
vertices lie on $E$.  Write $x_{a-1}=e_g$.  Then either
$x_{a-1+h}=e_{g+h}$ for all $0\le h\le b-a+2$, or $x_{a-1+h}=e_{g-h}$ for
all $0\le h\le b-a+2$.  The vertices $x_{a-1}$ and $x_{b+1}$ lie on $E$
but not in its interior, so each of them is $e_0$ or $e_\ell$.  They are
distinct, since $P$ is a path.

Suppose first that $x_{a-1}=e_0$, so $g=0$.  The second form would give
$x_a=e_{-1}$, which is impossible.  Hence $x_{a-1+h}=e_h$ for every $h$.
Taking $h=b-a+2$ gives $x_{b+1}=e_{b-a+2}$.  Since $x_{b+1}\ne e_0$, we
have $x_{b+1}=e_\ell$, so $b-a+2=\ell$.  Thus the segment
$x_{a-1},\dots,x_{b+1}$ is $e_0,e_1,\dots,e_\ell$, which is $E$.

Suppose now that $x_{a-1}=e_\ell$, so $g=\ell$.  The first form would give
$x_a=e_{\ell+1}$, which is impossible.  Hence $x_{a-1+h}=e_{\ell-h}$ for
every $h$.  Taking $h=b-a+2$ gives $x_{b+1}=e_{\ell-(b-a+2)}$.  Since
$x_{b+1}\ne e_\ell$, we have $x_{b+1}=e_0$, so $b-a+2=\ell$.  Thus the
segment is $e_\ell,e_{\ell-1},\dots,e_0$, which is $E$ reversed.
\end{proof}

\begin{definition}[twins]\label{def:N01-twins}
Two distinct vertices $x,y$ of a graph are \emph{twins} when
$N(x)\setminus\{y\}=N(y)\setminus\{x\}$.
\end{definition}

\begin{lemma}[twins and automorphisms]\label{lem:N01-twins}
Let $G$ be a connected graph and $x,y$ twins
(Definition~\ref{def:N01-twins}).
\begin{enumerate}
\item[\rm(i)] The map $\psi_{xy}$ exchanging $x$ and $y$ and fixing every
  other vertex is an automorphism of $G$.
\item[\rm(ii)] $d(w,x)=d(w,y)$ for every $w\in V(G)\setminus\{x,y\}$.
\item[\rm(iii)] Every resolving set of $G$ contains $x$ or $y$.
\item[\rm(iv)] If $W$ resolves $G$ and $\psi$ is an automorphism of $G$,
  then $\psi(W)$ resolves $G$ and $|\psi(W)|=|W|$.
\end{enumerate}
\end{lemma}

\begin{proof}
\emph{Distances.}  An automorphism $\varphi$ of $G$ maps every path to a
path with the same number of edges, and so does $\varphi^{-1}$.  Hence
$d(\varphi(w),\varphi(z))=d(w,z)$ for all vertices $w,z$.

(i) $\psi_{xy}$ is a bijection of $V(G)$, and applying it twice gives
the identity.  We first check that it maps every edge $uv$ to an edge.
If neither $u$ nor $v$ is in $\{x,y\}$, the edge is fixed.  If $u=x$ and
$v\ne y$, then $v\in N(x)\setminus\{y\}=N(y)\setminus\{x\}$, so the image
$yv$ is an edge.  If $u=y$ and $v\ne x$, the same argument with $x$ and
$y$ exchanged shows that the image $xv$ is an edge.  If $\{u,v\}=\{x,y\}$,
the edge is mapped to itself.  Conversely, if the image of $uv$ is an
edge, then applying $\psi_{xy}$ to it gives back $uv$, which is therefore
an edge.  So $\psi_{xy}$ is an automorphism.

(ii) Let $w\notin\{x,y\}$.  Then $\psi_{xy}(w)=w$ and $\psi_{xy}(x)=y$.
By the first paragraph,
$d(w,x)=d(\psi_{xy}(w),\psi_{xy}(x))=d(w,y)$.

(iii) Let $W$ be a set of vertices with $W\cap\{x,y\}=\emptyset$.  By
(ii), every $w\in W$ has $d(w,x)=d(w,y)$, so $W$ does not separate $x$
and $y$, and $W$ does not resolve $G$ (Definition~\ref{def:N01-resolve}).

(iv) Write $\psi(W):=\{\psi(w):w\in W\}$.  As $\psi$ is injective,
$|\psi(W)|=|W|$.  Let $u\ne v$.  Then $\psi^{-1}(u)\ne\psi^{-1}(v)$, so
some $w\in W$ separates them.  By the first paragraph,
$d(\psi(w),u)=d(w,\psi^{-1}(u))\ne d(w,\psi^{-1}(v))=d(\psi(w),v)$.  So
$\psi(w)\in\psi(W)$ separates $u$ and $v$, and $\psi(W)$ resolves $G$.
\end{proof}


\section{The template}\label{sec:N02}

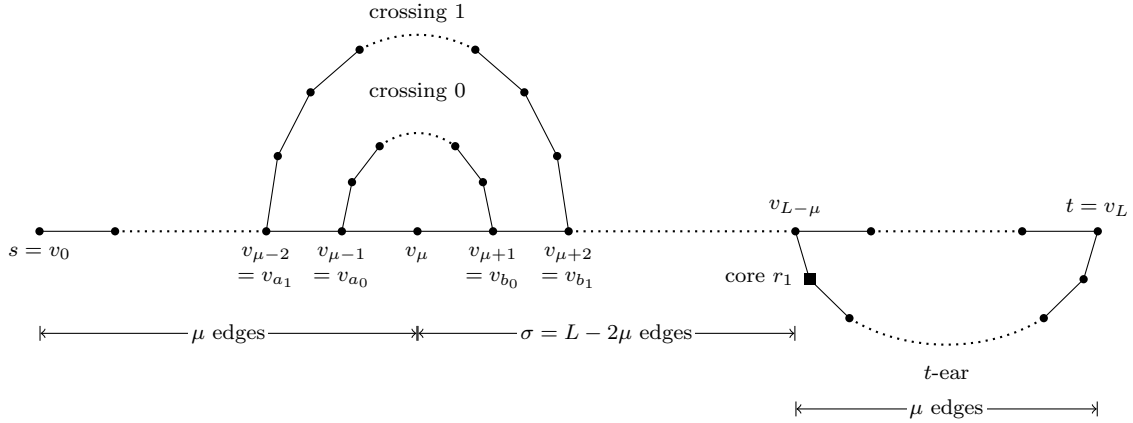
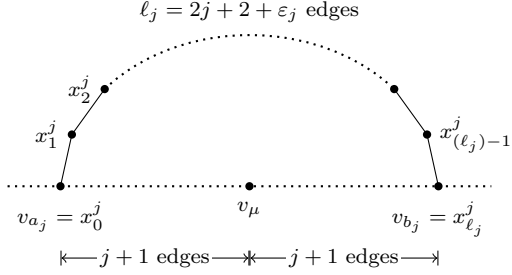
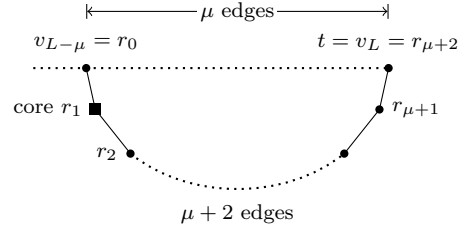
\begin{figure}[ht]
\centering
\begin{subfigure}{\textwidth}
\centering
\begin{tikzpicture}[font=\scriptsize,x=1cm,y=1cm]
\foreach \p/\i in {0/0,1/1,3/a1,4/a0,5/mu,6/b0,7/b1,10/Lm,11/Lm1,13/L1,14/L}
  \coordinate (v\i) at (\p,0);
\draw (v0) -- (v1); \draw[dotted,thick] (v1) -- (va1);
\draw (va1) -- (vb1); \draw[dotted,thick] (vb1) -- (vLm);
\draw (vLm) -- (vLm1); \draw[dotted,thick] (vLm1) -- (vL1); \draw (vL1) -- (vL);
\foreach \k/\a in {1/157.5,2/135,3/112.5,4/67.5,5/45,6/22.5}
  \coordinate (c\k) at ($(5,0)+(\a:2 and 2.6)$);
\draw (va1) -- (c1) -- (c2) -- (c3); \draw (c4) -- (c5) -- (c6) -- (vb1);
\draw[dotted,thick] (c3) arc[start angle=112.5,end angle=67.5,x radius=2,y radius=2.6];
\foreach \k/\a in {1/150,2/120,3/60,4/30}
  \coordinate (d\k) at ($(5,0)+(\a:1 and 1.3)$);
\draw (va0) -- (d1) -- (d2); \draw (d3) -- (d4) -- (vb0);
\draw[dotted,thick] (d2) arc[start angle=120,end angle=60,x radius=1,y radius=1.3];
\foreach \k/\a in {1/205,2/230,3/310,4/335}
  \coordinate (r\k) at ($(12,0)+(\a:2 and 1.5)$);
\draw (vLm) -- (r1) -- (r2); \draw (r3) -- (r4) -- (vL);
\draw[dotted,thick] (r2) arc[start angle=230,end angle=310,x radius=2,y radius=1.5];
\foreach \i in {0,1,a1,a0,mu,b0,b1,Lm,Lm1,L1,L} \node[vx] at (v\i) {};
\foreach \k in {1,...,6} \node[vx] at (c\k) {};
\foreach \k in {1,...,4} \node[vx] at (d\k) {};
\foreach \k in {2,3,4} \node[vx] at (r\k) {};
\node[lm] at (r1) {};
\node[below=2pt] at (v0) {$s=v_0$};
\node[below=2pt,align=center] at (va1) {$v_{\mu-2}$\\$=v_{a_1}$};
\node[below=2pt,align=center] at (va0) {$v_{\mu-1}$\\$=v_{a_0}$};
\node[below=2pt] at (vmu) {$v_\mu$};
\node[below=2pt,align=center] at (vb0) {$v_{\mu+1}$\\$=v_{b_0}$};
\node[below=2pt,align=center] at (vb1) {$v_{\mu+2}$\\$=v_{b_1}$};
\node[above=2pt] at (vLm) {$v_{L-\mu}$};
\node[above=2pt] at (vL) {$t=v_L$};
\node at (5,1.85) {crossing $0$};
\node[above=1pt] at (5,2.6) {crossing $1$};
\node[left=2pt] at (r1) {core $r_1$};
\node[below=4pt] at (12,-1.5) {$t$-ear};
\draw[|<->|] (0,-1.35) -- (5,-1.35) node[midway,fill=white,inner sep=1pt] {$\mu$ edges};
\draw[|<->|] (5,-1.35) -- (10,-1.35) node[midway,fill=white,inner sep=1pt] {$\sigma=L-2\mu$ edges};
\draw[|<->|] (10,-2.35) -- (14,-2.35) node[midway,fill=white,inner sep=1pt] {$\mu$ edges};
\end{tikzpicture}
\caption{The whole template.}
\end{subfigure}

\bigskip
\begin{subfigure}{0.5\textwidth}
\centering
\begin{tikzpicture}[font=\scriptsize,x=1cm,y=1cm]
\coordinate (a) at (0,0); \coordinate (m) at (2.5,0); \coordinate (b) at (5,0);
\draw[dotted,thick] (-0.7,0) -- (a); \draw[dotted,thick] (b) -- (5.7,0);
\draw[dotted,thick] (a) -- (b);
\foreach \k/\an in {1/160,2/140,3/40,4/20}
  \coordinate (x\k) at ($(2.5,0)+(\an:2.5 and 2)$);
\draw (a) -- (x1) -- (x2); \draw (x3) -- (x4) -- (b);
\draw[dotted,thick] (x2) arc[start angle=140,end angle=40,x radius=2.5,y radius=2];
\foreach \p in {a,m,b} \node[vx] at (\p) {};
\foreach \k in {1,...,4} \node[vx] at (x\k) {};
\node[below=2pt] at (a) {$v_{a_j}=x^j_0$};
\node[below=2pt] at (m) {$v_\mu$};
\node[below=2pt] at (b) {$v_{b_j}=x^j_{\ell_j}$};
\node[left=1pt] at (x1) {$x^j_1$};
\node[left=1pt] at (x2) {$x^j_2$};
\node[right=1pt] at (x4) {$x^j_{(\ell_j)-1}$};
\node[above=1pt] at (2.5,2) {$\ell_j=2j+2+\varepsilon_j$ edges};
\draw[|<->|] (0,-0.95) -- (2.5,-0.95) node[midway,fill=white,inner sep=1pt] {$j+1$ edges};
\draw[|<->|] (2.5,-0.95) -- (5,-0.95) node[midway,fill=white,inner sep=1pt] {$j+1$ edges};
\end{tikzpicture}
\caption{Crossing $j$.}
\end{subfigure}%
\begin{subfigure}{0.5\textwidth}
\centering
\begin{tikzpicture}[font=\scriptsize,x=1cm,y=1cm]
\coordinate (p) at (0,0); \coordinate (q) at (4,0);
\draw[dotted,thick] (-0.7,0) -- (p); \draw[dotted,thick] (p) -- (q);
\foreach \k/\an in {1/200,2/225,3/315,4/340}
  \coordinate (r\k) at ($(2,0)+(\an:2 and 1.6)$);
\draw (p) -- (r1) -- (r2); \draw (r3) -- (r4) -- (q);
\draw[dotted,thick] (r2) arc[start angle=225,end angle=315,x radius=2,y radius=1.6];
\foreach \p in {p,q} \node[vx] at (\p) {};
\foreach \k in {2,3,4} \node[vx] at (r\k) {};
\node[lm] at (r1) {};
\node[above=2pt] at (p) {$v_{L-\mu}=r_0$};
\node[above=2pt] at (q) {$t=v_L=r_{\mu+2}$};
\node[left=1pt] at (r1) {core $r_1$};
\node[left=1pt] at (r2) {$r_2$};
\node[right=1pt] at (r4) {$r_{\mu+1}$};
\node[below=2pt] at (2,-1.6) {$\mu+2$ edges};
\draw[|<->|] (0,0.75) -- (4,0.75) node[midway,fill=white,inner sep=1pt] {$\mu$ edges};
\end{tikzpicture}
\caption{The $t$-ear.}
\end{subfigure}
\caption{The general shape of a template $\T(L,\mu;E)$
(Definition~\ref{def:N02-template}), drawn with two crossings ($e=2$);
when $e=1$ only crossing $0$ is present, and when $e=0$ there is none.
Dotted lines are paths whose inner vertices are not drawn.  (a)~The
spine $v_0,\dots,v_L$ runs from $s$ to $t$.  Crossing $j$ joins
$v_{a_j}=v_{\mu-1-j}$ to $v_{b_j}=v_{\mu+1+j}$, around $v_\mu$; the
$t$-ear joins $v_{L-\mu}$ to $v_L$.  Each of the three paths meets the
rest of the template only at its two ends, which are spine vertices, and
by Lemma~\ref{lem:N02-ends} the crossings end before $v_{L-\mu}$.
(b)~Crossing $j$ has $\ell_j=2j+2+\varepsilon_j$ edges, against $2j+2$
spine edges between its ends.  (c)~The $t$-ear has $\mu+2$ edges, against
$\mu$ spine edges between its ends; its first member $r_1$ is the core.}
\label{fig:N02-shape}
\end{figure}

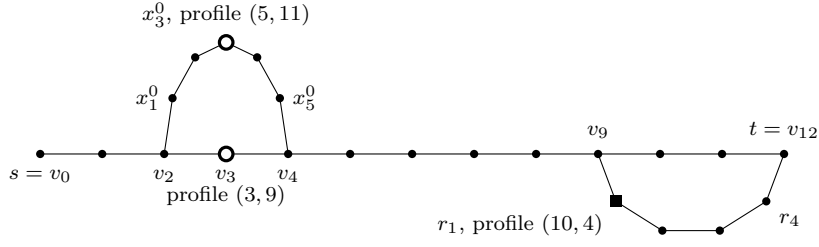
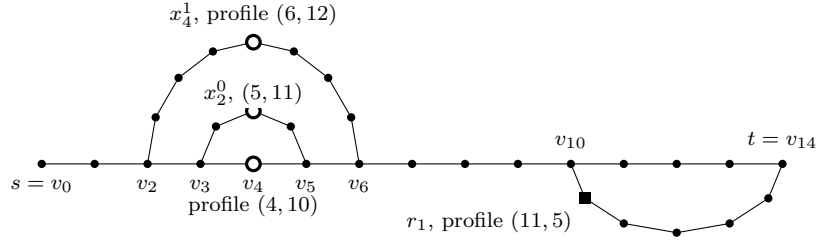
\begin{figure}[ht]
\centering
\begin{subfigure}{\textwidth}
\centering
\begin{tikzpicture}[font=\scriptsize,x=0.82cm,y=0.82cm]
\foreach \i in {0,...,12} \coordinate (v\i) at (\i,0);
\draw (v0) -- (v12);
\foreach \k in {1,...,5} \coordinate (x\k) at ({3+cos(180-30*\k)},{1.8*sin(180-30*\k)});
\draw (v2) -- (x1) -- (x2) -- (x3) -- (x4) -- (x5) -- (v4);
\foreach \k in {1,...,4} \coordinate (r\k) at ({10.5+1.5*cos(180+36*\k)},{1.3*sin(180+36*\k)});
\draw (v9) -- (r1) -- (r2) -- (r3) -- (r4) -- (v12);
\foreach \i in {0,...,12} \node[vx] at (v\i) {};
\foreach \k in {1,2,4,5} \node[vx] at (x\k) {};
\foreach \k in {2,3,4} \node[vx] at (r\k) {};
\node[dg] at (v3) {}; \node[dg] at (x3) {}; \node[lm] at (r1) {};
\node[below=2pt] at (v0) {$s=v_0$};
\node[below=2pt] at (v2) {$v_2$};
\node[below=2pt] at (v3) {$v_3$};
\node[below=2pt] at (v4) {$v_4$};
\node[above=2pt] at (v9) {$v_9$};
\node[above=2pt] at (v12) {$t=v_{12}$};
\node[left=1pt] at (x1) {$x^0_1$};
\node[right=1pt] at (x5) {$x^0_5$};
\node[above=2pt] at (x3) {$x^0_3$, profile $(5,11)$};
\node[below=8pt] at (v3) {profile $(3,9)$};
\node[below left=1pt] at (r1) {$r_1$, profile $(10,4)$};
\node[below right=1pt] at (r4) {$r_4$};
\end{tikzpicture}
\caption{$\T(12,3;(4))$: $\sigma=6$; one crossing, from $v_2$ to $v_4$, of
length $6$; the $t$-ear from $v_9$ to $v_{12}$, of length $5$.}
\end{subfigure}

\medskip
\begin{subfigure}{\textwidth}
\centering
\begin{tikzpicture}[font=\scriptsize,x=0.7cm,y=0.7cm]
\foreach \i in {0,...,14} \coordinate (v\i) at (\i,0);
\draw (v0) -- (v14);
\foreach \k in {1,...,3} \coordinate (a\k) at ({4+cos(180-45*\k)},{1.0*sin(180-45*\k)});
\draw (v3) -- (a1) -- (a2) -- (a3) -- (v5);
\foreach \k in {1,...,7} \coordinate (b\k) at ({4+2*cos(180-22.5*\k)},{2.3*sin(180-22.5*\k)});
\draw (v2) -- (b1) -- (b2) -- (b3) -- (b4) -- (b5) -- (b6) -- (b7) -- (v6);
\foreach \k in {1,...,5} \coordinate (r\k) at ({12+2*cos(180+30*\k)},{1.3*sin(180+30*\k)});
\draw (v10) -- (r1) -- (r2) -- (r3) -- (r4) -- (r5) -- (v14);
\foreach \i in {0,...,14} \node[vx] at (v\i) {};
\foreach \k in {1,3} \node[vx] at (a\k) {};
\foreach \k in {1,2,3,5,6,7} \node[vx] at (b\k) {};
\foreach \k in {2,...,5} \node[vx] at (r\k) {};
\node[dg] at (v4) {}; \node[dg] at (a2) {}; \node[dg] at (b4) {}; \node[lm] at (r1) {};
\node[below=2pt] at (v0) {$s=v_0$};
\node[below=2pt] at (v2) {$v_2$};
\node[below=2pt] at (v3) {$v_3$};
\node[below=2pt] at (v4) {$v_4$};
\node[below=2pt] at (v5) {$v_5$};
\node[below=2pt] at (v6) {$v_6$};
\node[above=2pt] at (v10) {$v_{10}$};
\node[above=2pt] at (v14) {$t=v_{14}$};
\node[below=8pt] at (v4) {profile $(4,10)$};
\node[above=1pt,fill=white,inner sep=1pt] at (a2) {$x^0_2$, $(5,11)$};
\node[above=2pt] at (b4) {$x^1_4$, profile $(6,12)$};
\node[below left=1pt] at (r1) {$r_1$, profile $(11,5)$};
\end{tikzpicture}
\caption{$\T(14,4;(2,4))$: $\sigma=6$; crossing $0$ from $v_3$ to $v_5$, of
length $4$; crossing $1$ from $v_2$ to $v_6$, of length $8$; the $t$-ear
from $v_{10}$ to $v_{14}$, of length $6$.}
\end{subfigure}
\caption{Two templates (Definition~\ref{def:N02-template}), with $s=v_0$ on
the left and $t=v_L$ on the right.  The core $r_1$ is the black square,
with profile $(L-\mu+1,\mu+1)$ (Lemma~\ref{lem:N02-core}).  The circled
vertices are the diagonal members, those with $\kappa=\hb-h=\sigma$,
labelled with their profiles $(h,\hb)$: in (a) $v_3$ and $x^0_3$, in (b)
$v_4$, $x^0_2$ and $x^1_4$ (Lemma~\ref{lem:N02-skew}(iii)).}
\label{fig:N02-template}
\end{figure}

\begin{definition}[template]\label{def:N02-template}
Let $L,\mu$ be integers, put $\sigma:=L-2\mu$, and let
$E=(\varepsilon_j)_{0\le j<e}$ with $e\in\{0,1,2\}$ be a list of even
integers $\varepsilon_j\ge2$, with $\varepsilon_0<\varepsilon_1$ if
$e=2$.  Assume
\begin{equation}\label{eq:N02-hyp}
 \mu\ge3,\qquad \sigma\ge2,\qquad \sigma\ge4\ \text{ if }e=2 .
\end{equation}
The \emph{template} $\T=\T(L,\mu;E)$ (Figure~\ref{fig:N02-shape}; two
examples in Figure~\ref{fig:N02-template}) is the union of the following
paths.  Each path is given by the list of its vertices in order; its
edges join consecutive vertices of the list.
\begin{itemize}
\item The \emph{spine} $v_0,v_1,\dots,v_L$, with $L$ edges.  Its ends are
  the \emph{terminals} $s:=v_0$ and $t:=v_L$.
\item The \emph{$t$-ear} $r_0,r_1,\dots,r_{\mu+2}$, with $\mu+2$ edges,
  where $r_0:=v_{L-\mu}$ and $r_{\mu+2}:=v_L$.  Between these two
  vertices the spine has $\mu$ edges.
\item For each $j<e$, the \emph{crossing} $j$:
  $x^j_0,x^j_1,\dots,x^j_{\ell_j}$, with $\ell_j$ edges, where
  \[
  a_j:=\mu-1-j,\qquad b_j:=\mu+1+j,\qquad \ell_j:=2j+2+\varepsilon_j,
  \]
  $x^j_0:=v_{a_j}$ and $x^j_{\ell_j}:=v_{b_j}$.  Between these two
  vertices the spine has $b_j-a_j=2j+2$ edges.
\end{itemize}
The vertices $v_0,\dots,v_L$, $r_1,\dots,r_{\mu+1}$ and
$x^j_1,\dots,x^j_{(\ell_j)-1}$ ($j<e$) are pairwise distinct.  So the
$t$-ear meets the rest of $\T$ only at $v_{L-\mu}$ and $v_L$, and the
crossing $j$ only at $v_{a_j}$ and $v_{b_j}$.

The \emph{members} of $\T$ are the vertices of its interior
$\Int(\T)=V(\T)\setminus\{s,t\}$ (Definition~\ref{def:N01-attach}, with the
terminals $s$ and $t$).  As the vertices above are pairwise distinct,
$\Int(\T)$ is the disjoint union
\[
 \Int(\T)=\{v_1,\dots,v_{L-1}\}\cup\{r_1,\dots,r_{\mu+1}\}\cup
 \bigcup_{j<e}\{x^j_1,\dots,x^j_{(\ell_j)-1}\}
\]
of the \emph{spine members} $v_i$, the \emph{$t$-ear members} $r_i$ and
the \emph{crossing members} $x^j_i$ of the crossing $j$.  The spine
members and the $t$-ear members are the \emph{common} members of $\T$,
and the member $r_1$, the vertex of the $t$-ear next to $v_{L-\mu}$, is
its \emph{core}.
\end{definition}

\begin{definition}[diagonal vertices]\label{def:N02-diagonal}
In a template $\T=\T(L,\mu;E)$ the terminals are $s$ and $t$, and we
write $h$, $\hb$, $\kappa$ for $h_\T$, $\hb_\T$, $\kappa_\T$
(Definition~\ref{def:N01-profile}):
\[
 h(u)=d_\T(s,u),\qquad \hb(u)=d_\T(t,u),\qquad \kappa(u)=\hb(u)-h(u).
\]
A vertex $u$ of $\T$ is \emph{diagonal} when $\kappa(u)=\sigma$, where
$\sigma=L-2\mu$.
\end{definition}

Two templates $\T(L,\mu;E)$ and $\T(L,\mu;E')$ with the same $L$ and
$\mu$ have the same spine and the same $t$-ear: both have the vertices
$v_0,\dots,v_L$ and $r_1,\dots,r_{\mu+1}$, with the same edges between
them.  They differ only in their crossings, which depend on $E$ and
$E'$; this is why the spine members and the $t$-ear members are called
common.  Common does not mean shared: when several such templates are
used as separate graphs, each has its own spine and its own $t$-ear.

The rest of this section computes the profiles and proves the facts used
later.  The profile of a common member depends on $L$ and $\mu$ only
(Lemma~\ref{lem:N02-profiles}).  In Section~\ref{sec:N03}, the three
copies of a level are templates with the same $L$ and $\mu$; each copy
has its own spine, $t$-ear and crossings, and the copies meet only at
their terminals (Definition~\ref{def:N03-graph}).  A common member of one
copy and its counterparts in the other two copies, different vertices
with the same name (Definition~\ref{def:N03-names}), have the same
profile.  The diagonal
members are $v_\mu$ and one member on each crossing
(Lemma~\ref{lem:N02-skew}; in the plane of profiles they lie on the line
$\hb-h=\sigma$).  Two distinct members
have distinct profiles (Lemma~\ref{lem:N02-distinct}).  The core reaches
every other member by a path shorter than any walk through $s$ or $t$
(Lemma~\ref{lem:N02-core}).  In $G_\varphi$ the landmarks placed inside
copies are cores (Definition~\ref{def:N03-pattern}), and a resolving set
of $2m+5$ vertices has no vertex inside a copy other than its core
(Lemma~\ref{lem:N05-forcing}).

\begin{lemma}[positions of the ends]\label{lem:N02-ends}
In $\T=\T(L,\mu;E)$, for every $j<e$,
\[
 1\le a_j<b_j\le L-\mu-1 .
\]
So the ends $v_{a_j},v_{b_j}$ of the crossing $j$ lie strictly between
$s=v_0$ and $v_{L-\mu}$, the first end of the $t$-ear.  Moreover
$L-\mu=\mu+\sigma\ge5$.
\end{lemma}

\begin{proof}
As $j\le1$ and $\mu\ge3$ by \eqref{eq:N02-hyp},
$a_j=\mu-1-j\ge\mu-2\ge1$.  Next, $b_j-a_j=2j+2>0$, so $a_j<b_j$.
Finally,
\[
 b_j\le L-\mu-1\iff\mu+1+j\le L-\mu-1\iff2+j\le L-2\mu=\sigma .
\]
If $j=0$, this is $\sigma\ge2$, which holds by \eqref{eq:N02-hyp}.  If
$j=1$, then $e=2$, and $\sigma\ge4>3$ by \eqref{eq:N02-hyp}.  Moreover
$L-\mu=\mu+(L-2\mu)=\mu+\sigma\ge3+2=5$ by \eqref{eq:N02-hyp}.
\end{proof}

\begin{lemma}[the ears of the template]\label{lem:N02-ears}
In $\T=\T(L,\mu;E)$ (Definition~\ref{def:N02-template}):
\begin{enumerate}
\item[\rm(i)] $R:=(r_0,\dots,r_{\mu+2})$ is an ear of length $\mu+2$ with
  ends $v_{L-\mu},v_L$, and for $j<e$, $X^j:=(x^j_0,\dots,x^j_{\ell_j})$ is
  an ear of length $\ell_j$ with ends $v_{a_j},v_{b_j}$
  (Definition~\ref{def:N01-ear});
\item[\rm(ii)] every member is a spine member or lies in the interior of
  exactly one of these ears;
\item[\rm(iii)] $\T$ is connected.
\end{enumerate}
\end{lemma}

\begin{proof}
(i) $R$ is a path of $\T$.  Its consecutive vertices are adjacent by
Definition~\ref{def:N02-template}.  Its vertices are pairwise distinct:
$r_1,\dots,r_{\mu+1}$ are distinct from each other and from the spine
vertices (Definition~\ref{def:N02-template}), and its ends
$r_0=v_{L-\mu}$ and $r_{\mu+2}=v_L$ are distinct, as $\mu\ge3$.  Likewise
$X^j$ is a path of $\T$, taking $x^j_i$ in place of $r_i$; its ends
$x^j_0=v_{a_j}$ and $x^j_{\ell_j}=v_{b_j}$ are distinct, as $a_j<b_j$.

The edges of $\T$ are those of the spine, the $t$-ear and the crossings
(Definition~\ref{def:N02-template}).  For $1\le i\le\mu+1$, $r_i$ is not a
vertex of the spine or of a crossing, so its only edges are
$r_{i-1}r_i$ and $r_ir_{i+1}$, and it has degree two.  Likewise $x^j_i$
has degree two for $1\le i\le\ell_j-1$.  The lengths are $\mu+2\ge5$ and
$\ell_j=2j+2+\varepsilon_j\ge4$.  So $R$ and $X^j$ are ears
(Definition~\ref{def:N01-ear}).

(ii) The vertex list of Definition~\ref{def:N02-template} is a disjoint
union of the spine vertices and the interiors of $R$ and the $X^j$.

(iii) Any two spine vertices are joined along the spine.  Every vertex of
$R$ is joined along $R$ to its end $r_0=v_{L-\mu}$, and every vertex of
$X^j$ is joined along $X^j$ to its end $v_{a_j}$; both ends are spine
vertices.  By Definition~\ref{def:N02-template}, every vertex of $\T$
lies on the spine, on $R$ or on some $X^j$.  So every vertex is joined to
a spine vertex, any two vertices are joined through the spine, and $\T$
is connected.
\end{proof}

\begin{lemma}[the spine is geodesic]\label{lem:N02-spine}
In $\T=\T(L,\mu;E)$, $d_\T(v_i,v_{i'})=|i-i'|$ for $0\le i,i'\le L$; in
particular $d_\T(s,t)=L$.
\end{lemma}

\begin{proof}
The spine path from $v_i$ to $v_{i'}$ has $|i-i'|$ edges, so
$d_\T(v_i,v_{i'})\le|i-i'|$.  Let $P=(x_0,\dots,x_k)$ be a shortest
$v_i$--$v_{i'}$ path of $\T$.  We show that every vertex of $P$ is a
spine vertex.

Suppose not.  By Lemma~\ref{lem:N02-ears}(ii), some vertex $x_c$ of $P$
lies in the interior of one of the ears $E$ of
Lemma~\ref{lem:N02-ears}(i).  Let $x_a,\dots,x_b$ be the maximal run of
$P$ in the interior of $E$ that contains $x_c$
(Definition~\ref{def:N01-run}).  The ends $v_i$ and $v_{i'}$ of $P$ are
spine vertices, so they are not in the interior of $E$, and by
Lemma~\ref{lem:N01-earcross} the subpath $Q:=(x_{a-1},\dots,x_{b+1})$ of
$P$ runs along $E$ from one end to the other.

$E$ has more edges than the spine path $S$ between its ends.  The $t$-ear
$R$ has $\mu+2$ edges, and the spine path from $v_{L-\mu}$ to $v_L$ has
$\mu$.  The crossing $X^j$ has $\ell_j=2j+2+\varepsilon_j$ edges, where
$\varepsilon_j\ge2$ (Definition~\ref{def:N02-template}), and the spine path from $v_{a_j}$ to $v_{b_j}$ has $b_j-a_j=2j+2$.  Replacing $Q$
by $S$ in $P$ gives a $v_i$--$v_{i'}$ walk with fewer edges than $P$.  A
walk contains a path between its ends with at most as many edges, so some
$v_i$--$v_{i'}$ path is shorter than $P$, which is a contradiction.

So every vertex of $P$ is a spine vertex.  We check that every edge of
$P$ is a spine edge.  By Definition~\ref{def:N02-template}, an edge of
$\T$ is a spine edge or an edge $e_ke_{k+1}$ of an ear
$(e_0,\dots,e_\ell)$ ($0\le k<\ell$).  Such an ear edge has an end in the
interior of the ear: $e_k$ if $k\ge1$, and $e_1$ if $k=0$, as $\ell\ge2$.
Interior vertices of the ears are not spine vertices
(Definition~\ref{def:N02-template}).  So an ear edge never joins two spine
vertices, and every edge of $P$ is a spine edge.  So $P$ is a path of the
spine $v_0,\dots,v_L$ from $v_i$ to $v_{i'}$, that is,
$v_i,v_{i\pm1},\dots,v_{i'}$, with $|i-i'|$ edges.  Hence
$d_\T(v_i,v_{i'})=|i-i'|$.
\end{proof}

\begin{lemma}[profiles]\label{lem:N02-profiles}
In $\T=\T(L,\mu;E)$ (Definition~\ref{def:N02-template}), the profiles
$(h,\hb)$ (Definition~\ref{def:N02-diagonal}) are as follows:
\begin{enumerate}
\item[\rm(i)] $h(v_i)=i$ and $\hb(v_i)=L-i$ for $0\le i\le L$;
\item[\rm(ii)] $h(r_i)=L-\mu+i$ and $\hb(r_i)=\mu+2-i$ for $1\le i\le\mu+1$;
\item[\rm(iii)] for $j<e$ and $1\le i\le\ell_j-1$, with
  $a:=a_j$, $b:=b_j$, $\ell:=\ell_j$, $\varepsilon:=\varepsilon_j$,
  \[
  h(x^j_i)=\min(a+i,\ b+\ell-i),\qquad
  \hb(x^j_i)=\min(L-a+i,\ L-b+\ell-i),
  \]
  that is,
  \[
  (h,\hb)(x^j_i)=\begin{cases}
  (a+i,\ L-a+i) & \text{zone A: } 1\le i\le\varepsilon/2,\\
  (a+i,\ L-b+\ell-i) & \text{zone B: } \varepsilon/2<i\le 2j+2+\varepsilon/2,\\
  (b+\ell-i,\ L-b+\ell-i) & \text{zone C: } 2j+2+\varepsilon/2<i\le\ell-1.
  \end{cases}
  \]
\end{enumerate}
In particular the profiles of the common members depend on $L$ and $\mu$
only.
\end{lemma}

\begin{proof}
By Definition~\ref{def:N02-diagonal}, $h(u)=d_\T(s,u)$ and
$\hb(u)=d_\T(t,u)$, with $s=v_0$ and $t=v_L$.  The ranges of (i)--(iii)
cover every vertex of $\T$ exactly once: (ii) and (iii) are the ranges
of the $t$-ear members and of the crossing members in the disjoint union
$\Int(\T)$ of Definition~\ref{def:N02-template}, and (i) covers the spine
members $v_1,\dots,v_{L-1}$ and the terminals $s=v_0$ and $t=v_L$.

(i) By Lemma~\ref{lem:N02-spine}, $h(v_i)=d_\T(v_0,v_i)=i$ and
$\hb(v_i)=d_\T(v_L,v_i)=L-i$.

(ii) By Lemma~\ref{lem:N02-ears}(i), $R=(r_0,\dots,r_{\mu+2})$ is an ear of
length $\mu+2$ with ends $r_0=v_{L-\mu}$ and $r_{\mu+2}=v_L$; $\T$ is
connected by Lemma~\ref{lem:N02-ears}(iii), and $s,t$ are spine vertices,
so they are not in the interior of $R$.  For $1\le i\le\mu+1$,
Lemma~\ref{lem:N01-eardist} with $y:=s$ and (i) give
\[
 h(r_i)=d_\T(s,r_i)=\min\bigl(i+d_\T(s,v_{L-\mu}),\ (\mu+2-i)+d_\T(s,v_L)\bigr)
 =\min(i+L-\mu,\ \mu+2-i+L).
\]
The first term is the smaller: $(\mu+2-i+L)-(i+L-\mu)=2\mu+2-2i\ge0$, as
$i\le\mu+1$.  So $h(r_i)=L-\mu+i$.  Likewise, with $y:=t$,
\[
 \hb(r_i)=d_\T(t,r_i)=\min\bigl(i+d_\T(t,v_{L-\mu}),\ (\mu+2-i)+d_\T(t,v_L)\bigr)
 =\min(i+\mu,\ \mu+2-i),
\]
and $(i+\mu)-(\mu+2-i)=2i-2\ge0$, as $i\ge1$.  So $\hb(r_i)=\mu+2-i$.

(iii) By Lemma~\ref{lem:N02-ears}(i), $X^j=(x^j_0,\dots,x^j_\ell)$ is an ear
of length $\ell$ with ends $x^j_0=v_a$ and $x^j_\ell=v_b$, and $s,t$ are
not in its interior.  For $1\le i\le\ell-1$, Lemma~\ref{lem:N01-eardist}
and (i) give
\[
\begin{aligned}
 h(x^j_i)&=\min\bigl(i+d_\T(s,v_a),\ (\ell-i)+d_\T(s,v_b)\bigr)
  =\min(a+i,\ b+\ell-i),\\
 \hb(x^j_i)&=\min\bigl(i+d_\T(t,v_a),\ (\ell-i)+d_\T(t,v_b)\bigr)
  =\min(L-a+i,\ L-b+\ell-i).
\end{aligned}
\]
To see which terms are the smaller, use $b-a=2j+2$ and
$\ell=2j+2+\varepsilon$:
\[
\begin{aligned}
 (b+\ell-i)-(a+i)&=4j+4+\varepsilon-2i, &&\text{which is }\ge0\text{ iff }
  i\le2j+2+\varepsilon/2;\\
 (L-b+\ell-i)-(L-a+i)&=\varepsilon-2i, &&\text{which is }\ge0\text{ iff }
  i\le\varepsilon/2.
\end{aligned}
\]
As $\varepsilon$ is even, $\varepsilon/2$ is an integer, and
$\varepsilon/2<2j+2+\varepsilon/2$.  So the range $1\le i\le\ell-1$ splits
into three zones.  If $i\le\varepsilon/2$, both minima are the first
terms (zone A).  If $\varepsilon/2<i\le2j+2+\varepsilon/2$, the minimum for
$h$ is its first term and the minimum for $\hb$ its second (zone B).  If
$i>2j+2+\varepsilon/2$, both minima are the second terms (zone C).

Finally, by (i) and (ii), the profiles of the spine members and of the
$t$-ear members are given by formulas in $L$, $\mu$ and $i$ only, so they
do not depend on $E$.
\end{proof}

\begin{lemma}[skews and diagonal members]\label{lem:N02-skew}
In $\T=\T(L,\mu;E)$ (Definition~\ref{def:N02-template}), with $\kappa$ and
$\sigma=L-2\mu$ as in Definition~\ref{def:N02-diagonal}:
\begin{enumerate}
\item[\rm(i)] $\kappa(v_i)=L-2i$ for $0\le i\le L$, and
  $\kappa(r_i)=2-\sigma-2i\le-\sigma$ for $1\le i\le\mu+1$;
\item[\rm(ii)] on the crossing $j$, $\kappa(x^j_i)$ is $\sigma+2+2j$ in
  zone A, $\sigma+2+2j+\varepsilon_j-2i\in[\sigma-2-2j,\ \sigma+2j]$ in
  zone B, and $\sigma-2-2j$ in zone C; so
  $\sigma-2-2j\le\kappa(x^j_i)\le\sigma+2+2j\le\sigma+4$;
\item[\rm(iii)] the diagonal members are exactly $v_\mu$, with
  $h(v_\mu)=\mu$, and, for each $j<e$, the member
  $x^j_{j+1+\varepsilon_j/2}$, with
  $h(x^j_{j+1+\varepsilon_j/2})=\mu+\varepsilon_j/2$;
\item[\rm(iv)] every member $u$ with $\kappa(u)>\sigma+4$ is a spine member.
\end{enumerate}
\end{lemma}

\begin{proof}
By Definition~\ref{def:N02-diagonal}, $\kappa(u)=\hb(u)-h(u)$.  On the
crossing $j$ we write $a=\mu-1-j$, $b=\mu+1+j$, $\ell=2j+2+\varepsilon$
and $\varepsilon=\varepsilon_j$, as in Lemma~\ref{lem:N02-profiles}(iii);
so $a+b=2\mu$.

(i) By Lemma~\ref{lem:N02-profiles}(i), $\kappa(v_i)=(L-i)-i=L-2i$.  By
Lemma~\ref{lem:N02-profiles}(ii),
$\kappa(r_i)=(\mu+2-i)-(L-\mu+i)=2-(L-2\mu)-2i=2-\sigma-2i$, and
$2-\sigma-2i\le-\sigma$ as $i\ge1$.

(ii) By Lemma~\ref{lem:N02-profiles}(iii):
\[
\begin{aligned}
 &\text{zone A:} &\kappa&=(L-a+i)-(a+i)=L-2a=L-2\mu+2+2j=\sigma+2+2j;\\
 &\text{zone B:} &\kappa&=(L-b+\ell-i)-(a+i)=L-(a+b)+\ell-2i
   =\sigma+2+2j+\varepsilon-2i;\\
 &\text{zone C:} &\kappa&=(L-b+\ell-i)-(b+\ell-i)=L-2b=L-2\mu-2-2j
   =\sigma-2-2j.
\end{aligned}
\]
In zone B, $\varepsilon/2+1\le i\le2j+2+\varepsilon/2$.  At
$i=\varepsilon/2+1$ the skew is
$\sigma+2+2j+\varepsilon-(\varepsilon+2)=\sigma+2j$; at
$i=2j+2+\varepsilon/2$ it is
$\sigma+2+2j+\varepsilon-(4j+4+\varepsilon)=\sigma-2-2j$; in between it
decreases by $2$ at each step.  So it lies in
$[\sigma-2-2j,\ \sigma+2j]$.  Hence
$\sigma-2-2j\le\kappa(x^j_i)\le\sigma+2+2j$ on the whole crossing, and
$\sigma+2+2j\le\sigma+4$ as $j\le1$.

(iii) By Definition~\ref{def:N02-template}, every member is a spine
member, a $t$-ear member or a member of a crossing $j$, and by
Lemma~\ref{lem:N02-profiles}(iii) every member $x^j_i$ of the crossing
$j$ lies in exactly one of the zones A, B, C.  We find the diagonal
members in each of these classes.
\begin{itemize}
\item \emph{Spine members.}  $\kappa(v_i)=L-2i=\sigma$ iff
  $2i=L-\sigma=2\mu$, i.e.\ $i=\mu$.  As $\mu\ge3$ and $L-\mu\ge5$
  (Lemma~\ref{lem:N02-ends}), $1\le\mu\le L-1$, so $v_\mu$ is a member.
  So $v_\mu$ is the only diagonal spine member, and $h(v_\mu)=\mu$ by
  Lemma~\ref{lem:N02-profiles}(i).
\item \emph{$t$-ear members.}  $\kappa\le-\sigma<\sigma$ by (i), as
  $\sigma\ge2$.  So no $t$-ear member is diagonal.
\item \emph{Zones A and C of the crossing $j$.}  By (ii),
  $\kappa=\sigma+(2+2j)$ in zone A and $\kappa=\sigma-(2+2j)$ in zone C;
  neither is $\sigma$.  So no member of these zones is diagonal.
\item \emph{Zone B of the crossing $j$.}  By (ii), $\kappa=\sigma$ iff
  $2+2j+\varepsilon-2i=0$, i.e.\ $i=j+1+\varepsilon/2$.  This $i$ lies in
  zone B: $\varepsilon/2<j+1+\varepsilon/2\le2j+2+\varepsilon/2$, as
  $1\le j+1\le2j+2$.  So $x^j_{j+1+\varepsilon/2}$ is the only diagonal
  member of the crossing $j$, and
  $h(x^j_{j+1+\varepsilon/2})=a+i=(\mu-1-j)+(j+1+\varepsilon/2)
  =\mu+\varepsilon/2$.
\end{itemize}

(iv) By (i), $t$-ear members have $\kappa\le-\sigma<\sigma+4$; by (ii),
crossing members have $\kappa\le\sigma+4$.  By
Definition~\ref{def:N02-template} every member is a spine, $t$-ear or
crossing member, so a member with $\kappa>\sigma+4$ is a spine member.
\end{proof}

\begin{lemma}[distinct profiles]\label{lem:N02-distinct}
In $\T=\T(L,\mu;E)$, two distinct members have distinct profiles.
\end{lemma}

\begin{proof}
(1) The profile of a vertex $u$ is $(h(u),\hb(u))$ and its skew is
$\kappa(u)=\hb(u)-h(u)$ (Definitions~\ref{def:N01-profile}
and~\ref{def:N02-diagonal}).  Put $\Sigma(u):=h(u)+\hb(u)$.  Then
\[
 \Sigma(u)-\kappa(u)=2h(u),\qquad \Sigma(u)+\kappa(u)=2\hb(u),
\]
so $h(u)=(\Sigma(u)-\kappa(u))/2$ and $\hb(u)=(\Sigma(u)+\kappa(u))/2$.
Hence two members $u,u'$ have the same profile iff
$\Sigma(u)=\Sigma(u')$ and $\kappa(u)=\kappa(u')$.  Indeed, if the
profiles are equal, so are their sums and their differences.
Conversely, if $\Sigma(u)=\Sigma(u')$ and $\kappa(u)=\kappa(u')$, the two
formulas give
\[
\begin{aligned}
 h(u)&=\tfrac12\bigl(\Sigma(u)-\kappa(u)\bigr)
  =\tfrac12\bigl(\Sigma(u')-\kappa(u')\bigr)=h(u'),\\
 \hb(u)&=\tfrac12\bigl(\Sigma(u)+\kappa(u)\bigr)
  =\tfrac12\bigl(\Sigma(u')+\kappa(u')\bigr)=\hb(u').
\end{aligned}
\]
We show that two
distinct members differ in $\Sigma$ or in $\kappa$.  In the table and the
cases below, $\Sigma$ and $\kappa$ stand for $\Sigma(u)$ and $\kappa(u)$,
where $u$ is the member named in the row or in the case.

(2) On the crossing $j$ write $a,b,\ell,\varepsilon$ as in
Lemma~\ref{lem:N02-profiles}(iii).  Lemma~\ref{lem:N02-profiles} gives
$\Sigma$ and Lemma~\ref{lem:N02-skew} gives $\kappa$, with $i$ in the
ranges of Lemma~\ref{lem:N02-profiles}:
{\setlength{\arraycolsep}{3pt}
\[
\begin{array}{lll}
\text{spine } v_i: & \Sigma=i+(L-i)=L, & \kappa=L-2i;\\[2pt]
t\text{-ear } r_i: & \Sigma=(L-\mu+i)+(\mu+2-i)=L+2, &
 \kappa=2-\sigma-2i;\\[2pt]
\text{zone A}: & \Sigma=(a+i)+(L-a+i)=L+2i, & \kappa=\sigma+2+2j;\\[2pt]
\text{zone B}: & \Sigma=(a+i)+(L-b+\ell-i)=L+\ell-(b-a)=L+\varepsilon, &
 \kappa=\sigma+2+2j+\varepsilon-2i;\\[2pt]
\text{zone C}: & \Sigma=(b+\ell-i)+(L-b+\ell-i)=L+2(\ell-i), &
 \kappa=\sigma-2-2j.
\end{array}
\]}%
In zone A, $\Sigma$ runs over $L+2,L+4,\dots,L+\varepsilon$.  In zone C,
$\ell-i$ runs from $\ell-(2j+3+\varepsilon/2)=\varepsilon/2-1$ down to
$1$, so $\Sigma$ runs over $L+\varepsilon-2,\dots,L+2$.  If
$\varepsilon=2$, then $2j+2+\varepsilon/2=2j+3=\ell-1$, so no index lies
in zone C.  So, on the crossing $j$,
\[
 (\ast)\qquad \Sigma\le L+\varepsilon_j\quad\text{and}\quad
 \sigma-2-2j\le\kappa\le\sigma+2+2j\ \ \text{(Lemma~\ref{lem:N02-skew}(ii))}.
\]

(3) Let $u\ne u'$ be members.  By the disjoint union $\Int(\T)$ of
Definition~\ref{def:N02-template}, each of them is a spine member, a
$t$-ear member, or a member of the crossing $j$ for exactly one $j<e$.
By Lemma~\ref{lem:N02-profiles}(iii), a member of the crossing $j$ lies
in exactly one of the zones A, B, C, so the table in (2) gives $\Sigma$
and $\kappa$ for $u$ and for $u'$.  We distinguish cases according to the
classes of $u$ and $u'$.
\begin{itemize}
\item[(a)] \emph{Two spine members:} $\kappa=L-2i$ differs for different
  $i$.
\item[(b)] \emph{A spine member and a member of another class:}
  $\Sigma=L$ against $\Sigma\ge L+2$.
\item[(c)] \emph{Two $t$-ear members:} $\kappa=2-\sigma-2i$ differs for
  different $i$.
\item[(d)] \emph{A $t$-ear member and a member of the crossing $j$.}
  Exchanging $u$ and $u'$ if necessary, $u$ is the $t$-ear member and $u'$
  the member of the crossing $j$.  By Lemma~\ref{lem:N02-skew}(i),
  $\kappa(u)\le-\sigma$, and by $(\ast)$, $\kappa(u')\ge\sigma-2-2j$.  We
  compare the two bounds:
  \[
   (\sigma-2-2j)-(-\sigma)=2\sigma-2-2j=2(\sigma-1-j).
  \]
  If $j=0$, then $\sigma\ge2$ by \eqref{eq:N02-hyp}, so
  $\sigma-1-j=\sigma-1\ge1$.  If $j=1$, then $e=2$, as $j<e\le2$, so
  $\sigma\ge4$ by \eqref{eq:N02-hyp}, and $\sigma-1-j=\sigma-2\ge2$.  In
  both cases $\sigma-1-j>0$, so $-\sigma<\sigma-2-2j$.  Hence
  \[
   \kappa(u)\le-\sigma<\sigma-2-2j\le\kappa(u').
  \]
\item[(e)] \emph{Two members of the crossing $j$.}  We write $\kappa_A$,
  $\kappa_B$, $\kappa_C$ for the skew of a member in zone A, B, C, and
  $\Sigma_B$, $\Sigma_C$ for the value of $\Sigma$ at a member in zone B,
  C, as given by the table in (2).
  \begin{itemize}
  \item in one zone: $\Sigma$ differs for different $i$ in zones A and C,
    and $\kappa$ differs in zone B;
  \item zones A and B: $\kappa_A=\sigma+2+2j>\sigma+2j\ge\kappa_B$
    (Lemma~\ref{lem:N02-skew}(ii));
  \item zones A and C:
    $\kappa_A-\kappa_C=(\sigma+2+2j)-(\sigma-2-2j)=4+4j>0$;
  \item zones B and C: by the table, $\Sigma_B=L+\varepsilon$ and
    $\Sigma_C=L+2(\ell-i)$, where $i\ge2j+3+\varepsilon/2$ in zone C.  So
    \[
    \begin{aligned}
     \ell-i&\le\ell-(2j+3+\varepsilon/2)
       &&\text{as } i\ge2j+3+\varepsilon/2,\\
     &=(2j+2+\varepsilon)-(2j+3+\varepsilon/2)
       &&\text{as } \ell=2j+2+\varepsilon,\\
     &=\varepsilon/2-1,
    \end{aligned}
    \]
    and
    \[
     \Sigma_C\le L+2(\varepsilon/2-1)=L+\varepsilon-2<L+\varepsilon=\Sigma_B .
    \]
  \end{itemize}
\item[(f)] \emph{A member of the crossing $j$ and a member of the
  crossing $j'\ne j$.}  Then $e=2$ and $\{j,j'\}=\{0,1\}$.  Exchanging $u$
  and $u'$ if necessary, $u$ is a member of the crossing $0$ and $u'$ of
  the crossing $1$, and $\varepsilon_0<\varepsilon_1$
  (Definition~\ref{def:N02-template}).  By $(\ast)$ with $j=0$, $\kappa(u)\in[\sigma-2,\sigma+2]$ and
  $\Sigma(u)\le L+\varepsilon_0$.
  \begin{itemize}
  \item If $u'$ is in zone A or C of the crossing $1$,
    $\kappa(u')=\sigma+4$ or $\sigma-4$, outside $[\sigma-2,\sigma+2]$.
  \item If $u'$ is in zone B of the crossing $1$,
    $\Sigma(u')=L+\varepsilon_1>L+\varepsilon_0\ge\Sigma(u)$.
  \end{itemize}
\end{itemize}
These cases cover all pairs of classes.  In each, $\Sigma$ or $\kappa$
differs, so the profiles differ.
\end{proof}

\begin{remark}[the hypothesis $\sigma\ge4$ for two crossings]\label{rem:N02-sigma}
The hypothesis $\sigma\ge4$ when $e=2$ in \eqref{eq:N02-hyp} is used only
in case~(d), and it cannot be dropped.  Let $e=2$ and $\sigma=2$.  By
Lemma~\ref{lem:N02-profiles}(ii) with $i=1$, $(h,\hb)(r_1)=(L-\mu+1,\mu+1)$.
The last member $x^1_{\ell_1-1}$ of the crossing $1$ lies in zone~C: here
$j=1$ and $\ell_1=4+\varepsilon_1$, so $\ell_1-1=\varepsilon_1+3$ and
$2j+2+\varepsilon_1/2=4+\varepsilon_1/2$, and $\varepsilon_1+3>4+\varepsilon_1/2$
as $\varepsilon_1\ge4$, since $\varepsilon_1>\varepsilon_0\ge2$ and
$\varepsilon_1$ is even.  So by Lemma~\ref{lem:N02-profiles}(iii) with
$b=\mu+2$ and $i=\ell_1-1$,
\[
 (h,\hb)(x^1_{\ell_1-1})=(b+1,\ L-b+1)=(\mu+3,\ L-\mu-1)=(L-\mu+1,\ \mu+1),
\]
the last equality as $L-\mu=\mu+\sigma=\mu+2$.  So $r_1$ and
$x^1_{\ell_1-1}$ have the same profile.
\end{remark}

In $G_\varphi$ the copies with two crossings are those of the variable
levels, where $\sigma_k=2k\ge4$ (Lemma~\ref{lem:N03-params}(ii)).

\begin{lemma}[the core]\label{lem:N02-core}
In $\T=\T(L,\mu;E)$, with $h,\hb$ as in Definition~\ref{def:N02-diagonal},
put $\alpha_\T:=L-\mu+1$ and $\beta_\T:=\mu+1$.  Then
$(h,\hb)(r_1)=(\alpha_\T,\beta_\T)$, $\alpha_\T-\beta_\T=\sigma$, and for every member
$u\ne r_1$,
\[
 d_\T(r_1,u)\ <\ \min\bigl(\alpha_\T+h(u),\ \beta_\T+\hb(u)\bigr).
\]
\end{lemma}

\begin{proof}
By Lemma~\ref{lem:N02-profiles}(ii) with $i=1$,
$(h,\hb)(r_1)=(L-\mu+1,\mu+1)=(\alpha_\T,\beta_\T)$, and
$\alpha_\T-\beta_\T=L-2\mu=\sigma$.

As $\alpha_\T=d_\T(s,r_1)$ and $h(u)=d_\T(s,u)$
(Definition~\ref{def:N02-diagonal}), $\alpha_\T+h(u)$ is the length of a walk
from $r_1$ to $u$ through $s$: a shortest path from $r_1$ to $s$ followed
by a shortest path from $s$ to $u$.  Likewise $\beta_\T+\hb(u)$ is the length
of a walk from $r_1$ to $u$ through $t$.  In each case below we give a
path $P$ from $r_1$ to $u$ and show that it is shorter than both walks:
$(\alpha_\T+h(u))-|P|>0$ and $(\beta_\T+\hb(u))-|P|>0$.  Then
$d_\T(r_1,u)\le|P|<\min(\alpha_\T+h(u),\ \beta_\T+\hb(u))$.  By
Definition~\ref{def:N02-template}, a member
$u\ne r_1$ is a spine member, a $t$-ear member $r_i$ with
$2\le i\le\mu+1$, or a crossing member.  We treat the three kinds in turn.

\emph{Spine members.}  Let $u=v_p$ with $1\le p\le L-1$.  The path
$P$, the edge $r_1r_0$, where $r_0=v_{L-\mu}$, followed by the spine path
from $v_{L-\mu}$ to $v_p$, has $|P|=1+|L-\mu-p|$ edges.  By
Lemma~\ref{lem:N02-profiles}(i), $\alpha_\T+h(v_p)=(L-\mu+1)+p$ and
$\beta_\T+\hb(v_p)=(\mu+1)+(L-p)$.  If $p\le L-\mu$, then
$|P|=1+L-\mu-p$, and
\[
\begin{aligned}
 (\alpha_\T+h(v_p))-|P|&=(L-\mu+1+p)-(1+L-\mu-p)=2p>0,\\
 (\beta_\T+\hb(v_p))-|P|&=(\mu+1+L-p)-(1+L-\mu-p)=2\mu>0.
\end{aligned}
\]
If $p>L-\mu$, then $|P|=1+p-L+\mu$, and
\[
\begin{aligned}
 (\alpha_\T+h(v_p))-|P|&=(L-\mu+1+p)-(1+p-L+\mu)=2(L-\mu)>0,\\
 (\beta_\T+\hb(v_p))-|P|&=(\mu+1+L-p)-(1+p-L+\mu)=2(L-p)>0,
\end{aligned}
\]
since $L-\mu\ge5$ by Lemma~\ref{lem:N02-ends} and $p\le L-1$.

\emph{$t$-ear members.}  Let $u=r_i$ with $2\le i\le\mu+1$.  The path
$P=r_1,r_2,\dots,r_i$ along $R$ has $|P|=i-1$ edges.  By
Lemma~\ref{lem:N02-profiles}(ii),
\[
\begin{aligned}
 (\alpha_\T+h(r_i))-|P|&=(2(L-\mu)+1+i)-(i-1)=2(L-\mu)+2>0,\\
 (\beta_\T+\hb(r_i))-|P|&=(2\mu+3-i)-(i-1)=2(\mu+2-i)>0,
\end{aligned}
\]
the latter as $i\le\mu+1$.

\emph{Crossing members.}  Let $u=x^j_i$ with $1\le i\le\ell-1$, and
$a,b,\ell$ as in Lemma~\ref{lem:N02-profiles}(iii).  By
Lemma~\ref{lem:N02-ears}, $X^j$ is an ear of the connected graph $\T$
with ends $v_a$ and $v_b$, and the spine vertex $v_{L-\mu}$ is not in its
interior.  By Lemma~\ref{lem:N02-ends}, $a<b<L-\mu$, so
$d_\T(v_{L-\mu},v_a)=L-\mu-a$ and $d_\T(v_{L-\mu},v_b)=L-\mu-b$ by
Lemma~\ref{lem:N02-spine}.  Let $P$ be the edge $r_1r_0$ followed by a
shortest path from $v_{L-\mu}$ to $u$.  Lemma~\ref{lem:N01-eardist} with
$y:=v_{L-\mu}$ gives
\[
 |P|=1+d_\T(v_{L-\mu},u)=\min(c_1,c_2),\qquad
 c_1:=1+L-\mu-a+i,\quad c_2:=1+L-\mu-b+\ell-i,
\]
where $c_1$ and $c_2$ are the lengths of the two candidates for $P$,
entering the crossing at $v_a$ and at $v_b$.  By
Lemma~\ref{lem:N02-profiles}(iii),
\[
 \alpha_\T+h(u)=\min(c_1+2a,\ c_2+2b),\qquad
 \beta_\T+\hb(u)=\min(c_1+2\mu,\ c_2+2\mu);
\]
the two terms of $\alpha_\T+h(u)$ are the lengths of the walks through $s$
entering the crossing at $v_a$ and at $v_b$, and likewise for
$\beta_\T+\hb(u)$ through $t$.  As $a$, $b$ and $\mu$ are positive, each term
on the right exceeds the corresponding $c_k$.  Since $|P|=\min(c_1,c_2)\le
c_k$ for $k=1,2$, $|P|$ is smaller than both minima.

In all three cases, $d_\T(r_1,u)<\min\bigl(\alpha_\T+h(u),\ \beta_\T+\hb(u)\bigr)$.
\end{proof}

\begin{lemma}[distances to $r_{\mu+1}$]\label{lem:N02-to-r}
In $\T=\T(L,\mu;E)$:
\begin{enumerate}
\item[\rm(i)] $d_\T(v_p,r_{\mu+1})=L-p+1$ for $0\le p\le L$;
\item[\rm(ii)] $d_\T(x^j_i,r_{\mu+1})=\hb(x^j_i)+1$ for $j<e$ and
  $1\le i\le\ell_j-1$.
\end{enumerate}
\end{lemma}

\begin{proof}
We use the ear $R=(r_0,\dots,r_{\mu+2})$, with $r_0=v_{L-\mu}$ and
$r_{\mu+2}=v_L$ (Lemma~\ref{lem:N02-ears}(i)), and
Lemma~\ref{lem:N01-eardist} in $\T$, which is connected
(Lemma~\ref{lem:N02-ears}(iii)).

(i) The spine vertex $v_p$ is not in the interior of $R$.
Lemma~\ref{lem:N01-eardist}, with $r_{\mu+1}$ in the place of $e_i$ and
$y:=v_p$, and Lemma~\ref{lem:N02-spine} give
\[
 d_\T(v_p,r_{\mu+1})
 =\min\bigl((\mu+1)+d_\T(v_{L-\mu},v_p),\ 1+d_\T(v_L,v_p)\bigr)
 =\min\bigl(\mu+1+|L-\mu-p|,\ 1+L-p\bigr).
\]
If $p\le L-\mu$, the two terms are equal:
$\mu+1+(L-\mu-p)=1+L-p$.  If $p>L-\mu$, the first exceeds the second:
$(\mu+1+p-L+\mu)-(1+L-p)=2\bigl(p-(L-\mu)\bigr)>0$.  So
$d_\T(v_p,r_{\mu+1})=L-p+1$.

(ii) Let $a,b,\ell$ be as in Lemma~\ref{lem:N02-profiles}(iii).  By
Lemma~\ref{lem:N02-ears}(i), $X^j$ is an ear with ends $v_a$ and $v_b$, and
$r_{\mu+1}$ is not in its interior.  By (i) with $p=a$ and $p=b$,
$d_\T(v_a,r_{\mu+1})=L-a+1$ and $d_\T(v_b,r_{\mu+1})=L-b+1$.  So
Lemma~\ref{lem:N01-eardist} with $y:=r_{\mu+1}$ gives
\[
\begin{aligned}
 d_\T(x^j_i,r_{\mu+1})
 &=\min\bigl(i+d_\T(v_a,r_{\mu+1}),\ (\ell-i)+d_\T(v_b,r_{\mu+1})\bigr)\\
 &=\min\bigl(i+(L-a+1),\ (\ell-i)+(L-b+1)\bigr)\\
 &=1+\min(L-a+i,\ L-b+\ell-i).
\end{aligned}
\]
By Lemma~\ref{lem:N02-profiles}(iii), the last minimum is $\hb(x^j_i)$.
\end{proof}

\begin{lemma}[distances to $v_{L-1}$]\label{lem:N02-to-v}
In $\T=\T(L,\mu;E)$, $d_\T(r_i,v_{L-1})=\mu+3-i$ for $2\le i\le\mu+1$.
\end{lemma}

\begin{proof}
By Lemma~\ref{lem:N02-ears}(i), $R=(r_0,\dots,r_{\mu+2})$ is an ear with
ends $r_0=v_{L-\mu}$ and $r_{\mu+2}=v_L$, and the spine vertex $v_{L-1}$ is
not in its interior; $\T$ is connected (Lemma~\ref{lem:N02-ears}(iii)).
Lemma~\ref{lem:N01-eardist} with $y:=v_{L-1}$, and
Lemma~\ref{lem:N02-spine}, give
\[
 d_\T(r_i,v_{L-1})
 =\min\bigl(i+d_\T(v_{L-\mu},v_{L-1}),\ (\mu+2-i)+d_\T(v_L,v_{L-1})\bigr)
 =\min(i+\mu-1,\ \mu+3-i).
\]
The first term is at least the second: $(i+\mu-1)-(\mu+3-i)=2i-4\ge0$,
as $i\ge2$.  So $d_\T(r_i,v_{L-1})=\mu+3-i$.
\end{proof}

\begin{corollary}[members other than the core]\label{cor:N02-noncore}
In $\T=\T(L,\mu;E)$, with $h,\hb$ as in Definition~\ref{def:N02-diagonal},
let $w\ne r_1$ be a member.  Put $x:=r_{\mu+1}$ if $w$ is a spine or
crossing member, and $x:=v_{L-1}$ if $w$ is a $t$-ear member.  Then $x$
is a common member, $x\ne w$, and
\[
 d_\T(w,x)=\min\bigl(h(w)+h(x),\ \hb(w)+\hb(x)\bigr).
\]
\end{corollary}

\begin{proof}
By Definition~\ref{def:N02-template}, $w$ is a spine member, a crossing
member, or a $t$-ear member $r_i$ with $2\le i\le\mu+1$.  Also $r_{\mu+1}$
is a $t$-ear member and $v_{L-1}$ a spine member, so $x$ is common, and
$x\ne w$.

We first show that $d_\T(w,x)=\hb(w)+\hb(x)$, case by case.
\begin{itemize}
\item $w=v_p$ with $1\le p\le L-1$; then $x=r_{\mu+1}$.  By
  Lemma~\ref{lem:N02-to-r}(i), $d_\T(v_p,r_{\mu+1})=L-p+1$.  By
  Lemma~\ref{lem:N02-profiles}(i),(ii), $\hb(v_p)=L-p$ and
  $\hb(r_{\mu+1})=\mu+2-(\mu+1)=1$.  So
  $d_\T(v_p,r_{\mu+1})=(L-p)+1=\hb(v_p)+\hb(r_{\mu+1})$.
\item $w=x^j_i$, a crossing member; then $x=r_{\mu+1}$.  By
  Lemma~\ref{lem:N02-to-r}(ii) and $\hb(r_{\mu+1})=1$,
  $d_\T(x^j_i,r_{\mu+1})=\hb(x^j_i)+1=\hb(x^j_i)+\hb(r_{\mu+1})$.
\item $w=r_i$ with $2\le i\le\mu+1$; then $x=v_{L-1}$.  By
  Lemma~\ref{lem:N02-to-v}, $d_\T(r_i,v_{L-1})=\mu+3-i$.  By
  Lemma~\ref{lem:N02-profiles}(ii),(i), $\hb(r_i)=\mu+2-i$ and
  $\hb(v_{L-1})=L-(L-1)=1$.  So
  $d_\T(r_i,v_{L-1})=(\mu+2-i)+1=\hb(r_i)+\hb(v_{L-1})$.
\end{itemize}

It remains to see that $\hb(w)+\hb(x)$ is the minimum.  By
Definition~\ref{def:N02-diagonal}, $h(w)=d_\T(s,w)$ and $h(x)=d_\T(s,x)$,
so $h(w)+h(x)$ is the length of a walk from $w$ to $x$ through $s$: a
shortest path from $w$ to $s$ followed by a shortest path from $s$ to $x$.
Hence $d_\T(w,x)\le h(w)+h(x)$, that is, $\hb(w)+\hb(x)\le h(w)+h(x)$.
So
\[
 \min\bigl(h(w)+h(x),\ \hb(w)+\hb(x)\bigr)=\hb(w)+\hb(x)=d_\T(w,x).
\]
\end{proof}


\section{Restricted formulas and the graph \texorpdfstring{$G_\varphi$}{G\_phi}}\label{sec:N03}

This section builds the graph $G_\varphi$ of the reduction from a
restricted formula $\varphi$ (Definition~\ref{def:N03-formula}) with $n$
clauses and $r$ variables.  The formula has $m=n+r+1$ levels: one for
each clause, one for each variable, and a last level, the idle level,
which receives no cell and is there to carry the pair $q,q'$ at its
end.  Each level has three options: at
the level of a clause an option selects a literal of the clause, and at
the level of a variable it sets the variable true or false.  Each option
receives a set of integers, its cells, so that two options at different
levels share a cell exactly when the literal selected by one of them is
falsified by the value set by the other
(Definition~\ref{def:N03-options} and Lemma~\ref{lem:N03-options}).
In $G_\varphi$, each option is a copy of a template (Section~\ref{sec:N02})
attached at two vertices $s_k$ and $t_k$ of its level $k$
(Definition~\ref{def:N03-graph}); its crossings put its diagonal members
at the distances from a fixed vertex $T_1$ given by its cells
(Lemma~\ref{lem:N08-cells}).  A pattern selects one option at each level
and leaves its copy without a landmark (Definition~\ref{def:N03-pattern});
it resolves $G_\varphi$ exactly when no two selected options share a cell
(Lemma~\ref{lem:N08-choice}), so such a pattern exists exactly when
$\varphi$ is satisfiable (Corollary~\ref{cor:N08-sat}).
Sections~\ref{sec:N04}--\ref{sec:N08} prove that
$\md(G_\varphi)=2m+5$ if $\varphi$ is satisfiable and $2m+6$ otherwise
(Theorem~\ref{thm:N08-value}).

\begin{definition}[restricted formulas]\label{def:N03-formula}
Let $\xi_1,\dots,\xi_r$ be variables.  A \emph{literal} is $\xi_g$ or $\neg\xi_g$;
its \emph{variable} is $\xi_g$.  A \emph{clause} is a nonempty sequence
$C=(\lambda_1,\dots,\lambda_k)$ of literals whose variables are pairwise
distinct; we write $|C|:=k$, and $\lambda\in C$ when $\lambda$ is one of
$\lambda_1,\dots,\lambda_k$.  An \emph{assignment} is
a map $A:\{\xi_1,\dots,\xi_r\}\to\{\mathrm{true},\mathrm{false}\}$; it
\emph{satisfies} $\xi_g$ when $A(\xi_g)=\mathrm{true}$, $\neg\xi_g$ when
$A(\xi_g)=\mathrm{false}$, a clause when it satisfies one of its literals,
and a sequence $\varphi=(C_1,\dots,C_n)$ of clauses when it satisfies
every $C_i$; $\varphi$ is \emph{satisfiable} when some assignment
satisfies it.  We write $C_i=(\lambda_{i,1},\dots,\lambda_{i,|C_i|})$, so
$\lambda_{i,c}$ is the $c$-th literal of the $i$-th clause.  The sequence
$\varphi$ is a \emph{restricted formula} when
\begin{enumerate}
\item[\rm(a)] $n\ge1$ and $|C_i|\le3$ for $1\le i\le n$;
\item[\rm(b)] every variable $\xi_g$, $1\le g\le r$, is the variable of a
  literal of some clause;
\item[\rm(c)] every literal belongs to at most two of the clauses
  $C_1,\dots,C_n$.
\end{enumerate}
For a restricted formula we put $m:=n+r+1$.
\end{definition}

\begin{definition}[levels, options, cells]\label{def:N03-options}
Let $\varphi$ be a restricted formula (Definition~\ref{def:N03-formula}).
Its \emph{levels} are the integers $1,\dots,m$: the level $i$ is the
level of the clause $C_i$ ($1\le i\le n$), the level $n+g$ is the
level of the variable $\xi_g$ ($1\le g\le r$), and the level
$m=n+r+1$ is the \emph{idle level}.  Each level has three
\emph{options} $1,2,3$:
\begin{itemize}
\item at the level $i$ of a clause, the option $c$ \emph{selects} the
  literal $\lambda_{i,\min(c,|C_i|)}$, that is, the $c$-th literal of
  $C_i$, or its last literal if $|C_i|<c$;
\item at the level $n+g$ of a variable, the option $c$ sets $\xi_g$ to
  the value $v(c)$, where $v(1):=\mathrm{true}$ and
  $v(2):=v(3):=\mathrm{false}$; the option $c$ \emph{falsifies} a literal
  $\lambda$ with variable $\xi_g$ when $\lambda$ is not satisfied by the
  assignment that gives $\xi_g$ the value $v(c)$
  (Definition~\ref{def:N03-formula}), that is, when $\lambda=\neg\xi_g$
  and $c=1$, or $\lambda=\xi_g$ and $c\in\{2,3\}$;
\item at the idle level, an option selects nothing and falsifies
  nothing.
\end{itemize}
An \emph{occurrence} of $\varphi$ is a pair $(i,\lambda)$ with
$1\le i\le n$ and $\lambda\in C_i$, that is, a place where a literal is
written in $\varphi$; $\Omega$ denotes the set of occurrences.  The
\emph{cell} of the occurrence $(i,\lambda_{i,c})$ is
\[
 \cell(i,\lambda_{i,c}):=2m+5+|C_1|+\dots+|C_{i-1}|+c;
\]
so, reading the clauses from $C_1$ to $C_n$, the occurrences receive the
cells $2m+6,2m+7,\dots,2m+5+|\Omega|$.  The \emph{option set} $Z_{k,c}$
of the option $c$ at the level $k$ is the following set of cells: at the
level of a clause, the cell of the selected occurrence; at the level of a
variable, the cells of the occurrences of the literal that the option
falsifies; at the idle level, no cell.  That is,
\[
\begin{aligned}
 Z_{i,c}&:=\{\cell(i,\lambda_{i,\min(c,|C_i|)})\} &&(1\le i\le n),\\
 Z_{n+g,1}&:=\{\cell(i,\neg\xi_g):(i,\neg\xi_g)\in\Omega\} &&(1\le g\le r),\\
 Z_{n+g,2}&:=Z_{n+g,3}:=\{\cell(i,\xi_g):(i,\xi_g)\in\Omega\} &&(1\le g\le r),\\
 Z_{m,c}&:=\emptyset &&(c\in\{1,2,3\}).
\end{aligned}
\]
For example, $\varphi=((\xi_1,\xi_2),(\neg\xi_1))$ has $m=5$, the
occurrences $(1,\xi_1)$, $(1,\xi_2)$, $(2,\neg\xi_1)$ with the cells $16$,
$17$, $18$, and the option sets $Z_{1,1}=\{16\}$,
$Z_{1,2}=Z_{1,3}=\{17\}$; $Z_{2,c}=\{18\}$; $Z_{3,1}=\{18\}$,
$Z_{3,2}=Z_{3,3}=\{16\}$; $Z_{4,1}=\emptyset$, $Z_{4,2}=Z_{4,3}=\{17\}$;
$Z_{5,c}=\emptyset$
(Figure~\ref{fig:N08-cells}).
\end{definition}

A variable option thus lists the occurrences that it falsifies: an option
of a clause and an option of a variable share a cell exactly when the
value set by the second falsifies the literal selected by the first
(Lemma~\ref{lem:N03-options}(iv)).

\begin{lemma}[option sets]\label{lem:N03-options}
Let $\varphi$ be a restricted formula, with the notation of
Definition~\ref{def:N03-options}.
\begin{enumerate}
\item[\rm(i)] Distinct occurrences have distinct cells, and every cell
  lies between $2m+6$ and $2m+5+|\Omega|\le5m+5$.
\item[\rm(ii)] $|Z_{k,c}|\le2$ for every level $k$ and option $c$, and
  $|Z_{k,c}|=1$ if $k$ is the level of a clause.
\item[\rm(iii)] Let $(i,\lambda)$ be an occurrence, and $\xi_g$ the
  variable of $\lambda$.  The cell of $(i,\lambda)$ lies in exactly the
  following option sets:
  \begin{itemize}
  \item $Z_{i,c}$, for each option $c$ of the level $i$ (of the clause
    $C_i$) that selects $\lambda$;
  \item if $\lambda=\xi_g$: $Z_{n+g,2}$ and $Z_{n+g,3}$, the sets of the
    options of the level $n+g$ that falsify $\lambda$;
  \item if $\lambda=\neg\xi_g$: $Z_{n+g,1}$, the set of the option of the
    level $n+g$ that falsifies $\lambda$.
  \end{itemize}
\item[\rm(iv)] Let $k<k'$ be levels and $c,c'$ options.  Then
  $Z_{k,c}\cap Z_{k',c'}\ne\emptyset$ iff $k$ is the level $i$ of a clause,
  $k'$ is the level of the variable of the literal $\lambda$ that the
  option $c$ at the level $i$ selects, and the option $c'$ at the level
  $k'$ falsifies $\lambda$.
\end{enumerate}
\end{lemma}

\begin{proof}
(i) By Definition~\ref{def:N03-options}, the occurrences receive the cells
$2m+6,\dots,2m+5+|\Omega|$, one each.  Every clause has at most three
literals (Definition~\ref{def:N03-formula}(a)), so
$|\Omega|\le3n\le3m$ and $2m+5+|\Omega|\le5m+5$.

(ii) At the level $i$ of a clause, $Z_{i,c}$ is one cell.  At the level
$n+g$ of a variable, $Z_{n+g,c}$ has one cell for each clause that contains
the literal falsified by the option $c$; there are at most two such
clauses by Definition~\ref{def:N03-formula}(c).

(iii) By (i), the cell of $(i,\lambda)$ is the cell of no other
occurrence.  We read off Definition~\ref{def:N03-options} which option
sets contain it.

The level $i'$ of a clause has the option sets
$Z_{i',c}=\{\cell(i',\lambda_{i',\min(c,|C_{i'}|)})\}$.  So the cell of
$(i,\lambda)$ lies in $Z_{i',c}$ iff $i'=i$ and the option $c$ selects
$\lambda$: the first item.

The level $n+g'$ of a variable has the option sets $Z_{n+g',1}$, made of
the cells of the occurrences of $\neg\xi_{g'}$, and
$Z_{n+g',2}=Z_{n+g',3}$, made of the cells of the occurrences of
$\xi_{g'}$.  So the cell of $(i,\lambda)$ lies in $Z_{n+g',1}$ iff
$\lambda=\neg\xi_{g'}$, and in $Z_{n+g',2}$ and $Z_{n+g',3}$ iff
$\lambda=\xi_{g'}$.  In both cases $g'=g$, as $\xi_g$ is the variable of
$\lambda$: these are the second and third items.  The idle level has
the option sets $Z_{m,c}=\emptyset$, which contain no cell.

(iv) Suppose first that $Z_{k,c}\cap Z_{k',c'}\ne\emptyset$.  The common
cell lies in an option set, so it is the cell of an occurrence
$(i,\lambda)$ (Definition~\ref{def:N03-options}); by (i), this occurrence
is unique.  Let $\xi_g$ be
the variable of $\lambda$.  By (iii), applied to $k$ and $c$, either $k=i$
and the option $c$ selects $\lambda$, or $k=n+g$ and the option $c$
falsifies $\lambda$ (the second and third items of (iii)).  Likewise for
$k'$, taking $c'$ in place of $c$.  As
$i\le n<n+g$ and $k\ne k'$, one of $k,k'$ is $i$ and the other is $n+g$;
as $k<k'$, $k=i$ and $k'=n+g$.  So $k$ is the level $i$ of a clause, the
option $c$ selects $\lambda$, $k'$ is the level of the variable of
$\lambda$, and the option $c'$ falsifies $\lambda$.

Conversely, suppose that the option $c$ at the level $i$ selects
$\lambda$, that $k'=n+g$ is the level of the variable $\xi_g$ of $\lambda$,
and that the option $c'$ at the level $k'$ falsifies $\lambda$.  Then
$\lambda\in C_i$, so $(i,\lambda)$ is an occurrence, and by (iii) its cell
lies in $Z_{i,c}$ and in $Z_{k',c'}$.
\end{proof}

\begin{definition}[the numbers of the construction]\label{def:N03-params}
Let $\varphi$ be a restricted formula, with the levels and option sets of
Definition~\ref{def:N03-options}.  Put
\[
 D:=4m+15,\qquad p:=2m+4,\qquad \alpha:=2m+5,
\]
and for each level $k$
\[
 L_k:=4m+8-2k,\qquad \mu_k:=2m+4-2k,\qquad \sigma_k:=2k,\qquad
 \beta_k:=2m+5-2k .
\]
For a level $k$ and an option $c$, let $\zeta_0<\dots<\zeta_{e-1}$ be the
cells of the option set $Z_{k,c}$, so that $e=|Z_{k,c}|\le2$
(Lemma~\ref{lem:N03-options}(ii)), and put
\[
 E_{k,c}:=\bigl(2(\zeta_j-k-\mu_k)\bigr)_{j<e}.
\]
\end{definition}

The copies of the level $k$ in $G_\varphi$ are the templates
$\T(L_k,\mu_k;E_{k,c})$, $c\in\{1,2,3\}$
(Definition~\ref{def:N03-graph}).  The numbers $\alpha$ and $\beta_k$ are
the profile of the core of such a copy (Lemmas~\ref{lem:N02-core}
and~\ref{lem:N03-params}(i)).  The list $E_{k,c}$ gives the excesses of
the crossings of the copy of the option $c$ at the level $k$; it puts their
diagonal members at the distances $\zeta\in Z_{k,c}$ from $T_1$
(Lemma~\ref{lem:N08-cells}).  $D$ is the length of the four arms, and $p$
and $7$ are the positions on the arms of the rows $w_j$
(Definition~\ref{def:N03-graph}).

\begin{lemma}[the numbers]\label{lem:N03-params}
Let $\varphi$ be a restricted formula, with the notation of
Definitions~\ref{def:N03-options} and~\ref{def:N03-params}, and
$1\le k\le m$.
\begin{enumerate}
\item[\rm(i)] $L_k=2\mu_k+\sigma_k$, $\alpha=L_k-\mu_k+1$,
  $\beta_k=\mu_k+1$, $\alpha-\beta_k=\sigma_k$, $L_k=D-7-2k$ and
  $D-p=p+7$.
\item[\rm(ii)] For $c\in\{1,2,3\}$, the triple $(L_k,\mu_k,E_{k,c})$
  satisfies the hypotheses of Definition~\ref{def:N02-template}, with
  $\sigma=L_k-2\mu_k=\sigma_k$ and $e=|Z_{k,c}|$: $e\le2$; the entries of
  $E_{k,c}$ are even and at least $2$, and $\varepsilon_0<\varepsilon_1$ if
  $e=2$; $\mu_k\ge3$, $\sigma_k\ge2$, and $\sigma_k\ge4$ if $e=2$.
\item[\rm(iii)] $m\ge3$.
\end{enumerate}
\end{lemma}

\begin{proof}
(i) By Definition~\ref{def:N03-params}:
\[
\begin{aligned}
 2\mu_k+\sigma_k&=2(2m+4-2k)+2k=4m+8-2k=L_k,\\
 L_k-\mu_k+1&=(4m+8-2k)-(2m+4-2k)+1=2m+5=\alpha,\\
 \mu_k+1&=2m+5-2k=\beta_k,\\
 \alpha-\beta_k&=(2m+5)-(2m+5-2k)=2k=\sigma_k,\\
 D-7-2k&=(4m+15)-7-2k=4m+8-2k=L_k,\\
 D-p&=(4m+15)-(2m+4)=2m+11=p+7.
\end{aligned}
\]

(ii) By (i), $L_k-2\mu_k=\sigma_k$, so $\sigma_k$ is the number $\sigma$ of
Definition~\ref{def:N02-template} for this template.
\begin{itemize}
\item $e=|Z_{k,c}|\le2$ by Lemma~\ref{lem:N03-options}(ii).
\item Write $\varepsilon_j:=2(\zeta_j-k-\mu_k)$ for the entries of
  $E_{k,c}$; each is even.  The cells of Definition~\ref{def:N03-options}
  start at $2m+6$, so $\zeta_j\ge2m+6$ (Lemma~\ref{lem:N03-options}(i)).
  Hence $\zeta_j-k-\mu_k=\zeta_j-2m-4+k\ge2+k\ge3$, and
  $\varepsilon_j\ge6\ge2$, the bound required by
  Definition~\ref{def:N02-template}.  If $e=2$, then $\zeta_0<\zeta_1$ gives
  $\varepsilon_0<\varepsilon_1$.
\item $\mu_k=2m+4-2k\ge4\ge3$, as $k\le m$.
\item $\sigma_k=2k\ge2$.  If $e=2$, then $k$ is not the level of a clause,
  by Lemma~\ref{lem:N03-options}(ii), so $k\ge n+1\ge2$ by
  Definition~\ref{def:N03-formula}(a), and $\sigma_k=2k\ge4$.
\end{itemize}

(iii) By Definition~\ref{def:N03-formula}(a), $n\ge1$, so $\varphi$ has a
clause $C_1$.  It is nonempty, and the variable of its first literal is one
of $\xi_1,\dots,\xi_r$, so $r\ge1$.  By Definition~\ref{def:N03-formula},
$m=n+r+1$: one level per clause, one per variable, and the idle level
(Definition~\ref{def:N03-options}).  Hence $m\ge1+1+1=3$.
\end{proof}

\begin{definition}[the graph $G_\varphi$]\label{def:N03-graph}
Let $\varphi$ be a restricted formula, with the notation of
Definitions~\ref{def:N03-options} and~\ref{def:N03-params}.  The graph
$G_\varphi$ (Figure~\ref{fig:N03-shape}) is the union of the following
parts.  Paths
are given by the lists of their vertices in order; their edges join
consecutive vertices of the list.
\begin{itemize}
\item The \emph{frame}: the path $T_1,s_1,s_2,\dots,s_m$ and the path
  $t_m,t_{m-1},\dots,t_1,z_1,z_2,\dots,z_7,T_2$.  The vertices
  $z_1,\dots,z_7$ form the \emph{tail}.
\item The \emph{copies}: for each level $k$ and option $c\in\{1,2,3\}$,
  the \emph{copy of the option $c$ at the level $k$}, a copy $U_{k,c}$ of
  the template $\T(L_k,\mu_k;E_{k,c})$ (Definition~\ref{def:N02-template};
  its hypotheses hold by Lemma~\ref{lem:N03-params}(ii)) whose terminals
  $s$ and $t$ are $s_k$ and $t_k$.  Its vertices are $v^{k,c}_i$
  ($0\le i\le L_k$), $r^{k,c}_i$ ($1\le i\le\mu_k+1$) and $x^{k,c,j}_i$
  ($j<|Z_{k,c}|$, $1\le i\le\ell^{k,c}_j-1$, where
  $\ell^{k,c}_j:=2j+2+\varepsilon$ is the length of the crossing $j$ of the
  template, $\varepsilon$ being the $j$-th entry of $E_{k,c}$
  (Definition~\ref{def:N02-template})), named after the vertices $v_i$,
  $r_i$, $x^j_i$ of the template, with $v^{k,c}_0=s_k$ and $v^{k,c}_{L_k}=t_k$; its edges are
  those of the template, read through these names.
\item The \emph{pair} $q,q'$: the vertices $q$ and $q'$, with the four
  edges $t_mq$, $t_{m-1}q$, $t_mq'$ and $t_{m-1}q'$ (the vertex $t_{m-1}$
  exists, as $m\ge3$ by Lemma~\ref{lem:N03-params}(iii)).
\item The \emph{arms}: for $j\in\{1,2,3,4\}$, the path
  $y^j_0,y^j_1,\dots,y^j_D$, of length $D$, where $y^j_0:=T_1$ and
  $y^j_D:=T_2$, so that it runs from $T_1$ to $T_2$; and a vertex
  $\tilde y^j$ with the edges $\tilde y^jy^j_{\pi_j-1}$ and
  $\tilde y^jy^j_{\pi_j+1}$, where $\pi_1:=\pi_2:=p$ and
  $\pi_3:=\pi_4:=7$.
\end{itemize}
The vertices $T_1$, $T_2$, $s_k$, $t_k$ ($1\le k\le m$), $z_1,\dots,z_7$,
$q$, $q'$, the vertices of the copies other than their terminals, and
$y^j_i$ ($1\le i\le D-1$), $\tilde y^j$ ($1\le j\le4$) are pairwise
distinct.  So a copy $U_{k,c}$ meets the rest of $G_\varphi$ only at $s_k$
and $t_k$, an arm only at $T_1$ and $T_2$, and $q,q'$ are adjacent
to $t_m$, to $t_{m-1}$ and to nothing else.  The three copies of the idle
level $m$ are copies of $\T(L_m,\mu_m;())$, the template without
crossing, as $Z_{m,c}=\emptyset$ gives $E_{m,c}=()$
(Definitions~\ref{def:N03-options} and~\ref{def:N03-params}).
\end{definition}

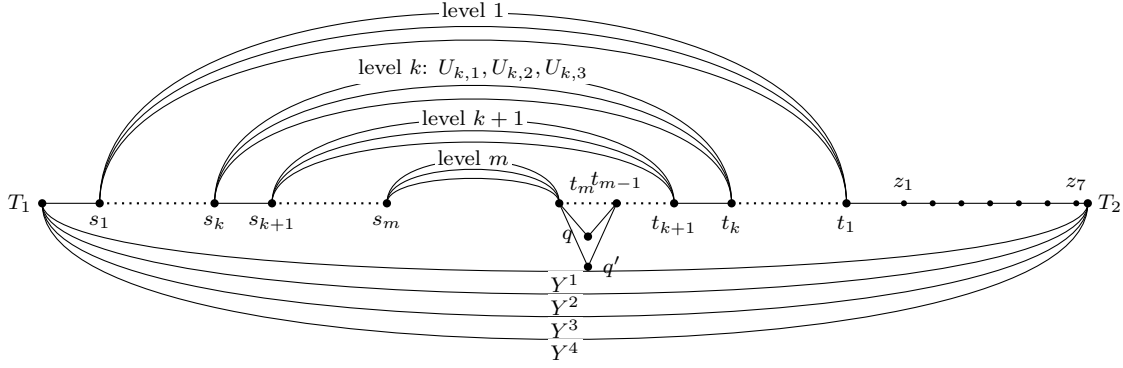
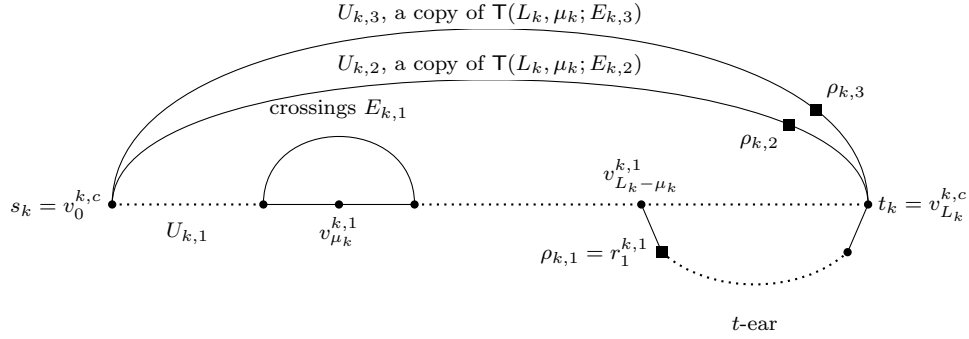
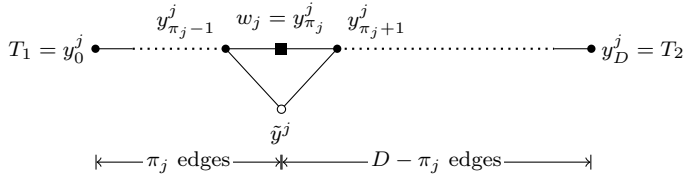
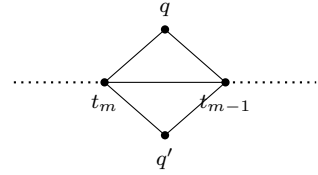
\begin{figure}[!htbp]
\centering
\begin{subfigure}{\textwidth}
\centering
\begin{tikzpicture}[font=\scriptsize,x=0.76cm,y=0.8cm]
\coordinate (T1) at (0,0); \coordinate (T2) at (18.2,0);
\coordinate (s1) at (1,0); \coordinate (sk) at (3,0); \coordinate (sk1) at (4,0); \coordinate (sm) at (6,0);
\coordinate (tm) at (9,0); \coordinate (tk1) at (11,0); \coordinate (tk) at (12,0); \coordinate (t1) at (14,0);
\foreach \i in {1,...,7} \coordinate (z\i) at ({14.5+0.5*\i},0);
\draw (T1) -- (s1); \draw[dotted,thick] (s1) -- (sk); \draw (sk) -- (sk1); \draw[dotted,thick] (sk1) -- (sm);
\draw[dotted,thick] (tm) -- (tk1); \draw (tk1) -- (tk); \draw[dotted,thick] (tk) -- (t1);
\draw (t1) -- (z1) -- (z2) -- (z3) -- (z4) -- (z5) -- (z6) -- (z7) -- (T2);
\foreach \a/\b/\H in {s1/t1/4.2,s1/t1/3.9,s1/t1/3.6, sk/tk/2.9,sk/tk/2.6,sk/tk/2.3, sk1/tk1/1.85,sk1/tk1/1.6,sk1/tk1/1.35, sm/tm/0.95,sm/tm/0.75,sm/tm/0.55}
  \draw (\a) .. controls ($(\a)+(0,\H)$) and ($(\b)+(0,\H)$) .. (\b);
\node[fill=white,inner sep=1pt] at (7.5,3.2) {level $1$};
\node[fill=white,inner sep=1pt] at (7.5,2.2) {level $k$: $U_{k,1},U_{k,2},U_{k,3}$};
\node[fill=white,inner sep=1pt] at (7.5,1.4) {level $k+1$};
\node[fill=white,inner sep=1pt] at (7.5,0.75) {level $m$};
\coordinate (tm1) at (10,0);
\coordinate (q) at (9.5,-0.55); \coordinate (qq) at (9.5,-1.05);
\draw (tm) -- (q) -- (tm1); \draw (tm) -- (qq) -- (tm1);
\node[vx,label=left:$q$] at (q) {}; \node[vx,label=right:$q'$] at (qq) {};
\foreach \j/\D in {1/1.5,2/2.0,3/2.5,4/3.0}
  \draw (T1) .. controls (0,-\D) and (18.2,-\D) .. (T2) node[pos=0.5,below=0pt,fill=white,inner sep=1pt] {$Y^{\j}$};
\foreach \v in {T1,T2,s1,sk,sk1,sm,tm,tm1,tk1,tk,t1} \node[vx] at (\v) {};
\foreach \i in {1,...,7} \node[vx,inner sep=0.8pt] at (z\i) {};
\node[left] at (T1) {$T_1$}; \node[right] at (T2) {$T_2$};
\node[below=1pt] at (s1) {$s_1$}; \node[below=1pt] at (sk) {$s_k$}; \node[below=1pt] at (sk1) {$s_{k+1}$}; \node[below=1pt] at (sm) {$s_m$};
\node[above right=0pt] at (tm) {$t_m$}; \node[above=1pt] at (tm1) {$t_{m-1}$}; \node[below=1pt] at (tk1) {$t_{k+1}$}; \node[below=1pt] at (tk) {$t_k$}; \node[below=1pt] at (t1) {$t_1$};
\node[above=1pt] at (z1) {$z_1$}; \node[above=1pt] at (z7) {$z_7$};
\end{tikzpicture}
\caption{The whole graph.}
\end{subfigure}

\bigskip
\begin{subfigure}{\textwidth}
\centering
\begin{tikzpicture}[font=\scriptsize,x=1cm,y=1cm]
\coordinate (s) at (0,0); \coordinate (t) at (10,0);
\coordinate (a1) at (2,0); \coordinate (mu) at (3,0); \coordinate (b1) at (4,0); \coordinate (Lm) at (7,0);
\draw[dotted,thick] (s) -- (a1); \draw (a1) -- (b1); \draw[dotted,thick] (b1) -- (Lm); \draw[dotted,thick] (Lm) -- (t);
\draw (a1) .. controls (2,1.2) and (4,1.2) .. (b1);
\foreach \k/\a in {1/215,2/240,3/300,4/325} \coordinate (r\k) at ($(8.5,0)+(\a:1.5 and 1.1)$);
\draw (Lm) -- (r1); \draw[dotted,thick] (r1) .. controls (7.9,-1.2) and (9.1,-1.2) .. (r4); \draw (r4) -- (t);
\node[lm] at (r1) {}; \node[left=1pt] at (r1) {$\rho_{k,1}=r^{k,1}_1$};
\node[vx] at (a1) {}; \node[vx] at (b1) {}; \node[vx] at (mu) {}; \node[vx] at (Lm) {}; \node[vx] at (r4) {};
\node[below=1pt] at (mu) {$v^{k,1}_{\mu_k}$};
\node[above=1pt] at (Lm) {$v^{k,1}_{L_k-\mu_k}$};
\node at (3,1.25) {crossings $E_{k,1}$};
\node[below=7pt] at (8.5,-1.1) {$t$-ear};
\foreach \c/\H in {2/2.2,3/3.1} {
  \draw (s) .. controls (0,\H) and (10,\H) .. (t) node[pos=0.5,above=0pt,fill=white,inner sep=1pt] {$U_{k,\c}$, a copy of $\T(L_k,\mu_k;E_{k,\c})$};
}
\path (s) .. controls (0,2.2) and (10,2.2) .. (t) node[lm,pos=0.8] (c2) {};
\node[below left=0pt] at (c2) {$\rho_{k,2}$};
\path (s) .. controls (0,3.1) and (10,3.1) .. (t) node[lm,pos=0.84] (c3) {};
\node[above right=0pt] at (c3) {$\rho_{k,3}$};
\node[vx] at (s) {}; \node[vx] at (t) {};
\node[left] at (s) {$s_k=v^{k,c}_0$}; \node[right] at (t) {$t_k=v^{k,c}_{L_k}$};
\node[below=3pt] at (1,0) {$U_{k,1}$};
\end{tikzpicture}
\caption{One level $k$.}
\end{subfigure}

\bigskip
\begin{subfigure}{0.62\textwidth}
\centering
\begin{tikzpicture}[font=\scriptsize,x=0.82cm,y=1cm]
\coordinate (T1) at (0,0); \coordinate (T2) at (8,0);
\coordinate (ym) at (2.1,0); \coordinate (w) at (3,0); \coordinate (yp) at (3.9,0); \coordinate (tw) at (3,-0.8);
\draw (T1) -- (0.6,0); \draw[dotted,thick] (0.6,0) -- (ym); \draw (ym) -- (w) -- (yp); \draw[dotted,thick] (yp) -- (7.4,0); \draw (7.4,0) -- (T2);
\draw (ym) -- (tw) -- (yp);
\foreach \v in {T1,T2,ym,yp} \node[vx] at (\v) {};
\node[lm] at (w) {}; \node[hv] at (tw) {};
\node[left] at (T1) {$T_1=y^j_0$}; \node[right] at (T2) {$y^j_D=T_2$};
\node[above left=0pt] at (ym) {$y^j_{\pi_j-1}$}; \node[above=2pt] at (w) {$w_j=y^j_{\pi_j}$}; \node[above right=0pt] at (yp) {$y^j_{\pi_j+1}$};
\node[below=1pt] at (tw) {$\tilde y^j$};
\draw[|<->|] (0,-1.5) -- (3,-1.5) node[midway,fill=white,inner sep=1pt] {$\pi_j$ edges};
\draw[|<->|] (3,-1.5) -- (8,-1.5) node[midway,fill=white,inner sep=1pt] {$D-\pi_j$ edges};
\end{tikzpicture}
\caption{An arm $Y^j$.}
\end{subfigure}%
\begin{subfigure}{0.36\textwidth}
\centering
\begin{tikzpicture}[font=\scriptsize,x=1cm,y=1cm]
\coordinate (tm) at (0,0); \coordinate (tm1) at (1.6,0); \coordinate (q) at (0.8,-0.7); \coordinate (qq) at (0.8,0.7);
\draw[dotted,thick] (-1.2,0) -- (tm); \draw[dotted,thick] (tm1) -- (2.8,0); \draw (tm) -- (tm1);
\draw (tm) -- (q) -- (tm1); \draw (tm) -- (qq) -- (tm1);
\node[vx] at (tm) {}; \node[vx] at (tm1) {}; \node[vx] at (q) {}; \node[vx] at (qq) {};
\node[below=1pt] at (tm) {$t_m$}; \node[below=1pt] at (tm1) {$t_{m-1}$}; \node[below=1pt] at (q) {$q'$}; \node[above=1pt] at (qq) {$q$};
\end{tikzpicture}
\caption{The pair $q,q'$.}
\end{subfigure}
\caption{The general shape of $G_\varphi$ (Definition~\ref{def:N03-graph}), not
to scale; dotted lines are paths whose inner vertices are not drawn.
(a)~The frame is the path $T_1,s_1,\dots,s_m$ and the path
$t_m,\dots,t_1,z_1,\dots,z_7,T_2$.  The three copies of each level $k$
join $s_k$ to $t_k$, and the levels are nested.  Each curve stands for a
whole copy, a template from $s_k$ to $t_k$ (panel (b)), not for an edge.  Below the axis are
the pair $q,q'$ across the edge $t_mt_{m-1}$ and the four arms from $T_1$ to $T_2$.  (b)~The copies
$U_{k,1},U_{k,2},U_{k,3}$ of the options $1,2,3$ at the level $k$ share
only the terminals $s_k$ and $t_k$.  Each is a template
$\T(L_k,\mu_k;E_{k,c})$ with its own spine, $t$-ear and crossings, drawn
in full for $c=1$ and as a single curve for $c=2,3$.  The copies of the
idle level $m$ have no crossing.  The common members
of one copy have counterparts, with the same names and profiles, in the
other two (Definition~\ref{def:N03-names}).  The black squares are the
rows $\rho_{k,c}$, the cores of the copies, each on the $t$-ear of its own
copy.  (c)~An arm, a path of length $D$ from $T_1$ to
$T_2$, with the twin $\tilde y^j$ of the row $w_j=y^j_{\pi_j}$; here
$\pi_j=p$ for $j=1,2$ and $\pi_j=7$ for $j=3,4$.  (d)~The pair $q,q'$, each adjacent to $t_m$ and to $t_{m-1}$
and to nothing else; together with the edge $t_mt_{m-1}$ they form three
parallel paths from $t_m$ to $t_{m-1}$.}
\label{fig:N03-shape}
\end{figure}

\begin{definition}[copies in $G_\varphi$]\label{def:N03-names}
Let $\varphi$ be a restricted formula, $G_\varphi$ the graph of
Definition~\ref{def:N03-graph}, $k$ a level and $c$ an option.
\begin{itemize}
\item $U_{k,c}$ also denotes the subgraph of $G_\varphi$ formed by the
  vertices and the edges of the copy of the option $c$ at the level $k$.
  Its \emph{terminals} are $s_k$ and $t_k$, and its \emph{members} are its
  other vertices, that is, the vertices of $\Int(U_{k,c})$
  (Definition~\ref{def:N01-attach}).
\item The \emph{spine members} of $U_{k,c}$ are the $v^{k,c}_i$
  ($1\le i\le L_k-1$), its \emph{$t$-ear members} the $r^{k,c}_i$
  ($1\le i\le\mu_k+1$), and its \emph{crossing members} the
  $x^{k,c,j}_i$ ($j<|Z_{k,c}|$, $1\le i\le\ell^{k,c}_j-1$); spine and $t$-ear
  members are \emph{common}.  Its \emph{core} is $r^{k,c}_1$.
\item For a vertex $u$ of $U_{k,c}$, the profile and the skew of $u$ in
  $U_{k,c}$, with the terminals $s_k$ and $t_k$, are
  (Definition~\ref{def:N01-profile})
  \[
   h_{U_{k,c}}(u)=d_{U_{k,c}}(s_k,u),\qquad
   \hb_{U_{k,c}}(u)=d_{U_{k,c}}(t_k,u),\qquad
   \kappa_{U_{k,c}}(u)=\hb_{U_{k,c}}(u)-h_{U_{k,c}}(u);
  \]
  $u$ is \emph{diagonal} when $\kappa_{U_{k,c}}(u)=\sigma_k$.
\item For a common member $u$ of $U_{k,c}$ and an option $c'$, the vertex
  of $U_{k,c'}$ with the same indices ($v^{k,c'}_i$ for $u=v^{k,c}_i$,
  $r^{k,c'}_i$ for $u=r^{k,c}_i$) is the \emph{counterpart} of $u$ in
  $U_{k,c'}$; for $c'\ne c$ it is a different vertex
  (Definition~\ref{def:N03-graph}).
\end{itemize}
\end{definition}

\begin{definition}[arms, rows, skeleton]\label{def:N03-rows}
Let $\varphi$ be a restricted formula.
\begin{itemize}
\item For $1\le j\le4$, the \emph{arm} $Y^j$ is the subgraph of
  $G_\varphi$ formed by the vertices $y^j_0,\dots,y^j_D$, $\tilde y^j$ and
  the edges of the arm $j$ in Definition~\ref{def:N03-graph}, that is,
  $y^j_iy^j_{i+1}$ ($0\le i<D$), $\tilde y^jy^j_{\pi_j-1}$ and
  $\tilde y^jy^j_{\pi_j+1}$.
\item The \emph{rows} are the cores of the copies and the vertices of the
  arms at the positions $\pi_j$:
  \[
   \rho_{k,c}:=r^{k,c}_1\quad(1\le k\le m,\ c\in\{1,2,3\}),\qquad
   w_j:=y^j_{\pi_j}\quad(1\le j\le4).
  \]
\item The \emph{skeleton} is
  \[
   S:=\{T_1,T_2\}\cup\{s_k,t_k:1\le k\le m\}\cup\{z_1,\dots,z_7\}\cup\{q'\}.
  \]
\end{itemize}
\end{definition}

The skeleton consists of the vertices of $G_\varphi$ that are neither
members of a copy, nor in the interior of an arm, nor $q$: by
Lemma~\ref{lem:N03-structure}(i), $V(G_\varphi)$ is the disjoint union of
$S$, $\{q\}$, the interiors of the copies and the interiors of the arms.
The vertex $q$ is left out because it belongs to every pattern
(Definition~\ref{def:N03-pattern}).

\begin{lemma}[copies are templates]\label{lem:N03-copies}
Let $\varphi$ be a restricted formula, $k$ a level and $c$ an option, and
write $\T:=\T(L_k,\mu_k;E_{k,c})$.  Let $\psi$ map the vertices of $\T$ to
those of $U_{k,c}$ by
\[
\begin{aligned}
 \psi(v_i)&:=v^{k,c}_i &&(0\le i\le L_k),\\
 \psi(r_i)&:=r^{k,c}_i &&(1\le i\le\mu_k+1),\\
 \psi(x^j_i)&:=x^{k,c,j}_i &&(j<|Z_{k,c}|,\ 1\le i\le\ell^{k,c}_j-1);
\end{aligned}
\]
here $\ell^{k,c}_j$ is the length $\ell_j$ of the crossing $j$ of $\T$
(Definition~\ref{def:N03-graph}); in particular $\psi(s)=s_k$ and
$\psi(t)=t_k$.
\begin{enumerate}
\item[\rm(i)] $\psi$ is an isomorphism from $\T$ onto $U_{k,c}$.
\item[\rm(ii)] For all vertices $u,u'$ of $\T$,
  $d_{U_{k,c}}(\psi(u),\psi(u'))=d_\T(u,u')$.  In particular, for every
  vertex $u$ of $\T$, $h_{U_{k,c}}(\psi(u))=h_\T(u)$,
  $\hb_{U_{k,c}}(\psi(u))=\hb_\T(u)$ and
  $\kappa_{U_{k,c}}(\psi(u))=\kappa_\T(u)$.
\item[\rm(iii)] $\psi$ maps the spine, $t$-ear, crossing and common
  members, the core and the diagonal members of $\T$ onto those of
  $U_{k,c}$.
\end{enumerate}
\end{lemma}

\begin{proof}
(i) By Definition~\ref{def:N03-graph}, the vertices of the copy are the
images under $\psi$ of the vertices of $\T$, and they are pairwise
distinct.  Its edges are the edges of $\T$ read through these names, that
is, the pairs $\psi(u)\psi(u')$ for the edges $uu'$ of $\T$.  By
Definition~\ref{def:N03-names}, $U_{k,c}$ has exactly these vertices and
edges.  So $\psi$ is a bijection from $V(\T)$ onto $V(U_{k,c})$, and
$uu'$ is an edge of $\T$ iff $\psi(u)\psi(u')$ is an edge of $U_{k,c}$.

(ii) An isomorphism maps every path to a path with the same number of
edges, and so does its inverse (as in the proof of
Lemma~\ref{lem:N01-twins}).  Hence
$d_{U_{k,c}}(\psi(u),\psi(u'))=d_\T(u,u')$ for all vertices $u,u'$ of
$\T$.  This is the first claim.  With $u':=s$ and $\psi(s)=s_k$,
\[
 h_{U_{k,c}}(\psi(u))=d_{U_{k,c}}(s_k,\psi(u))=d_\T(s,u)=h_\T(u)
\]
(Definitions~\ref{def:N01-profile} and~\ref{def:N02-diagonal}).  Likewise
$\hb_{U_{k,c}}(\psi(u))=\hb_\T(u)$, taking $t$ and $t_k$ in place of $s$
and $s_k$; and $\kappa$ is the difference of the two.

(iii) The kinds of members are given by the same indices in
Definitions~\ref{def:N02-template} and~\ref{def:N03-names}, and $\psi$
keeps the indices; so $\psi$ maps each kind onto the same kind, and $r_1$
onto $r^{k,c}_1$.  By Lemma~\ref{lem:N03-params}(i), the number $\sigma$ of
$\T$ is $L_k-2\mu_k=\sigma_k$.  So a member $u$ of $\T$ is diagonal iff
$\kappa_\T(u)=\sigma_k$, iff $\kappa_{U_{k,c}}(\psi(u))=\sigma_k$ by (ii),
iff $\psi(u)$ is diagonal (Definition~\ref{def:N03-names}).
\end{proof}

\begin{lemma}[inside a copy]\label{lem:N03-inside}
Let $\varphi$ be a restricted formula, $k$ a level and $c$ an option.  Let
$\alpha$, $L_k$, $\beta_k$, $\mu_k$ and $\sigma_k$ be the numbers of
Definition~\ref{def:N03-params}, and $\rho_{k,c}=r^{k,c}_1$ the core of
$U_{k,c}$ (Definition~\ref{def:N03-rows}).  Profiles, skews, and the
spine, $t$-ear, crossing, common and diagonal members of $U_{k,c}$ are
those of Definition~\ref{def:N03-names}.
\begin{enumerate}
\item[\rm(i)] $U_{k,c}$ is connected, and $d_{U_{k,c}}(s_k,t_k)=L_k$.
\item[\rm(ii)] $(h_{U_{k,c}},\hb_{U_{k,c}})(\rho_{k,c})=(\alpha,\beta_k)$.
  For every member $u\ne\rho_{k,c}$ of $U_{k,c}$,
  \[
   d_{U_{k,c}}(\rho_{k,c},u)\ <\ \min\bigl(\alpha+h_{U_{k,c}}(u),\
   \beta_k+\hb_{U_{k,c}}(u)\bigr).
  \]
\item[\rm(iii)] Two distinct members of $U_{k,c}$ have distinct profiles
  in $U_{k,c}$.
\item[\rm(iv)] $h_{U_{k,c}}(v^{k,c}_i)=i$ and $\hb_{U_{k,c}}(v^{k,c}_i)=L_k-i$
  for $0\le i\le L_k$; $h_{U_{k,c}}(r^{k,c}_i)=L_k-\mu_k+i$ and
  $\hb_{U_{k,c}}(r^{k,c}_i)=\mu_k+2-i$ for $1\le i\le\mu_k+1$; and every
  member $u$ of $U_{k,c}$ with $\kappa_{U_{k,c}}(u)>\sigma_k+4$ is a spine
  member.
\item[\rm(v)] Every diagonal member $u$ of $U_{k,c}$ has
  $h_{U_{k,c}}(u)\ge\mu_k$.
\item[\rm(vi)] Let $w\ne\rho_{k,c}$ be a member of $U_{k,c}$.  Put
  $x:=r^{k,c}_{\mu_k+1}$ if $w$ is a spine or crossing member, and
  $x:=v^{k,c}_{L_k-1}$ if $w$ is a $t$-ear member.  Then $x$ is a common
  member of $U_{k,c}$, $x\ne w$, and
  \[
   d_{U_{k,c}}(w,x)=\min\bigl(h_{U_{k,c}}(w)+h_{U_{k,c}}(x),\
   \hb_{U_{k,c}}(w)+\hb_{U_{k,c}}(x)\bigr).
  \]
\end{enumerate}
\end{lemma}

\begin{proof}
Let $U:=U_{k,c}$.  Let $\T:=\T(L_k,\mu_k;E_{k,c})$ and $\psi$ be as in
Lemma~\ref{lem:N03-copies}.  By Lemma~\ref{lem:N03-params}(ii), the triple
$(L_k,\mu_k,E_{k,c})$ satisfies the hypotheses of
Definition~\ref{def:N02-template}, so $\T$ is a template and the lemmas of
Section~\ref{sec:N02} apply to it.  By Lemma~\ref{lem:N03-copies},
$\psi(s)=s_k$ and $\psi(t)=t_k$, and by Lemma~\ref{lem:N03-copies}(i),
$\psi$ is an isomorphism from $\T$ onto $U$.

(1) We prove (i).  By Lemma~\ref{lem:N02-ears}(iii), $\T$ is connected.  As
$\psi$ is an isomorphism from $\T$ onto $U$, $U$ is connected.  By
Lemma~\ref{lem:N03-copies}(ii), and by Lemma~\ref{lem:N02-spine} for $\T$,
\[
 d_U(s_k,t_k)=d_U(\psi(s),\psi(t))=d_\T(s,t)=L_k.
\]

(2) We prove (ii).  We take Lemma~\ref{lem:N02-core} in $\T$, that is, with
$L=L_k$, $\mu=\mu_k$ and $E=E_{k,c}$.  Lemma~\ref{lem:N02-core} puts
$\alpha_\T:=L-\mu+1$ and $\beta_\T:=\mu+1$, so here
\[
 \alpha_\T=L_k-\mu_k+1=\alpha,\qquad \beta_\T=\mu_k+1=\beta_k,
\]
where the second equalities are Lemma~\ref{lem:N03-params}(i).  So
Lemma~\ref{lem:N02-core} gives $(h_\T,\hb_\T)(r_1)=(\alpha,\beta_k)$, and,
for every member $u_0\ne r_1$ of $\T$,
\[
 d_\T(r_1,u_0)<\min\bigl(\alpha+h_\T(u_0),\ \beta_k+\hb_\T(u_0)\bigr).
\]
By Lemma~\ref{lem:N03-copies} and Definition~\ref{def:N03-rows},
$\psi(r_1)=r^{k,c}_1=\rho_{k,c}$.  So, by Lemma~\ref{lem:N03-copies}(ii) at
$r_1$,
\[
 (h_U,\hb_U)(\rho_{k,c})=(h_U,\hb_U)(\psi(r_1))=(h_\T,\hb_\T)(r_1)
 =(\alpha,\beta_k).
\]
Now let $u\ne\rho_{k,c}$ be a member of $U$.  The members of $\T$ are its
vertices other than $s$ and $t$ (Definition~\ref{def:N02-template}).  The
members of $U$ are its vertices other than $s_k$ and $t_k$
(Definition~\ref{def:N03-names}).  The map $\psi$ is a bijection with
$\psi(s)=s_k$, $\psi(t)=t_k$ and $\psi(r_1)=\rho_{k,c}$, so $u=\psi(u_0)$
for a member $u_0\ne r_1$ of $\T$.  By Lemma~\ref{lem:N03-copies}(ii),
\[
\begin{aligned}
 d_U(\rho_{k,c},u)&=d_U(\psi(r_1),\psi(u_0))=d_\T(r_1,u_0)\\
 &<\min\bigl(\alpha+h_\T(u_0),\ \beta_k+\hb_\T(u_0)\bigr)
 =\min\bigl(\alpha+h_U(u),\ \beta_k+\hb_U(u)\bigr).
\end{aligned}
\]
This proves (ii).

(3) We prove (iii).  Let $u\ne u'$ be members of $U$.  As in (2),
$u=\psi(u_0)$ and $u'=\psi(u_0')$ for members $u_0,u_0'$ of $\T$, and
$u_0\ne u_0'$ as $\psi$ is injective.  By Lemma~\ref{lem:N02-distinct} in
$\T$, $(h_\T,\hb_\T)(u_0)\ne(h_\T,\hb_\T)(u_0')$.  By
Lemma~\ref{lem:N03-copies}(ii),
\[
 (h_U,\hb_U)(u)=(h_\T,\hb_\T)(u_0)\ne(h_\T,\hb_\T)(u_0')=(h_U,\hb_U)(u').
\]

(4) We prove (iv).  Let $0\le i\le L_k$.  By Lemma~\ref{lem:N03-copies},
$v^{k,c}_i=\psi(v_i)$, so by Lemma~\ref{lem:N03-copies}(ii) and
Lemma~\ref{lem:N02-profiles}(i) in $\T$,
\[
 h_U(v^{k,c}_i)=h_\T(v_i)=i,\qquad \hb_U(v^{k,c}_i)=\hb_\T(v_i)=L_k-i.
\]
Let $1\le i\le\mu_k+1$.  By Lemma~\ref{lem:N03-copies},
$r^{k,c}_i=\psi(r_i)$, so by Lemma~\ref{lem:N03-copies}(ii) and
Lemma~\ref{lem:N02-profiles}(ii) in $\T$,
\[
 h_U(r^{k,c}_i)=h_\T(r_i)=L_k-\mu_k+i,\qquad
 \hb_U(r^{k,c}_i)=\hb_\T(r_i)=\mu_k+2-i.
\]
Now let $u$ be a member of $U$ with $\kappa_U(u)>\sigma_k+4$, and
$u=\psi(u_0)$ with $u_0$ a member of $\T$, as in (2).  By
Lemma~\ref{lem:N03-copies}(ii), $\kappa_\T(u_0)=\kappa_U(u)>\sigma_k+4$.
By Lemma~\ref{lem:N03-params}(ii), the number $\sigma=L-2\mu$ of $\T$ is
$\sigma_k$, so Lemma~\ref{lem:N02-skew}(iv) in $\T$ shows that $u_0$ is a
spine member of $\T$.  By Lemma~\ref{lem:N03-copies}(iii), $u=\psi(u_0)$
is a spine member of $U$.

(5) We prove (v).  Let $u$ be a diagonal member of $U$.  By
Lemma~\ref{lem:N03-copies}(iii), $u=\psi(u_0)$ for a diagonal member $u_0$
of $\T$.  By Lemma~\ref{lem:N02-skew}(iii) in $\T$, the diagonal members
of $\T$ are $v_{\mu_k}$, with $h_\T(v_{\mu_k})=\mu_k$, and, for each
crossing $j$ of $\T$, one member of the crossing $j$ with
$h_\T=\mu_k+\varepsilon_j/2\ge\mu_k$, as $\varepsilon_j\ge2$
(Definition~\ref{def:N02-template}).  So $h_\T(u_0)\ge\mu_k$, and by
Lemma~\ref{lem:N03-copies}(ii), $h_U(u)=h_\T(u_0)\ge\mu_k$.

(6) We prove (vi).  Let $w\ne\rho_{k,c}$ be a member of $U$.  The
members of $\T$ and of $U$ are their vertices other than the terminals
(Definitions~\ref{def:N02-template} and~\ref{def:N03-names}), and $\psi$
is a bijection (Lemma~\ref{lem:N03-copies}(i)) with $\psi(s)=s_k$,
$\psi(t)=t_k$ and $\psi(r_1)=r^{k,c}_1=\rho_{k,c}$
(Lemma~\ref{lem:N03-copies}, Definition~\ref{def:N03-rows}).  So
$w=\psi(w_0)$ for a member $w_0\ne r_1$ of $\T$.  By Lemma~\ref{lem:N03-copies}(iii), $\psi$ maps the spine,
crossing and $t$-ear members of $\T$ onto those of $U$, so $w_0$ is a
member of the same kind as $w$.  Put $x_0:=r_{\mu_k+1}$ if $w_0$ is a
spine or crossing member, and $x_0:=v_{L_k-1}$ if $w_0$ is a $t$-ear
member.  By Lemma~\ref{lem:N03-copies}, $\psi(r_{\mu_k+1})=r^{k,c}_{\mu_k+1}$
and $\psi(v_{L_k-1})=v^{k,c}_{L_k-1}$, so $\psi(x_0)=x$.
Corollary~\ref{cor:N02-noncore} in $\T$, that is, with $L=L_k$ and
$\mu=\mu_k$, applied to $w_0$, shows that $x_0$ is a common member of
$\T$, $x_0\ne w_0$, and
\[
 d_\T(w_0,x_0)=\min\bigl(h_\T(w_0)+h_\T(x_0),\ \hb_\T(w_0)+\hb_\T(x_0)\bigr).
\]
By Lemma~\ref{lem:N03-copies}(iii), $x=\psi(x_0)$ is a common member of
$U$, and $x\ne w$ as $\psi$ is injective.  By
Lemma~\ref{lem:N03-copies}(ii), at $w_0,x_0$ and for the profiles,
\[
\begin{aligned}
 d_U(w,x)&=d_\T(w_0,x_0)
 =\min\bigl(h_\T(w_0)+h_\T(x_0),\ \hb_\T(w_0)+\hb_\T(x_0)\bigr)\\
 &=\min\bigl(h_U(w)+h_U(x),\ \hb_U(w)+\hb_U(x)\bigr).
\end{aligned}
\]
\end{proof}

The parts of $G_\varphi$ play the following roles.  A path from $T_1$ to
$T_2$ through a level $k$, from $T_1$ along $s_1,\dots,s_k$, then along the
spine of a copy of the level $k$, then along $t_k,\dots,t_1,z_1,\dots,z_7$
to $T_2$, has
\[
 k+L_k+(k-1)+8=2k+7+L_k=D
\]
edges, as $L_k=D-7-2k$ (Lemma~\ref{lem:N03-params}(i)); so it is as long
as an arm.  Lemma~\ref{lem:N04-potentials} gives the distances $\delta_1$
and $\delta_2$ from $T_1$ and $T_2$ to every vertex.  The pairs
$\{q,q'\}$ and $\{w_j,\tilde y^j\}$ are twins
(Lemma~\ref{lem:N03-structure}(iii)), so every resolving set contains a
vertex of each pair (Lemma~\ref{lem:N05-lower}(ii)); a pattern takes $q$
and $w_1,\dots,w_4$ (Definition~\ref{def:N03-pattern}).  The row $w_j$
sees a vertex $u$ outside its arm only through $\delta_1(u)$ and
$\delta_2(u)$ (Lemma~\ref{lem:N04-rows}(iii)).  The two positions of the
rows of the arms make them read a member $u$ of a copy differently:
$w_3$ and $w_4$, at the position $7$, read it at $7+\delta_1(u)$, and $w_1$
and $w_2$, at the position $p$, at $\min(p+\delta_1(u),\ p+7+\delta_2(u))$
(Corollary~\ref{cor:N04-levels}(iii),(iv)).  The rows $\rho_{k,c}$ are the
landmarks that patterns place inside the copies, and
Lemma~\ref{lem:N05-forcing} shows that a resolving set of $2m+5$ vertices
has no vertex inside a copy other than its core.

\begin{definition}[choices and patterns]\label{def:N03-pattern}
Let $\varphi$ be a restricted formula.  A \emph{choice} is
$f=(f_1,\dots,f_m)\in\{1,2,3\}^m$; it selects the option $f_k$ at each
level $k$ (Definition~\ref{def:N03-options}).  Its \emph{pattern} is
\[
 W_f:=\bigcup_{k=1}^{m}\{\rho_{k,c}:c\ne f_k\}\ \cup\ \{q,w_1,w_2,w_3,w_4\},
\]
a set of $2m+5$ vertices; the rows $\rho_{k,c}$ and $w_j$ are defined in
Definition~\ref{def:N03-rows}.  The copy
$U_{k,f_k}$ of the selected option is the \emph{free copy} of level $k$,
and the two others are its \emph{rowed copies}.
\end{definition}

\begin{lemma}[structure of $G_\varphi$]\label{lem:N03-structure}
Let $\varphi$ be a restricted formula and $G:=G_\varphi$
(Definitions~\ref{def:N03-graph}, \ref{def:N03-names}
and~\ref{def:N03-rows}).
\begin{enumerate}
\item[\rm(i)] $V(G)$ is the disjoint union of $S$, $\{q\}$, the sets
  $\Int(U_{k,c})$ of members of the copies, and the sets
  $\Int(Y^j)=\{y^j_1,\dots,y^j_{D-1},\tilde y^j\}$.
\item[\rm(ii)] Each $U_{k,c}$ is attached at $\{s_k,t_k\}$ and each $Y^j$
  at $\{T_1,T_2\}$ (Definition~\ref{def:N01-attach}).
\item[\rm(iii)] The five pairs $\{q,q'\}$ and $\{w_j,\tilde y^j\}$
  ($1\le j\le4$) are twins (Definition~\ref{def:N01-twins}); no vertex
  lies in two of these pairs, and no vertex of a pair lies in
  $\Int(U_{k,c})$ for any level $k$ and option $c$.
\item[\rm(iv)] $|V(G)|\le78m^2+60m+71$.
\end{enumerate}
\end{lemma}

\begin{proof}
(i) By Definition~\ref{def:N03-graph}, the vertices of $G$ are $T_1$,
$T_2$, the $s_k$ and $t_k$, $z_1,\dots,z_7$, $q$, $q'$, the vertices of
the copies other than their terminals, and the $y^j_i$
($1\le i\le D-1$) and $\tilde y^j$; and these are pairwise distinct.  The
vertices of the copies other than their terminals are the members, that
is, the vertices of the sets $\Int(U_{k,c})$
(Definition~\ref{def:N03-names}).  The $y^j_i$ ($1\le i\le D-1$) and
$\tilde y^j$ are the vertices of
$\Int(Y^j)=V(Y^j)\setminus\{T_1,T_2\}$ (Definitions~\ref{def:N01-attach}
and~\ref{def:N03-rows}).  The remaining vertices other than $q$, namely
$T_1$, $T_2$, the $s_k$, $t_k$, $z_i$ and $q'$, form $S$.

(ii) Every edge of $G$ is an edge of the frame, of a copy, an edge at
$q$ or $q'$, or an edge of an arm (Definition~\ref{def:N03-graph}).  A member of $U_{k,c}$ is a
vertex of no other part, so every edge at it is an edge of $U_{k,c}$;
likewise every edge at a vertex of $\Int(Y^j)$ is an edge of $Y^j$.  This
is Definition~\ref{def:N01-attach}, with $s_k\ne t_k$ and $T_1\ne T_2$.

(iii) By Definition~\ref{def:N03-graph}, $N(q)=N(q')=\{t_m,t_{m-1}\}$ and
$q,q'$ are not adjacent, so
$N(q)\setminus\{q'\}=\{t_m,t_{m-1}\}=N(q')\setminus\{q\}$.  Next, $\pi_j\in\{p,7\}$,
and $7\le2m+4=p\le4m+13=D-2$ as $m\ge3$ (Lemma~\ref{lem:N03-params}(iii)),
so $2\le\pi_j\le D-2$.  Hence
$w_j=y^j_{\pi_j}$ lies on exactly the edges $y^j_{\pi_j-1}y^j_{\pi_j}$ and
$y^j_{\pi_j}y^j_{\pi_j+1}$ (the edges of $\tilde y^j$ go to
$y^j_{\pi_j\pm1}$), so $N(w_j)=\{y^j_{\pi_j-1},y^j_{\pi_j+1}\}=N(\tilde y^j)$.
As $w_j$ and $\tilde y^j$ are not adjacent,
$N(w_j)\setminus\{\tilde y^j\}=N(\tilde y^j)\setminus\{w_j\}$.
The last two claims follow from (i): $q,q'\in S\cup\{q\}$ and
$w_j,\tilde y^j\in\Int(Y^j)$.

(iv) We count the five sets of (i).  $|S|=2+2m+7+1=2m+10$, so
$|S|+|\{q\}|=2m+11$.  Each $\Int(Y^j)$ has $(D-1)+1=D$ vertices, so the
arms contribute $4D=16m+60$.  A copy $U_{k,c}$ has
$(L_k-1)+(\mu_k+1)+\sum_j(\ell^{k,c}_j-1)$ members, where
$\ell^{k,c}_j-1=2j+1+2(\zeta_j-k-\mu_k)$, $\zeta_0<\cdots$ being the
cells of $Z_{k,c}$ (Definitions~\ref{def:N03-params}
and~\ref{def:N03-graph}).  As $\zeta_j\le5m+5$
(Lemma~\ref{lem:N03-options}(i)) and $k\le m$,
\[
 \zeta_j-k-\mu_k=\zeta_j-2m-4+k\le(5m+5)-2m-4+m=4m+1\le5m,
\]
so $\ell^{k,c}_j-1\le2\cdot1+1+10m=10m+3$, as $j\le1$.  With
$L_k-1\le L_1-1=4m+5$, $\mu_k+1\le\mu_1+1=2m+3$ and at most two crossings
(Lemma~\ref{lem:N03-options}(ii)), a copy has at most
$(4m+5)+(2m+3)+2(10m+3)=26m+14$ members, and the $3m$ copies at most
$3m(26m+14)=78m^2+42m$.  The sum is at most
$(2m+11)+(16m+60)+(78m^2+42m)=78m^2+60m+71$.
\end{proof}


\section{Distances in \texorpdfstring{$G_\varphi$}{G\_phi}}\label{sec:N04}

\begin{lemma}[distances to $T_1$ and $T_2$]\label{lem:N04-potentials}
Let $\varphi$ be a restricted formula and $G:=G_\varphi$
(Definition~\ref{def:N03-graph}), and let
\[
 \delta_1(u):=d(T_1,u),\qquad \delta_2(u):=d(T_2,u)\qquad(u\in V(G)).
\]
For a member $u$ of a copy $U_{k,c}$,
$(h_{U_{k,c}}(u),\hb_{U_{k,c}}(u))=(d_{U_{k,c}}(s_k,u),\,d_{U_{k,c}}(t_k,u))$
is its profile (Definition~\ref{def:N03-names}).  Let
$F_1,F_2:V(G)\to\mathbb Z$ be given by the table
\[
\begin{array}{lll}
\text{vertex} & F_1 & F_2\\\hline
T_1 & 0 & D\\
T_2 & D & 0\\
s_k\ (1\le k\le m) & k & D-k\\
t_k\ (1\le k\le m) & D-7-k & 7+k\\
z_i\ (1\le i\le7) & D-8+i & 8-i\\
q,\ q' & D-6-m & 7+m\\
u\in\Int(U_{k,c}) & k+h_{U_{k,c}}(u) & 7+k+\hb_{U_{k,c}}(u)\\
y^j_i\ (1\le i\le D-1) & i & D-i\\
\tilde y^j & \pi_j & D-\pi_j
\end{array}
\tag{T1}\label{eq:N04-values}
\]
Then $G$ is connected, $\delta_1=F_1$ and $\delta_2=F_2$.
\end{lemma}

\begin{proof}
By Definition~\ref{def:N03-graph}, the vertices of $G$ are $T_1$, $T_2$,
the $s_k$, $t_k$ and $z_i$, $q$, $q'$, the members of the copies, the
$y^j_i$ ($1\le i\le D-1$) and the $\tilde y^j$, and these are pairwise
distinct; the members of $U_{k,c}$ form $\Int(U_{k,c})$
(Definition~\ref{def:N03-names}).  So the table \eqref{eq:N04-values} lists every vertex of $G$
exactly once, and $F_1,F_2$ are well defined.  We check the conditions
(a)--(c) of Lemma~\ref{lem:N01-potential} for $F_1$ with $o:=T_1$, and for
$F_2$ with $o:=T_2$.  These conditions involve only the values of $F_1$, $F_2$ at a
vertex and at its neighbours, so steps (1)--(5) use no distance of $G$.  Distances of $G$ appear only
in (6), as the conclusion of Lemma~\ref{lem:N01-potential}.

(1) \emph{Condition (a).}  By the table \eqref{eq:N04-values},
$F_1(T_1)=0$ and $F_2(T_2)=0$.

(2) \emph{The formulas of a copy.}  Let $U:=U_{k,c}$.  For a member $u$ of
$U$, that is, $u\in\Int(U)$, the table \eqref{eq:N04-values} gives $F_1(u)=k+h_U(u)$ and
$F_2(u)=7+k+\hb_U(u)$.  By Definition~\ref{def:N01-profile}, applied to $U$
with the terminals $s_k$ and $t_k$ (Definition~\ref{def:N03-names}),
$h_U(u)=d_U(s_k,u)$ and $\hb_U(u)=d_U(t_k,u)$, which are distances in $U$,
not in $G$.

At the terminals, the table \eqref{eq:N04-values} gives other formulas, with the same values.
By Lemmas~\ref{lem:N02-profiles}(i) and~\ref{lem:N03-copies}(ii), taken at
$v^{k,c}_0=s_k$ and at $v^{k,c}_{L_k}=t_k$,
\[
 (h_U,\hb_U)(s_k)=(0,L_k),\qquad (h_U,\hb_U)(t_k)=(L_k,0).
\]
With $L_k=D-7-2k$ (Lemma~\ref{lem:N03-params}(i)) and the values of $F_1$
and $F_2$ at $s_k$ and $t_k$ given in the table \eqref{eq:N04-values},
\[
\begin{aligned}
 k+h_U(s_k)&=k=F_1(s_k), & 7+k+\hb_U(s_k)&=7+k+L_k=D-k=F_2(s_k),\\
 k+h_U(t_k)&=k+L_k=D-7-k=F_1(t_k), & 7+k+\hb_U(t_k)&=7+k=F_2(t_k).
\end{aligned}
\]
Hence, for every vertex $u$ of $U$,
\[
 F_1(u)=k+h_U(u)=k+d_U(s_k,u),\qquad F_2(u)=7+k+\hb_U(u)=7+k+d_U(t_k,u).
 \tag{$*$}
\]

(3) \emph{Condition (b): along every edge, $F_1$ and $F_2$ change by at
most one.}  By Definition~\ref{def:N03-graph}, every edge of $G$ is an edge
of the frame, of a copy, an edge at $q$ or $q'$, or of an arm.

\emph{The edges outside the copies.}  In the table \eqref{eq:N04-edges} below, the first six
lines are the edges of the frame, the next two the four edges at $q$
and $q'$, and the last five the edges of the arms; by
Definition~\ref{def:N03-graph}, these are all the edges outside the
copies.  For an edge $xy$, the table \eqref{eq:N04-edges} gives the values at $x$ and at $y$,
then the differences $F_1(y)-F_1(x)$ and $F_2(y)-F_2(x)$:
\[
\begin{array}{llll}
 \text{edge } xy & F_1(x),\ F_1(y) & F_2(x),\ F_2(y)
   & \text{differences}\\\hline
 T_1s_1 & 0,\ 1 & D,\ D-1 & 1,\ -1\\
 s_ks_{k+1}\ (1\le k<m) & k,\ k+1 & D-k,\ D-k-1 & 1,\ -1\\
 t_{k+1}t_k\ (1\le k<m) & D-8-k,\ D-7-k & 8+k,\ 7+k & 1,\ -1\\
 t_1z_1 & D-8,\ D-7 & 8,\ 7 & 1,\ -1\\
 z_iz_{i+1}\ (1\le i\le6) & D-8+i,\ D-7+i & 8-i,\ 7-i & 1,\ -1\\
 z_7T_2 & D-1,\ D & 1,\ 0 & 1,\ -1\\
 t_mq,\ t_mq' & D-7-m,\ D-6-m & 7+m,\ 7+m & 1,\ 0\\
 t_{m-1}q,\ t_{m-1}q' & D-6-m,\ D-6-m & 6+m,\ 7+m & 0,\ 1\\
 T_1y^j_1 & 0,\ 1 & D,\ D-1 & 1,\ -1\\
 y^j_iy^j_{i+1}\ (1\le i\le D-2) & i,\ i+1 & D-i,\ D-i-1 & 1,\ -1\\
 y^j_{D-1}T_2 & D-1,\ D & 1,\ 0 & 1,\ -1\\
 \tilde y^jy^j_{\pi_j-1} & \pi_j,\ \pi_j-1 & D-\pi_j,\ D-\pi_j+1 & -1,\ 1\\
 \tilde y^jy^j_{\pi_j+1} & \pi_j,\ \pi_j+1 & D-\pi_j,\ D-\pi_j-1 & 1,\ -1
\end{array}
\tag{T2}\label{eq:N04-edges}
\]
Each value is the formula of the table \eqref{eq:N04-values} at that vertex.  At
the vertices with a shifted index:
\[
\begin{aligned}
 F_1(s_{k+1})&=k+1, & F_2(s_{k+1})&=D-(k+1)=D-k-1,\\
 F_1(t_{k+1})&=D-7-(k+1)=D-8-k, & F_2(t_{k+1})&=7+(k+1)=8+k,\\
 F_1(t_1)&=D-7-1=D-8, & F_2(t_1)&=7+1=8,\\
 F_1(z_1)&=D-8+1=D-7, & F_2(z_1)&=8-1=7,\\
 F_1(z_{i+1})&=D-8+(i+1)=D-7+i, & F_2(z_{i+1})&=8-(i+1)=7-i,\\
 F_1(z_7)&=D-8+7=D-1, & F_2(z_7)&=8-7=1,\\
 F_1(y^j_1)&=1, & F_2(y^j_1)&=D-1,\\
 F_1(y^j_{i+1})&=i+1, & F_2(y^j_{i+1})&=D-(i+1)=D-i-1,\\
 F_1(y^j_{D-1})&=D-1, & F_2(y^j_{D-1})&=D-(D-1)=1,\\
 F_1(y^j_{\pi_j\pm1})&=\pi_j\pm1, & F_2(y^j_{\pi_j\pm1})&=D-(\pi_j\pm1)=D-\pi_j\mp1.
\end{aligned}
\]
For the arms, the values $F_1(y^j_{\pi_j\pm1})=\pi_j\pm1$ and
$F_2(y^j_{\pi_j\pm1})=D-\pi_j\mp1$ need $1\le\pi_j\pm1\le D-1$.  This holds
as $\pi_j\in\{p,7\}$ (Definition~\ref{def:N03-graph}), and $2\le7\le D-2$
and $2\le p=2m+4\le4m+13=D-2$ (Definition~\ref{def:N03-params}).  In the
last column, every difference is $-1$, $0$ or $1$.

\emph{The edges of a copy.}  Let $uu'$ be an edge of a copy $U:=U_{k,c}$.
By ($*$) and the triangle inequality in $U$
(Definition~\ref{def:N01-graph}),
\[
 F_1(u)-k=d_U(s_k,u)\le d_U(s_k,u')+d_U(u',u)=(F_1(u')-k)+1,
\]
so $F_1(u)-F_1(u')\le1$; exchanging $u$ and $u'$ gives
$F_1(u')-F_1(u)\le1$.  Likewise for $F_2$, taking $t_k$ and $7+k$ in place
of $s_k$ and $k$.

(4) \emph{Condition (c) for $F_1$: every vertex $x\ne T_1$ has a neighbour
$x'$ with $F_1(x')=F_1(x)-1$.}  By Definition~\ref{def:N03-graph}, every
vertex $x\ne T_1$ of $G$ either is listed in the first column of the table
\eqref{eq:N04-down1} below, or is $t_m$, or is a member of a copy.  We treat the three cases in
turn.

\emph{The vertices of the table.}  For each $x$, the table \eqref{eq:N04-down1} gives
$x'$.  The pair $x,x'$ is an edge in the table \eqref{eq:N04-edges}, and the values are the ones
computed there:
\[
\begin{array}{lll}
 x & x' & F_1(x),\ F_1(x')\\\hline
 s_1 & T_1 & 1,\ 0\\
 s_{k+1}\ (1\le k<m) & s_k & k+1,\ k\\
 t_k\ (1\le k<m) & t_{k+1} & D-7-k,\ D-8-k\\
 z_1 & t_1 & D-7,\ D-8\\
 z_{i+1}\ (1\le i\le6) & z_i & D-7+i,\ D-8+i\\
 T_2 & z_7 & D,\ D-1\\
 q,\ q' & t_m & D-6-m,\ D-7-m\\
 y^j_1 & T_1 & 1,\ 0\\
 y^j_{i+1}\ (1\le i\le D-2) & y^j_i & i+1,\ i\\
 \tilde y^j & y^j_{\pi_j-1} & \pi_j,\ \pi_j-1
\end{array}
\tag{T3}\label{eq:N04-down1}
\]
In each line, $F_1(x')=F_1(x)-1$.

\emph{The vertex $t_m$.}  As $m\ge3$ (Lemma~\ref{lem:N03-params}(iii)),
its neighbours outside the copies are $t_{m-1}$, $q$ and $q'$; by the table
\eqref{eq:N04-edges}, they have the value $D-6-m=F_1(t_m)+1$.  So we
take $x'$ in a copy.  Let $x':=v^{m,1}_{L_m-1}$, the neighbour of
$t_m=v^{m,1}_{L_m}$ on the spine of $U_{m,1}$
(Definition~\ref{def:N03-graph}); any option of level $m$ would do.  By
Lemmas~\ref{lem:N02-profiles}(i) and~\ref{lem:N03-copies}(ii),
$h_{U_{m,1}}(x')=L_m-1$.  So by ($*$) and $L_m=D-7-2m$
(Lemma~\ref{lem:N03-params}(i)),
\[
 F_1(x')=m+(L_m-1)=m+(D-7-2m)-1=D-8-m=F_1(t_m)-1.
\]

\emph{The members of the copies.}  Let $x$ be a member of a copy
$U:=U_{k,c}$.  The copy $U$ is connected (Lemma~\ref{lem:N03-inside}(i)),
and $x\ne s_k$.  So there is a shortest
$s_k$--$x$ path $P$ of $U$, with $h_U(x)=d_U(s_k,x)\ge1$ edges.  Let $x'$
be the vertex before $x$ on $P$; then $x'x$ is an edge of $U$.  The part of
$P$ from $s_k$ to $x'$ has $h_U(x)-1$ edges, so $h_U(x')\le h_U(x)-1$.  By
the triangle inequality in $U$ (Definition~\ref{def:N01-graph}),
\[
 h_U(x)\le h_U(x')+d_U(x',x)=h_U(x')+1.
\]
Hence $h_U(x')=h_U(x)-1$.  The vertex $x'$ is a member of $U$ or one of
its terminals $s_k$, $t_k$; in both cases ($*$) applies, as it holds at
every vertex of $U$.  So
\[
 F_1(x')=k+h_U(x')=k+\bigl(h_U(x)-1\bigr)=\bigl(k+h_U(x)\bigr)-1=F_1(x)-1.
\]

(5) \emph{Condition (c) for $F_2$: every vertex $x\ne T_2$ has a neighbour
$x'$ with $F_2(x')=F_2(x)-1$.}  By Definition~\ref{def:N03-graph}, every
vertex $x\ne T_2$ of $G$ either is listed in the first column of the table
\eqref{eq:N04-down2} below, or is $s_m$, or is a member of a copy.  We
treat the three cases in turn.

\emph{The vertices of the table.}  For each $x$, the table
\eqref{eq:N04-down2} gives $x'$.  The pair $x,x'$ is an edge in the table
\eqref{eq:N04-edges}, and the values are the ones computed there:
\[
\begin{array}{lll}
 x & x' & F_2(x),\ F_2(x')\\\hline
 T_1 & s_1 & D,\ D-1\\
 s_k\ (1\le k<m) & s_{k+1} & D-k,\ D-k-1\\
 t_{k+1}\ (1\le k<m) & t_k & 8+k,\ 7+k\\
 t_1 & z_1 & 8,\ 7\\
 z_i\ (1\le i\le6) & z_{i+1} & 8-i,\ 7-i\\
 z_7 & T_2 & 1,\ 0\\
 q,\ q' & t_{m-1} & 7+m,\ 6+m\\
 y^j_i\ (1\le i\le D-2) & y^j_{i+1} & D-i,\ D-i-1\\
 y^j_{D-1} & T_2 & 1,\ 0\\
 \tilde y^j & y^j_{\pi_j+1} & D-\pi_j,\ D-\pi_j-1
\end{array}
\tag{T4}\label{eq:N04-down2}
\]
In each line, $F_2(x')=F_2(x)-1$.

\emph{The vertex $s_m$.}  As $m\ge3$ (Lemma~\ref{lem:N03-params}(iii)),
its only neighbour outside the copies is $s_{m-1}$; by the table
\eqref{eq:N04-edges}, it has the value $D-m+1=F_2(s_m)+1$.  So we
take $x'$ in a copy.  Let $x':=v^{m,1}_1$, the neighbour of
$s_m=v^{m,1}_0$ on the spine of $U_{m,1}$
(Definition~\ref{def:N03-graph}); any option of level $m$ would do.  By
Lemmas~\ref{lem:N02-profiles}(i) and~\ref{lem:N03-copies}(ii),
$\hb_{U_{m,1}}(x')=L_m-1$.  So by ($*$) and $L_m=D-7-2m$
(Lemma~\ref{lem:N03-params}(i)),
\[
 F_2(x')=7+m+(L_m-1)=7+m+(D-7-2m)-1=D-m-1=F_2(s_m)-1.
\]

\emph{The members of the copies.}  Let $x$ be a member of a copy
$U:=U_{k,c}$.  As in (4), taking $t_k$ and $\hb_U$ in place of $s_k$ and
$h_U$, the vertex $x'$ before $x$ on a shortest $t_k$--$x$ path of $U$ has
$\hb_U(x')=\hb_U(x)-1$.  So $F_2(x')=F_2(x)-1$, by ($*$).

(6) \emph{Conclusion.}  By (1), (3), (4), (5) and
Lemma~\ref{lem:N01-potential}, $G$ is connected, $\delta_1=F_1$ and
$\delta_2=F_2$.
\end{proof}

\begin{lemma}[distances on the frame]\label{lem:N04-skeleton}
Let $\varphi$ be a restricted formula and $G:=G_\varphi$
(Definition~\ref{def:N03-graph}).  For $1\le j,k\le m$ and $1\le i\le7$:
\begin{enumerate}
\item[\rm(i)] $d(s_j,s_k)=d(t_j,t_k)=|j-k|$;
\item[\rm(ii)] $d(s_j,t_k)=L_{\max(j,k)}+|j-k|=D-7-j-k$;
\item[\rm(iii)] $d(t_1,z_i)=i$;
\item[\rm(iv)] $d(T_1,T_2)=D$.
\end{enumerate}
\end{lemma}

\begin{proof}
We prove each equality with a path of $G$ for the upper bound and the
triangle inequality for the lower bound.  For vertices $o,x,y$ of $G$, the
triangle inequality (Definition~\ref{def:N01-graph}) gives
$d(o,y)\le d(o,x)+d(x,y)$, that is,
\[
 d(x,y)\ge d(o,y)-d(o,x).
\]
We take $o=T_1$ or $o=T_2$, so that $d(o,\cdot)$ is $\delta_1$ or
$\delta_2$, given by the table \eqref{eq:N04-values} of
Lemma~\ref{lem:N04-potentials}.

(1) We first compute the right-hand side of (ii).  By
Lemma~\ref{lem:N03-params}(i), $L_k=D-7-2k$ for every level $k$.  Since
$2\max(j,k)-|j-k|=j+k$, we get
\[
 L_{\max(j,k)}+|j-k|=D-7-2\max(j,k)+|j-k|=D-7-j-k.
\]

(2) We prove (i).  By symmetry we may assume $j\le k$.  The frame
contains the paths $s_j,s_{j+1},\dots,s_k$ and $t_k,t_{k-1},\dots,t_j$
(Definition~\ref{def:N03-graph}).  Each has $k-j$ edges, so
$d(s_j,s_k)\le k-j$ and $d(t_j,t_k)\le k-j$.  Conversely, the triangle
inequality with $o=T_1$ and with $o=T_2$ gives
\[
 d(s_j,s_k)\ge\delta_1(s_k)-\delta_1(s_j)=k-j,\qquad
 d(t_j,t_k)\ge\delta_2(t_k)-\delta_2(t_j)=(7+k)-(7+j)=k-j.
\]
With the upper bounds, this proves (i).

(3) We prove (ii).  Suppose first that $j\le k$.  We go from $s_j$ to
$s_k$ along the frame, with $k-j$ edges.  Then we go from $s_k$ to $t_k$
along the spine $s_k=v^{k,1}_0,v^{k,1}_1,\dots,v^{k,1}_{L_k}=t_k$ of
$U_{k,1}$ (Definition~\ref{def:N03-graph}), with $L_k$ edges.  This path
has $(k-j)+L_k=L_{\max(j,k)}+|j-k|$ edges.  Suppose now that $j>k$.  We go
from $s_j$ to $t_j$ along the spine of $U_{j,1}$, with $L_j$ edges.  Then
we go from $t_j$ to $t_k$ along the frame, with $j-k$ edges.  This path has
$L_j+(j-k)=L_{\max(j,k)}+|j-k|$ edges.  In both cases, by (1),
\[
 d(s_j,t_k)\le L_{\max(j,k)}+|j-k|=D-7-j-k.
\]
Conversely, the triangle inequality with $o=T_1$ gives, for all $j$ and
$k$,
\[
 d(s_j,t_k)\ge\delta_1(t_k)-\delta_1(s_j)=(D-7-k)-j=D-7-j-k.
\]
With the upper bound, this proves (ii).

(4) We prove (iii).  The path $t_1,z_1,\dots,z_i$ of the frame has $i$
edges, so $d(t_1,z_i)\le i$.  Conversely, the triangle inequality with
$o=T_2$ gives
\[
 d(t_1,z_i)\ge\delta_2(t_1)-\delta_2(z_i)=8-(8-i)=i.
\]
With the upper bound, this proves (iii).

(5) Finally, (iv) is a value of the table \eqref{eq:N04-values}:
$d(T_1,T_2)=\delta_1(T_2)=D$.
\end{proof}

\begin{lemma}[the copies are isometric]\label{lem:N04-isometric}
Let $\varphi$ be a restricted formula and $G:=G_\varphi$
(Definition~\ref{def:N03-graph}).  Let $k$ be a level and $c$ an option,
with $h_{U_{k,c}}(u)=d_{U_{k,c}}(s_k,u)$ and
$\hb_{U_{k,c}}(u)=d_{U_{k,c}}(t_k,u)$ for $u\in V(U_{k,c})$
(Definition~\ref{def:N03-names}); let $\alpha=2m+5$ and $\beta_k=2m+5-2k$ be the numbers of
Definition~\ref{def:N03-params}, and $\rho_{k,c}=r^{k,c}_1$ the core of
$U_{k,c}$ (Definition~\ref{def:N03-rows}).
\begin{enumerate}
\item[\rm(i)] $d(u,u')=d_{U_{k,c}}(u,u')$ for all $u,u'\in V(U_{k,c})$.
\item[\rm(ii)] $d(s_k,u)=h_{U_{k,c}}(u)$ and $d(t_k,u)=\hb_{U_{k,c}}(u)$
  for every $u\in V(U_{k,c})$.
\item[\rm(iii)] $(h_{U_{k,c}},\hb_{U_{k,c}})(\rho_{k,c})=(\alpha,\beta_k)$,
  $d(\rho_{k,c},s_k)=\alpha$ and $d(\rho_{k,c},t_k)=\beta_k$.
\item[\rm(iv)] For every member $u\ne\rho_{k,c}$ of $U_{k,c}$,
  \[
   d(\rho_{k,c},u)\ <\ \min\bigl(\alpha+h_{U_{k,c}}(u),\
   \beta_k+\hb_{U_{k,c}}(u)\bigr).
  \]
\end{enumerate}
\end{lemma}

\begin{proof}
Let $U:=U_{k,c}$.  We prove (i) with Lemma~\ref{lem:N01-isometric}, and
(ii) follows from (i).  Then (iii) and (iv) are
Lemma~\ref{lem:N03-inside}(ii), moved to $G$ by (i) and (ii).

(1) We prove (i).  We apply Lemma~\ref{lem:N01-isometric} to $G$ and
$X:=U$, with $s:=s_k$ and $t:=t_k$, and we check its four hypotheses.

First, by Lemma~\ref{lem:N04-potentials}, $G$ is connected.  Second, by
Lemma~\ref{lem:N03-inside}(i), $U$ is connected.  Third, by
Lemma~\ref{lem:N03-structure}(ii), $U$ is attached at $\{s_k,t_k\}$.

Finally, we show $d_{G-\Int(U)}(s_k,t_k)\ge d_U(s_k,t_k)$.  Every path of
$G-\Int(U)$ is a path of $G$, so
\[
 d_{G-\Int(U)}(s_k,t_k)\ge d(s_k,t_k).
\]
By Lemma~\ref{lem:N04-skeleton}(ii) with $j=k$,
\[
 d(s_k,t_k)=L_{\max(k,k)}+|k-k|=L_k.
\]
By Lemma~\ref{lem:N03-inside}(i), $d_U(s_k,t_k)=L_k$.  Hence
$d_{G-\Int(U)}(s_k,t_k)\ge d(s_k,t_k)=L_k=d_U(s_k,t_k)$.

So the hypotheses of Lemma~\ref{lem:N01-isometric} hold, and it gives
$d(u,u')=d_U(u,u')$ for all $u,u'\in V(U)$, which is (i).

(2) We prove (ii).  Let $u\in V(U)$.  By (i) with $u'=s_k$ and with
$u'=t_k$, and by Definition~\ref{def:N03-names},
\[
 d(s_k,u)=d_U(s_k,u)=h_U(u),\qquad d(t_k,u)=d_U(t_k,u)=\hb_U(u).
\]

(3) We prove (iii).  By Lemma~\ref{lem:N03-inside}(ii),
$(h_U,\hb_U)(\rho_{k,c})=(\alpha,\beta_k)$.  By (ii) at $u=\rho_{k,c}$,
$d(s_k,\rho_{k,c})=h_U(\rho_{k,c})=\alpha$ and
$d(t_k,\rho_{k,c})=\hb_U(\rho_{k,c})=\beta_k$.  This proves (iii).

(4) We prove (iv).  Let $u\ne\rho_{k,c}$ be a member of $U$.  By (i) and
Lemma~\ref{lem:N03-inside}(ii),
\[
 d(\rho_{k,c},u)=d_U(\rho_{k,c},u)<\min\bigl(\alpha+h_U(u),\
 \beta_k+\hb_U(u)\bigr).
\]
This proves (iv).
\end{proof}

\begin{lemma}[first-exit readings]\label{lem:N04-reading}
Let $\varphi$ be a restricted formula and $G:=G_\varphi$
(Definition~\ref{def:N03-graph}), with $\delta_1=d(T_1,\cdot)$ and
$\delta_2=d(T_2,\cdot)$ (Lemma~\ref{lem:N04-potentials}).
\begin{enumerate}
\item[\rm(i)] Let $k$ be a level, $c$ an option, $u$ a member of $U_{k,c}$,
  with the profile
  $(h_{U_{k,c}}(u),\hb_{U_{k,c}}(u))=(d_{U_{k,c}}(s_k,u),\,d_{U_{k,c}}(t_k,u))$
  of Definition~\ref{def:N03-names}, and $w\in V(G)\setminus\Int(U_{k,c})$.
  Then
  \[
   d(w,u)=\min\bigl(d(w,s_k)+h_{U_{k,c}}(u),\ d(w,t_k)+\hb_{U_{k,c}}(u)\bigr).
  \]
\item[\rm(ii)] Let $1\le j\le4$, $u\in\Int(Y^j)$ and
  $w\in V(G)\setminus\Int(Y^j)$.  Then
  \[
   d(w,u)=\min\bigl(d(w,T_1)+\delta_1(u),\ d(w,T_2)+\delta_2(u)\bigr).
  \]
\item[\rm(iii)] For every $u\in V(G)\setminus\{q,q'\}$,
  $d(q,u)=d(q',u)=1+\min\bigl(d(t_m,u),\ d(t_{m-1},u)\bigr)$; and
  $d(q,q')=2$.
\end{enumerate}
\end{lemma}

\begin{proof}
By Lemma~\ref{lem:N04-potentials}, $G$ is connected.  Items (i) and (ii)
are Lemma~\ref{lem:N01-exit} (first exit), applied to a copy and to an
arm.  We prove (iii) directly, by the argument of the lower bound of
Lemma~\ref{lem:N01-exit}.

(1) We prove (i).  Let $U:=U_{k,c}$.  We apply Lemma~\ref{lem:N01-exit} to
$G$ and $X:=U$, with $s:=s_k$ and $t:=t_k$.  Its hypotheses hold:
\begin{itemize}
\item $G$ is connected;
\item $U$ is attached at $\{s_k,t_k\}$ (Lemma~\ref{lem:N03-structure}(ii));
\item $u$ is a member of $U$, so $u\in\Int(U)$
  (Definition~\ref{def:N03-names});
\item $w\in V(G)\setminus\Int(U)$.
\end{itemize}
Lemma~\ref{lem:N01-exit} gives
\[
 d(u,w)=\min\bigl(d(u,s_k)+d(s_k,w),\ d(u,t_k)+d(t_k,w)\bigr).
\]
By Lemma~\ref{lem:N04-isometric}(ii), $d(s_k,u)=h_U(u)$ and
$d(t_k,u)=\hb_U(u)$.  As $d$ is symmetric (Definition~\ref{def:N01-graph}),
\[
 d(w,u)=d(u,w)=\min\bigl(d(w,s_k)+h_U(u),\ d(w,t_k)+\hb_U(u)\bigr).
\]

(2) We prove (ii).  We apply Lemma~\ref{lem:N01-exit} to $G$ and
$X:=Y^j$, with $s:=T_1$ and $t:=T_2$.  Its hypotheses hold:
\begin{itemize}
\item $G$ is connected;
\item $Y^j$ is attached at $\{T_1,T_2\}$ (Lemma~\ref{lem:N03-structure}(ii));
\item $u\in\Int(Y^j)$ and $w\in V(G)\setminus\Int(Y^j)$.
\end{itemize}
Lemma~\ref{lem:N01-exit} gives
\[
 d(u,w)=\min\bigl(d(u,T_1)+d(T_1,w),\ d(u,T_2)+d(T_2,w)\bigr).
\]
By the definition of $\delta_1$ and $\delta_2$, $d(T_1,u)=\delta_1(u)$ and
$d(T_2,u)=\delta_2(u)$.  As $d$ is symmetric,
\[
 d(w,u)=d(u,w)=\min\bigl(d(w,T_1)+\delta_1(u),\ d(w,T_2)+\delta_2(u)\bigr).
\]

(3) We prove (iii).  Let $X$ be the subgraph of $G$ with the vertices
$t_m$, $t_{m-1}$, $q$, $q'$ and the four edges $t_mq$, $t_{m-1}q$, $t_mq'$,
$t_{m-1}q'$.  By Definition~\ref{def:N03-graph}, every edge of $G$ with an
end in $\{q,q'\}$ is one of these four, so $X$ is attached at
$\{t_m,t_{m-1}\}$ with $\Int(X)=\{q,q'\}$ (Definition~\ref{def:N01-attach}).
Let $u\in V(G)\setminus\{q,q'\}$ and $w\in\{q,q'\}$.  Lemma~\ref{lem:N01-exit}
at the interior vertex $w$ and the vertex $u$ gives
\[
 d(w,u)=\min\bigl(d(w,t_m)+d(t_m,u),\ d(w,t_{m-1})+d(t_{m-1},u)\bigr)
 =1+\min\bigl(d(t_m,u),\ d(t_{m-1},u)\bigr),
\]
as $w$ is adjacent to $t_m$ and to $t_{m-1}$.  Finally $q$ and $q'$ are
not adjacent and have the common neighbour $t_m$, so $d(q,q')=2$.  This
proves (iii).
\end{proof}

\begin{lemma}[the rows]\label{lem:N04-rows}
Let $\varphi$ be a restricted formula and $G:=G_\varphi$
(Definition~\ref{def:N03-graph}).  Let $\alpha$ and $\beta_j$
($1\le j\le m$) be the numbers of Definition~\ref{def:N03-params},
$\rho_{j,c}=r^{j,c}_1$ and $w_j=y^j_{\pi_j}$ the rows of
Definition~\ref{def:N03-rows}, and $\delta_1=d(T_1,\cdot)$,
$\delta_2=d(T_2,\cdot)$ (Lemma~\ref{lem:N04-potentials}).
\begin{enumerate}
\item[\rm(i)] For $1\le j,k\le m$ and every option $c$,
  \[
   d(\rho_{j,c},s_k)=\alpha+|j-k|,\qquad d(\rho_{j,c},t_k)=\beta_j+|j-k|.
  \]
\item[\rm(ii)] For $1\le k\le m$,
  \[
   d(q,s_k)=m-k+1+L_m,\qquad d(q,t_k)=\begin{cases}m-k&\text{if }k<m,\\1&\text{if }k=m.\end{cases}
  \]
\item[\rm(iii)] For $1\le j\le4$,
  \[
  \begin{aligned}
   d(w_j,u)&=\min\bigl(\pi_j+\delta_1(u),\ D-\pi_j+\delta_2(u)\bigr)
     &&(u\in V(G)\setminus\Int(Y^j)),\\
   d(w_j,y^j_i)&=|i-\pi_j| &&(0\le i\le D),\\
   d(w_j,\tilde y^j)&=2.
  \end{aligned}
  \]
\end{enumerate}
\end{lemma}

\begin{proof}
Items (i) and (ii) follow from Lemma~\ref{lem:N04-reading}(i) and (iii),
with the distances on the frame (Lemma~\ref{lem:N04-skeleton}).  Item (iii)
follows
from Lemma~\ref{lem:N04-reading}(ii) and the table \eqref{eq:N04-values}
of Lemma~\ref{lem:N04-potentials}.

(1) We prove (i).  Let $1\le j,k\le m$, $c$ an option,
$\rho:=\rho_{j,c}$, $U:=U_{j,c}$ and $M:=\max(j,k)$.  We apply
Lemma~\ref{lem:N04-reading}(i) at the level $j$, with $u:=\rho$, and with
$w:=s_k$ and then $w:=t_k$.  Its hypotheses hold:
\begin{itemize}
\item $\rho=r^{j,c}_1$ is a $t$-ear member of $U$
  (Definitions~\ref{def:N03-names} and~\ref{def:N03-rows});
\item $s_k$ and $t_k$ lie in $S$, so they are not in $\Int(U)$
  (Lemma~\ref{lem:N03-structure}(i)).  This also holds when $k=j$, where
  $s_k$ and $t_k$ are the terminals of $U$.
\end{itemize}
By Lemma~\ref{lem:N04-isometric}(iii), $(h_U,\hb_U)(\rho)=(\alpha,\beta_j)$,
so Lemma~\ref{lem:N04-reading}(i) gives
\[
\begin{aligned}
 d(s_k,\rho)&=\min\bigl(d(s_k,s_j)+\alpha,\ d(s_k,t_j)+\beta_j\bigr),\\
 d(t_k,\rho)&=\min\bigl(d(t_k,s_j)+\alpha,\ d(t_k,t_j)+\beta_j\bigr).
\end{aligned}
\]
By Lemma~\ref{lem:N04-skeleton}(i), $d(s_k,s_j)=d(t_k,t_j)=|j-k|$.  By
Lemma~\ref{lem:N04-skeleton}(ii), taken at $(k,j)$ and at $(j,k)$,
$d(s_k,t_j)=d(s_j,t_k)=L_M+|j-k|$.  Hence
\[
 d(s_k,\rho)=|j-k|+\min(\alpha,\ L_M+\beta_j),\qquad
 d(t_k,\rho)=|j-k|+\min(L_M+\alpha,\ \beta_j).
\]
We compare the two terms of each minimum.  By
Lemma~\ref{lem:N03-params}(i) and Definition~\ref{def:N03-params},
$\alpha-\beta_j=\sigma_j=2j$ and $L_M=4m+8-2M$.  As $j\le m$ and $M\le m$,
\[
 L_M-(\alpha-\beta_j)=4m+8-2M-2j\ge8>0,
\]
so $\alpha<L_M+\beta_j$, and the first minimum is $\alpha$.  As
$\alpha-\beta_j=2j>0$ and $L_M>0$, we have $\beta_j<\alpha<L_M+\alpha$, and
the second minimum is $\beta_j$.  As $d$ is symmetric
(Definition~\ref{def:N01-graph}),
\[
 d(\rho,s_k)=\alpha+|j-k|,\qquad d(\rho,t_k)=\beta_j+|j-k|.
\]

(2) We prove (ii).  Let $1\le k\le m$.  By Definition~\ref{def:N03-graph},
$s_k$ and $t_k$ are distinct from $q$ and $q'$, so
Lemma~\ref{lem:N04-reading}(iii) applies at $u=s_k$ and at $u=t_k$:
\[
 d(q,s_k)=1+\min\bigl(d(t_m,s_k),d(t_{m-1},s_k)\bigr),\qquad
 d(q,t_k)=1+\min\bigl(d(t_m,t_k),d(t_{m-1},t_k)\bigr).
\]
As $k\le m$, Lemma~\ref{lem:N04-skeleton}(ii) at $(k,m)$ and at
$(k,m-1)$ gives $d(s_k,t_m)=L_m+(m-k)$ and
$d(s_k,t_{m-1})=L_{\max(k,m-1)}+|k-m+1|$, which is
$L_{m-1}+(m-1-k)=L_m+(m-k)+1$ if $k\le m-1$ and $L_m+1$ if $k=m$
(Definition~\ref{def:N03-params}: $L_{m-1}=L_m+2$); in both cases
$d(s_k,t_{m-1})>d(s_k,t_m)$, so $d(q,s_k)=1+L_m+(m-k)$.
Lemma~\ref{lem:N04-skeleton}(i) gives $d(t_k,t_m)=m-k$ and
$d(t_k,t_{m-1})=|m-1-k|$, which is $m-k-1$ if $k<m$ and $1$ if $k=m$; so
$d(q,t_k)=1+(m-k-1)=m-k$ if $k<m$, and $d(q,t_m)=1+\min(0,1)=1$.

(3) We prove (iii).  Let $1\le j\le4$.  By Definition~\ref{def:N03-graph},
$\pi_j\in\{p,7\}$.  By Definition~\ref{def:N03-params}, $1\le7\le4m+14=D-1$
and $1\le2m+4=p\le D-1$.  So $1\le\pi_j\le D-1$, and $w_j=y^j_{\pi_j}$ lies
in $\Int(Y^j)=\{y^j_1,\dots,y^j_{D-1},\tilde y^j\}$
(Lemma~\ref{lem:N03-structure}(i)).  By the line $y^j_i$ of the table
\eqref{eq:N04-values} with $i=\pi_j$,
\[
 (\delta_1,\delta_2)(w_j)=(\pi_j,\ D-\pi_j).
\]

\emph{The first formula.}  Let $u\in V(G)\setminus\Int(Y^j)$.  We apply
Lemma~\ref{lem:N04-reading}(ii) with $w_j$ in the role of its $u$, and $u$
in the role of its $w$.  By the definition of $\delta_1,\delta_2$ and the
symmetry of $d$,
\[
 d(w_j,u)=\min\bigl(d(u,T_1)+\delta_1(w_j),\ d(u,T_2)+\delta_2(w_j)\bigr)
 =\min\bigl(\pi_j+\delta_1(u),\ D-\pi_j+\delta_2(u)\bigr).
\]

\emph{The second formula.}  Let $0\le i\le D$.  The arm contains the path
from $y^j_i$ to $y^j_{\pi_j}$ along $y^j_0,\dots,y^j_D$
(Definition~\ref{def:N03-graph}), with $|i-\pi_j|$ edges, so
$d(w_j,y^j_i)\le|i-\pi_j|$.  Conversely, $\delta_1(y^j_i)=i$:
\begin{itemize}
\item for $1\le i\le D-1$, this is the line $y^j_i$ of the table
  \eqref{eq:N04-values};
\item $y^j_0=T_1$ and $y^j_D=T_2$, with $\delta_1(T_1)=0$ and
  $\delta_1(T_2)=D$ by the lines $T_1$ and $T_2$.
\end{itemize}
The triangle inequality (Definition~\ref{def:N01-graph}), with $T_1$ as the
base point, gives
\[
 d(w_j,y^j_i)\ge|\delta_1(y^j_i)-\delta_1(w_j)|=|i-\pi_j|.
\]

\emph{The third formula.}  By Definition~\ref{def:N03-graph}, the
neighbours of $\tilde y^j$ are $y^j_{\pi_j-1}$ and $y^j_{\pi_j+1}$, and
$\tilde y^j\ne w_j$.  So $\tilde y^j$ is distinct from $w_j=y^j_{\pi_j}$
and not adjacent to it, and $d(w_j,\tilde y^j)\ge2$.  The vertex
$y^j_{\pi_j-1}$ is adjacent to $w_j$, by an edge of the arm, and to
$\tilde y^j$, so $d(w_j,\tilde y^j)\le2$.  Hence $d(w_j,\tilde y^j)=2$.
\end{proof}

\begin{corollary}[how the rows read a copy]\label{cor:N04-levels}
Let $\varphi$ be a restricted formula and $G:=G_\varphi$
(Definition~\ref{def:N03-graph}).  Let $\alpha$, $\beta_j$, $\sigma_j$
($1\le j\le m$), $p$ and $L_m$ be the numbers of
Definition~\ref{def:N03-params}, $\rho_{j,c}=r^{j,c}_1$ and
$w_j=y^j_{\pi_j}$ the rows of Definition~\ref{def:N03-rows}, and
$\delta_1=d(T_1,\cdot)$, $\delta_2=d(T_2,\cdot)$
(Lemma~\ref{lem:N04-potentials}).  Let $k$ be a level, $c$ an option and
$u$ a member of $U_{k,c}$, with the profile
$(h_{U_{k,c}}(u),\hb_{U_{k,c}}(u))=(d_{U_{k,c}}(s_k,u),\,d_{U_{k,c}}(t_k,u))$
of Definition~\ref{def:N03-names}.
\begin{enumerate}
\item[\rm(i)] For every row $\rho_{j,c'}$ with $(j,c')\ne(k,c)$,
  \[
   d(\rho_{j,c'},u)=|j-k|+\beta_j+\min\bigl(\sigma_j+h_{U_{k,c}}(u),\
   \hb_{U_{k,c}}(u)\bigr).
  \]
\item[\rm(ii)] If $k<m$, then
  $d(q,u)=m-k+1+\min\bigl(L_m+h_{U_{k,c}}(u),\ \hb_{U_{k,c}}(u)-1\bigr)$;
  if $k=m$, then $d(q,u)=1+\hb_{U_{m,c}}(u)$.
\item[\rm(iii)] $d(w_3,u)=d(w_4,u)=7+\delta_1(u)=7+k+h_{U_{k,c}}(u)$.
\item[\rm(iv)] $d(w_1,u)=d(w_2,u)=\min\bigl(p+\delta_1(u),\ p+7+\delta_2(u)\bigr)$.
\end{enumerate}
\end{corollary}

\begin{proof}
Let $U:=U_{k,c}$.  Items (i) and (ii) follow from
Lemma~\ref{lem:N04-reading}(i) and the distances of
Lemma~\ref{lem:N04-rows}(i),(ii); items (iii) and (iv) follow from
Lemma~\ref{lem:N04-rows}(iii) and the values of the table
\eqref{eq:N04-values}.  By Lemma~\ref{lem:N03-structure}(i), the set
$\Int(U)$ of the members of $U$ is disjoint from $S$, from $\{q\}$, from
every $\Int(U_{j,c'})$ with $(j,c')\ne(k,c)$, and from every
$\Int(Y^j)$.  This gives the hypothesis $w\notin\Int(U)$ of
Lemma~\ref{lem:N04-reading}(i) in (1) and (2), and the hypothesis
$u\notin\Int(Y^j)$ of Lemma~\ref{lem:N04-rows}(iii) in (3) and (4).

(1) We prove (i).  Let $(j,c')\ne(k,c)$ and $\rho:=\rho_{j,c'}$.  The row
$\rho$ is a member of $U_{j,c'}$ (Definitions~\ref{def:N03-names}
and~\ref{def:N03-rows}), so $\rho\in\Int(U_{j,c'})$ and
$\rho\notin\Int(U)$.  Lemma~\ref{lem:N04-reading}(i) with $w:=\rho$ gives
\[
 d(\rho,u)=\min\bigl(d(\rho,s_k)+h_U(u),\ d(\rho,t_k)+\hb_U(u)\bigr).
\]
By Lemma~\ref{lem:N04-rows}(i), $d(\rho,s_k)=\alpha+|j-k|$ and
$d(\rho,t_k)=\beta_j+|j-k|$.  By Lemma~\ref{lem:N03-params}(i),
$\alpha-\beta_j=\sigma_j$, that is, $\alpha=\beta_j+\sigma_j$.  Hence
\[
 d(\rho,u)=\min\bigl(\beta_j+\sigma_j+|j-k|+h_U(u),\
 \beta_j+|j-k|+\hb_U(u)\bigr)
 =|j-k|+\beta_j+\min\bigl(\sigma_j+h_U(u),\ \hb_U(u)\bigr).
\]

(2) We prove (ii).  As $q\notin\Int(U)$, Lemma~\ref{lem:N04-reading}(i)
with $w:=q$ gives
\[
 d(q,u)=\min\bigl(d(q,s_k)+h_U(u),\ d(q,t_k)+\hb_U(u)\bigr).
\]
By Lemma~\ref{lem:N04-rows}(ii), $d(q,s_k)=m-k+1+L_m$, and $d(q,t_k)=m-k$
if $k<m$, $d(q,t_m)=1$.  If $k<m$,
\[
 d(q,u)=\min\bigl(m-k+1+L_m+h_U(u),\ m-k+\hb_U(u)\bigr)
 =m-k+1+\min\bigl(L_m+h_U(u),\ \hb_U(u)-1\bigr).
\]
If $k=m$, then $d(q,u)=\min(1+L_m+h_U(u),\ 1+\hb_U(u))=1+\hb_U(u)$: by
Lemma~\ref{lem:N03-inside}(iv), a spine member $v_i$ has
$\hb_U-h_U=L_m-2i\le L_m$, a $t$-ear member has $\hb_U<h_U$, and $U$ has
no crossing (Definition~\ref{def:N03-graph}: $E_{m,c}=()$), so
$\hb_U(u)\le L_m+h_U(u)$.

(3) We prove (iii).  Let $j\in\{3,4\}$; then $\pi_j=7$
(Definition~\ref{def:N03-graph}).  As $u\notin\Int(Y^j)$, the first
formula of Lemma~\ref{lem:N04-rows}(iii) gives
\[
 d(w_j,u)=\min\bigl(7+\delta_1(u),\ D-7+\delta_2(u)\bigr).
\]
By the line $u\in\Int(U_{k,c})$ of the table \eqref{eq:N04-values},
$\delta_1(u)=k+h_U(u)$ and $\delta_2(u)=7+k+\hb_U(u)$.  We compare the two
terms of the minimum:
\[
 \bigl(D-7+\delta_2(u)\bigr)-\bigl(7+\delta_1(u)\bigr)
 =D-14+(7+k+\hb_U(u))-(k+h_U(u))=D-7-\bigl(h_U(u)-\hb_U(u)\bigr).
\]
By Lemma~\ref{lem:N04-isometric}(ii), $h_U(u)=d(s_k,u)$ and
$\hb_U(u)=d(t_k,u)$, and by the triangle inequality
(Definition~\ref{def:N01-graph}) and Lemma~\ref{lem:N04-skeleton}(ii) with
$j=k$,
\[
 h_U(u)-\hb_U(u)=d(s_k,u)-d(t_k,u)\le d(s_k,t_k)=L_k .
\]
By Lemma~\ref{lem:N03-params}(i), $L_k=D-7-2k$, so
\[
 D-7-\bigl(h_U(u)-\hb_U(u)\bigr)\ge D-7-L_k=2k>0 .
\]
So the minimum is its first term:
$d(w_j,u)=7+\delta_1(u)=7+k+h_U(u)$, for $j=3$ and for $j=4$.

(4) We prove (iv).  Let $j\in\{1,2\}$; then $\pi_j=p$
(Definition~\ref{def:N03-graph}).  As $u\notin\Int(Y^j)$, the first
formula of Lemma~\ref{lem:N04-rows}(iii) gives
\[
 d(w_j,u)=\min\bigl(p+\delta_1(u),\ D-p+\delta_2(u)\bigr),
\]
and $D-p=p+7$ by Lemma~\ref{lem:N03-params}(i).  This proves (iv), for
$j=1$ and for $j=2$.
\end{proof}

\begin{lemma}[small distances from $T_1$ and $T_2$]\label{lem:N04-small}
Let $\varphi$ be a restricted formula and $G:=G_\varphi$
(Definition~\ref{def:N03-graph}), with $\delta_1=d(T_1,\cdot)$ and
$\delta_2=d(T_2,\cdot)$ (Lemma~\ref{lem:N04-potentials}).  For a member $u$
of a copy $U_{j,c}$, $h_{U_{j,c}}(u)=d_{U_{j,c}}(s_j,u)$
(Definition~\ref{def:N03-names}).  Let $u$ be a vertex of $G$ that lies in
no $\Int(Y^j)$, $1\le j\le4$.
\begin{enumerate}
\item[\rm(i)] If $1\le\delta_1(u)\le m$, then $u=s_k$ with $k=\delta_1(u)$,
  or $u$ is a member of a copy $U_{j,c}$ with $j<\delta_1(u)$ and
  $h_{U_{j,c}}(u)=\delta_1(u)-j$.
\item[\rm(ii)] If $1\le\delta_2(u)\le7$, then $u=z_i$ with
  $i=8-\delta_2(u)$.
\end{enumerate}
\end{lemma}

\begin{proof}
By Lemma~\ref{lem:N03-structure}(i), $V(G)$ is the disjoint union of $S$,
$\{q\}$, the sets $\Int(U_{j,c})$ and the sets $\Int(Y^j)$.  As $u$ lies in
no $\Int(Y^j)$, $u$ lies in $S$, or $u=q$, or $u$ is a member of a copy.
By Definition~\ref{def:N03-rows}, $S$ consists of $T_1$, $T_2$, the $s_k$
and $t_k$ ($1\le k\le m$), $z_1,\dots,z_7$ and $q'$.  So $u$ is one of
\[
 T_1,\ T_2,\ s_k,\ t_k\ (1\le k\le m),\ z_i\ (1\le i\le7),\ q,\ q',\
 \text{or a member of a copy }U_{j,c}.
 \tag{$\ast$}\label{eq:N04-list}
\]
By Lemma~\ref{lem:N04-potentials}, $\delta_1=F_1$ and $\delta_2=F_2$, so
the values of $\delta_1,\delta_2$ at these vertices are those of the table
\eqref{eq:N04-values}; there $D=4m+15$ (Definition~\ref{def:N03-params}).
For a member $u$ of $U_{j,c}$, $u\ne s_j,t_j$
(Definition~\ref{def:N03-names}), and the distance between two distinct
vertices is at least $1$ (Definition~\ref{def:N01-graph}), so
\[
 h_{U_{j,c}}(u)=d_{U_{j,c}}(s_j,u)\ge1,\qquad
 \hb_{U_{j,c}}(u)=d_{U_{j,c}}(t_j,u)\ge1.
\]

(1) We prove (i).  Let $1\le\delta_1(u)\le m$.  We first exclude the
vertices of \eqref{eq:N04-list} whose $\delta_1$ is $0$ or larger than $m$.
By the column $F_1$ of the table \eqref{eq:N04-values}, with $1\le k\le m$
and $1\le i\le7$,
\[
\begin{aligned}
 \delta_1(T_1)&=0,\\
 \delta_1(T_2)&=D=4m+15>m,\\
 \delta_1(t_k)&=D-7-k=4m+8-k\ge3m+8>m,\\
 \delta_1(z_i)&=D-8+i=4m+7+i>m,\\
 \delta_1(q)=\delta_1(q')&=D-6-m=3m+9>m.
\end{aligned}
\]
So $u$ is none of $T_1$, $T_2$, $t_k$, $z_i$, $q$, $q'$, and by
\eqref{eq:N04-list} it is some $s_k$ or a member of a copy.

If $u=s_k$, then $\delta_1(u)=k$ by the line $s_k$ of the table
\eqref{eq:N04-values}, so $k=\delta_1(u)$.

If $u$ is a member of $U_{j,c}$, then $\delta_1(u)=j+h_{U_{j,c}}(u)$ by the
line $u\in\Int(U_{k,c})$ of the table \eqref{eq:N04-values}, at the level
$j$.  So $h_{U_{j,c}}(u)=\delta_1(u)-j$, and as $h_{U_{j,c}}(u)\ge1$,
$j\le\delta_1(u)-1<\delta_1(u)$.  This proves (i).

(2) We prove (ii).  Let $1\le\delta_2(u)\le7$.  We exclude the vertices of
\eqref{eq:N04-list} whose $\delta_2$ is $0$ or larger than $7$.  By the
column $F_2$ of the table \eqref{eq:N04-values}, with $1\le k\le m$, and
with $j\ge1$ and $\hb_{U_{j,c}}(u)\ge1$ for a member $u$ of $U_{j,c}$,
\[
\begin{aligned}
 \delta_2(T_1)&=D=4m+15>7,\\
 \delta_2(T_2)&=0,\\
 \delta_2(s_k)&=D-k=4m+15-k\ge3m+15>7,\\
 \delta_2(t_k)&=7+k\ge8>7,\\
 \delta_2(q)=\delta_2(q')&=7+m>7,\\
 \delta_2(u)&=7+j+\hb_{U_{j,c}}(u)\ge9>7
   \qquad(u\text{ a member of }U_{j,c}).
\end{aligned}
\]
So by \eqref{eq:N04-list}, $u=z_i$ for some $1\le i\le7$.  By the line $z_i$
of the table \eqref{eq:N04-values}, $\delta_2(u)=8-i$, so
$i=8-\delta_2(u)$.  This proves (ii).
\end{proof}

\begin{corollary}[the cores of a level]\label{cor:N04-cores}
Let $\varphi$ be a restricted formula and $G:=G_\varphi$
(Definition~\ref{def:N03-graph}).  Let $\alpha$ and $\beta_k$ be the
numbers of Definition~\ref{def:N03-params}, $k$ a level, and for each
option $c$ let $\rho_{k,c}=r^{k,c}_1$ be the core of $U_{k,c}$
(Definition~\ref{def:N03-rows}).  For a member $u$ of $U_{k,c}$,
$(h_{U_{k,c}}(u),\hb_{U_{k,c}}(u))=(d_{U_{k,c}}(s_k,u),\,d_{U_{k,c}}(t_k,u))$
is its profile (Definition~\ref{def:N03-names}).
\begin{enumerate}
\item[\rm(i)] For every option $c$ and every vertex $v\notin\Int(U_{k,c})$,
  \[
   d(\rho_{k,c},v)=\min\bigl(\alpha+d(s_k,v),\ \beta_k+d(t_k,v)\bigr).
  \]
\item[\rm(ii)] For two options $a\ne b$ and every member $u$ of $U_{k,a}$,
  $d(\rho_{k,a},u)<d(\rho_{k,b},u)$.
\end{enumerate}
\end{corollary}

\begin{proof}
(1) We prove (i).  Let $c$ be an option and $v\notin\Int(U_{k,c})$.  The
core $\rho_{k,c}$ is a member of $U_{k,c}$
(Definitions~\ref{def:N03-names} and~\ref{def:N03-rows}), so
Lemma~\ref{lem:N04-reading}(i) applies to it with $w:=v$:
\[
 d(v,\rho_{k,c})=\min\bigl(d(v,s_k)+h_{U_{k,c}}(\rho_{k,c}),\
 d(v,t_k)+\hb_{U_{k,c}}(\rho_{k,c})\bigr).
\]
By Lemma~\ref{lem:N04-isometric}(iii),
$(h_{U_{k,c}},\hb_{U_{k,c}})(\rho_{k,c})=(\alpha,\beta_k)$.  As $d$ is
symmetric (Definition~\ref{def:N01-graph}), this is (i).

(2) We prove (ii).  Let $a\ne b$ be options and $u$ a member of
$U_{k,a}$.  The sets $\Int(U_{k,a})$ and $\Int(U_{k,b})$ are disjoint
(Lemma~\ref{lem:N03-structure}(i)), so $u\notin\Int(U_{k,b})$.

If $u=\rho_{k,a}$, then $d(\rho_{k,a},u)=0$.  The core $\rho_{k,b}$ lies in
$\Int(U_{k,b})$ (Definitions~\ref{def:N03-names} and~\ref{def:N03-rows}),
so $\rho_{k,b}\ne u$, and the distance between two distinct vertices is at
least $1$ (Definition~\ref{def:N01-graph}).  Hence
$d(\rho_{k,a},u)=0<1\le d(\rho_{k,b},u)$.

If $u\ne\rho_{k,a}$, then (i) applies at $c:=b$ and $v:=u$, as
$u\notin\Int(U_{k,b})$.  By Lemma~\ref{lem:N04-isometric}(ii) in $U_{k,a}$,
$d(s_k,u)=h_{U_{k,a}}(u)$ and $d(t_k,u)=\hb_{U_{k,a}}(u)$.  Hence
\[
 d(\rho_{k,b},u)=\min\bigl(\alpha+d(s_k,u),\ \beta_k+d(t_k,u)\bigr)
 =\min\bigl(\alpha+h_{U_{k,a}}(u),\ \beta_k+\hb_{U_{k,a}}(u)\bigr).
\]
By Lemma~\ref{lem:N04-isometric}(iv) in $U_{k,a}$,
$d(\rho_{k,a},u)<\min(\alpha+h_{U_{k,a}}(u),\ \beta_k+\hb_{U_{k,a}}(u))$.
Hence $d(\rho_{k,a},u)<d(\rho_{k,b},u)$.
\end{proof}

\begin{corollary}[the rows of two arms]\label{cor:N04-arms}
Let $\varphi$ be a restricted formula and $G:=G_\varphi$
(Definition~\ref{def:N03-graph}).  Let $D$ and $p$ be the numbers of
Definition~\ref{def:N03-params}, $Y^1,\dots,Y^4$ the arms and
$w_a=y^a_{\pi_a}$ their rows (Definition~\ref{def:N03-rows}), and
$\delta_1=d(T_1,\cdot)$, $\delta_2=d(T_2,\cdot)$
(Lemma~\ref{lem:N04-potentials}).  Let $a\ne b$ be indices with
$\pi_a=\pi_b$, that is, $\{a,b\}=\{1,2\}$ or $\{a,b\}=\{3,4\}$
(Definition~\ref{def:N03-graph}).
\begin{enumerate}
\item[\rm(i)] For every vertex $v\notin\Int(Y^a)\cup\Int(Y^b)$,
  $d(w_a,v)=d(w_b,v)$.
\item[\rm(ii)] For every $u\in\Int(Y^a)$, $d(w_a,u)<d(w_b,u)$.
\end{enumerate}
\end{corollary}

\begin{proof}
Put $\pi:=\pi_a=\pi_b$.  By Definition~\ref{def:N03-graph}, $\pi=7$ or
$\pi=p$.  By Definition~\ref{def:N03-params}, $p=2m+4$ and $D=4m+15$, and
by Lemma~\ref{lem:N03-params}(iii), $m\ge3$.  So
\[
\begin{aligned}
 \pi=7:&\quad \pi=7, && D-\pi=4m+8\ge16,\\
 \pi=p:&\quad \pi=2m+4\ge10, && D-\pi=2m+11\ge17.
\end{aligned}
\]
In both cases $\pi\ge7$ and $D-\pi\ge15$.

(1) We prove (i).  As $v\notin\Int(Y^a)$ and $v\notin\Int(Y^b)$, the first
formula of Lemma~\ref{lem:N04-rows}(iii), at $w_a$ and at $w_b$, gives
\[
 d(w_a,v)=\min\bigl(\pi+\delta_1(v),\ D-\pi+\delta_2(v)\bigr)=d(w_b,v).
\]

(2) We prove (ii).  Let $u\in\Int(Y^a)$.  By
Lemma~\ref{lem:N03-structure}(i), $\Int(Y^a)\cap\Int(Y^b)=\emptyset$, so
$u\notin\Int(Y^b)$, and the first formula of
Lemma~\ref{lem:N04-rows}(iii) gives
\[
 d(w_b,u)=\min\bigl(\pi+\delta_1(u),\ D-\pi+\delta_2(u)\bigr).
\]
By Lemma~\ref{lem:N03-structure}(i), $u=y^a_i$ with $1\le i\le D-1$, or
$u=\tilde y^a$.

If $u=y^a_i$, then $d(w_a,u)=|i-\pi|$ by the second formula of
Lemma~\ref{lem:N04-rows}(iii), and $\delta_1(u)=i$, $\delta_2(u)=D-i$ by the
line $y^j_i$ of the table \eqref{eq:N04-values}.  As $i\ge1$ and
$\pi\ge1$,
\[
 |i-\pi|=\max(i,\pi)-\min(i,\pi)<\max(i,\pi)+\min(i,\pi)=\pi+\delta_1(u),
\]
and as $D-i\ge1$ and $D-\pi\ge1$, in the same way,
\[
 |i-\pi|=|(D-i)-(D-\pi)|<(D-i)+(D-\pi)=D-\pi+\delta_2(u).
\]
So $d(w_a,u)<d(w_b,u)$.

If $u=\tilde y^a$, then $d(w_a,u)=2$ by the third formula of
Lemma~\ref{lem:N04-rows}(iii), and $(\delta_1,\delta_2)(u)=(\pi,D-\pi)$ by
the line $\tilde y^j$ of the table \eqref{eq:N04-values}.  So
$d(w_b,u)=\min(2\pi,\ 2(D-\pi))\ge14>2=d(w_a,u)$.
\end{proof}

\begin{corollary}[a row and a vertex outside its arm]\label{cor:N04-other-arm}
Let $\varphi$ be a restricted formula and $G:=G_\varphi$
(Definition~\ref{def:N03-graph}).  Let $D$ be the number of
Definition~\ref{def:N03-params}, $Y^1,\dots,Y^4$ the arms and
$w_a=y^a_{\pi_a}$ their rows (Definition~\ref{def:N03-rows}), and
$\delta_1=d(T_1,\cdot)$, $\delta_2=d(T_2,\cdot)$
(Lemma~\ref{lem:N04-potentials}).  Let $1\le a\le4$ and
$u\notin\Int(Y^a)$.
\begin{enumerate}
\item[\rm(i)] If $\delta_1(u)\le D-\pi_a$, then $d(w_a,u)=\pi_a+\delta_1(u)$.
\item[\rm(ii)] If $\delta_2(u)\le\pi_a$, then $d(w_a,u)=D-\pi_a+\delta_2(u)$.
\end{enumerate}
\end{corollary}

\begin{proof}
As $u\notin\Int(Y^a)$, the first formula of Lemma~\ref{lem:N04-rows}(iii)
gives
\[
 d(w_a,u)=\min\bigl(\pi_a+\delta_1(u),\ D-\pi_a+\delta_2(u)\bigr).
\]
By Lemma~\ref{lem:N04-skeleton}(iv), the triangle inequality and the
symmetry of $d$ (Definition~\ref{def:N01-graph}),
\[
 D=d(T_1,T_2)\le d(T_1,u)+d(u,T_2)=\delta_1(u)+\delta_2(u).
\]

(1) We prove (i).  As $\delta_1(u)\le D-\pi_a$, the first term satisfies
$\pi_a+\delta_1(u)\le D$.  As $\delta_2(u)\ge D-\delta_1(u)$, the second
term satisfies
\[
 D-\pi_a+\delta_2(u)\ge2D-\bigl(\pi_a+\delta_1(u)\bigr)\ge2D-D=D,
\]
where the last inequality comes from the preceding one.  So the minimum
is its first term, $\pi_a+\delta_1(u)$.

(2) We prove (ii).  As $\delta_2(u)\le\pi_a$, the second term satisfies
$D-\pi_a+\delta_2(u)\le D$.  As $\delta_1(u)\ge D-\delta_2(u)\ge D-\pi_a$,
the first term satisfies $\pi_a+\delta_1(u)\ge D$.  So the minimum is its
second term, $D-\pi_a+\delta_2(u)$.
\end{proof}


\section{The lower bound and the forcing}\label{sec:N05}

\begin{lemma}[the lower bound]\label{lem:N05-lower}
Let $\varphi$ be a restricted formula and $G:=G_\varphi$
(Definition~\ref{def:N03-graph}); $G$ is connected
(Lemma~\ref{lem:N04-potentials}).  Let $W$ be a set resolving $G$
(Definition~\ref{def:N01-resolve}).  For each level $k$
(Definition~\ref{def:N03-options}) let
\[
 I_k:=\Int(U_{k,1})\cup\Int(U_{k,2})\cup\Int(U_{k,3})
\]
be the set of the members of the copies of the level $k$
(Definition~\ref{def:N03-names}).  Let $w_j=y^j_{\pi_j}$ ($1\le j\le4$) be
the rows of the arms (Definition~\ref{def:N03-rows}), and $\tilde y^j$ the
vertex of the arm $Y^j$ adjacent to $y^j_{\pi_j-1}$ and $y^j_{\pi_j+1}$
(Definition~\ref{def:N03-graph}).  Let
\[
 P_0:=\{q,q'\},\qquad P_j:=\{w_j,\tilde y^j\}\quad(1\le j\le4)
\]
be the five pairs of Lemma~\ref{lem:N03-structure}(iii).
\begin{enumerate}
\item[\rm(i)] For every level $k$, at most one option $c$ has
  $W\cap\Int(U_{k,c})=\emptyset$.
\item[\rm(ii)] $W$ meets each of the pairs $P_0,\dots,P_4$.
\item[\rm(iii)] $|W|\ge2m+5$.
\item[\rm(iv)] If $|W|=2m+5$, then for every level $k$, exactly two of the
  sets $\Int(U_{k,1})$, $\Int(U_{k,2})$, $\Int(U_{k,3})$ meet $W$, each in
  exactly one vertex; $W$ meets each pair $P_j$ in exactly one vertex;
  and
  $W\subseteq I_1\cup\dots\cup I_m\cup P_0\cup\dots\cup P_4$.
\end{enumerate}
\end{lemma}

\begin{proof}
(1) We prove (i).  Let $k$ be a level, and suppose that two distinct
options $c,c'$ have $W\cap\Int(U_{k,c})=W\cap\Int(U_{k,c'})=\emptyset$.
Put $u:=v^{k,c}_1$ and $u':=v^{k,c'}_1$.  By
Definition~\ref{def:N03-params}, $L_k=4m+8-2k\ge2m+8$, so $1\le L_k-1$,
and $u$, $u'$ are spine members of $U_{k,c}$ and of $U_{k,c'}$
(Definition~\ref{def:N03-names}).  They are distinct, as
$\Int(U_{k,c})\cap\Int(U_{k,c'})=\emptyset$
(Lemma~\ref{lem:N03-structure}(i)).  By Lemma~\ref{lem:N03-inside}(iv)
with $i=1$, in $U_{k,c}$ and in $U_{k,c'}$,
\[
 (h_{U_{k,c}},\hb_{U_{k,c}})(u)=(1,\ L_k-1)
 =(h_{U_{k,c'}},\hb_{U_{k,c'}})(u').
\]
Let $w\in W$.  By the supposition, $w\notin\Int(U_{k,c})$ and
$w\notin\Int(U_{k,c'})$.  So Lemma~\ref{lem:N04-reading}(i) applies to the
member $u$ of $U_{k,c}$ and to the member $u'$ of $U_{k,c'}$, with this
$w$:
\[
\begin{aligned}
 d(w,u)&=\min\bigl(d(w,s_k)+1,\ d(w,t_k)+L_k-1\bigr)
   &&(\text{Lemma~\ref{lem:N04-reading}(i) in }U_{k,c}),\\
 d(w,u')&=\min\bigl(d(w,s_k)+1,\ d(w,t_k)+L_k-1\bigr)
   &&(\text{Lemma~\ref{lem:N04-reading}(i) in }U_{k,c'}).
\end{aligned}
\]
The right-hand sides are equal, so $d(w,u)=d(w,u')$.  So no vertex of $W$ separates $u$ and $u'$, and $W$ does not resolve $G$
(Definition~\ref{def:N01-resolve}), a contradiction.  This proves (i).

(2) We prove (ii).  By Lemma~\ref{lem:N04-potentials}, $G$ is connected,
and by Lemma~\ref{lem:N03-structure}(iii), the two vertices of each pair
$P_j$ are twins.  By Lemma~\ref{lem:N01-twins}(iii), $W$ contains a vertex
of each pair $P_j$.

(3) We prove (iii).  Put
$X:=I_1\cup\dots\cup I_m\cup P_0\cup\dots\cup P_4$.  The $m+5$ sets
$I_1,\dots,I_m,P_0,\dots,P_4$ are pairwise disjoint: the sets $\Int(U_{k,c})$ are pairwise disjoint
(Lemma~\ref{lem:N03-structure}(i)), so $I_k\cap I_{k'}=\emptyset$ for
$k\ne k'$; no vertex lies in two of the pairs, and no vertex of a pair lies
in a set $\Int(U_{k,c})$ (Lemma~\ref{lem:N03-structure}(iii)).  Hence
\begin{equation}\label{eq:N05-sum}
 |W|\ \ge\ |W\cap X|\ =\ \sum_{k=1}^m|W\cap I_k|+\sum_{j=0}^4|W\cap P_j| .
\end{equation}
Let $k$ be a level.  The three sets $W\cap\Int(U_{k,c})$ are pairwise
disjoint, and by (i) at least two of them are nonempty, so
\[
 |W\cap I_k|=\sum_{c=1}^3|W\cap\Int(U_{k,c})|\ \ge\ 2 .
\]
By (ii), $|W\cap P_j|\ge1$ for $0\le j\le4$.  With \eqref{eq:N05-sum},
$|W|\ge2m+5$.

(4) We prove (iv).  Let $|W|=2m+5$.  By \eqref{eq:N05-sum} and the
bounds of (3),
\[
 |W|\ \ge\ |W\cap X|\ =\ \sum_{k=1}^m|W\cap I_k|+\sum_{j=0}^4|W\cap P_j|
 \ \ge\ 2m+5\ =\ |W|,
\]
so both inequalities are equalities.  The first gives $|W\cap X|=|W|$,
that is, $W\subseteq X$.  In the second, each of the $m+5$ terms is at
least its bound ($2$ for $I_k$, $1$ for $P_j$), and their sum equals the
sum $2m+5$ of the bounds.  So each term equals its bound: $|W\cap I_k|=2$
for every level $k$, and $|W\cap P_j|=1$ for $0\le j\le4$.

Let $k$ be a level.  The three sets $W\cap\Int(U_{k,c})$,
$c\in\{1,2,3\}$, are pairwise disjoint and their union is $W\cap I_k$, so
\[
 |W\cap\Int(U_{k,1})|+|W\cap\Int(U_{k,2})|+|W\cap\Int(U_{k,3})|
 =|W\cap I_k|=2 .
\]
By (i), at least two of these three numbers are at least $1$.  So exactly
two of them are $1$ and the third is $0$; that is, exactly two of the sets
$\Int(U_{k,1})$, $\Int(U_{k,2})$, $\Int(U_{k,3})$ meet $W$, each in
exactly one vertex.  This proves (iv).
\end{proof}

\begin{lemma}[the forcing]\label{lem:N05-forcing}
Let $\varphi$ be a restricted formula and $G:=G_\varphi$
(Definition~\ref{def:N03-graph}); $G$ is connected
(Lemma~\ref{lem:N04-potentials}).  Let $W$ be a set resolving $G$
(Definition~\ref{def:N01-resolve}) with $|W|=2m+5$.  Let $k$ be a level
and $c$ an option (Definition~\ref{def:N03-options}) with
$W\cap\Int(U_{k,c})\ne\emptyset$, and let $\rho_{k,c}=r^{k,c}_1$ be the
core of $U_{k,c}$ (Definition~\ref{def:N03-rows}).  Then
\[
 W\cap\Int(U_{k,c})=\{\rho_{k,c}\}.
\]
\end{lemma}

\begin{proof}
Let $U:=U_{k,c}$, and let $L_k$ and $\mu_k$ be the numbers of
Definition~\ref{def:N03-params}.

(1) \emph{The landmark $w$ and the copy $U_0$ with no landmark.}  By
Lemma~\ref{lem:N05-lower}(iv), exactly two of the sets
$\Int(U_{k,1})$, $\Int(U_{k,2})$, $\Int(U_{k,3})$ meet $W$, each in exactly
one vertex.  One of them is $\Int(U)$, as $W\cap\Int(U)\ne\emptyset$; so
$W\cap\Int(U)=\{w\}$ for one vertex $w$, and there is an option $c_0\ne c$
with $W\cap\Int(U_{k,c_0})=\emptyset$.  Put $U_0:=U_{k,c_0}$.  Suppose,
for a contradiction, that $w\ne\rho_{k,c}$.  In (2) and (3) we find two
distinct vertices $x$ and $x'$, and in (4) we show that no landmark
separates them.

(2) \emph{The vertex $x$.}  The landmark $w$ is a member of $U$, and
$w\ne\rho_{k,c}$ by the supposition.  Put $x:=r^{k,c}_{\mu_k+1}$ if $w$ is
a spine or crossing member of $U$, and $x:=v^{k,c}_{L_k-1}$ if $w$ is a
$t$-ear member of $U$ (Definition~\ref{def:N03-names}).  By
Lemma~\ref{lem:N03-inside}(vi), $x$ is a common member of $U$, $x\ne w$,
and
\[
 d_U(w,x)=\min\bigl(h_U(w)+h_U(x),\ \hb_U(w)+\hb_U(x)\bigr).
\]
By Lemma~\ref{lem:N04-isometric}(i), $d(w,x)=d_U(w,x)$.  Hence
\begin{equation}\label{eq:N05-wx}
 d(w,x)=\min\bigl(h_U(w)+h_U(x),\ \hb_U(w)+\hb_U(x)\bigr).
\end{equation}

(3) \emph{The counterpart $x'$.}  Let $x'$ be the counterpart of $x$ in
$U_0$ (Definition~\ref{def:N03-names}), that is,
$x':=r^{k,c_0}_{\mu_k+1}$ if $x=r^{k,c}_{\mu_k+1}$, and
$x':=v^{k,c_0}_{L_k-1}$ if $x=v^{k,c}_{L_k-1}$.  As $c_0\ne c$, $x'\ne x$
(Definition~\ref{def:N03-names}).  By Lemma~\ref{lem:N03-inside}(iv) in $U$
and in $U_0$, at $i=\mu_k+1$ for the $t$-ear members and at $i=L_k-1$ for
the spine members,
\[
\begin{aligned}
 (h_U,\hb_U)(r^{k,c}_{\mu_k+1})
   &=\bigl(L_k-\mu_k+(\mu_k+1),\,\mu_k+2-(\mu_k+1)\bigr)=(L_k+1,1)
   =(h_{U_0},\hb_{U_0})(r^{k,c_0}_{\mu_k+1}),\\
 (h_U,\hb_U)(v^{k,c}_{L_k-1})
   &=\bigl(L_k-1,\,L_k-(L_k-1)\bigr)=(L_k-1,1)
   =(h_{U_0},\hb_{U_0})(v^{k,c_0}_{L_k-1}).
\end{aligned}
\]
So in both cases
\begin{equation}\label{eq:N05-same}
 (h_{U_0},\hb_{U_0})(x')=(h_U,\hb_U)(x).
\end{equation}

(4) \emph{No landmark separates $x$ and $x'$.}  Let $w'\in W$.

\emph{The case $w'=w$.}  As $w\in\Int(U)$ and
$\Int(U)\cap\Int(U_0)=\emptyset$ (Lemma~\ref{lem:N03-structure}(i)),
$w\notin\Int(U_0)$, so Lemma~\ref{lem:N04-reading}(i) applies to the
member $x'$ of $U_0$ with this $w$.  By Lemma~\ref{lem:N04-isometric}(ii)
at $w\in V(U)$ and the symmetry of $d$ (Definition~\ref{def:N01-graph}),
$d(w,s_k)=h_U(w)$ and $d(w,t_k)=\hb_U(w)$.  Hence
\[
\begin{aligned}
 d(w,x')&=\min\bigl(d(w,s_k)+h_{U_0}(x'),\ d(w,t_k)+\hb_{U_0}(x')\bigr)
   &&(\text{Lemma~\ref{lem:N04-reading}(i)})\\
 &=\min\bigl(h_U(w)+h_U(x),\ \hb_U(w)+\hb_U(x)\bigr)
   &&(\text{Lemma~\ref{lem:N04-isometric}(ii), \eqref{eq:N05-same}})\\
 &=d(w,x) &&(\text{by }\eqref{eq:N05-wx}).
\end{aligned}
\]

\emph{The case $w'\ne w$.}  By (1), $W\cap\Int(U)=\{w\}$ and
$W\cap\Int(U_0)=\emptyset$, so $w'\notin\Int(U)$ and $w'\notin\Int(U_0)$.
Hence
\[
\begin{aligned}
 d(w',x)&=\min\bigl(d(w',s_k)+h_U(x),\ d(w',t_k)+\hb_U(x)\bigr)
   &&(\text{Lemma~\ref{lem:N04-reading}(i) in }U)\\
 &=\min\bigl(d(w',s_k)+h_{U_0}(x'),\ d(w',t_k)+\hb_{U_0}(x')\bigr)
   &&(\text{by }\eqref{eq:N05-same})\\
 &=d(w',x') &&(\text{Lemma~\ref{lem:N04-reading}(i) in }U_0).
\end{aligned}
\]

(5) By (4), no landmark separates the distinct vertices $x$ and $x'$, so
$W$ does not resolve $G$ (Definition~\ref{def:N01-resolve}), a
contradiction.  Hence $w=\rho_{k,c}$, that is,
$W\cap\Int(U)=\{\rho_{k,c}\}$.
\end{proof}

\begin{corollary}[small resolving sets are patterns]\label{cor:N05-pattern}
Let $\varphi$ be a restricted formula and $G:=G_\varphi$
(Definition~\ref{def:N03-graph}); $G$ is connected
(Lemma~\ref{lem:N04-potentials}).  Let $W$ be a set resolving $G$
(Definition~\ref{def:N01-resolve}) with $|W|\le2m+5$.
\begin{enumerate}
\item[\rm(i)] $|W|=2m+5$.
\item[\rm(ii)] There are an automorphism $\psi$ of $G$ and a choice $f$
  with $\psi(W)=W_f$, where $W_f$ is the pattern of $f$
  (Definition~\ref{def:N03-pattern}).
\end{enumerate}
\end{corollary}

\begin{proof}
Let $I_k$ ($1\le k\le m$), $P_0:=\{q,q'\}$ and
$P_j:=\{w_j,\tilde y^j\}$ ($1\le j\le4$) be the sets of
Lemma~\ref{lem:N05-lower}.

(1) We prove (i).  By Lemma~\ref{lem:N05-lower}(iii), $|W|\ge2m+5$; with
$|W|\le2m+5$, $|W|=2m+5$.

(2) We prove (ii).  We first show that $W$ is a pattern up to the
vertices it takes in the pairs $P_j$.  Let $k$ be a level.  By (i) and
Lemma~\ref{lem:N05-lower}(iv), exactly two of the sets $\Int(U_{k,1})$,
$\Int(U_{k,2})$, $\Int(U_{k,3})$ meet $W$.  Let $f_k$ be the option of the
set that does not meet $W$:
\[
 W\cap\Int(U_{k,f_k})=\emptyset,\qquad W\cap\Int(U_{k,c})\ne\emptyset
 \quad(c\ne f_k).
\]
For $c\ne f_k$, the hypotheses of Lemma~\ref{lem:N05-forcing} hold
($|W|=2m+5$ by (i), and $W\cap\Int(U_{k,c})\ne\emptyset$), so
$W\cap\Int(U_{k,c})=\{\rho_{k,c}\}$.  Hence
\[
 W\cap I_k=\bigcup_{c=1}^{3}\bigl(W\cap\Int(U_{k,c})\bigr)
 =\{\rho_{k,c}:c\ne f_k\}.
\]
By Lemma~\ref{lem:N05-lower}(iv), $W$ meets each pair $P_j$ in exactly one
vertex, which we call $p_j$, so $p_0\in\{q,q'\}$ and
$p_j\in\{w_j,\tilde y^j\}$ for $1\le j\le4$.  Also
$W\subseteq I_1\cup\dots\cup I_m\cup P_0\cup\dots\cup P_4$, so $W$ is the
union of the sets $W\cap I_k$ and $W\cap P_j=\{p_j\}$.  With the choice
$f:=(f_1,\dots,f_m)$, this gives
\begin{equation}\label{eq:N05-W}
 W=\{\rho_{k,c}:1\le k\le m,\ c\ne f_k\}\cup\{p_0,p_1,p_2,p_3,p_4\}.
\end{equation}
By Definition~\ref{def:N03-pattern},
$W_f=\{\rho_{k,c}:1\le k\le m,\ c\ne f_k\}\cup\{q,w_1,w_2,w_3,w_4\}$.
This is the same set with $q,w_1,\dots,w_4$ in place of $p_0,\dots,p_4$.

(3) \emph{The map $\psi$.}  For $0\le j\le4$, let $a_j$ be the vertex of
$P_j$ that lies in $W_f$, and $b_j$ the other one:
\[
 a_0:=q,\quad b_0:=q';\qquad a_j:=w_j,\quad b_j:=\tilde y^j\quad
 (1\le j\le4).
\]
So $p_j=a_j$ or $p_j=b_j$.  Let $J:=\{j:p_j=b_j\}$ be the set of the
pairs where $W$ takes the other vertex.  Define $\psi:V(G)\to V(G)$ by
\[
 \psi(a_j):=b_j,\quad\psi(b_j):=a_j\quad(j\in J);\qquad
 \psi(u):=u\quad\text{for every other vertex }u.
\]
This is well defined: the pairs $P_0,\dots,P_4$ are pairwise disjoint
(Lemma~\ref{lem:N03-structure}(iii)), so every vertex lies in at most one
pair.  So $\psi$ exchanges $a_j$ and $b_j$ for $j\in J$, and fixes every
other vertex.  In particular $\psi(p_j)=a_j$ for every $j$.  If $j\in J$,
then $p_j=b_j$ and $\psi(b_j)=a_j$.  If $j\notin J$, then $p_j=a_j$,
which $\psi$ fixes.

(4) \emph{$\psi$ is an automorphism of $G$.}  For $j\in J$, let $\psi_j$
exchange $a_j$ and $b_j$ and fix every other vertex.  As $a_j,b_j$ are
twins (Lemma~\ref{lem:N03-structure}(iii)) and $G$ is connected, $\psi_j$
is an automorphism of $G$ (Lemma~\ref{lem:N01-twins}(i)).  Let
$j_1<\dots<j_t$ be the elements of $J$.  We show
\[
 \psi=\psi_{j_1}\circ\dots\circ\psi_{j_t}
\]
(the identity if $J=\emptyset$).  Let $u$ be a vertex of $G$.

If $u\in P_j$ with $j\in J$: for $j'\ne j$, the map $\psi_{j'}$ moves only
the two vertices of $P_{j'}$, which are not in $P_j$, so it fixes both $u$
and $\psi_j(u)$.  Hence the composition sends $u$ to $\psi_j(u)$, the
other vertex of $P_j$, which is $\psi(u)$.

Otherwise, every $\psi_{j'}$ with $j'\in J$ fixes $u$, and the composition
sends $u$ to $u=\psi(u)$.

A composition of automorphisms is an automorphism, so $\psi$ is an
automorphism of $G$.

(5) \emph{$\psi(W)=W_f$.}  Let $k$ be a level and $c$ an option.  The
row $\rho_{k,c}$ lies in $\Int(U_{k,c})$
(Definitions~\ref{def:N03-names} and~\ref{def:N03-rows}).  By
Lemma~\ref{lem:N03-structure}(iii), no vertex of the pairs
$P_0,\dots,P_4$ lies in $\Int(U_{k,c})$.  So $\rho_{k,c}$ is none of the
vertices $a_j,b_j$ ($0\le j\le4$), and by (3),
$\psi(\rho_{k,c})=\rho_{k,c}$.  Also by (3), $\psi(p_j)=a_j$.  By
\eqref{eq:N05-W} and Definition~\ref{def:N03-pattern},
\[
\begin{aligned}
 \psi(W)&=\{\rho_{k,c}:1\le k\le m,\ c\ne f_k\}\cup\{a_0,a_1,a_2,a_3,a_4\}\\
 &=\{\rho_{k,c}:1\le k\le m,\ c\ne f_k\}\cup\{q,w_1,w_2,w_3,w_4\}=W_f.
\end{aligned}
\]
This proves (ii).
\end{proof}


\section{The pairs a pattern leaves unseparated}\label{sec:N06}

\begin{lemma}[pairs inside a copy]\label{lem:N06-copy}
Let $\varphi$ be a restricted formula and $G:=G_\varphi$
(Definition~\ref{def:N03-graph}); $G$ is connected
(Lemma~\ref{lem:N04-potentials}).  Let $k$ be a level and $c$ an option
(Definition~\ref{def:N03-options}), and $x\ne y$ two members of $U_{k,c}$
(Definition~\ref{def:N03-names}).  Then $q$ or the row $w_3=y^3_7$
(Definitions~\ref{def:N03-graph} and~\ref{def:N03-rows}) separates $x$ and
$y$ (Definition~\ref{def:N01-resolve}).
\end{lemma}

\begin{proof}
Let $U:=U_{k,c}$, and write $h$, $\hb$, $\kappa$ for $h_U$, $\hb_U$,
$\kappa_U$ (Definition~\ref{def:N03-names}).  Let $L_k$, $L_m$ and
$\sigma_k$ be the numbers of Definition~\ref{def:N03-params}.  By
Corollary~\ref{cor:N04-levels}(iii) and~(ii), for every member $u$ of $U$,
\begin{equation}\label{eq:N06-copy}
 d(w_3,u)=7+k+h(u),\qquad
 d(q,u)=\begin{cases}m-k+1+\min\bigl(L_m+h(u),\ \hb(u)-1\bigr)&\text{if }k<m,\\
 1+\hb(u)&\text{if }k=m.\end{cases}
\end{equation}

(1) \emph{The case $h(x)\ne h(y)$.}  By \eqref{eq:N06-copy},
$d(w_3,x)=7+k+h(x)\ne7+k+h(y)=d(w_3,y)$, so $w_3$ separates $x$ and $y$.

(2) \emph{The case $h(x)=h(y)=:\eta$.}  We show that $q$ separates $x$
and $y$.  By Lemma~\ref{lem:N03-inside}(iii), the profiles of $x$ and $y$
are distinct, so $\hb(x)\ne\hb(y)$.  The claim is symmetric in $x$ and
$y$, so we may exchange their names and assume $\hb(x)<\hb(y)$.  If
$k=m$, then by \eqref{eq:N06-copy} $d(q,x)=1+\hb(x)<1+\hb(y)=d(q,y)$,
and $q$ separates $x$ and $y$.  Let $k<m$, and suppose, for a
contradiction, that $d(q,x)=d(q,y)$.  By \eqref{eq:N06-copy} this means
\[
 \min\bigl(L_m+\eta,\ \hb(x)-1\bigr)=\min\bigl(L_m+\eta,\ \hb(y)-1\bigr).
\]
If $\hb(x)-1<L_m+\eta$, the left side is $\hb(x)-1$, while both terms of
the right side exceed $\hb(x)-1$, a contradiction.  So
$\hb(x)-1\ge L_m+\eta$, and $\hb(y)-1>\hb(x)-1\ge L_m+\eta$.  Hence, with
the skew $\kappa=\hb-h$ of Definition~\ref{def:N03-names},
\[
 \kappa(x)=\hb(x)-\eta\ge L_m+1,\qquad \kappa(y)=\hb(y)-\eta\ge L_m+1 .
\]

(3) \emph{The contradiction.}  By Definition~\ref{def:N03-params},
$L_m=4m+8-2m=2m+8$ and $\sigma_k=2k\le2m$, so
$L_m+1\ge\sigma_k+9>\sigma_k+4$.  Hence $\kappa(x)>\sigma_k+4$ and
$\kappa(y)>\sigma_k+4$, so $x$ and $y$ are spine members of $U$
(Lemma~\ref{lem:N03-inside}(iv)), that is, $x=v^{k,c}_i$ and
$y=v^{k,c}_{i'}$ with $1\le i,i'\le L_k-1$
(Definition~\ref{def:N03-names}).  Then $i=h(x)=h(y)=i'$
(Lemma~\ref{lem:N03-inside}(iv)), so $x=y$, which contradicts $x\ne y$.
Hence $q$ separates $x$ and $y$.
\end{proof}

\begin{lemma}[members of rowed copies]\label{lem:N06-rowed}
Let $\varphi$ be a restricted formula and $G:=G_\varphi$
(Definition~\ref{def:N03-graph}); $G$ is connected
(Lemma~\ref{lem:N04-potentials}).  Let $f$ be a choice and $W:=W_f$ its
pattern (Definition~\ref{def:N03-pattern}).  Let $k$ be a level and
$c\ne f_k$ an option (Definition~\ref{def:N03-options}), so that $U_{k,c}$
is a rowed copy of the level $k$ (Definition~\ref{def:N03-pattern}), and
let $x$ be a member of $U_{k,c}$ (Definition~\ref{def:N03-names}).  Then
$W$ separates $x$ from every vertex $y\ne x$ of $G$
(Definition~\ref{def:N01-resolve}).
\end{lemma}

\begin{proof}
Let $c'$ be the option other than $c$ and $f_k$.  Then $U_{k,c'}$ is the
other rowed copy of the level $k$, and by
Definition~\ref{def:N03-pattern} the cores $\rho_{k,c}$ and $\rho_{k,c'}$
lie in $W$, as do $q$ and $w_3$.  Let $\alpha$ and $\beta_k$ be the numbers
of Definition~\ref{def:N03-params}, and $y\ne x$ a vertex of $G$.  By
Lemma~\ref{lem:N03-structure}(i), $y$ lies in $\Int(U_{k,c})$, in
$\Int(U_{k,c'})$, or in neither, and (1)--(3) treat these cases.

(1) \emph{The case $y\in\Int(U_{k,c})$.}  Then $x$ and $y$ are two
distinct members of $U_{k,c}$.  By Lemma~\ref{lem:N06-copy}, $q$ or $w_3$
separates them, and both lie in $W$.

(2) \emph{The case $y\in\Int(U_{k,c'})$.}  Corollary~\ref{cor:N04-cores}(ii)
with $(a,b)=(c,c')$ at the member $u:=x$ of $U_{k,c}$, and with
$(a,b)=(c',c)$ at the member $u:=y$ of $U_{k,c'}$, gives
\[
 d(\rho_{k,c},x)<d(\rho_{k,c'},x),\qquad d(\rho_{k,c'},y)<d(\rho_{k,c},y).
\]
Suppose that neither $\rho_{k,c}$ nor $\rho_{k,c'}$ separates $x$ and $y$:
$d(\rho_{k,c},x)=d(\rho_{k,c},y)$ and $d(\rho_{k,c'},x)=d(\rho_{k,c'},y)$.
Then
\[
 d(\rho_{k,c},x)<d(\rho_{k,c'},x)=d(\rho_{k,c'},y)<d(\rho_{k,c},y)
 =d(\rho_{k,c},x),
\]
a contradiction.  So $\rho_{k,c}$ or $\rho_{k,c'}$ separates $x$ and $y$.

(3) \emph{The case $y\notin\Int(U_{k,c})\cup\Int(U_{k,c'})$.}  Suppose that
neither $\rho_{k,c}$ nor $\rho_{k,c'}$ separates $x$ and $y$:
$d(\rho_{k,c},x)=d(\rho_{k,c},y)$ and $d(\rho_{k,c'},x)=d(\rho_{k,c'},y)$.
By Corollary~\ref{cor:N04-cores}(i) at $c$ and at $c'$, with $v:=y$,
\[
 d(\rho_{k,c},y)=\min\bigl(\alpha+d(s_k,y),\ \beta_k+d(t_k,y)\bigr)
 =d(\rho_{k,c'},y).
\]
Hence
\[
 d(\rho_{k,c},x)=d(\rho_{k,c},y)=d(\rho_{k,c'},y)=d(\rho_{k,c'},x).
\]
This contradicts Corollary~\ref{cor:N04-cores}(ii) with $(a,b)=(c,c')$ at
the member $u:=x$ of $U_{k,c}$, which gives
$d(\rho_{k,c},x)<d(\rho_{k,c'},x)$.  So $\rho_{k,c}$ or $\rho_{k,c'}$
separates $x$ and $y$.
\end{proof}

\begin{lemma}[pairs inside an arm]\label{lem:N06-arm-pairs}
Let $\varphi$ be a restricted formula and $G:=G_\varphi$
(Definition~\ref{def:N03-graph}); $G$ is connected
(Lemma~\ref{lem:N04-potentials}).  Let $1\le j\le4$, and $j':=3$ if
$j\le2$, $j':=1$ if $j\ge3$, so that $\{\pi_j,\pi_{j'}\}=\{p,7\}$
(Definition~\ref{def:N03-graph}).  Let $x\ne y$ be two vertices of
$\Int(Y^j)$ (Definition~\ref{def:N03-rows}).  Then the row $w_j$ or the
row $w_{j'}$ separates $x$ and $y$ (Definition~\ref{def:N01-resolve}).
\end{lemma}

\begin{proof}
Put $\pi:=\pi_j$, and let $D$ and $p$ be the numbers of
Definition~\ref{def:N03-params}.  By Lemma~\ref{lem:N03-structure}(i),
$\Int(Y^j)=\{y^j_1,\dots,y^j_{D-1},\tilde y^j\}$, and by
Lemma~\ref{lem:N04-rows}(iii),
\[
 d(w_j,y^j_i)=|i-\pi|\quad(1\le i\le D-1),\qquad d(w_j,\tilde y^j)=2.
\]
Either $d(w_j,x)\ne d(w_j,y)$ or $d(w_j,x)=d(w_j,y)$, and (1)--(3) treat
the two cases.

(1) \emph{The case $d(w_j,x)\ne d(w_j,y)$.}  Then $w_j$ separates $x$ and
$y$.

(2) \emph{The case $d(w_j,x)=d(w_j,y)$: $x$ and $y$ lie symmetrically
around $\pi$.}  We show that $\{x,y\}$ is
$\{y^j_{\pi-e},y^j_{\pi+e}\}$ for some $1\le e\le\pi-1$, or
$\{\tilde y^j,y^j_{\pi-2}\}$, or $\{\tilde y^j,y^j_{\pi+2}\}$.

If $x=y^j_i$ and $y=y^j_{i'}$ with $i\ne i'$, the distances above give
\[
 |i-\pi|=d(w_j,x)=d(w_j,y)=|i'-\pi|,
\]
so $i-\pi=-(i'-\pi)$, as $i\ne i'$.  Put $e:=|i-\pi|$.  Then the numbers
$i-\pi$ and $i'-\pi$ are $e$ and $-e$ in some order, so
$\{i,i'\}=\{\pi-e,\pi+e\}$ and $\{x,y\}=\{y^j_{\pi-e},y^j_{\pi+e}\}$.  Here
$e\ge1$, as $e=0$ would give $i=\pi=i'$.  Also $e\le\pi-1$, as $\pi-e$ is
one of $i,i'$, which are at least $1$.

If $\tilde y^j$ is one of $x,y$, the other is some $y^j_i$, and in the
same way $2=d(w_j,\tilde y^j)=d(w_j,y^j_i)=|i-\pi|$.  So $i=\pi-2$ or
$i=\pi+2$.

(3) \emph{The case $d(w_j,x)=d(w_j,y)$: $w_{j'}$ separates $x$ and $y$.}
By the lines $y^j_i$ and $\tilde y^j$ of the table \eqref{eq:N04-values},
$\delta_1(y^j_i)=i$ and $\delta_1(\tilde y^j)=\pi$.  So by (2), the two
numbers $\delta_1(x),\delta_1(y)$ are $\pi-e$ and $\pi+e$, or $\pi$ and
$\pi-2$, or $\pi$ and $\pi+2$.

These two numbers differ, as $e\ge1$.  We show that they are at most
$D-\pi_{j'}$.  First, $\pi+e\le2\pi-1$ as $e\le\pi-1$, and
$\pi+2\le2\pi-1$ as $\pi\ge7$ ($\pi\in\{p,7\}$ and $p=2m+4\ge8$,
Definition~\ref{def:N03-params} and Lemma~\ref{lem:N03-params}(iii)).  So
both are at most $2\pi-1$.  Next, by Definition~\ref{def:N03-params},
$p=2m+4$ and $D=4m+15$, so
\[
\begin{aligned}
 \pi=p,\ \pi_{j'}=7:&\quad 2\pi-1=4m+7, && D-\pi_{j'}=4m+8,\\
 \pi=7,\ \pi_{j'}=p:&\quad 2\pi-1=13, && D-\pi_{j'}=2m+11\ge15.
\end{aligned}
\]
In both cases $2\pi-1\le D-\pi_{j'}$, so $\delta_1(x),\delta_1(y)\le
D-\pi_{j'}$.

We now apply Corollary~\ref{cor:N04-other-arm}(i) with $a:=j'$.  Its
hypotheses are $u\notin\Int(Y^{j'})$ and $\delta_1(u)\le D-\pi_{j'}$.  By
the definition of $j'$ in the statement, $j'\ne j$, so
$\Int(Y^j)\cap\Int(Y^{j'})=\emptyset$ (Lemma~\ref{lem:N03-structure}(i)),
and $x,y\notin\Int(Y^{j'})$.

At $u:=x$: $x\notin\Int(Y^{j'})$ and $\delta_1(x)\le D-\pi_{j'}$, so
Corollary~\ref{cor:N04-other-arm}(i) gives
$d(w_{j'},x)=\pi_{j'}+\delta_1(x)$.

At $u:=y$: $y\notin\Int(Y^{j'})$ and $\delta_1(y)\le D-\pi_{j'}$, so
Corollary~\ref{cor:N04-other-arm}(i) gives
$d(w_{j'},y)=\pi_{j'}+\delta_1(y)$.

As $\delta_1(x)\ne\delta_1(y)$, $d(w_{j'},x)\ne d(w_{j'},y)$, so $w_{j'}$
separates $x$ and $y$.
\end{proof}

\begin{lemma}[the arms]\label{lem:N06-arms}
Let $\varphi$ be a restricted formula and $G:=G_\varphi$
(Definition~\ref{def:N03-graph}); $G$ is connected
(Lemma~\ref{lem:N04-potentials}).  Let $f$ be a choice and $W:=W_f$ its
pattern (Definition~\ref{def:N03-pattern}).  Let $1\le j\le4$ and
$x\in\Int(Y^j)$, where $Y^j$ is the arm $j$
(Definition~\ref{def:N03-rows}).  Then $W$ separates $x$ from every vertex
$y\ne x$ of $G$ (Definition~\ref{def:N01-resolve}).
\end{lemma}

\begin{proof}
Let $j^*$ be the other index with $\pi_{j^*}=\pi_j$, that is,
$\{j,j^*\}=\{1,2\}$ or $\{j,j^*\}=\{3,4\}$
(Definition~\ref{def:N03-graph}).  By Definition~\ref{def:N03-pattern},
$w_1,\dots,w_4$ lie in $W$.  As $x\in\Int(Y^j)$,
Corollary~\ref{cor:N04-arms}(ii) with $(a,b)=(j,j^*)$ at $u:=x$ gives
\begin{equation}\label{eq:N06-arm-x}
 d(w_j,x)<d(w_{j^*},x).
\end{equation}
Let $y\ne x$ be a vertex of $G$.  By Lemma~\ref{lem:N03-structure}(i), $y$
lies in $\Int(Y^j)$, in $\Int(Y^{j^*})$, or in neither, and (1)--(3) treat
these cases.

(1) \emph{The case $y\in\Int(Y^j)$.}  Let $j':=3$ if $j\le2$ and $j':=1$
if $j\ge3$.  By Lemma~\ref{lem:N06-arm-pairs}, $w_j$ or $w_{j'}$ separates
$x$ and $y$, and both lie in $W$.

(2) \emph{The case $y\in\Int(Y^{j^*})$.}  Corollary~\ref{cor:N04-arms}(ii)
with $(a,b)=(j^*,j)$ at $u:=y\in\Int(Y^{j^*})$ gives
$d(w_{j^*},y)<d(w_j,y)$.  Suppose that neither $w_j$ nor $w_{j^*}$
separates $x$ and $y$.  We then have
\[
\begin{aligned}
 d(w_j,x)&<d(w_{j^*},x) &&\text{by \eqref{eq:N06-arm-x}},\\
 d(w_{j^*},x)&=d(w_{j^*},y) &&\text{by the supposition},\\
 d(w_{j^*},y)&<d(w_j,y) &&\text{by the inequality above},\\
 d(w_j,y)&=d(w_j,x) &&\text{by the supposition}.
\end{aligned}
\]
Chaining these four relations gives $d(w_j,x)<d(w_j,x)$, a
contradiction.  So $w_j$ or $w_{j^*}$ separates $x$ and $y$.

(3) \emph{The case $y\notin\Int(Y^j)\cup\Int(Y^{j^*})$.}  Suppose that
neither $w_j$ nor $w_{j^*}$ separates $x$ and $y$.  By the case
hypothesis, $y$ lies neither in $\Int(Y^j)$ nor in $\Int(Y^{j^*})$, so
Corollary~\ref{cor:N04-arms}(i) applies with $(a,b)=(j,j^*)$ at $v:=y$ and
gives $d(w_j,y)=d(w_{j^*},y)$.  Hence
\[
 d(w_j,x)=d(w_j,y)=d(w_{j^*},y)=d(w_{j^*},x).
\]
This contradicts \eqref{eq:N06-arm-x}.  So $w_j$ or $w_{j^*}$ separates
$x$ and $y$.
\end{proof}

\begin{lemma}[four vertices of the skeleton]\label{lem:N06-skel-easy}
Let $\varphi$ be a restricted formula and $G:=G_\varphi$
(Definition~\ref{def:N03-graph}); $G$ is connected
(Lemma~\ref{lem:N04-potentials}).  Let $w_1,w_3$ be the rows of the arms
$1$ and $3$, and for a level $k$ and an option $c$, $\rho_{k,c}$ the core
of $U_{k,c}$ (Definition~\ref{def:N03-rows}).  Let $y$ be a vertex of
$G$.
\begin{enumerate}
\item[\rm(i)] If $y\ne T_1$ and $y\notin\Int(Y^3)$, then $w_3$ separates
  $T_1$ and $y$.
\item[\rm(ii)] If $y\ne T_2$ and $y\notin\Int(Y^1)\cup\Int(Y^3)$, then
  $w_1$ or $w_3$ separates $T_2$ and $y$.
\item[\rm(iii)] Let $k$ be a level and $c$ an option.  If $y\ne t_k$ and
  $y\notin\Int(U_{k,c})$, then $\rho_{k,c}$ separates $t_k$ and $y$.
\item[\rm(iv)] If $y\ne q'$, then $q$, $w_1$ or $w_3$ separates $q'$
  and $y$.
\end{enumerate}
\end{lemma}

\begin{proof}
Let $D$, $p$, $\alpha$ and $\beta_k$ be the numbers of
Definition~\ref{def:N03-params}, and $\delta_1=d(T_1,\cdot)$,
$\delta_2=d(T_2,\cdot)$ (Lemma~\ref{lem:N04-potentials}).  By
Definition~\ref{def:N03-params} and Lemma~\ref{lem:N03-params}(i),(iii),
$D-p=p+7$, $\alpha=\beta_k+\sigma_k=\beta_k+2k$, and $D-7=4m+8\ge16$ as
$m\ge3$ (Lemma~\ref{lem:N03-params}(iii)).

(1) We prove (i).  By the second formula of
Lemma~\ref{lem:N04-rows}(iii) with $j=3$ and $i=0$, as $y^3_0=T_1$ and
$\pi_3=7$ (Definition~\ref{def:N03-graph}), $d(w_3,T_1)=|0-\pi_3|=7$.  As
$y\notin\Int(Y^3)$, the first formula of Lemma~\ref{lem:N04-rows}(iii) with
$j=3$ gives
\[
 d(w_3,y)=\min\bigl(\pi_3+\delta_1(y),\ D-\pi_3+\delta_2(y)\bigr)
 =\min\bigl(7+\delta_1(y),\ D-7+\delta_2(y)\bigr).
\]
Here $\delta_1(y)\ge1$ as $y\ne T_1$, and $D-7\ge16$.  So
$d(w_3,y)\ge\min(8,16)=8>7=d(w_3,T_1)$.

(2) We prove (ii).  We first compute the distances from $w_1$ and $w_3$
to $T_2$.  The second formula of Lemma~\ref{lem:N04-rows}(iii) gives
$d(w_j,y^j_i)=|i-\pi_j|$ for $1\le j\le4$ and $0\le i\le D$.

With $j=1$ and $i=D$: as $y^1_D=T_2$ (Definition~\ref{def:N03-graph}), it
gives $d(w_1,T_2)=|D-\pi_1|$.  As $\pi_1=p$
(Definition~\ref{def:N03-graph}) and $D-p=p+7>0$,
\[
 d(w_1,T_2)=|D-p|=D-p=p+7.
\]

With $j=3$ and $i=D$: as $y^3_D=T_2$, it gives $d(w_3,T_2)=|D-\pi_3|$.
As $\pi_3=7$ (Definition~\ref{def:N03-graph}) and $D-7\ge16>0$,
\begin{equation}\label{eq:N06-w3T2}
 d(w_3,T_2)=|D-7|=D-7.
\end{equation}

Suppose that neither $w_1$
nor $w_3$ separates $T_2$ and $y$; in particular
$d(w_1,y)=d(w_1,T_2)=p+7$.  As
$y\notin\Int(Y^1)$, the first formula of Lemma~\ref{lem:N04-rows}(iii) with
$j=1$ gives, as $\pi_1=p$ and $D-p=p+7$,
\[
 d(w_1,y)=\min\bigl(\pi_1+\delta_1(y),\ D-\pi_1+\delta_2(y)\bigr)
 =\min\bigl(p+\delta_1(y),\ p+7+\delta_2(y)\bigr).
\]
So $p+7=\min(p+\delta_1(y),\ p+7+\delta_2(y))$.  As $y\ne T_2$, $\delta_2(y)\ge1$, so the second term exceeds $p+7$, and
$p+\delta_1(y)=p+7$, that is, $\delta_1(y)=7$.  As $y\notin\Int(Y^3)$ and
$\delta_1(y)=7\le D-7=D-\pi_3$, Corollary~\ref{cor:N04-other-arm}(i) with
$a:=3$ gives $d(w_3,y)=\pi_3+\delta_1(y)=7+7=14$.  With
\eqref{eq:N06-w3T2} and $D-7\ge16$, $d(w_3,T_2)\ge16$.  This contradicts
$d(w_3,y)=d(w_3,T_2)$.  So $w_1$ or $w_3$ separates $T_2$ and $y$.

(3) We prove (iii).  By Lemma~\ref{lem:N04-rows}(i),
$d(\rho_{j,c},t_k)=\beta_j+|j-k|$; with $j=k$,
$d(\rho_{k,c},t_k)=\beta_k$.  As $y\notin\Int(U_{k,c})$,
Corollary~\ref{cor:N04-cores}(i) gives
\[
 d(\rho_{k,c},y)=\min\bigl(\alpha+d(s_k,y),\ \beta_k+d(t_k,y)\bigr).
\]
Both terms exceed $\beta_k$: $\alpha+d(s_k,y)\ge\alpha>\beta_k$, and
$\beta_k+d(t_k,y)\ge\beta_k+1$ as $y\ne t_k$.  So $d(\rho_{k,c},y)>d(\rho_{k,c},t_k)$, and
$\rho_{k,c}$ separates $t_k$ and $y$.

(4) We prove (iv).  Let $y\ne q'$.  If $y=q$: $d(q,q)=0\ne2=d(q,q')$
(Lemma~\ref{lem:N04-reading}(iii)).  Let $y\notin\{q,q'\}$.  By
Lemma~\ref{lem:N04-reading}(iii), $d(q,q')=2$ and
$d(q,y)=1+\min(d(t_m,y),d(t_{m-1},y))$.  So $q$ separates $q'$ and $y$
unless $\min(d(t_m,y),d(t_{m-1},y))=1$.
By Definition~\ref{def:N03-graph}, the neighbours of $t_m$ are
\[
 t_{m-1},\quad q,\quad q',\quad v^{m,c}_{L_m-1},\quad r^{m,c}_{\mu_m+1}
 \qquad(c\in\{1,2,3\}),
\]
and the neighbours of $t_{m-1}$ are
\[
 t_m,\quad t_{m-2},\quad q,\quad q',\quad v^{m-1,c}_{L_{m-1}-1},\quad
 r^{m-1,c}_{\mu_{m-1}+1}\qquad(c\in\{1,2,3\}),
\]
where $t_{m-2}$ exists as $m\ge3$ (Lemma~\ref{lem:N03-params}(iii)).  As
$y\notin\{q,q'\}$, the minimum is $1$ exactly when $y$ is a vertex of one
of the five kinds
\begin{equation}\label{eq:N06-fivekinds}
 t_{m-2},\quad v^{m,c}_{L_m-1},\quad r^{m,c}_{\mu_m+1},\quad
 v^{m-1,c}_{L_{m-1}-1},\quad r^{m-1,c}_{\mu_{m-1}+1}\qquad(c\in\{1,2,3\}).
\end{equation}
Let $y$ be such a vertex.  We prove that $w_3$ or $w_1$ separates $q'$
and $y$.

Neither $q'$ nor a vertex of \eqref{eq:N06-fivekinds} lies in
$\Int(Y^1)\cup\Int(Y^3)$, since the vertices of an arm $Y^j$ other than
$T_1,T_2$ are the $y^j_i$ and $\tilde y^j$ (Definition~\ref{def:N03-graph}).
So the first formula of Lemma~\ref{lem:N04-rows}(iii), with $\pi_3=7$ and
$\pi_1=p$, gives for each such vertex $v$
\[
 d(w_3,v)=\min\bigl(7+\delta_1(v),\ D-7+\delta_2(v)\bigr),\qquad
 d(w_1,v)=\min\bigl(p+\delta_1(v),\ D-p+\delta_2(v)\bigr),
\]
where $D-7=4m+8$, $p=2m+4$ and $D-p=2m+11$ (Definition~\ref{def:N03-params}).

We compute $\delta_1(v)$ and $\delta_2(v)$ from the table
\eqref{eq:N04-values}.  The line $q,q'$ gives
$(\delta_1,\delta_2)(q')=(D-6-m,\,7+m)=(3m+9,\,m+7)$.  The line $t_k$ at
$k=m-2$ gives $(\delta_1,\delta_2)(t_{m-2})=(D-7-(m-2),\,7+m-2)=(3m+10,\,m+5)$.
For a member $u$ of $U_{j,c}$, the line $u\in\Int(U_{j,c})$ gives
$\delta_1(u)=j+h_{U_{j,c}}(u)$ and $\delta_2(u)=7+j+\hb_{U_{j,c}}(u)$, and
by Lemma~\ref{lem:N03-inside}(iv) the profiles are
$(h,\hb)(v^{j,c}_{L_j-1})=(L_j-1,\,1)$ and
$(h,\hb)(r^{j,c}_{\mu_j+1})=(L_j+1,\,1)$.  With $L_m=2m+8$ and
$L_{m-1}=2m+10$ (Definition~\ref{def:N03-params}, $L_k=4m+8-2k$):
\[
\begin{array}{lll}
 v & \delta_1(v) & \delta_2(v)\\\hline
 q' & 3m+9 & m+7\\
 t_{m-2} & 3m+10 & m+5\\
 v^{m,c}_{L_m-1} & m+(L_m-1)=3m+7 & 7+m+1=m+8\\
 r^{m,c}_{\mu_m+1} & m+(L_m+1)=3m+9 & 7+m+1=m+8\\
 v^{m-1,c}_{L_{m-1}-1} & (m-1)+(L_{m-1}-1)=3m+8 & 7+(m-1)+1=m+7\\
 r^{m-1,c}_{\mu_{m-1}+1} & (m-1)+(L_{m-1}+1)=3m+10 & 7+(m-1)+1=m+7
\end{array}
\]

The row $w_3$.  In every line of the table, $\delta_2(v)\ge m+5$ and
$\delta_1(v)\le3m+10$, so $D-7+\delta_2(v)\ge5m+13$ and
$7+\delta_1(v)\le3m+17$; as $m\ge3$, $5m+13>3m+17$.  So
$d(w_3,v)=7+\delta_1(v)$ for every $v$ of the table:
\begin{align*}
 d(w_3,q')&=3m+16, & d(w_3,t_{m-2})&=3m+17,\\
 d(w_3,v^{m,c}_{L_m-1})&=3m+14, & d(w_3,r^{m,c}_{\mu_m+1})&=3m+16,\\
 d(w_3,v^{m-1,c}_{L_{m-1}-1})&=3m+15, & d(w_3,r^{m-1,c}_{\mu_{m-1}+1})&=3m+17 .
\end{align*}
So $w_3$ separates $q'$ from every vertex of \eqref{eq:N06-fivekinds}
except $r^{m,c}_{\mu_m+1}$.

The row $w_1$.  By the table, $\delta_1(q')=\delta_1(r^{m,c}_{\mu_m+1})=3m+9$,
so $p+\delta_1(v)=2m+4+3m+9=5m+13$ for both $v=q'$ and
$v=r^{m,c}_{\mu_m+1}$.  At $v=q'$, $\delta_2(v)=m+7$ and
$D-p+\delta_2(v)=2m+11+m+7=3m+18$.  At $v=r^{m,c}_{\mu_m+1}$,
$\delta_2(v)=m+8$ and $D-p+\delta_2(v)=2m+11+m+8=3m+19$.  As $m\ge3$,
$5m+13\ge3m+19$, so the minimum is the second term in both cases:
\[
 d(w_1,q')=3m+18,\qquad d(w_1,r^{m,c}_{\mu_m+1})=3m+19 .
\]
So $w_1$ separates $q'$ and $r^{m,c}_{\mu_m+1}$.  This proves (iv).
\end{proof}

\begin{lemma}[the vertices $z_i$]\label{lem:N06-skel-z}
Let $\varphi$ be a restricted formula and $G:=G_\varphi$
(Definition~\ref{def:N03-graph}); $G$ is connected
(Lemma~\ref{lem:N04-potentials}).  Let $w_1$ be the row of the arm $1$
and, for an option $c$, $\rho_{1,c}$ the core of $U_{1,c}$
(Definition~\ref{def:N03-rows}).  Let $1\le i\le7$, and $y\ne z_i$ a
vertex of $G$ in no $\Int(Y^j)$, $1\le j\le4$, with $y\notin\Int(U_{1,c})$.
Then $\rho_{1,c}$ or $w_1$ separates $z_i$ and $y$.
\end{lemma}

\begin{proof}
Let $D$, $p$, $\alpha$ and $\beta_1$ be the numbers of
Definition~\ref{def:N03-params}, and $\delta_1=d(T_1,\cdot)$,
$\delta_2=d(T_2,\cdot)$ (Lemma~\ref{lem:N04-potentials}).  By
Definition~\ref{def:N03-params} and Lemma~\ref{lem:N03-params}(iii),
$D=4m+15\ge27$.  By
Lemma~\ref{lem:N03-params}(i), $\alpha=\beta_1+\sigma_1=\beta_1+2$ and
$D-p=p+7$.  By Definition~\ref{def:N03-graph}, $\pi_1=p$.

(1) \emph{The distances from $z_i$.}  By the line $z_i$ of the table
\eqref{eq:N04-values}, $\delta_1(z_i)=D-8+i$ and $\delta_2(z_i)=8-i$.  The
vertex $z_i$ lies in $S$, so $z_i\notin\Int(U_{1,c})$ and
$z_i\notin\Int(Y^1)$ (Lemma~\ref{lem:N03-structure}(i)).

Corollary~\ref{cor:N04-cores}(i) therefore applies at the level $1$, the
option $c$ and $v:=z_i$, and gives
\[
 d(\rho_{1,c},z_i)=\min\bigl(\alpha+d(s_1,z_i),\ \beta_1+d(t_1,z_i)\bigr).
\]
By Lemma~\ref{lem:N04-skeleton}(iii), $d(t_1,z_i)=i$.  By the triangle
inequality and the symmetry of $d$ (Definition~\ref{def:N01-graph}),
\[
 \delta_1(z_i)=d(T_1,z_i)\le d(T_1,s_1)+d(s_1,z_i)=\delta_1(s_1)+d(s_1,z_i).
\]
By the lines $z_i$ and $s_k$ of the table \eqref{eq:N04-values}, at $k=1$,
$\delta_1(z_i)=D-8+i$ and $\delta_1(s_1)=1$.  Hence
\[
 d(s_1,z_i)\ge\delta_1(z_i)-\delta_1(s_1)=(D-8+i)-1=D-9+i.
\]
So
$\alpha+d(s_1,z_i)\ge\alpha+D-9+i>\beta_1+i$, as $\alpha>\beta_1$ and
$D\ge9$.  Hence $d(\rho_{1,c},z_i)=\beta_1+i$.

As $z_i\notin\Int(Y^1)$ and $\delta_2(z_i)=8-i\le7\le p=\pi_1$,
Corollary~\ref{cor:N04-other-arm}(ii) with $a:=1$ gives
\[
 d(w_1,z_i)=D-\pi_1+\delta_2(z_i)=(D-p)+(8-i)=p+15-i.
\]

(2) \emph{If neither separates.}  Suppose that neither $\rho_{1,c}$ nor
$w_1$ separates $z_i$ and $y$.  As $y\notin\Int(U_{1,c})$,
Corollary~\ref{cor:N04-cores}(i) applies at the level $1$, the option $c$
and $v:=y$, and gives, with $\alpha=\beta_1+2$,
\[
 d(\rho_{1,c},y)=\min\bigl(\alpha+d(s_1,y),\ \beta_1+d(t_1,y)\bigr)
 =\min\bigl(\beta_1+2+d(s_1,y),\ \beta_1+d(t_1,y)\bigr).
\]
As $d(\rho_{1,c},y)=d(\rho_{1,c},z_i)=\beta_1+i$, one of the two terms
equals $\beta_1+i$, so $d(s_1,y)=i-2$ or $d(t_1,y)=i$.

(3) \emph{The case $d(s_1,y)=i-2$.}  By the triangle inequality and the
line $s_k$ of the table \eqref{eq:N04-values} at $k=1$,
\[
 \delta_1(y)\le\delta_1(s_1)+d(s_1,y)=1+(i-2)=i-1\le6.
\]
As $y\notin\Int(Y^1)$ and $\delta_1(y)\le6\le D-p=D-\pi_1$,
Corollary~\ref{cor:N04-other-arm}(i) with $a:=1$ gives
$d(w_1,y)=\pi_1+\delta_1(y)=p+\delta_1(y)\le p+6$.  By (1),
$d(w_1,z_i)=p+15-i\ge p+8$, as $i\le7$.  This contradicts
$d(w_1,y)=d(w_1,z_i)$.

(4) \emph{The case $d(t_1,y)=i$.}  By the triangle inequality and the line
$t_k$ of the table \eqref{eq:N04-values} at $k=1$,
\[
 \delta_1(y)\ge\delta_1(t_1)-d(t_1,y)=(D-8)-i.
\]
As $y\notin\Int(Y^1)$, the first formula of Lemma~\ref{lem:N04-rows}(iii)
with $j=1$ gives, as $\pi_1=p$ and $D-p=p+7$,
\[
 d(w_1,y)=\min\bigl(\pi_1+\delta_1(y),\ D-\pi_1+\delta_2(y)\bigr)
 =\min\bigl(p+\delta_1(y),\ p+7+\delta_2(y)\bigr).
\]
The first term is at least $p+D-8-i>p+15-i$, as $D\ge27>23$.  As
$d(w_1,y)=d(w_1,z_i)=p+15-i$, the minimum is its second term.  So
$p+7+\delta_2(y)=p+15-i$, that is,
$\delta_2(y)=(p+15-i)-(p+7)=8-i$.  As $1\le i\le7$, $1\le8-i\le7$.
As $y$ lies in no $\Int(Y^j)$, Lemma~\ref{lem:N04-small}(ii) gives
$y=z_{8-(8-i)}=z_i$.  This contradicts $y\ne z_i$.
\end{proof}

\begin{lemma}[the vertices $s_k$]\label{lem:N06-skel-s}
Let $\varphi$ be a restricted formula and $G:=G_\varphi$
(Definition~\ref{def:N03-graph}); $G$ is connected
(Lemma~\ref{lem:N04-potentials}).  Let $w_3$ be the row of the arm $3$,
$k$ a level, $c$ an option and $\rho_{k,c}$ the core of $U_{k,c}$
(Definition~\ref{def:N03-rows}).  Let $y\ne s_k$ be a vertex of $G$ in no
$\Int(Y^j)$, $1\le j\le4$.  Then $\rho_{k,c}$ or $w_3$ separates $s_k$ and
$y$.
\end{lemma}

\begin{proof}
Let $D$, $\alpha$, $\beta_k$, $\sigma_k$ and $\mu_j$ ($1\le j\le m$) be the
numbers of Definition~\ref{def:N03-params}, and $\delta_1=d(T_1,\cdot)$,
$\delta_2=d(T_2,\cdot)$ (Lemma~\ref{lem:N04-potentials}).  By
Definition~\ref{def:N03-params} and Lemma~\ref{lem:N03-params}(i),
$\alpha=\beta_k+\sigma_k=\beta_k+2k$ and $D-7=4m+8$.  By
Definition~\ref{def:N03-graph}, $\pi_3=7$.

(1) \emph{The distances from $s_k$.}  By Lemma~\ref{lem:N04-rows}(i),
$d(\rho_{j,c},s_k)=\alpha+|j-k|$; with $j=k$, $d(\rho_{k,c},s_k)=\alpha$.
By the line $s_k$ of the table \eqref{eq:N04-values},
$\delta_1(s_k)=k\le m\le D-7=D-\pi_3$.  As $s_k\notin\Int(Y^3)$
(Lemma~\ref{lem:N03-structure}(i)), Corollary~\ref{cor:N04-other-arm}(i)
with $a:=3$ gives
\[
 d(w_3,s_k)=\pi_3+\delta_1(s_k)=7+k.
\]

(2) \emph{If neither separates, then $\delta_1(y)=k$.}  Suppose that
neither $\rho_{k,c}$ nor $w_3$ separates $s_k$ and $y$.  As
$y\notin\Int(Y^3)$, the first formula of Lemma~\ref{lem:N04-rows}(iii)
with $j=3$ gives, as $\pi_3=7$,
\[
 d(w_3,y)=\min\bigl(\pi_3+\delta_1(y),\ D-\pi_3+\delta_2(y)\bigr)
 =\min\bigl(7+\delta_1(y),\ D-7+\delta_2(y)\bigr).
\]
As $\delta_2(y)\ge0$ and $k\le m$, the second term is at least
$D-7=4m+8>7+k$.  As $d(w_3,y)=d(w_3,s_k)=7+k$, the minimum is its first
term.  So $7+\delta_1(y)=7+k$, that is, $\delta_1(y)=k$.

(3) \emph{Where $y$ lies.}  We apply Lemma~\ref{lem:N04-small} with
$u:=y$.  Its hypothesis holds, as $y$ lies in no $\Int(Y^j)$,
$1\le j\le4$.  By (2), $\delta_1(y)=k$, and $1\le k\le m$, so
Lemma~\ref{lem:N04-small}(i) applies.  It gives two possibilities.

The first is $y=s_{\delta_1(y)}=s_k$.  It is excluded, as $y\ne s_k$.

The second is that $y$ is a member of a copy $U_{j,c'}$ with
$j<\delta_1(y)=k$ and $h_{U_{j,c'}}(y)=\delta_1(y)-j=k-j$.

(4) \emph{$y$ is diagonal.}  As $j<k$, $(k,c)\ne(j,c')$, so
Corollary~\ref{cor:N04-levels}(i), at the level $j$, the option $c'$ and
the member $y$ of $U_{j,c'}$, applies to the row $\rho_{k,c}$:
\[
 d(\rho_{k,c},y)=|k-j|+\beta_k
 +\min\bigl(\sigma_k+h_{U_{j,c'}}(y),\ \hb_{U_{j,c'}}(y)\bigr).
\]
Here $|k-j|=k-j$, as $j<k$, and
$\sigma_k+h_{U_{j,c'}}(y)=2k+(k-j)=3k-j$.  So
\[
 d(\rho_{k,c},y)=(k-j)+\beta_k+\min\bigl(3k-j,\ \hb_{U_{j,c'}}(y)\bigr).
\]
As $\rho_{k,c}$ does not separate $s_k$ and $y$,
$d(\rho_{k,c},y)=d(\rho_{k,c},s_k)=\alpha=\beta_k+2k$.  So
\[
 \min\bigl(3k-j,\ \hb_{U_{j,c'}}(y)\bigr)=(\beta_k+2k)-(k-j)-\beta_k=k+j.
\]
As $j<k$, $(3k-j)-(k+j)=2(k-j)>0$, so the minimum is its second term:
$\hb_{U_{j,c'}}(y)=k+j$.  So, with the skew of
Definition~\ref{def:N03-names},
\[
 \kappa_{U_{j,c'}}(y)=\hb_{U_{j,c'}}(y)-h_{U_{j,c'}}(y)=(k+j)-(k-j)=2j
 =\sigma_j,
\]
that is, $y$ is a diagonal member of $U_{j,c'}$.

(5) \emph{The contradiction.}  By (4), $y$ is a diagonal member of
$U_{j,c'}$.  So Lemma~\ref{lem:N03-inside}(v), at the level $j$ and the
option $c'$, applies to the member $y$, and gives
$h_{U_{j,c'}}(y)\ge\mu_j$, where $\mu_j=2m+4-2j$
(Definition~\ref{def:N03-params}).  But by (3),
$h_{U_{j,c'}}(y)=k-j\le m-j$, and $m-j<2(m-j)+4=\mu_j$ as $m\ge j$.  So
$h_{U_{j,c'}}(y)<\mu_j$, a contradiction.
\end{proof}

\begin{lemma}[the skeleton]\label{lem:N06-skeleton}
Let $\varphi$ be a restricted formula and $G:=G_\varphi$
(Definition~\ref{def:N03-graph}); $G$ is connected
(Lemma~\ref{lem:N04-potentials}).  Let $f$ be a choice and $W:=W_f$ its
pattern (Definition~\ref{def:N03-pattern}), and $S$ the skeleton
(Definition~\ref{def:N03-rows}).  Let $x\in S$.  Then $W$ separates $x$
from every vertex $y\ne x$ of $G$ (Definition~\ref{def:N01-resolve}).
\end{lemma}

\begin{proof}
For each level $k$ fix an option $c_k\ne f_k$.  By
Definition~\ref{def:N03-pattern}, $\rho_{k,c_k}$ ($1\le k\le m$), $q$,
$w_1$ and $w_3$ lie in $W$.  Let $y\ne x$ be a vertex of $G$.  If $y$ is a
member of a rowed copy, Lemma~\ref{lem:N06-rowed}, applied to $y$, shows
that $W$ separates $y$ from $x$.  If $y\in\Int(Y^j)$ for some $j$,
Lemma~\ref{lem:N06-arms}, applied to $y$, does the same.  So from now on
$y$ lies in no $\Int(Y^j)$, and in no $\Int(U_{k,c_k})$, as these copies
are rowed.

By Definition~\ref{def:N03-rows}, $x$ is $T_1$, $T_2$, $t_k$ or $s_k$ for
a level $k$, $q'$, or $z_i$ for $1\le i\le7$.
\begin{itemize}
\item If $x=T_1$: as $y\notin\Int(Y^3)$, Lemma~\ref{lem:N06-skel-easy}(i)
  shows that $w_3$ separates $x$ and $y$.
\item If $x=T_2$: as $y\notin\Int(Y^1)\cup\Int(Y^3)$,
  Lemma~\ref{lem:N06-skel-easy}(ii) shows that $w_1$ or $w_3$ separates
  them.
\item If $x=t_k$: as $y\notin\Int(U_{k,c_k})$,
  Lemma~\ref{lem:N06-skel-easy}(iii) with $c:=c_k$ shows that
  $\rho_{k,c_k}$ separates them.
\item If $x=q'$: Lemma~\ref{lem:N06-skel-easy}(iv) shows that
  $q$, $w_1$ or $w_3$ separates them.
\item If $x=z_i$: $y$ lies in no arm and $y\notin\Int(U_{1,c_1})$, so
  Lemma~\ref{lem:N06-skel-z} with $c:=c_1$ shows
  that $\rho_{1,c_1}$ or $w_1$ separates them.
\item If $x=s_k$: as $y$ lies in no arm, Lemma~\ref{lem:N06-skel-s} with
  $c:=c_k$ shows that $\rho_{k,c_k}$ or $w_3$ separates them.
\end{itemize}
In each case $y\ne x$, as assumed.
All these landmarks lie in $W$.
\end{proof}

\begin{lemma}[the equations of two levels]\label{lem:N06-cross-eqs}
Let $\varphi$ be a restricted formula and $G:=G_\varphi$
(Definition~\ref{def:N03-graph}); $G$ is connected
(Lemma~\ref{lem:N04-potentials}).  Let $L_m$, $\beta_j$ and $\sigma_j$
($1\le j\le m$) be the numbers of Definition~\ref{def:N03-params}, $w_3$
the row of the arm $3$ and $\rho_{j,c}$ the cores
(Definition~\ref{def:N03-rows}), and $\delta_1=d(T_1,\cdot)$
(Lemma~\ref{lem:N04-potentials}).  Let $k<k'$ be levels, $\Delta:=k'-k$,
$a$ and $b$ options (Definition~\ref{def:N03-options}), $x$ a member of
$U_{k,a}$ and $y$ a member of $U_{k',b}$; profiles and skews of $x$ are
taken in $U_{k,a}$, and those of $y$ in $U_{k',b}$
(Definition~\ref{def:N03-names}).  Let $c\ne a$ and $c'\ne b$ be options.
\begin{enumerate}
\item[\rm(i)] If $w_3$ does not separate $x$ and $y$, then
  $\delta_1(x)=\delta_1(y)$, that is, $h_{U_{k',b}}(y)=h_{U_{k,a}}(x)-\Delta$.
\item[\rm(ii)] Suppose that neither $w_3$ nor $q$ separates $x$ and
  $y$.  If $k'<m$, then
  \[
   \min\bigl(L_m,\ \kappa_{U_{k',b}}(y)-1\bigr)
   =\min\bigl(L_m,\ \kappa_{U_{k,a}}(x)-1\bigr)+2\Delta ;
  \]
  if $k'=m$, then
  \[
   \kappa_{U_{m,b}}(y)=\min\bigl(L_m,\ \kappa_{U_{k,a}}(x)-1\bigr)+2\Delta .
  \]
\item[\rm(iii)] If neither $w_3$ nor $\rho_{k,c}$ separates $x$ and $y$,
  then $\min(\sigma_k,\kappa_{U_{k,a}}(x))=\min(\sigma_k,\kappa_{U_{k',b}}(y))$.
\item[\rm(iv)] If neither $w_3$ nor $\rho_{k',c'}$ separates $x$ and $y$,
  then
  $\min(\sigma_{k'},\kappa_{U_{k',b}}(y))=\min(\sigma_{k'},\kappa_{U_{k,a}}(x))+2\Delta$.
\end{enumerate}
\end{lemma}

\begin{proof}
Put $U:=U_{k,a}$ and $U':=U_{k',b}$.  We first record an identity.  By
Definition~\ref{def:N03-names}, $\kappa_U(x)=\hb_U(x)-h_U(x)$, so
$\hb_U(x)=h_U(x)+\kappa_U(x)$.  Adding $h_U(x)$ to both terms of a minimum
adds it to the minimum, so for every integer $s$,
\begin{equation}\label{eq:N06-skewmin}
 \min\bigl(s+h_U(x),\ \hb_U(x)\bigr)=h_U(x)+\min\bigl(s,\ \kappa_U(x)\bigr),
\end{equation}
and likewise at $y$ in $U'$:
$\min(s+h_{U'}(y),\ \hb_{U'}(y))=h_{U'}(y)+\min(s,\ \kappa_{U'}(y))$.

(1) We prove (i).  Suppose that $w_3$ does not separate $x$ and $y$.
Corollary~\ref{cor:N04-levels}(iii) at the level $k$, the option $a$ and
the member $x$ gives $d(w_3,x)=7+\delta_1(x)=7+k+h_U(x)$; at the level
$k'$, the option $b$ and the member $y$ it gives
$d(w_3,y)=7+\delta_1(y)=7+k'+h_{U'}(y)$.  As $w_3$ does not separate $x$
and $y$, $d(w_3,x)=d(w_3,y)$.  So $7+\delta_1(x)=7+\delta_1(y)$, that is,
$\delta_1(x)=\delta_1(y)$; and $7+k+h_U(x)=7+k'+h_{U'}(y)$, that is,
$h_{U'}(y)=h_U(x)-(k'-k)=h_U(x)-\Delta$.

(2) We prove (ii).  Suppose that neither $w_3$ nor $q$ separates $x$ and
$y$.  By (i), $h_{U'}(y)=h_U(x)-\Delta$.  As $k<k'\le m$, $k<m$, so
Corollary~\ref{cor:N04-levels}(ii) gives
$d(q,x)=m-k+1+\min(L_m+h_U(x),\ \hb_U(x)-1)$.  As
$\hb_U(x)-1=h_U(x)+(\kappa_U(x)-1)$ (Definition~\ref{def:N03-names}),
\begin{equation}\label{eq:N06-dqx}
 d(q,x)=m-k+1+h_U(x)+\min\bigl(L_m,\ \kappa_U(x)-1\bigr).
\end{equation}
As $q$ does not separate $x$ and $y$, $d(q,y)=d(q,x)$.  We compute
$d(q,y)$ in the two cases of (ii).

\emph{Case $k'<m$.}  Corollary~\ref{cor:N04-levels}(ii) at the level
$k'$ gives $d(q,y)=m-k'+1+\min(L_m+h_{U'}(y),\ \hb_{U'}(y)-1)$, and
$\hb_{U'}(y)-1=h_{U'}(y)+(\kappa_{U'}(y)-1)$, so
\[
 d(q,y)=m-k'+1+h_{U'}(y)+\min\bigl(L_m,\ \kappa_{U'}(y)-1\bigr).
\]
With $k'=k+\Delta$, $h_{U'}(y)=h_U(x)-\Delta$ and \eqref{eq:N06-dqx},
$d(q,y)=d(q,x)$ reads
\[
 m-k+1+h_U(x)-2\Delta+\min\bigl(L_m,\ \kappa_{U'}(y)-1\bigr)
 =m-k+1+h_U(x)+\min\bigl(L_m,\ \kappa_U(x)-1\bigr),
\]
that is, $\min(L_m,\kappa_{U'}(y)-1)=\min(L_m,\kappa_U(x)-1)+2\Delta$.

\emph{Case $k'=m$.}  Corollary~\ref{cor:N04-levels}(ii) at the level $m$
gives $d(q,y)=1+\hb_{U'}(y)$, and $\hb_{U'}(y)=h_{U'}(y)+\kappa_{U'}(y)$,
so
\[
 d(q,y)=1+h_{U'}(y)+\kappa_{U'}(y)=1+h_U(x)-\Delta+\kappa_{U'}(y).
\]
Here $m-k=\Delta$, so \eqref{eq:N06-dqx} reads
$d(q,x)=\Delta+1+h_U(x)+\min(L_m,\kappa_U(x)-1)$, and $d(q,y)=d(q,x)$
reads
\[
 1+h_U(x)-\Delta+\kappa_{U'}(y)
 =\Delta+1+h_U(x)+\min\bigl(L_m,\ \kappa_U(x)-1\bigr),
\]
that is, $\kappa_{U'}(y)=\min(L_m,\kappa_U(x)-1)+2\Delta$.  This proves (ii).

(3) We prove (iii).  Suppose that neither $w_3$ nor $\rho_{k,c}$
separates $x$ and $y$.  By (i), $h_{U'}(y)=h_U(x)-\Delta$.  We apply
Corollary~\ref{cor:N04-levels}(i) to the row $\rho_{k,c}$; it requires the
member's copy to differ from $U_{k,c}$.  At the member $x$ of $U_{k,a}$,
$(k,c)\ne(k,a)$ as $c\ne a$, so, with \eqref{eq:N06-skewmin} at
$s:=\sigma_k$,
\[
 d(\rho_{k,c},x)=|k-k|+\beta_k+\min\bigl(\sigma_k+h_U(x),\ \hb_U(x)\bigr)
 =\beta_k+h_U(x)+\min\bigl(\sigma_k,\ \kappa_U(x)\bigr).
\]
At the member $y$ of $U_{k',b}$, $(k,c)\ne(k',b)$ as $k\ne k'$, so, with
\eqref{eq:N06-skewmin} and then $h_{U'}(y)=h_U(x)-\Delta$,
\[
 d(\rho_{k,c},y)=|k-k'|+\beta_k+h_{U'}(y)+\min\bigl(\sigma_k,\ \kappa_{U'}(y)\bigr)
 =\Delta+\beta_k+h_U(x)-\Delta+\min\bigl(\sigma_k,\ \kappa_{U'}(y)\bigr).
\]
As $\rho_{k,c}$ does not separate $x$ and $y$,
$d(\rho_{k,c},x)=d(\rho_{k,c},y)$.  Subtracting $\beta_k+h_U(x)$ from both
sides gives (iii).

(4) We prove (iv).  Suppose that neither $w_3$ nor $\rho_{k',c'}$
separates $x$ and $y$.  By (i), $h_{U'}(y)=h_U(x)-\Delta$.  We apply
Corollary~\ref{cor:N04-levels}(i) to the row $\rho_{k',c'}$.  At the member
$x$ of $U_{k,a}$, $(k',c')\ne(k,a)$ as $k'\ne k$, so, with
\eqref{eq:N06-skewmin} at $s:=\sigma_{k'}$,
\[
 d(\rho_{k',c'},x)=|k'-k|+\beta_{k'}+h_U(x)+\min\bigl(\sigma_{k'},\ \kappa_U(x)\bigr)
 =\Delta+\beta_{k'}+h_U(x)+\min\bigl(\sigma_{k'},\ \kappa_U(x)\bigr).
\]
At the member $y$ of $U_{k',b}$, $(k',c')\ne(k',b)$ as $c'\ne b$, so
\[
 d(\rho_{k',c'},y)=|k'-k'|+\beta_{k'}+h_{U'}(y)+\min\bigl(\sigma_{k'},\ \kappa_{U'}(y)\bigr)
 =\beta_{k'}+h_U(x)-\Delta+\min\bigl(\sigma_{k'},\ \kappa_{U'}(y)\bigr).
\]
As $\rho_{k',c'}$ does not separate $x$ and $y$, these are equal.
Subtracting $\beta_{k'}+h_U(x)$ from both sides,
\[
 \Delta+\min\bigl(\sigma_{k'},\ \kappa_U(x)\bigr)
 =-\Delta+\min\bigl(\sigma_{k'},\ \kappa_{U'}(y)\bigr),
\]
and adding $\Delta$ to both sides gives (iv).
\end{proof}

\begin{lemma}[free members of two levels]\label{lem:N06-cross}
Let $\varphi$ be a restricted formula and $G:=G_\varphi$
(Definition~\ref{def:N03-graph}); $G$ is connected
(Lemma~\ref{lem:N04-potentials}).  Let $f$ be a choice and $W:=W_f$ its
pattern (Definition~\ref{def:N03-pattern}), and $\delta_1=d(T_1,\cdot)$
(Lemma~\ref{lem:N04-potentials}).  Let $1\le k<k'\le m$, $x$ a member of
the free copy $U_{k,f_k}$ and $y$ a member of the free copy $U_{k',f_{k'}}$
(Definition~\ref{def:N03-pattern}).  If $W$ does not separate $x$ and $y$,
then $k'<m$, $x$ is a diagonal member of $U_{k,f_k}$, $y$ is a diagonal member of
$U_{k',f_{k'}}$ (Definition~\ref{def:N03-names}), and
$\delta_1(x)=\delta_1(y)$.
\end{lemma}

\begin{proof}
Put $U:=U_{k,f_k}$, $U':=U_{k',f_{k'}}$ and $\Delta:=k'-k$; $\Delta\ge1$ as
$k<k'$.  Let $L_m$, $\sigma_k$ and $\sigma_{k'}$ be the numbers of
Definition~\ref{def:N03-params}; by that definition, $\sigma_k=2k$,
$\sigma_{k'}=2k'$ and $L_m=4m+8-2m=2m+8$, so $\sigma_{k'}=\sigma_k+2\Delta$,
and $\sigma_{k'}\le2m<L_m$ as $k'\le m$.  Choose options $c\ne f_k$ and
$c'\ne f_{k'}$.  By Definition~\ref{def:N03-pattern}, $w_3$, $q$,
$\rho_{k,c}$ and $\rho_{k',c'}$ lie in $W$.

Suppose that $W$ does not separate $x$ and $y$.  Then none of these four
landmarks separates them.  So Lemma~\ref{lem:N06-cross-eqs} applies with
$a:=f_k$, $b:=f_{k'}$ and the options $c,c'$: its hypotheses hold, as
$k<k'$, $x$ is a member of $U_{k,f_k}$, $y$ is a member of $U_{k',f_{k'}}$,
$c\ne f_k$ and $c'\ne f_{k'}$.  By Lemma~\ref{lem:N06-cross-eqs}(i),
$\delta_1(x)=\delta_1(y)$.  By Lemma~\ref{lem:N06-cross-eqs}(ii), (iii)
and (iv),
\begin{align}
 \min\bigl(L_m,\ \kappa_{U'}(y)-1\bigr)&=\min\bigl(L_m,\ \kappa_U(x)-1\bigr)+2\Delta
   &&\text{if }k'<m,\label{eq:N06-Eq}\\
 \kappa_{U'}(y)&=\min\bigl(L_m,\ \kappa_U(x)-1\bigr)+2\Delta
   &&\text{if }k'=m,\label{eq:N06-Eqm}\\
 \min\bigl(\sigma_k,\ \kappa_U(x)\bigr)
 &=\min\bigl(\sigma_k,\ \kappa_{U'}(y)\bigr),\label{eq:N06-Erho}\\
 \min\bigl(\sigma_{k'},\ \kappa_{U'}(y)\bigr)
 &=\min\bigl(\sigma_{k'},\ \kappa_U(x)\bigr)+2\Delta.\label{eq:N06-Erho2}
\end{align}

(1) \emph{$\kappa_U(x)\le\sigma_k$.}  The left side of \eqref{eq:N06-Erho2}
is a minimum one of whose terms is $\sigma_{k'}$, so it is at most
$\sigma_{k'}$.  Hence $\min(\sigma_{k'},\kappa_U(x))+2\Delta\le\sigma_{k'}$,
that is, $\min(\sigma_{k'},\kappa_U(x))\le\sigma_{k'}-2\Delta=\sigma_k$.  As
$\sigma_k<\sigma_{k'}$, this minimum is not $\sigma_{k'}$; so it is
$\kappa_U(x)$, and $\kappa_U(x)\le\sigma_k$.

(2) \emph{$\kappa_U(x)=\sigma_k$.}  By (1),
$\min(\sigma_k,\kappa_U(x))=\kappa_U(x)$, so \eqref{eq:N06-Erho} gives
$\kappa_U(x)=\min(\sigma_k,\kappa_{U'}(y))$.  So $\kappa_U(x)=\sigma_k$ or
$\kappa_U(x)=\kappa_{U'}(y)$.  Suppose that $\kappa_U(x)=\kappa_{U'}(y)=:\kappa$;
then $\kappa\le\sigma_k<L_m$, so $\min(L_m,\kappa-1)=\kappa-1$.  If
$k'<m$, \eqref{eq:N06-Eq} gives $\kappa-1=\kappa-1+2\Delta$; if $k'=m$,
\eqref{eq:N06-Eqm} gives $\kappa=\kappa-1+2\Delta$.  Both contradict
$\Delta\ge1$.  So $\kappa_U(x)=\sigma_k$, and $x$ is a diagonal member of
$U$ (Definition~\ref{def:N03-names}).  As $\sigma_k<L_m$,
$\min(L_m,\kappa_U(x)-1)=\sigma_k-1$.

(3) \emph{$k'<m$.}  Suppose that $k'=m$.  Then \eqref{eq:N06-Eqm} and (2)
give $\kappa_{U'}(y)=\sigma_k-1+2\Delta=\sigma_m-1$.  So
the left side of \eqref{eq:N06-Erho2} is $\min(\sigma_m,\sigma_m-1)=\sigma_m-1$,
while its right side is $\min(\sigma_m,\sigma_k)+2\Delta=\sigma_k+2\Delta=\sigma_m$,
as $\sigma_k<\sigma_m$.  This contradicts \eqref{eq:N06-Erho2}.  So $k'<m$.

(4) \emph{$\kappa_{U'}(y)=\sigma_{k'}$.}  By (3), $k'<m$, so
\eqref{eq:N06-Eq} and (2) give
\[
 \min\bigl(L_m,\kappa_{U'}(y)-1\bigr)=\sigma_k-1+2\Delta=\sigma_{k'}-1 .
\]
As
$\sigma_{k'}-1<L_m$, this minimum is not $L_m$; so it is
$\kappa_{U'}(y)-1$, and $\kappa_{U'}(y)=\sigma_{k'}$.  So $y$ is a diagonal
member of $U'$ (Definition~\ref{def:N03-names}).
\end{proof}

\begin{lemma}[diagonal members of free copies]\label{lem:N06-designed}
Let $\varphi$ be a restricted formula and $G:=G_\varphi$
(Definition~\ref{def:N03-graph}); $G$ is connected
(Lemma~\ref{lem:N04-potentials}).  Let $p$, $L_m$, $\beta_j$ and
$\sigma_j$ ($1\le j\le m$) be the numbers of
Definition~\ref{def:N03-params}, $\pi_i$ ($1\le i\le4$) the positions of
Definition~\ref{def:N03-graph}, $\rho_{j,c}$ and $w_i$ the rows
(Definition~\ref{def:N03-rows}), and $\delta_1=d(T_1,\cdot)$
(Lemma~\ref{lem:N04-potentials}).  Let $f$ be a choice and $W:=W_f$ its
pattern (Definition~\ref{def:N03-pattern}).  Let $k$ be a level and $x$ a
diagonal member of the free copy $U_{k,f_k}$
(Definitions~\ref{def:N03-names} and~\ref{def:N03-pattern}).
\begin{enumerate}
\item[\rm(i)] For every row $\rho_{j,c}\in W$,
  $d(\rho_{j,c},x)=\beta_j+j+\delta_1(x)$.
\item[\rm(ii)] If $k<m$, then $d(q,x)=m+\delta_1(x)$.
\item[\rm(iii)] For $1\le i\le4$, $d(w_i,x)=\pi_i+\delta_1(x)$.
\item[\rm(iv)] $d(T_2,x)=7+\delta_1(x)+2k$.
\item[\rm(v)] If $k<m$, and $k'\ne k$ is a level with $k'<m$, and
  $y$ a diagonal member of
  $U_{k',f_{k'}}$ with $\delta_1(y)=\delta_1(x)$, then $W$ does not
  separate $x$ and $y$.
\end{enumerate}
\end{lemma}

\begin{proof}
Put $U:=U_{k,f_k}$.  As $x$ is a diagonal member of $U$,
$\kappa_U(x)=\sigma_k=2k$ (Definitions~\ref{def:N03-names}
and~\ref{def:N03-params}), and $\kappa_U(x)=\hb_U(x)-h_U(x)$
(Definition~\ref{def:N03-names}); so
\begin{equation}\label{eq:N06-diag-skew}
 \hb_U(x)=h_U(x)+2k .
\end{equation}
Let $\delta_2=d(T_2,\cdot)$ (Lemma~\ref{lem:N04-potentials}).  By the line
$u\in\Int(U_{k,c})$ of the table \eqref{eq:N04-values} of
Lemma~\ref{lem:N04-potentials}, at the member $x$ of $U$, and by
\eqref{eq:N06-diag-skew},
\begin{equation}\label{eq:N06-diag}
 \delta_1(x)=k+h_U(x),\qquad
 \delta_2(x)=7+k+\hb_U(x)=7+k+h_U(x)+2k=7+\delta_1(x)+2k .
\end{equation}

(1) We prove (i).  Let $\rho_{j,c}\in W$.  By
Definition~\ref{def:N03-pattern}, $c\ne f_j$; so $(j,c)\ne(k,f_k)$, and
Corollary~\ref{cor:N04-levels}(i) applies at the level $k$, the option
$f_k$ and the member $x$ of $U$, to the row $\rho_{j,c}$:
\[
 d(\rho_{j,c},x)=|j-k|+\beta_j+\min\bigl(\sigma_j+h_U(x),\ \hb_U(x)\bigr).
\]
By Definition~\ref{def:N03-params}, $\sigma_j=2j$, and by
\eqref{eq:N06-diag-skew}, $\hb_U(x)=h_U(x)+2k$; so
\[
 \min\bigl(\sigma_j+h_U(x),\ \hb_U(x)\bigr)
 =\min\bigl(2j+h_U(x),\ 2k+h_U(x)\bigr)=h_U(x)+2\min(j,k).
\]
As $|j-k|=\max(j,k)-\min(j,k)$, we get
$|j-k|+2\min(j,k)=\max(j,k)+\min(j,k)=j+k$.  Hence, by
\eqref{eq:N06-diag},
\[
 d(\rho_{j,c},x)=j+k+\beta_j+h_U(x)=\beta_j+j+\delta_1(x).
\]

(2) We prove (ii).  Let $k<m$.  Corollary~\ref{cor:N04-levels}(ii) at the
level $k$, the option $f_k$ and the member $x$ of $U$ gives
\[
 d(q,x)=m-k+1+\min\bigl(L_m+h_U(x),\ \hb_U(x)-1\bigr).
\]
By Definition~\ref{def:N03-params}, $L_m=4m+8-2m=2m+8$, and
$2k-1<2m<L_m$; so, by \eqref{eq:N06-diag-skew},
$\min(L_m+h_U(x),\ \hb_U(x)-1)=\min(L_m+h_U(x),\ 2k-1+h_U(x))=2k-1+h_U(x)$.
Hence, by \eqref{eq:N06-diag},
\[
 d(q,x)=m-k+1+2k-1+h_U(x)=m+(k+h_U(x))=m+\delta_1(x).
\]

(3) We prove (iii) for $i=3,4$.  Corollary~\ref{cor:N04-levels}(iii) at
the level $k$, the option $f_k$ and the member $x$ of $U$ gives
$d(w_3,x)=d(w_4,x)=7+\delta_1(x)$.  By Definition~\ref{def:N03-graph},
$\pi_3=\pi_4=7$.  So $d(w_i,x)=\pi_i+\delta_1(x)$ for $i=3,4$.

(4) We prove (iii) for $i=1,2$.  Corollary~\ref{cor:N04-levels}(iv) at
the level $k$, the option $f_k$ and the member $x$ of $U$ gives
\[
 d(w_1,x)=d(w_2,x)=\min\bigl(p+\delta_1(x),\ p+7+\delta_2(x)\bigr).
\]
By \eqref{eq:N06-diag}, $p+7+\delta_2(x)=p+14+\delta_1(x)+2k>p+\delta_1(x)$.
So the minimum is $p+\delta_1(x)$.  By Definition~\ref{def:N03-graph},
$\pi_1=\pi_2=p$.  So $d(w_i,x)=\pi_i+\delta_1(x)$ for $i=1,2$.

(5) We prove (iv).  By \eqref{eq:N06-diag},
$d(T_2,x)=\delta_2(x)=7+\delta_1(x)+2k$.

(6) We prove (v).  Let $k<m$, $k'\ne k$ a level with $k'<m$, and $y$ a diagonal member of
$U_{k',f_{k'}}$ with $\delta_1(y)=\delta_1(x)$.  By
Definition~\ref{def:N03-pattern}, the landmarks of $W$ are the rows
$\rho_{j,c}$ with $c\ne f_j$, the vertex $q$ and the rows $w_1,\dots,w_4$.
Items (i), (ii) and (iii) hold at $x$, and also at $y$ with $k'$ in place
of $k$, as $y$ is a diagonal member of the free copy $U_{k',f_{k'}}$.
In these formulas the level $k$ of $x$ does not appear; so at $y$ they
give the same expressions with $\delta_1(y)$ in place of $\delta_1(x)$.
So, with $\delta_1(x)=\delta_1(y)$,
\[
\begin{aligned}
 d(\rho_{j,c},x)&=\beta_j+j+\delta_1(x)=\beta_j+j+\delta_1(y)=d(\rho_{j,c},y)
   &&(\rho_{j,c}\in W),\\
 d(q,x)&=m+\delta_1(x)=m+\delta_1(y)=d(q,y),\\
 d(w_i,x)&=\pi_i+\delta_1(x)=\pi_i+\delta_1(y)=d(w_i,y) &&(1\le i\le4).
\end{aligned}
\]
So no landmark of $W$ separates $x$ and $y$
(Definition~\ref{def:N01-resolve}).
\end{proof}

\begin{theorem}[validity of patterns]\label{thm:N06-validity}
Let $\varphi$ be a restricted formula and $G:=G_\varphi$
(Definition~\ref{def:N03-graph}); $G$ is connected
(Lemma~\ref{lem:N04-potentials}).  Let $f$ be a choice and $W_f$ its
pattern, with the free copies $U_{k,f_k}$
(Definition~\ref{def:N03-pattern}), and $\delta_1=d(T_1,\cdot)$
(Lemma~\ref{lem:N04-potentials}).
\begin{enumerate}
\item[\rm(i)] Two distinct vertices $x,y$ of $G$ are not separated by
  $W_f$ if and only if there are levels $k\ne k'$, both less than
  $m$, such that
  \begin{itemize}
  \item $x$ is a diagonal member of $U_{k,f_k}$
    (Definition~\ref{def:N03-names}),
  \item $y$ is a diagonal member of $U_{k',f_{k'}}$, and
  \item $\delta_1(x)=\delta_1(y)$.
  \end{itemize}
\item[\rm(ii)] $W_f\cup\{T_2\}$ resolves $G$.
\end{enumerate}
\end{theorem}

\begin{proof}
Put $W:=W_f$.

(1) We prove the ``if'' part of (i).  Let $k\ne k'$ be levels less
than $m$, $x$ a
diagonal member of $U_{k,f_k}$ and $y$ a diagonal member of
$U_{k',f_{k'}}$ with $\delta_1(x)=\delta_1(y)$.  By
Lemma~\ref{lem:N06-designed}(v) at the level $k$, the member $x$, the
level $k'$ and the member $y$ (its hypotheses $k<m$ and $k'<m$ hold),
$W$ does not separate $x$ and $y$.

(2) We prove the ``only if'' part of (i).  Let $x\ne y$ be two vertices
of $G$ not separated by $W$.  Let $v$ be $x$ or $y$, and $v'$ the other
one; so $v'\ne v$ and $W$ does not separate $v$ and $v'$.  By
Lemma~\ref{lem:N03-structure}(i), $v$ lies in $S$, or $v=q$, or $v$ is a
member of a copy $U_{\ell,c}$, or $v$ lies in some $\Int(Y^j)$.  Then:
\begin{itemize}
\item $v\notin S$: Lemma~\ref{lem:N06-skeleton} at the vertex $v$ would
  say that $W$ separates $v$ from $v'$;
\item $v\ne q$: $q\in W$ (Definition~\ref{def:N03-pattern}), and a
  landmark separates itself from every other vertex
  (Definition~\ref{def:N01-resolve}), so $W$ would separate $q$ from $v'$;
\item $v$ is not a member of a rowed copy: if $v$ were a member of
  $U_{\ell,c}$ with $c\ne f_\ell$, then Lemma~\ref{lem:N06-rowed} at the
  level $\ell$, the option $c$ and the member $v$ would say that $W$
  separates $v$ from $v'\ne v$;
\item $v$ lies in no $\Int(Y^j)$ ($1\le j\le4$): if it did,
  Lemma~\ref{lem:N06-arms} at $j$ and the vertex $v$ would say that $W$
  separates $v$ from $v'$.
\end{itemize}
So, for some level $\ell$, $v$ is a member of a copy $U_{\ell,c}$ with
$c=f_\ell$, that is, of the free copy $U_{\ell,f_\ell}$
(Definition~\ref{def:N03-pattern}).

Hence $x$ is a member of $U_{k,f_k}$ and $y$ a member of $U_{k',f_{k'}}$
for some levels $k,k'$.  If $k=k'$, then $x\ne y$ are two members of
$U_{k,f_k}$, and Lemma~\ref{lem:N06-copy} at the level $k$, the option
$f_k$ and the members $x,y$ says that $q$ or $w_3$ separates them; both
lie in $W$ (Definition~\ref{def:N03-pattern}), a contradiction.  So
$k\ne k'$.  Exchanging $x$ and $y$ if necessary, we may assume that
$k<k'$: the hypothesis ($W$ does not separate $x$ and $y$) and the
conclusion are symmetric in $x$ and $y$.  Then Lemma~\ref{lem:N06-cross}
applies at the levels $k<k'$, the member $x$ of $U_{k,f_k}$ and the
member $y$ of $U_{k',f_{k'}}$: $k<k'<m$, $x$ and $y$ are diagonal members and
$\delta_1(x)=\delta_1(y)$.

(3) We prove (ii).  Let $x\ne y$ be two vertices of $G$.  If $W$
separates them, so does $W\cup\{T_2\}$.  Otherwise, by (i), $x$ is a
diagonal member of $U_{k,f_k}$ and $y$ a diagonal member of
$U_{k',f_{k'}}$ with $k\ne k'$ and $\delta_1(x)=\delta_1(y)$.
Lemma~\ref{lem:N06-designed}(iv) at the level $k$ and the member $x$
gives $d(T_2,x)=7+\delta_1(x)+2k$, and at the level $k'$ and the member
$y$, $d(T_2,y)=7+\delta_1(y)+2k'$.  As $\delta_1(x)=\delta_1(y)$ and
$k\ne k'$, $d(T_2,x)\ne d(T_2,y)$.  So $T_2$ separates $x$ and $y$.
Hence $W\cup\{T_2\}$ separates every two distinct vertices of $G$, that
is, resolves $G$ (Definition~\ref{def:N01-resolve}).
\end{proof}


\section{Treewidth}\label{sec:N07}

\begin{definition}[tree decompositions]\label{def:N07-td}
Let $H$ be a graph.  A \emph{tree decomposition} of $H$ (Robertson and
Seymour; see \cite[Section~2]{Bodlaender98}) is a tree
(Definition~\ref{def:N01-graph})
$\mathcal T$ together with a set $B_\nu\subseteq V(H)$, the \emph{bag} of
$\nu$, for each node $\nu$ of $\mathcal T$, such that
\begin{enumerate}
\item[\rm(a)] every vertex of $H$ lies in some bag;
\item[\rm(b)] for every edge $xy$ of $H$, some bag contains both $x$ and
  $y$;
\item[\rm(c)] for every vertex $x$ of $H$, the nodes $\nu$ with
  $x\in B_\nu$ induce a connected subtree of $\mathcal T$.
\end{enumerate}
Its \emph{width} is $\max_\nu|B_\nu|-1$.  The \emph{treewidth} of $H$ is
the least width of a tree decomposition of $H$.
\end{definition}

\begin{definition}[two-terminal series--parallel graphs]\label{def:N07-sp}
A \emph{two-terminal graph} $(H;s,t)$ is a graph $H$ with two distinct
vertices $s,t$ of $H$, its \emph{terminals}.  For a path
$(u_0,u_1,\dots,u_\ell)$ in a graph, with $\ell\ge1$, we write
$\Pi(u_0,\dots,u_\ell)$ for the two-terminal graph $(P;u_0,u_\ell)$, where
$P$ is the graph of the path (Definition~\ref{def:N01-graph}).
Let $(H_1;s_1,t_1)$ and $(H_2;s_2,t_2)$ be two-terminal graphs whose
graphs are subgraphs of one graph.
\begin{itemize}
\item If $t_1=s_2$ and $V(H_1)\cap V(H_2)=\{t_1\}$, their \emph{series
  composition} is
  \[
   (H_1;s_1,t_1)\circ(H_2;s_2,t_2):=(H_1\cup H_2;\ s_1,t_2).
  \]
\item If $s_1=s_2$, $t_1=t_2$ and $V(H_1)\cap V(H_2)=\{s_1,t_1\}$, their
  \emph{parallel composition} is
  \[
   (H_1;s_1,t_1)\,\|\,(H_2;s_2,t_2):=(H_1\cup H_2;\ s_1,t_1).
  \]
\end{itemize}
Here $H_1\cup H_2$ is the graph whose vertices and edges are those of
$H_1$ and those of $H_2$.  A \emph{single edge} is a two-terminal graph
$(K_2;s,t)$, where $K_2$ is the graph with the two vertices $s,t$ and the
edge $st$.  A two-terminal graph is \emph{series--parallel}
(Eppstein~\cite{Eppstein92}), written $(H;s,t)\in\mathcal{SP}$, or $H\in\mathcal{SP}$ when the terminals are
fixed by the context, when it is obtained from single edges by a finite
sequence of series and parallel compositions.
\end{definition}

Definition~\ref{def:N07-sp} is the definition of two-terminal series
parallel graphs of \cite[Section~1]{Eppstein92} and of
\cite[Section~8.3]{Bodlaender98}, written for subgraphs of one graph:
there a composition takes a disjoint union and identifies the terminals,
here the two graphs already meet exactly in those vertices.  Two
differences do not matter: \cite{Eppstein92} orients every edge from the
source to the sink and takes undirected graphs through an orientation,
and both allow multigraphs, so an edge $s_1t_1$ lying in both graphs of a
parallel composition is one edge here and two parallel edges there.

\begin{lemma}[series--parallel graphs]\label{lem:N07-sp}
Let $\mathcal{SP}$ be the class of Definition~\ref{def:N07-sp}.
\begin{enumerate}
\item[\rm(i)] For every path $(u_0,\dots,u_\ell)$ with $\ell\ge1$,
  $\Pi(u_0,\dots,u_\ell)\in\mathcal{SP}$.
\item[\rm(ii)] If $(H;s,t)\in\mathcal{SP}$ and $\psi$ is an isomorphism
  from $H$ onto a graph $H'$, then the two-terminal graph
  $(H';\psi(s),\psi(t))$ is in $\mathcal{SP}$.
\item[\rm(iii)] Every $(H;s,t)\in\mathcal{SP}$ has treewidth at most two
  (Definition~\ref{def:N07-td}) \cite[Theorem~41]{Bodlaender98}.
\item[\rm(iv)] Let $(H;s,t)\in\mathcal{SP}$.  Then $H$ is connected, and
  for every vertex $v$ of $H$, every vertex $u\ne v$ of $H$ is joined to
  $s$ or to $t$ by a path of $H-v$.
\end{enumerate}
\end{lemma}

\begin{proof}
(1) We prove (i), by induction on $\ell$.  If $\ell=1$, then
$\Pi(u_0,u_1)$ is the single edge $(K_2;u_0,u_1)$, in $\mathcal{SP}$ by
Definition~\ref{def:N07-sp}.  Let $\ell\ge2$.  By the induction
hypothesis, $\Pi(u_0,\dots,u_{\ell-1})\in\mathcal{SP}$, and
$\Pi(u_{\ell-1},u_\ell)$ is a single edge.  The vertices of a path are
distinct, so the two paths share exactly the vertex $u_{\ell-1}$, and their
union is the graph of the path $(u_0,\dots,u_\ell)$.  So
\[
 \Pi(u_0,\dots,u_\ell)=\Pi(u_0,\dots,u_{\ell-1})\circ\Pi(u_{\ell-1},u_\ell)
 \in\mathcal{SP}.
\]

(2) We prove (ii), by induction on the number of compositions in a
construction of $(H;s,t)$ from single edges.  If there is none, $(H;s,t)$
is a single edge $(K_2;s,t)$, and $\psi$ maps it onto the single edge
$(\psi(K_2);\psi(s),\psi(t))$, which is in $\mathcal{SP}$.  Otherwise
the last composition is a series composition $(H_1\cup H_2;s_1,t_2)$ or
a parallel composition $(H_1\cup H_2;s_1,t_1)$ of two graphs
$(H_1;s_1,t_1),(H_2;s_2,t_2)\in\mathcal{SP}$ constructed with fewer
compositions, for which (ii) holds by the induction hypothesis.

In the series case, $t_1=s_2$ and $V(H_1)\cap V(H_2)=\{t_1\}$; let $\psi$
be an isomorphism from $H_1\cup H_2$ onto $H'$.  The restrictions of
$\psi$ to $H_1$ and to $H_2$ are isomorphisms onto the subgraphs
$\psi(H_1)$ and $\psi(H_2)$ of $H'$, so, by the induction hypothesis,
$(\psi(H_1);\psi(s_1),\psi(t_1))$ and $(\psi(H_2);\psi(s_2),\psi(t_2))$
are in $\mathcal{SP}$.  As $t_1=s_2$, $\psi(t_1)=\psi(s_2)$; as $\psi$ is
injective,
$V(\psi(H_1))\cap V(\psi(H_2))=\psi(V(H_1)\cap V(H_2))=\{\psi(t_1)\}$; and
as $\psi$ is onto, $\psi(H_1)\cup\psi(H_2)=\psi(H_1\cup H_2)=H'$.  So
$(H';\psi(s_1),\psi(t_2))$ is the series composition of two members of
$\mathcal{SP}$, and is in $\mathcal{SP}$.

In the parallel case, $s_1=s_2$, $t_1=t_2$ and
$V(H_1)\cap V(H_2)=\{s_1,t_1\}$; let $\psi$ be an isomorphism from
$H_1\cup H_2$ onto $H'$.  The restrictions of $\psi$ to $H_1$ and to $H_2$
are isomorphisms onto the subgraphs $\psi(H_1)$ and $\psi(H_2)$ of $H'$,
so, by the induction hypothesis, $(\psi(H_1);\psi(s_1),\psi(t_1))$ and
$(\psi(H_2);\psi(s_2),\psi(t_2))$ are in $\mathcal{SP}$.  As $s_1=s_2$ and $t_1=t_2$, $\psi(s_1)=\psi(s_2)$
and $\psi(t_1)=\psi(t_2)$; as $\psi$ is injective,
$V(\psi(H_1))\cap V(\psi(H_2))=\psi(V(H_1)\cap V(H_2))=\{\psi(s_1),\psi(t_1)\}$;
and as $\psi$ is onto, $\psi(H_1)\cup\psi(H_2)=\psi(H_1\cup H_2)=H'$.  So
$(H';\psi(s_1),\psi(t_1))$ is the parallel composition of two members of
$\mathcal{SP}$, and is in $\mathcal{SP}$.

(3) We prove (iii).  As noted after Definition~\ref{def:N07-sp}, the
compositions of that definition are those of
\cite[Section~8.3]{Bodlaender98}, except that an edge lying in both
graphs of a parallel composition is one edge of the union and two
parallel edges of the multigraph.  So every $(H;s,t)\in\mathcal{SP}$ is
isomorphic to a series parallel graph $G$ of \cite{Bodlaender98} with
some parallel edges merged.  By
\cite[Theorem~41]{Bodlaender98}, $G$ has a tree decomposition of width at
most two.  It is also one of $H$: Definition~\ref{def:N07-td} involves
the edges only in condition (b), which asks less of a graph with fewer
edges, and an isomorphism carries a tree decomposition onto one of the
same width.  So the treewidth of $H$ is at most two.

(4) We prove (iv), by induction on the number of compositions in a
construction of $(H;s,t)$ from single edges.  If there is none, $(H;s,t)$
is a single edge $(K_2;s,t)$: it is connected, and for $v\in\{s,t\}$ the
only vertex $u\ne v$ is the other terminal, joined to itself by the path
$(u)$.  Otherwise the last composition is a series composition
$(H_1\cup H_2;s_1,t_2)$ or a parallel composition $(H_1\cup H_2;s_1,t_1)$
of two graphs $(H_1;s_1,t_1),(H_2;s_2,t_2)\in\mathcal{SP}$ constructed
with fewer compositions, for which (iv) holds by the induction
hypothesis.  In both cases $H_1$ and $H_2$ are connected and share a
vertex ($t_1$, or $s_1$), so $H=H_1\cup H_2$ is connected.  Let $v$ be a
vertex of $H$ and $u\ne v$ a vertex of $H$.  A path of $H_1-v$ or of
$H_2-v$ is a path of $H-v$.

\emph{The series case.}  Here $s=s_1$, $t=t_2$, and
$V(H_1)\cap V(H_2)=\{c\}$ with $c:=t_1=s_2$.  Each of $H_1,H_2$ has the
terminal $c$ and one terminal of $H$: $s$ for $H_1$ and $t$ for $H_2$.
If $v=c$: $u$ lies in some $H_i$, and by the induction hypothesis $u$
is joined in $H_i-c$ to a terminal of $H_i$ other than $c$, which is a
terminal of $H$.  If $v\ne c$: let $i$ be the index with $v\in V(H_i)$
and $j$ the other one, so that $v\notin V(H_j)$ and $H_j-v=H_j$ is
connected; so every vertex of $H_j$, and in particular $c$, is joined in
$H_j$ to the terminal of $H$ in $H_j$.  If $u\in V(H_j)$, this joins
$u$.  If $u\in V(H_i)$, by the induction hypothesis $u$ is joined in
$H_i-v$ to a terminal of $H_i$: to the terminal of $H$ in $H_i$, or to
$c$, and then on to the terminal of $H$ in $H_j$.

\emph{The parallel case.}  Here $s=s_1=s_2$, $t=t_1=t_2$ and
$V(H_1)\cap V(H_2)=\{s,t\}$.  If $v\in\{s,t\}$: $u$ lies in some $H_i$,
and by the induction hypothesis $u$ is joined in $H_i-v$ to $s$ or to
$t$.  If $v\notin\{s,t\}$: let $i$ be the index with $v\in V(H_i)$ and
$j$ the other one, so that $H_j-v=H_j$ is connected and joins every
vertex of $H_j$ to $s$.  If $u\in V(H_j)$, this joins $u$.  If
$u\in V(H_i)$, by the induction hypothesis $u$ is joined in $H_i-v$ to
$s$ or to $t$.
\end{proof}

\begin{lemma}[the template is series--parallel]\label{lem:N07-template}
Let $\T=\T(L,\mu;E)$ be a template (Definition~\ref{def:N02-template}),
with the spine $v_0,\dots,v_L$, the terminals $s=v_0$ and $t=v_L$, the
$t$-ear $r_0,\dots,r_{\mu+2}$ ($r_0=v_{L-\mu}$, $r_{\mu+2}=v_L$) and the
crossings $x^j_0,\dots,x^j_{\ell_j}$, $j<e$, with $x^j_0=v_{a_j}$,
$x^j_{\ell_j}=v_{b_j}$, $a_j=\mu-1-j$, $b_j=\mu+1+j$ and
$\ell_j=2j+2+\varepsilon_j$.  Then $(\T;s,t)\in\mathcal{SP}$
(Definition~\ref{def:N07-sp}).
\end{lemma}

\begin{proof}
By Definition~\ref{def:N02-template}, the vertices $v_0,\dots,v_L$,
$r_1,\dots,r_{\mu+1}$ and $x^j_1,\dots,x^j_{\ell_j-1}$ ($j<e$) are pairwise
distinct, and $\T$ is the union of the spine, the $t$-ear and the
crossings.  By Lemma~\ref{lem:N02-ends}, $1\le a_j<b_j\le L-\mu-1$ for
$j<e$, and $L-\mu\ge5$; by Definition~\ref{def:N02-template},
\eqref{eq:N02-hyp}, $\mu\ge3$.  Put $X^j:=\Pi(x^j_0,\dots,x^j_{\ell_j})$
for $j<e$; it has $\ell_j=2j+2+\varepsilon_j\ge4$ edges, as
$\varepsilon_j\ge2$.  Every path below has at least one edge, so it is in
$\mathcal{SP}$ by Lemma~\ref{lem:N07-sp}(i).  For each composition we
give the intersection of the vertex sets, which follows from the
distinctness above.  By Definition~\ref{def:N02-template},
$e\in\{0,1,2\}$; steps (2), (4) and (5) treat the three values.

(1) \emph{The end $R$.}  Put
\[
 A:=\Pi(v_{L-\mu},v_{L-\mu+1},\dots,v_L),\qquad
 R':=\Pi(r_0,r_1,\dots,r_{\mu+2}),\qquad R:=A\,\|\,R',
\]
where $A$ has $\mu\ge3$ edges and $R'$ has $\mu+2$ edges, both with the
terminals $(v_{L-\mu},v_L)$, and
$V(A)\cap V(R')=\{v_{L-\mu},v_L\}$, as $r_1,\dots,r_{\mu+1}$ are not spine
vertices.  So $R\in\mathcal{SP}$, with the terminals $(v_{L-\mu},v_L)$, and
$R$ is the union of the spine from $v_{L-\mu}$ to $v_L$ and the $t$-ear.

(2) \emph{The case $e=0$.}  Put $P:=\Pi(v_0,\dots,v_{L-\mu})$, with
$L-\mu\ge5$ edges.  Then $V(P)\cap V(R)=\{v_{L-\mu}\}$, and
$P\cup R$ is the union of the spine and the $t$-ear, which is $\T$.  So
\[
 (\T;s,t)=P\circ R\in\mathcal{SP}.
\]

(3) \emph{The middle $M_0$, when $e\ge1$.}  Here $a_0=\mu-1$ and
$b_0=\mu+1$.  Put
\[
 M_0:=\Pi(v_{\mu-1},v_\mu,v_{\mu+1})\,\|\,X^0 ,
\]
where the spine path has $2$ edges and $X^0$ has $\ell_0\ge4$ edges, both
with the terminals $(v_{\mu-1},v_{\mu+1})$, and the intersection of their
vertex sets is $\{v_{\mu-1},v_{\mu+1}\}$, as $x^0_1,\dots,x^0_{\ell_0-1}$
are not spine vertices.  So $M_0\in\mathcal{SP}$, with the terminals
$(v_{\mu-1},v_{\mu+1})$; its vertices are $v_{\mu-1},v_\mu,v_{\mu+1}$ and
the members of the crossing $0$.

(4) \emph{The case $e=1$.}  Put
\[
 P_1:=\Pi(v_0,\dots,v_{\mu-1}),\qquad P_2:=\Pi(v_{\mu+1},\dots,v_{L-\mu}),
\]
with $\mu-1\ge2$ and $L-\mu-(\mu+1)\ge1$ edges, the latter as
$b_0=\mu+1\le L-\mu-1$.  Then
\[
\begin{aligned}
 V(P_1)\cap V(M_0)&=\{v_{\mu-1}\}, &
 V(P_1\cup M_0)\cap V(P_2)&=\{v_{\mu+1}\},\\
 V(P_1\cup M_0\cup P_2)\cap V(R)&=\{v_{L-\mu}\},
\end{aligned}
\]
so the series compositions
\[
 H_1:=P_1\circ M_0,\qquad H_2:=H_1\circ P_2,\qquad H_3:=H_2\circ R
\]
are defined, with the terminals $(v_0,v_{\mu+1})$, $(v_0,v_{L-\mu})$ and
$(v_0,v_L)=(s,t)$, and $H_3\in\mathcal{SP}$.  The union
$P_1\cup\Pi(v_{\mu-1},v_\mu,v_{\mu+1})\cup P_2\cup A$ is the spine, so
$H_3=P_1\cup M_0\cup P_2\cup R$ is the union of the spine, the crossing
$0$ and the $t$-ear, which is $\T$.  So $(\T;s,t)\in\mathcal{SP}$.

(5) \emph{The case $e=2$.}  Here $a_1=\mu-2$ and $b_1=\mu+2$.  Put
\[
 N:=\bigl(\Pi(v_{\mu-2},v_{\mu-1})\circ M_0\bigr)\circ\Pi(v_{\mu+1},v_{\mu+2}),
\]
the series compositions being defined as
$V(\Pi(v_{\mu-2},v_{\mu-1}))\cap V(M_0)=\{v_{\mu-1}\}$ and the vertices of
$\Pi(v_{\mu-2},v_{\mu-1})\cup M_0$ meet those of $\Pi(v_{\mu+1},v_{\mu+2})$
in $\{v_{\mu+1}\}$; $N$ has the terminals $(v_{\mu-2},v_{\mu+2})$ and the
vertices $v_{\mu-2},\dots,v_{\mu+2}$ and the members of the crossing $0$.
The crossing $X^1$ has the vertices $v_{\mu-2}$, $x^1_1,\dots,x^1_{\ell_1-1}$,
$v_{\mu+2}$, so $V(N)\cap V(X^1)=\{v_{\mu-2},v_{\mu+2}\}$, and
\[
 M_1:=N\,\|\,X^1\in\mathcal{SP},
\]
with the terminals $(v_{\mu-2},v_{\mu+2})$.  Put
$P_1:=\Pi(v_0,\dots,v_{\mu-2})$ and $P_2:=\Pi(v_{\mu+2},\dots,v_{L-\mu})$,
with $\mu-2\ge1$ and $L-\mu-(\mu+2)\ge1$ edges, as $1\le a_1$ and
$b_1\le L-\mu-1$.  Then
\[
\begin{aligned}
 V(P_1)\cap V(M_1)&=\{v_{\mu-2}\}, &
 V(P_1\cup M_1)\cap V(P_2)&=\{v_{\mu+2}\},\\
 V(P_1\cup M_1\cup P_2)\cap V(R)&=\{v_{L-\mu}\},
\end{aligned}
\]
so the series compositions
\[
 H_1:=P_1\circ M_1,\qquad H_2:=H_1\circ P_2,\qquad H_3:=H_2\circ R
\]
are defined, with the terminals $(v_0,v_{\mu+2})$, $(v_0,v_{L-\mu})$ and
$(v_0,v_L)=(s,t)$, and $H_3\in\mathcal{SP}$.  Its graph is
$P_1\cup M_1\cup P_2\cup R$: the union of
$P_1$, $\Pi(v_{\mu-2},v_{\mu-1})$, $\Pi(v_{\mu-1},v_\mu,v_{\mu+1})$,
$\Pi(v_{\mu+1},v_{\mu+2})$, $P_2$ and $A$ is the spine, and the rest is
$X^0$, $X^1$ and the $t$-ear; so the graph is $\T$, and
$(\T;s,t)\in\mathcal{SP}$.
\end{proof}

\begin{theorem}[treewidth of $G_\varphi$]\label{thm:N07-tw}
Let $\varphi$ be a restricted formula and $G:=G_\varphi$
(Definition~\ref{def:N03-graph}).
\begin{enumerate}
\item[\rm(i)] $(G;T_1,T_2)\in\mathcal{SP}$
  (Definition~\ref{def:N07-sp}).
\item[\rm(ii)] $G$ has a tree decomposition of width at most two
  (Definition~\ref{def:N07-td}); so the treewidth of $G$ is at most two.
\end{enumerate}
\end{theorem}

\begin{proof}
We use the notation of Definitions~\ref{def:N03-graph},
\ref{def:N03-names} and~\ref{def:N03-rows}.  By
Definition~\ref{def:N03-graph}, the vertices $T_1$, $T_2$, $s_k$, $t_k$,
$z_i$, $q$, $q'$, the members of the copies, the $y^j_i$ ($1\le i\le D-1$)
and the $\tilde y^j$ are pairwise distinct; this gives the intersections
of vertex sets below.

(1) \emph{The copies.}  Let $k$ be a level and $c$ an option.  By
Lemma~\ref{lem:N03-params}(ii), $\T:=\T(L_k,\mu_k;E_{k,c})$ is a template,
so $(\T;s,t)\in\mathcal{SP}$ by Lemma~\ref{lem:N07-template}.  By
Lemma~\ref{lem:N03-copies}, the map $\psi$ there has $\psi(s)=s_k$,
$\psi(t)=t_k$, and by Lemma~\ref{lem:N03-copies}(i) it is an isomorphism
from $\T$ onto $U_{k,c}$.  By Lemma~\ref{lem:N07-sp}(ii),
$(U_{k,c};s_k,t_k)\in\mathcal{SP}$.

(2) \emph{The arms.}  Let $1\le j\le4$ and $\pi:=\pi_j$.  By
Definition~\ref{def:N03-graph}, $\pi\in\{7,p\}$, and by
Definition~\ref{def:N03-params}, $p=2m+4$ and $D=4m+15$; so
$2\le\pi\le D-2$.  Put
\[
 C_j:=\Pi(y^j_{\pi-1},y^j_\pi,y^j_{\pi+1})\,\|\,\Pi(y^j_{\pi-1},\tilde y^j,y^j_{\pi+1}),
\]
two paths with $2$ edges and the terminals $(y^j_{\pi-1},y^j_{\pi+1})$,
whose vertex sets meet in $\{y^j_{\pi-1},y^j_{\pi+1}\}$.  Then
\[
 Y^j=\bigl(\Pi(y^j_0,\dots,y^j_{\pi-1})\circ C_j\bigr)\circ\Pi(y^j_{\pi+1},\dots,y^j_D),
\]
where the two paths have $\pi-1\ge1$ and $D-\pi-1\ge1$ edges, the first
meets $C_j$ in $\{y^j_{\pi-1}\}$, and the union of the first and $C_j$
meets the last in $\{y^j_{\pi+1}\}$; the union of the three is the arm
$Y^j$ of Definition~\ref{def:N03-rows}.  So
$(Y^j;T_1,T_2)\in\mathcal{SP}$, as $y^j_0=T_1$ and $y^j_D=T_2$.

(3) \emph{The levels.}  For $1\le k\le m$ let $Q_k$ be the subgraph of $G$
formed by the copies of the levels $k,\dots,m$ and the frame edges
$s_is_{i+1}$, $t_{i+1}t_i$ for $k\le i<m$; its vertices are the $s_i,t_i$
($k\le i\le m$) and the members of the copies of the levels $i\ge k$,
and, if $k\le m-1$, the vertices $q$ and $q'$ with their four edges.  We
show by downward induction on $k$ that $(Q_k;s_k,t_k)\in\mathcal{SP}$.
For $k=m$, the three copies of the level $m$ pairwise meet in
$\{s_m,t_m\}$, so
\[
 Q_m=U_{m,1}\cup U_{m,2}\cup U_{m,3}=(U_{m,1}\,\|\,U_{m,2})\,\|\,U_{m,3}
 \in\mathcal{SP},
\]
with the terminals $(s_m,t_m)$, by (1).  Let $k<m$ and
$(Q_{k+1};s_{k+1},t_{k+1})\in\mathcal{SP}$.  Let $R_k$ be the single
edge $\Pi(t_{k+1},t_k)$ if $k<m-1$, and for $k=m-1$ let
\[
 R_{m-1}:=\bigl(\Pi(t_m,t_{m-1})\,\|\,\Pi(t_m,q,t_{m-1})\bigr)\,\|\,\Pi(t_m,q',t_{m-1}),
\]
the three paths from $t_m$ to $t_{m-1}$ of Definition~\ref{def:N03-graph}
(the frame edge and the two paths through $q$ and $q'$), which pairwise
meet in $\{t_m,t_{m-1}\}$ as $q\ne q'$ are not frame vertices; so
$R_{m-1}\in\mathcal{SP}$ with the terminals $(t_m,t_{m-1})$ and the
vertices $t_m,t_{m-1},q,q'$.  The vertex $s_k$ is not a
vertex of $Q_{k+1}$, so $V(\Pi(s_k,s_{k+1}))\cap V(Q_{k+1})=\{s_{k+1}\}$.
The vertices of $R_k$ other than $t_{k+1}$ are $t_k$ if $k<m-1$, and
$t_{m-1},q,q'$ if $k=m-1$; none of them is a vertex of
$\Pi(s_k,s_{k+1})\cup Q_{k+1}$, as $t_k\notin V(Q_{k+1})$ and, for
$k=m-1$, $q,q'\notin V(Q_m)$.  So the vertices of
$\Pi(s_k,s_{k+1})\cup Q_{k+1}$ meet those of $R_k$ in $\{t_{k+1}\}$.
Hence
\[
 B_k:=\bigl(\Pi(s_k,s_{k+1})\circ Q_{k+1}\bigr)\circ R_k\in\mathcal{SP},
\]
with the terminals $(s_k,t_k)$ and the vertices $V(Q_{k+1})\cup\{s_k,t_k\}$,
plus $q,q'$ if $k=m-1$.
The copies $U_{k,1},U_{k,2},U_{k,3}$ pairwise meet in $\{s_k,t_k\}$, and
each meets $B_k$ in $\{s_k,t_k\}$, as its members are not vertices of
$Q_{k+1}$.  So
\[
 Q_k=\bigl((U_{k,1}\,\|\,U_{k,2})\,\|\,U_{k,3}\bigr)\,\|\,B_k\in\mathcal{SP},
\]
with the terminals $(s_k,t_k)$; its graph is the union of the copies of
the level $k$, of $Q_{k+1}$ and of the edges $s_ks_{k+1}$, $t_{k+1}t_k$,
and for $k=m-1$ of the vertices $q,q'$ and their edges,
which is $Q_k$.

(4) \emph{The graph $G$; we prove (i).}  Put
\[
 P:=\bigl(\Pi(T_1,s_1)\circ Q_1\bigr)\circ\Pi(t_1,z_1,\dots,z_7,T_2),
\]
the compositions being defined as $V(\Pi(T_1,s_1))\cap V(Q_1)=\{s_1\}$ and
the vertices of $\Pi(T_1,s_1)\cup Q_1$ meet those of
$\Pi(t_1,z_1,\dots,z_7,T_2)$ in $\{t_1\}$; so $P\in\mathcal{SP}$ with the
terminals $(T_1,T_2)$, and $P$ is the union of the frame and the copies.
By Lemma~\ref{lem:N03-structure}(i), the sets $\Int(Y^j)$ are pairwise
disjoint and disjoint from $V(P)$, so $P$ and the arms pairwise meet in
$\{T_1,T_2\}$.  Hence
\[
 H_0:=\bigl(\bigl(\bigl(P\,\|\,Y^1\bigr)\,\|\,Y^2\bigr)\,\|\,Y^3\bigr)\,\|\,Y^4
 \in\mathcal{SP},
\]
by (2) and (3).  By Definition~\ref{def:N03-graph}, $G$ is the union of
the frame, the copies, the pair $q,q'$ with its four edges, and
the arms; the first three form $P$ ($q,q'$ lie in $Q_{m-1}$, by
(3)); so $H_0=G$.  This proves (i).

(5) We prove (ii).  By (i) and Lemma~\ref{lem:N07-sp}(iii), $G$ has
treewidth at most two, that is, a tree decomposition of width at most
two (Definition~\ref{def:N07-td}).
\end{proof}


\section{Pathwidth}\label{sec:N10}

\begin{definition}[path decompositions, boundaries, layouts]\label{def:N10-pd}
Let $H$ be a graph.  A \emph{path decomposition} of $H$ is a tree
decomposition of $H$ (Definition~\ref{def:N07-td}) whose tree is a path
$\nu_1,\dots,\nu_r$; its width is its width as a tree decomposition.  The
\emph{pathwidth} of $H$ is the least width of a path decomposition of
$H$; so the treewidth of $H$ is at most its pathwidth.

For a set $P\subseteq V(H)$, $N(P)$ is the set of vertices of $H$ with a
neighbour in $P$, and the \emph{boundary} of $P$ is
$\partial P:=N(P)\setminus P$, the set of vertices outside $P$ with a
neighbour in $P$.  For an ordering $v_1,\dots,v_n$ of $V(H)$, its
\emph{prefixes} are the sets $P_i:=\{v_1,\dots,v_i\}$ ($1\le i\le n$); the
ordering is a \emph{layout of width $k$} when $|\partial P_i|\le k$ for
every $i$.
\end{definition}

\begin{lemma}[layouts give path decompositions]\label{lem:N10-layout}
Let $H$ be a graph and $v_1,\dots,v_n$ a layout of $H$ of width $k$, with
the prefixes $P_i$ and their boundaries $\partial P_i$
(Definition~\ref{def:N10-pd}).  Then the path $\nu_1,\dots,\nu_n$ with the
bags $X_i:=\{v_i\}\cup\partial P_i$ is a path decomposition of $H$ of
width at most $k$.
\end{lemma}

\begin{proof}
We check the conditions (a)--(c) of Definition~\ref{def:N07-td}.
Condition (a): $v_i\in X_i$.  Condition (b): let $v_iv_j$ be an edge with
$i<j$; then $v_j\in N(P_i)$ and $v_j\notin P_i$, so $v_j\in\partial P_i$,
and both $v_i$ and $v_j$ lie in $X_i$.  Condition (c): let $v=v_j$.  For
$i>j$, $v\in P_i$, so $v\notin\partial P_i$, and $v\ne v_i$; so
$v\notin X_i$.  For $i<j$, $v\in X_i$ iff $v\in\partial P_i$ iff some
neighbour of $v$ lies in $P_i$.  If $v$ has no neighbour among
$v_1,\dots,v_{j-1}$, the nodes whose bags contain $v$ are $\nu_j$ alone.
Otherwise let $i_0<j$ be the least index of a neighbour of $v$; then
$v\in X_i$ exactly for $i_0\le i\le j$, and these nodes form the subpath
$\nu_{i_0},\dots,\nu_j$.  The width: $|X_i|\le1+|\partial P_i|\le k+1$.
\end{proof}

\begin{lemma}[placing a vertex]\label{lem:N10-place}
Let $H$ be a graph, $P\subseteq V(H)$ and $w\in V(H)\setminus P$, with
boundaries as in Definition~\ref{def:N10-pd}.  Then
\[
 \partial(P\cup\{w\})=\bigl(\partial P\setminus\{w\}\bigr)\cup
 \bigl(N(w)\setminus(P\cup\{w\})\bigr).
\]
\end{lemma}

\begin{proof}
A vertex $x\notin P\cup\{w\}$ lies in $\partial(P\cup\{w\})$ iff it has a
neighbour in $P\cup\{w\}$, that is, iff it has a neighbour in $P$, which
means $x\in\partial P$, or $w$ is a neighbour of $x$, which means
$x\in N(w)$ (Definition~\ref{def:N10-pd}).
\end{proof}

\begin{definition}[paths with ears]\label{def:N10-pwe}
A graph $H$ is a \emph{path with ears} when it consists of a path
$W=(w_0,w_1,\dots,w_n)$ with $n\ge2$ and of paths $F_1,\dots,F_h$
($h\ge0$), the \emph{ears}, where
$F_j=(w_{\alpha_j},f^j_1,\dots,f^j_{\ell_j-1},w_{\beta_j})$ with
$\ell_j\ge2$ and $1\le\alpha_j<\beta_j\le n$, such that the inner vertices
$f^j_i$ of the ears are pairwise distinct, distinct from the $w_i$, and
of degree two in $H$, and $E(H)$ consists of the edges of $W$ and of the
$F_j$.  The \emph{depth} of $H$ is
$d:=\max_{1\le i<n}\bigl|\{j:\alpha_j\le i<\beta_j\}\bigr|$, with $d:=0$
if $h=0$.  The \emph{ear order} of $H$ is the list $w_1,\dots,w_{n-1}$ in
which, right after $w_i$, the inner vertices $f^j_1,\dots,f^j_{\ell_j-1}$
of each ear $F_j$ with $\alpha_j=i$ are inserted, in increasing $j$.
\end{definition}

\begin{lemma}[the ear order]\label{lem:N10-pwe}
Let $H$ be a path with ears of depth $d$, with the notation of
Definition~\ref{def:N10-pwe}; boundaries are taken in $H$.
\begin{enumerate}
\item[\rm(i)] The ear order lists every vertex of $H$ other than
  $w_0,w_n$ exactly once.
\item[\rm(ii)] For every nonempty prefix $Q$ of the ear order,
  $\partial Q\subseteq\{w_0,w_n\}\cup M$ for a set $M$ of at most $d+1$
  vertices.
\item[\rm(iii)] For the set $Q$ of all vertices other than $w_0,w_n$,
  $\partial Q=\{w_0,w_n\}$.
\end{enumerate}
\end{lemma}

\begin{proof}
(1) We prove (i).  The vertices of $H$ other than $w_0,w_n$ are
$w_1,\dots,w_{n-1}$ and the inner vertices of the ears, pairwise distinct
(Definition~\ref{def:N10-pwe}); the ear order lists each $w_i$ once and
the inner vertices of each ear once, after $w_{\alpha_j}$.

(2) \emph{Traversing an ear.}  Let $P\subseteq V(H)$ contain $w_{\alpha_j}$
but no inner vertex of the ear $F_j$, and let
$U_i:=\{f^j_1,\dots,f^j_i\}$ for $1\le i\le\ell_j-1$.  We claim
\[
 \partial(P\cup U_i)=\bigl(\partial P\setminus U_i\bigr)\cup
 \bigl(\{f^j_{i+1}\}\setminus P\bigr),
\]
where $f^j_{\ell_j}:=w_{\beta_j}$.  Indeed, a vertex $x\notin P\cup U_i$
lies in $\partial(P\cup U_i)$ iff it has a neighbour in $P$, which means
$x\in\partial P$, or a neighbour $f^j_{i'}$ with $1\le i'\le i$; the
neighbours of $f^j_{i'}$ are $f^j_{i'-1}$ and $f^j_{i'+1}$, as $f^j_{i'}$
has degree two, and the only one outside $P\cup U_i$ is $f^j_{i+1}$,
since $f^j_0=w_{\alpha_j}\in P$.  So, while the ear is traversed, the
boundary changes only by the replacement of $f^j_1$, which lies in
$\partial P$, by the current inner vertex, and at the end by
$w_{\beta_j}$ if that vertex is unplaced.

(3) We prove (ii), with the sets
\[
 S_i:=\{w_0,w_n,w_{i+1}\}\cup\{w_{\beta_j}:\alpha_j\le i<\beta_j\}
 \qquad(1\le i\le n-1).
\]
Let $Q_i$ be the prefix ending with the last inner vertex of the last
ear starting at $w_i$, or with $w_i$ if no ear starts there, and
$Q_0:=\emptyset$.  We show by induction on $i$ that
$\partial Q_i\subseteq S_i$, and that every prefix between $Q_{i-1}$ and
$Q_i$ has its boundary in $\{w_0,w_n\}\cup M$ with $|M|\le d+1$.

Let $1\le i\le n-1$ and assume $\partial Q_{i-1}\subseteq S_{i-1}$ (for
$i=1$, $\partial Q_0=\emptyset$).  The vertex placed after $Q_{i-1}$ is
$w_i$, so Lemma~\ref{lem:N10-place} with $P:=Q_{i-1}$ and $w:=w_i$ gives
\[
 \partial(Q_{i-1}\cup\{w_i\})=\bigl(\partial Q_{i-1}\setminus\{w_i\}\bigr)
 \cup\bigl(N(w_i)\setminus(Q_{i-1}\cup\{w_i\})\bigr).
\]
We bound the two sets.  First, by the induction hypothesis and the
definition of $S_{i-1}$, removing $w_i$ removes the vertices $w_{\beta_j}$
with $\beta_j=i$:
\[
 \partial Q_{i-1}\setminus\{w_i\}\subseteq S_{i-1}\setminus\{w_i\}
 =\{w_0,w_n\}\cup\{w_{\beta_j}:\alpha_j\le i-1,\ \beta_j>i\}.
\]
Second, by Definition~\ref{def:N10-pwe},
$N(w_i)=\{w_{i-1},w_{i+1}\}\cup\{f^j_1:\alpha_j=i\}\cup\{f^j_{\ell_j-1}:\beta_j=i\}$.
Here $w_{i-1}\in Q_{i-1}$ if $i\ge2$, and $w_{i-1}=w_0$ if $i=1$; and
each $f^j_{\ell_j-1}$ with $\beta_j=i$ lies in $Q_{i-1}$, since the ear
$F_j$ has $\alpha_j<\beta_j=i$ and its inner vertices were placed right
after $w_{\alpha_j}$.  So
\[
 N(w_i)\setminus(Q_{i-1}\cup\{w_i\})\subseteq\{w_0,w_{i+1}\}\cup\{f^j_1:\alpha_j=i\}.
\]
Together,
\[
 \partial(Q_{i-1}\cup\{w_i\})\subseteq\{w_0,w_n,w_{i+1}\}
 \cup\{w_{\beta_j}:\alpha_j\le i-1,\ \beta_j>i\}\cup\{f^j_1:\alpha_j=i\}.
\]
Then the ears $F_j$ with $\alpha_j=i$ are traversed in turn, and by (2)
each traversal replaces $f^j_1$ by the current inner vertex and, at the
end, by $w_{\beta_j}$, which is unplaced as $\beta_j>i$.  So at every
prefix between $Q_{i-1}$ and $Q_i$, the boundary is contained in
$\{w_0,w_n,w_{i+1}\}$ together with one vertex for each ear $F_j$ with
$\alpha_j\le i<\beta_j$: its far end $w_{\beta_j}$, its first inner vertex
$f^j_1$, or its current inner vertex.  There are at most $d$ such ears,
so $|M|\le d+1$; and at $Q_i$ every such ear contributes $w_{\beta_j}$, so
$\partial Q_i\subseteq S_i$.

(4) We prove (iii).  By (3) at $i=n-1$,
$\partial Q_{n-1}\subseteq S_{n-1}=\{w_0,w_n\}$, as $\beta_j\le n$ for every
$j$; and $w_0,w_n\in\partial Q_{n-1}$, being adjacent to
$w_1,w_{n-1}\in Q_{n-1}$ and unplaced.
\end{proof}

\begin{lemma}[the pieces of $G_\varphi$ as paths with ears]\label{lem:N10-pieces}
Let $\varphi$ be a restricted formula and $G:=G_\varphi$
(Definition~\ref{def:N03-graph}), with the copies $U_{k,c}$ and their
vertices $v^{k,c}_i$, $r^{k,c}_i$, $x^{k,c,j}_i$ (Definition~\ref{def:N03-names}),
the arms $Y^j$ (Definition~\ref{def:N03-rows}) and the tail
$z_1,\dots,z_7$.
\begin{enumerate}
\item[\rm(i)] Let $k$ be a level and $c$ an option, and write $L:=L_k$,
  $\mu:=\mu_k$, $e:=|Z_{k,c}|$, $a_j:=\mu-1-j$, $b_j:=\mu+1+j$, and
  $v_i$, $r_i$, $x^j_i$ for the vertices $v^{k,c}_i$, $r^{k,c}_i$,
  $x^{k,c,j}_i$ of $U_{k,c}$ (Definition~\ref{def:N03-names}).  Then
  $U_{k,c}$ is a path with ears (Definition~\ref{def:N10-pwe}) with the
  path $v_0,\dots,v_L$ ($v_0=s_k$, $v_L=t_k$), the ears
  $X^j=(v_{a_j},x^j_1,\dots,x^j_{\ell_j-1},v_{b_j})$ ($j<e$) and
  $R=(v_{L-\mu},r_1,\dots,r_{\mu+1},t_k)$, and depth $2$ if $e=2$ and
  $1$ otherwise.  Its ear order (Definition~\ref{def:N10-pwe}) is the list
  of its members $v_1,\dots,v_{L-1}$ in which,
  for each $j<e$, the members $x^j_1,\dots,x^j_{\ell_j-1}$ are inserted right
  after $v_{a_j}$, and $r_1,\dots,r_{\mu+1}$ right after $v_{L-\mu}$.
\item[\rm(ii)] Each arm $Y^j$ is a path with ears with the path
  $y^j_0,\dots,y^j_D$ ($y^j_0=T_1$, $y^j_D=T_2$) and the one ear
  $(y^j_{\pi_j-1},\tilde y^j,y^j_{\pi_j+1})$, of depth $1$; its ear order is
  $y^j_1,\dots,y^j_{\pi_j-1},\tilde y^j,y^j_{\pi_j},\dots,y^j_{D-1}$.
\item[\rm(iii)] The subgraph $Z$ of $G$ formed by the path
  $t_1,z_1,\dots,z_7,T_2$ is a path with ears with no ear, of depth $0$,
  with the ear order $z_1,\dots,z_7$; and $Z$ is attached at $\{t_1,T_2\}$
  (Definition~\ref{def:N01-attach}).
\end{enumerate}
\end{lemma}

\begin{proof}
(1) We prove (i).  By Definition~\ref{def:N03-graph}, $U_{k,c}$ is the
template $\T(L_k,\mu_k;E_{k,c})$ with its vertices renamed, so
Definition~\ref{def:N02-template} applies to it: it consists of the
spine, a path of length $L\ge2$, the $t$-ear and the crossings, paths
whose inner vertices $r_1,\dots,r_{\mu+1}$ and $x^j_1,\dots,x^j_{\ell_j-1}$
are pairwise distinct, distinct from the spine vertices, and of degree
two; the $t$-ear has $\mu+2\ge2$ edges and runs from $v_{L-\mu}$ to
$v_L$, with $1\le L-\mu<L$; the crossing $j$ has $\ell_j\ge2$ edges and
runs from $v_{a_j}$ to $v_{b_j}$, with $1\le a_j<b_j\le L-\mu-1$
(Lemma~\ref{lem:N02-ends}).  So $U_{k,c}$ is a path with ears as stated.
Its depth (Definition~\ref{def:N10-pwe}): for a position $i$ with
$1\le i<L$, the ear $R$, with $\alpha=L-\mu$ and $\beta=L$, satisfies
$\alpha\le i<\beta$ iff $L-\mu\le i$, and the ear $X^j$ iff
$a_j\le i<b_j$.  As $b_j\le L-\mu-1$, no position satisfies both
$L-\mu\le i$ and $i<b_j$.  So the number of ears $F$ with
$\alpha_F\le i<\beta_F$ is at most $1$ at the positions $i\ge L-\mu$, and
at most $e$ at the positions $i<L-\mu$.  When $e=2$, $a_1<a_0<b_0<b_1$ by
the formulas, so this number is $2$ at the positions $a_0\le i<b_0$ and
at most $2$ everywhere.  When $e\le1$ it is at most $1$ everywhere, and
it is $1$ at $i=L-\mu\le L-1$.  So the depth is $2$ if $e=2$ and $1$
otherwise.  By Definition~\ref{def:N10-pwe}, the ear order inserts the
inner vertices of each ear right after its first end: $x^j_1,\dots,x^j_{\ell_j-1}$
right after $v_{a_j}$ and $r_1,\dots,r_{\mu+1}$ right after $v_{L-\mu}$; no
two ears start at the same vertex, as $a_{e-1}<\dots<a_0<L-\mu$, so the
order of insertion at one vertex does not arise.  This is the list
stated in (i).

(2) We prove (ii).  By Definition~\ref{def:N03-graph}, the arm $Y^j$
consists of the path $y^j_0,\dots,y^j_D$ and the vertex $\tilde y^j$ with
the edges $\tilde y^jy^j_{\pi_j-1}$, $\tilde y^jy^j_{\pi_j+1}$, and
$\tilde y^j$ has no other neighbour in $G$; and $\pi_j\in\{7,p\}$ with
$p=2m+4$ and $D=4m+15$ (Definition~\ref{def:N03-params}), so
$1\le\pi_j-1<\pi_j+1\le D$.  So $Y^j$ is a path with the one ear
$(y^j_{\pi_j-1},\tilde y^j,y^j_{\pi_j+1})$, of depth $1$, and its ear order
is as stated.

(3) We prove (iii).  By Definition~\ref{def:N03-graph}, $t_1,z_1,\dots,z_7,T_2$
is a path of $G$, and the vertices $z_1,\dots,z_7$ have no other
neighbours in $G$: they belong to the frame only, and the parts of $G$
share only the terminals $s_k,t_k$ of the copies and $T_1,T_2$ of the
arms.  So $Z$ is a path with no ear, of depth $0$, with the ear order
$z_1,\dots,z_7$; and every edge of $G$ with an end in
$\Int(Z)=\{z_1,\dots,z_7\}$ is an edge of $Z$, so $Z$ is attached at
$\{t_1,T_2\}$.
\end{proof}

\begin{lemma}[gluing an attached subgraph]\label{lem:N10-glue}
Let $G$ be a graph and $X$ a subgraph of $G$ attached at $\{a,b\}$
(Definition~\ref{def:N01-attach}).  Let $P\subseteq V(G)\setminus\Int(X)$
and $Q\subseteq\Int(X)$.  Then, with $\partial$ the boundary in $G$ and
$\partial_X$ the boundary in $X$ (Definition~\ref{def:N10-pd}),
\[
 \partial(P\cup Q)\ \subseteq\ (\partial P\setminus Q)\cup\{a,b\}\cup
 \bigl(\partial_XQ\cap\Int(X)\bigr).
\]
\end{lemma}

\begin{proof}
By Definition~\ref{def:N10-pd}, $\partial S=N(S)\setminus S$ for every
$S\subseteq V(G)$, and $N(P\cup Q)=N(P)\cup N(Q)$, since a vertex has a
neighbour in $P\cup Q$ iff it has one in $P$ or one in $Q$.  Hence
\begin{align*}
 \partial(P\cup Q)&=\bigl(N(P)\cup N(Q)\bigr)\setminus(P\cup Q)
 =\bigl(N(P)\setminus(P\cup Q)\bigr)\cup\bigl(N(Q)\setminus(P\cup Q)\bigr)\\
 &\subseteq(\partial P\setminus Q)\cup\bigl(N(Q)\setminus Q\bigr).
\end{align*}
As $Q\subseteq\Int(X)$ and $X$ is attached at $\{a,b\}$, every edge of $G$
at a vertex of $Q$ is an edge of $X$ (Definition~\ref{def:N01-attach}), so
the neighbours in $G$ of the vertices of $Q$ are their neighbours in $X$:
$N(Q)=N_X(Q)$, and $N(Q)\setminus Q=N_X(Q)\setminus Q=\partial_XQ$.
Finally $\partial_XQ\subseteq V(X)=\Int(X)\cup\{a,b\}$
(Definition~\ref{def:N01-attach}), so
$\partial_XQ\subseteq(\partial_XQ\cap\Int(X))\cup\{a,b\}$.  Together,
\[
 \partial(P\cup Q)\subseteq(\partial P\setminus Q)\cup\{a,b\}\cup
 \bigl(\partial_XQ\cap\Int(X)\bigr).
\]
\end{proof}

\begin{theorem}[pathwidth of $G_\varphi$]\label{thm:N10-pw}
Let $\varphi$ be a restricted formula and $G:=G_\varphi$
(Definition~\ref{def:N03-graph}).  Then $G$ has a layout of width five,
hence a path decomposition of width at most five
(Definition~\ref{def:N10-pd}); so the pathwidth of $G$ is at most five.
\end{theorem}

\begin{proof}
We use the notation of Definitions~\ref{def:N03-graph}
and~\ref{def:N03-names} and of Lemma~\ref{lem:N10-pieces}; boundaries are
taken in $G$ unless said otherwise.

(1) \emph{The ordering.}  For $k=m,m-1,\dots,1$ in turn: the ear orders
of $U_{k,1}$, $U_{k,2}$, $U_{k,3}$; then $s_k$; then $t_k$; and, for
$k=m$ only, then $q$ and $q'$.  Then the ear order $z_1,\dots,z_7$ of $Z$.
Then the ear orders of $Y^1,Y^2,Y^3,Y^4$ in turn.  Then $T_1$, then
$T_2$.  By Lemma~\ref{lem:N03-structure}(i) and
Lemma~\ref{lem:N10-pwe}(i) applied to the pieces of
Lemma~\ref{lem:N10-pieces}, every vertex of $G$ appears exactly once.
Moreover, by Lemma~\ref{lem:N10-pwe}(i), the ear order of a piece $X$
lists exactly the vertices of $X$ other than its two attachment vertices,
that is, $\Int(X)$ (Definition~\ref{def:N01-attach}).  So, for every piece
$X$: the set $P$ of the vertices placed before the ear order of $X$ is
disjoint from $\Int(X)$; a nonempty prefix $Q$ of the ear order of $X$ is
a subset of $\Int(X)$; and $Q=\Int(X)$ at the end of the ear order.  We
bound the boundary $B$ of every prefix, going through the ordering, and
use at each piece the same two facts: Lemma~\ref{lem:N10-glue}, for the
piece $X$ attached at $\{a,b\}$, the set $P$ of the vertices placed before
the piece, and a prefix $Q$ of its ear order; and
Lemma~\ref{lem:N10-pwe}(ii), (iii) for $\partial_XQ$.

(2) \emph{A copy.}  Let $X:=U_{k,c}$, attached at $\{s_k,t_k\}$
(Lemma~\ref{lem:N03-structure}(ii)), and let $P$ be the set of the
vertices placed before the ear order of $X$.  We prove the following
claim: if $\partial P\subseteq\{s_k,t_k\}$, then $|B|\le5$ for every
prefix of the ordering ending in the ear order of $X$, and
$\partial(P\cup\Int(X))\subseteq\{s_k,t_k\}$.

Let $Q$ be a nonempty prefix of the ear order of $X$.  The prefix of the
ordering ending with $Q$ is $P\cup Q$, so $B=\partial(P\cup Q)$, and by
(1), $P\subseteq V(G)\setminus\Int(X)$ and $Q\subseteq\Int(X)$.  So
Lemma~\ref{lem:N10-glue} applies with $a:=s_k$ and $b:=t_k$ and gives
\[
 B\subseteq(\partial P\setminus Q)\cup\{s_k,t_k\}\cup
 \bigl(\partial_XQ\cap\Int(X)\bigr)
 \subseteq\{s_k,t_k\}\cup\bigl(\partial_XQ\cap\Int(X)\bigr),
\]
the second inclusion by the hypothesis $\partial P\subseteq\{s_k,t_k\}$.
It remains to bound $\partial_XQ\cap\Int(X)$.  By
Lemma~\ref{lem:N10-pieces}(i), $X$ is a path with ears with $w_0=s_k$,
$w_n=t_k$ and depth $d\le2$, and $Q$ is a nonempty prefix of its ear
order; so Lemma~\ref{lem:N10-pwe}(ii) gives
$\partial_XQ\subseteq\{s_k,t_k\}\cup M$ for a set $M$ with
$|M|\le d+1\le3$.  Hence $\partial_XQ\cap\Int(X)\subseteq M$ and
$|B|\le|\{s_k,t_k\}|+|M|\le2+3=5$.  At the end of the ear order,
$Q=\Int(X)$ by (1), and Lemma~\ref{lem:N10-pwe}(iii) gives
$\partial_XQ=\{s_k,t_k\}$; so
$\partial_XQ\cap\Int(X)=\emptyset$, and
$B=\partial(P\cup\Int(X))\subseteq\{s_k,t_k\}$.  This proves the claim.

(3) \emph{The frame vertices of the level $m$.}  Let $P$ be the set
of the vertices placed before $s_m$, and suppose
$\partial P\subseteq\{s_m,t_m\}$.  We prove the following claim: the four
prefixes $P\cup\{s_m\}$, $P\cup\{s_m,t_m\}$, $P\cup\{s_m,t_m,q\}$ and
$P\cup\{s_m,t_m,q,q'\}$ have at most four boundary vertices each, and
$\partial(P\cup\{s_m,t_m,q,q'\})\subseteq\{s_{m-1},t_{m-1}\}$.

By (1), $P=\Int(U_{m,1})\cup\Int(U_{m,2})\cup\Int(U_{m,3})$.  By
Definition~\ref{def:N03-graph},
\[
\begin{aligned}
 N(s_m)&=\{s_{m-1},v^{m,1}_1,v^{m,2}_1,v^{m,3}_1\},\\
 N(t_m)&=\{t_{m-1},q,q'\}\cup\{v^{m,c}_{L_m-1},r^{m,c}_{\mu_m+1}:c=1,2,3\},\\
 N(q)&=N(q')=\{t_m,t_{m-1}\},
\end{aligned}
\]
and the $v^{m,c}_1$, $v^{m,c}_{L_m-1}$, $r^{m,c}_{\mu_m+1}$ are members of
$U_{m,c}$, hence in $P$.  Lemma~\ref{lem:N10-place} with $w:=s_m$ gives
\[
 \partial(P\cup\{s_m\})=(\partial P\setminus\{s_m\})\cup
 \bigl(N(s_m)\setminus(P\cup\{s_m\})\bigr)\subseteq\{t_m\}\cup\{s_{m-1}\}.
\]
Lemma~\ref{lem:N10-place} with $w:=t_m$ gives
\begin{align*}
 \partial(P\cup\{s_m,t_m\})&=\bigl(\partial(P\cup\{s_m\})\setminus\{t_m\}\bigr)\cup\bigl(N(t_m)\setminus(P\cup\{s_m,t_m\})\bigr)\\
 &\subseteq\{s_{m-1}\}\cup\{t_{m-1},q,q'\}.
\end{align*}
Lemma~\ref{lem:N10-place} with $w:=q$ gives
\begin{align*}
 \partial(P\cup\{s_m,t_m,q\})&=\bigl(\partial(P\cup\{s_m,t_m\})\setminus\{q\}\bigr)\cup\bigl(N(q)\setminus(P\cup\{s_m,t_m,q\})\bigr)\\
 &\subseteq\{s_{m-1},t_{m-1},q'\}\cup\{t_{m-1}\}.
\end{align*}
Lemma~\ref{lem:N10-place} with $w:=q'$ gives
\begin{align*}
 \partial(P\cup\{s_m,t_m,q,q'\})&=\bigl(\partial(P\cup\{s_m,t_m,q\})\setminus\{q'\}\bigr)\cup\bigl(N(q')\setminus(P\cup\{s_m,t_m,q,q'\})\bigr)\\
 &\subseteq\{s_{m-1},t_{m-1}\}\cup\{t_{m-1}\}.
\end{align*}
The four boundaries
have at most four vertices.  This proves the claim.

(4) \emph{The frame vertices of a level $k$ with $1<k<m$.}  Let $k$
be a level with $1<k<m$, let $P$ be the set of the vertices placed before
$s_k$, and suppose $\partial P\subseteq\{s_k,t_k\}$.  We prove the
following claim: the two prefixes $P\cup\{s_k\}$ and $P\cup\{s_k,t_k\}$
have at most two boundary vertices each, and
$\partial(P\cup\{s_k,t_k\})\subseteq\{s_{k-1},t_{k-1}\}$.

By (1), $P$ consists of the vertices of the levels $k+1,\dots,m$, of
$q,q'$, and of $\Int(U_{k,c})$ for $c=1,2,3$.  By
Definition~\ref{def:N03-graph},
\[
\begin{aligned}
 N(s_k)&=\{s_{k-1},s_{k+1},v^{k,1}_1,v^{k,2}_1,v^{k,3}_1\},\\
 N(t_k)&=\{t_{k-1},t_{k+1}\}\cup\{v^{k,c}_{L_k-1},r^{k,c}_{\mu_k+1}:c=1,2,3\}
   \cup\{q,q':k=m-1\},
\end{aligned}
\]
and the $v^{k,c}_1$, $v^{k,c}_{L_k-1}$, $r^{k,c}_{\mu_k+1}$ are members of
$U_{k,c}$, hence in $P$, as are $s_{k+1}$, $t_{k+1}$, $q$, $q'$.
Lemma~\ref{lem:N10-place} with $w:=s_k$ gives
\[
 \partial(P\cup\{s_k\})=(\partial P\setminus\{s_k\})\cup
 \bigl(N(s_k)\setminus(P\cup\{s_k\})\bigr)\subseteq\{t_k\}\cup\{s_{k-1}\}.
\]
Lemma~\ref{lem:N10-place} with $w:=t_k$ gives
\[
 \partial(P\cup\{s_k,t_k\})=\bigl(\partial(P\cup\{s_k\})\setminus\{t_k\}\bigr)\cup
 \bigl(N(t_k)\setminus(P\cup\{s_k,t_k\})\bigr)\subseteq\{s_{k-1}\}\cup\{t_{k-1}\}.
\]
The two boundaries have at most two vertices.  This proves the
claim.

(5) \emph{The frame vertices of the level $1$.}  Let $P$ be the set of
the vertices placed before $s_1$, and suppose
$\partial P\subseteq\{s_1,t_1\}$.  We prove the following claim: the two
prefixes $P\cup\{s_1\}$ and $P\cup\{s_1,t_1\}$ have at most two boundary
vertices each, and $\partial(P\cup\{s_1,t_1\})\subseteq\{T_1,z_1\}$.

By (1), $P$ consists of the vertices of the levels $2,\dots,m$, of
$q,q'$, and of $\Int(U_{1,c})$ for $c=1,2,3$.  By
Definition~\ref{def:N03-graph}, the frame consists of the paths
$T_1,s_1,\dots,s_m$ and $t_m,\dots,t_1,z_1,\dots,z_7,T_2$, so
\[
\begin{aligned}
 N(s_1)&=\{T_1,s_2,v^{1,1}_1,v^{1,2}_1,v^{1,3}_1\},\\
 N(t_1)&=\{z_1,t_2\}\cup\{v^{1,c}_{L_1-1},r^{1,c}_{\mu_1+1}:c=1,2,3\},
\end{aligned}
\]
and the $v^{1,c}_1$, $v^{1,c}_{L_1-1}$, $r^{1,c}_{\mu_1+1}$ are members of
$U_{1,c}$, hence in $P$, as are $s_2$, $t_2$.
Lemma~\ref{lem:N10-place} with $w:=s_1$ gives
\[
 \partial(P\cup\{s_1\})=(\partial P\setminus\{s_1\})\cup
 \bigl(N(s_1)\setminus(P\cup\{s_1\})\bigr)\subseteq\{t_1\}\cup\{T_1\}.
\]
Lemma~\ref{lem:N10-place} with $w:=t_1$ gives
\[
 \partial(P\cup\{s_1,t_1\})=\bigl(\partial(P\cup\{s_1\})\setminus\{t_1\}\bigr)\cup
 \bigl(N(t_1)\setminus(P\cup\{s_1,t_1\})\bigr)\subseteq\{T_1\}\cup\{z_1\}.
\]
The two boundaries have at most two vertices.  This proves the
claim.

(6) \emph{The levels.}  Let $P_k$ be the set of the vertices placed
before the ear order of $U_{k,1}$; by (1), $P_m=\emptyset$, and for
$1<k\le m$, $P_{k-1}$ is the set $P\cup\{s_m,t_m,q,q'\}$ of (3) if $k=m$
and the set $P\cup\{s_k,t_k\}$ of (4) if $k<m$, where $P$ is the set
placed before $s_k$.  We prove by induction on $k=m,m-1,\dots,1$ that
$\partial P_k\subseteq\{s_k,t_k\}$ and that $|B|\le5$ for every prefix
ending with a vertex of the level $k$, that is, a member of $U_{k,1}$,
$U_{k,2}$, $U_{k,3}$, or $s_k$, $t_k$, or $q,q'$ when $k=m$.  For $k=m$,
$\partial P_m=\emptyset$.  Let $k$ be a level with
$\partial P_k\subseteq\{s_k,t_k\}$.  The claim of (2) for $U_{k,1}$
gives $|B|\le5$ during its ear order and
$\partial(P_k\cup\Int(U_{k,1}))\subseteq\{s_k,t_k\}$; as
$P_k\cup\Int(U_{k,1})$ is the set placed before the ear order of
$U_{k,2}$, the claim of (2) for $U_{k,2}$ gives the same, and then the
claim for $U_{k,3}$.  So the set $P$ placed before $s_k$ has
$\partial P\subseteq\{s_k,t_k\}$.  If $k>1$, the claim of (3) if $k=m$,
or of (4) if $k<m$, gives $|B|\le4$ for the prefixes ending with $s_k$,
$t_k$ and, for $k=m$, $q$, $q'$, and
$\partial P_{k-1}\subseteq\{s_{k-1},t_{k-1}\}$, the hypothesis for
$k-1$.  If $k=1$, the claim of (5) gives $|B|\le2$ for the prefixes
ending with $s_1$ and $t_1$, and $\partial(P\cup\{s_1,t_1\})
\subseteq\{T_1,z_1\}$ for the prefix ending with $t_1$.

(7) \emph{The tail.}  Let $X:=Z$, attached at $\{t_1,T_2\}$
(Lemma~\ref{lem:N10-pieces}(iii)), and let $P$ be the set of the
vertices placed before $z_1$; by (6), $\partial P\subseteq\{T_1,z_1\}$,
and $t_1\in P$.  We prove the following claim: $|B|\le3$ for every
prefix of the ordering ending in $z_1,\dots,z_7$, and
$\partial(P\cup\Int(X))\subseteq\{T_1,T_2\}$.

Let $Q$ be a nonempty prefix of $z_1,\dots,z_7$, the ear order of $X$.
The prefix of the ordering ending with $Q$ is $P\cup Q$, so
$B=\partial(P\cup Q)$, and by (1), $P\subseteq V(G)\setminus\Int(X)$ and
$Q\subseteq\Int(X)$.  So Lemma~\ref{lem:N10-glue} applies with $a:=t_1$
and $b:=T_2$ and gives
\[
 B\subseteq(\partial P\setminus Q)\cup\{t_1,T_2\}\cup
 \bigl(\partial_XQ\cap\Int(X)\bigr)
 \subseteq\{T_1\}\cup\{t_1,T_2\}\cup\bigl(\partial_XQ\cap\Int(X)\bigr),
\]
the second inclusion as $\partial P\subseteq\{T_1,z_1\}$ and $z_1\in Q$.
Moreover $t_1\notin B$, as $t_1\in P$; so
$B\subseteq\{T_1,T_2\}\cup(\partial_XQ\cap\Int(X))$.  It remains to bound
$\partial_XQ\cap\Int(X)$.  By Lemma~\ref{lem:N10-pieces}(iii), $X$ is a
path with ears with $w_0=t_1$, $w_n=T_2$ and depth $0$, and $Q$ is a
nonempty prefix of its ear order; so Lemma~\ref{lem:N10-pwe}(ii) gives
$\partial_XQ\subseteq\{t_1,T_2\}\cup M$ for a set $M$ with $|M|\le1$.
Hence $\partial_XQ\cap\Int(X)\subseteq M$ and
$|B|\le|\{T_1,T_2\}|+|M|\le2+1=3$.  At the end of the ear order,
$Q=\Int(X)$ by (1), and Lemma~\ref{lem:N10-pwe}(iii) gives
$\partial_XQ=\{t_1,T_2\}$; so
$\partial_XQ\cap\Int(X)=\emptyset$, and
$B=\partial(P\cup\Int(X))\subseteq\{T_1,T_2\}$.  This proves the claim.

(8) \emph{An arm.}  Let $1\le j\le4$, let $X:=Y^j$, attached at
$\{T_1,T_2\}$ (Lemma~\ref{lem:N03-structure}(ii)), and let $P$ be the
set of the vertices placed before the ear order of $X$.  We prove the
following claim: if $\partial P\subseteq\{T_1,T_2\}$, then $|B|\le4$ for
every prefix of the ordering ending in the ear order of $X$, and
$\partial(P\cup\Int(X))\subseteq\{T_1,T_2\}$.

Let $Q$ be a nonempty prefix of the ear order of $X$.  The prefix of the
ordering ending with $Q$ is $P\cup Q$, so $B=\partial(P\cup Q)$, and by
(1), $P\subseteq V(G)\setminus\Int(X)$ and $Q\subseteq\Int(X)$.  So
Lemma~\ref{lem:N10-glue} applies with $a:=T_1$ and $b:=T_2$ and gives
\[
 B\subseteq(\partial P\setminus Q)\cup\{T_1,T_2\}\cup
 \bigl(\partial_XQ\cap\Int(X)\bigr)
 \subseteq\{T_1,T_2\}\cup\bigl(\partial_XQ\cap\Int(X)\bigr),
\]
the second inclusion by the hypothesis $\partial P\subseteq\{T_1,T_2\}$.
It remains to bound $\partial_XQ\cap\Int(X)$.  By
Lemma~\ref{lem:N10-pieces}(ii), $X$ is a path with ears with $w_0=T_1$,
$w_n=T_2$ and depth $1$, and $Q$ is a nonempty prefix of its ear order;
so Lemma~\ref{lem:N10-pwe}(ii) gives
$\partial_XQ\subseteq\{T_1,T_2\}\cup M$ for a set $M$ with $|M|\le2$.
Hence $\partial_XQ\cap\Int(X)\subseteq M$ and
$|B|\le|\{T_1,T_2\}|+|M|\le2+2=4$.  At the end of the ear order,
$Q=\Int(X)$ by (1), and Lemma~\ref{lem:N10-pwe}(iii) gives
$\partial_XQ=\{T_1,T_2\}$; so $\partial_XQ\cap\Int(X)=\emptyset$, and
$B=\partial(P\cup\Int(X))\subseteq\{T_1,T_2\}$.  This proves the claim.

(9) \emph{The four arms.}  Let $P^j$ be the set of the vertices placed
before the ear order of $Y^j$; by (1), $P^1=P\cup\Int(Z)$ for the set
$P$ of (7), and $P^{j+1}=P^j\cup\Int(Y^j)$ for $1\le j\le3$.  We prove
by induction on $j=1,\dots,4$ that $\partial P^j\subseteq\{T_1,T_2\}$
and that $|B|\le4$ for every prefix of the ordering ending in the ear
order of $Y^j$.  For $j=1$, $\partial P^1\subseteq\{T_1,T_2\}$ by (7).
Let $1\le j\le3$ with $\partial P^j\subseteq\{T_1,T_2\}$.  The claim of
(8) for $Y^j$ gives $|B|\le4$ during its ear order and
$\partial P^{j+1}=\partial(P^j\cup\Int(Y^j))\subseteq\{T_1,T_2\}$, the
hypothesis for $j+1$.  For $j=4$, the claim of (8) for $Y^4$ gives
$|B|\le4$ during its ear order.

(10) \emph{The end.}  Let $P$ be the set of the vertices placed before
$T_1$.  By (1), $T_1$ comes right after the ear order of $Y^4$, so
$P=P^4\cup\Int(Y^4)$ with $P^4$ the set of (9), and the claim of (8) for
$Y^4$, with $\partial P^4\subseteq\{T_1,T_2\}$ from (9), gives
$\partial P\subseteq\{T_1,T_2\}$.  By Definition~\ref{def:N03-graph},
\[
 N(T_1)=\{s_1,y^1_1,y^2_1,y^3_1,y^4_1\},\qquad
 N(T_2)=\{z_7,y^1_{D-1},y^2_{D-1},y^3_{D-1},y^4_{D-1}\},
\]
and all these vertices lie in $P$.  Lemma~\ref{lem:N10-place} with $w:=T_1$
gives
\[
 \partial(P\cup\{T_1\})=(\partial P\setminus\{T_1\})\cup
 \bigl(N(T_1)\setminus(P\cup\{T_1\})\bigr)\subseteq\{T_2\}\cup\emptyset .
\]
Lemma~\ref{lem:N10-place} with $w:=T_2$ gives
\[
 \partial(P\cup\{T_1,T_2\})=\bigl(\partial(P\cup\{T_1\})\setminus\{T_2\}\bigr)\cup
 \bigl(N(T_2)\setminus(P\cup\{T_1,T_2\})\bigr)=\emptyset .
\]
The two boundaries have at most one vertex.

By (1), every prefix of the ordering ends with a vertex of a level, of
the tail, of an arm, or with $T_1$ or $T_2$; so, by (6), (7), (9) and
(10), every prefix has at most five boundary vertices.  The ordering is
a layout of width five (Definition~\ref{def:N10-pd}), and
Lemma~\ref{lem:N10-layout} gives a path decomposition of width at most
five.
\end{proof}


\section{Four properties of \texorpdfstring{$G_\varphi$}{G\_phi}}\label{sec:N09}

\begin{lemma}[four properties of $G_\varphi$]\label{lem:N09-placement}
Let $\varphi$ be a restricted formula and $G:=G_\varphi$
(Definition~\ref{def:N03-graph}); $G$ is connected
(Lemma~\ref{lem:N04-potentials}).  Let $D=4m+15$ be the number of
Definition~\ref{def:N03-params}.
\begin{enumerate}
\item[\rm(i)] $G$ is not outerplanar: it contains a subdivision of
  $K_{2,3}$.
\item[\rm(ii)] $d(T_1,T_2)=D$; so the diameter of $G$, the largest
  distance between two of its vertices, is at least $4m+15$.
\item[\rm(iii)] $|E(G)|-|V(G)|+1\le12m+9$; so deleting some set of at
  most $12m+9$ edges of $G$ leaves a spanning tree of $G$.
\item[\rm(iv)] $G$ is two-connected (Definition~\ref{def:N01-graph}).
\end{enumerate}
\end{lemma}

\begin{proof}
(1) We prove (i).  Let $k$ be a level.  By Definition~\ref{def:N03-graph},
the spines $v^{k,c}_0,v^{k,c}_1,\dots,v^{k,c}_{L_k}$ of the three copies
$U_{k,1},U_{k,2},U_{k,3}$ are paths from $v^{k,c}_0=s_k$ to
$v^{k,c}_{L_k}=t_k$ with $L_k=4m+8-2k\ge2$ edges, whose inner vertices
are members of distinct copies, hence pairwise distinct.  Three
internally disjoint $s_k$--$t_k$ paths with at least two edges each form
a subdivision of $K_{2,3}$, and a graph containing a subdivision of
$K_{2,3}$ has $K_{2,3}$ as a minor.  By \cite[Lemma~78 and the remark
after it]{Bodlaender98}, the class of outerplanar graphs is closed under
taking minors and $K_{2,3}$ is not outerplanar.  So $G$ is not
outerplanar.

(2) We prove (ii).  By Lemma~\ref{lem:N04-skeleton}(iv), $d(T_1,T_2)=D$,
and $D=4m+15$ by Definition~\ref{def:N03-params}.  So the largest
distance between two vertices of $G$ is at least $4m+15$.

(3) We prove (iii).  By Definition~\ref{def:N03-graph}, $G$ is the union
of the frame, the $3m$ copies, the pair $q,q'$ with its four edges and the four arms, whose
vertices other than the terminals $s_k,t_k$ of the copies and $T_1,T_2$
of the arms are pairwise distinct.  Every edge of a copy has an end among
its members (no edge of a copy joins $s_k$ to $t_k$, as $L_k\ge2$ and the
$t$-ear and the crossings end at spine vertices other than both
terminals), every edge of an arm has an end in $\Int(Y^j)$, and the
four edges at $q,q'$ have an end in $\{q,q'\}$; so the edge sets of the
parts are pairwise disjoint, and $|E(G)|$ is their sum.  We count the
edges minus the vertices, part by part, counting the terminals in the
frame:
\begin{itemize}
\item the frame, the paths $T_1,s_1,\dots,s_m$ and
  $t_m,\dots,t_1,z_1,\dots,z_7,T_2$: $2m+9$ vertices and $2m+7$ edges, a
  difference of $-2$;
\item a copy $U_{k,c}$ with $e:=|Z_{k,c}|$ crossings:
  $L_k+(\mu_k+2)+\sum_{j<e}\ell^{k,c}_j$ edges and
  $(L_k-1)+(\mu_k+1)+\sum_{j<e}(\ell^{k,c}_j-1)$ vertices other than
  $s_k,t_k$, a difference of $2+e\le4$, as $e\le2$
  (Lemma~\ref{lem:N03-options}(ii));
\item the pair $q,q'$: $4$ edges and $2$ vertices, a difference
  of $2$;
\item an arm $Y^j$: $D+2$ edges and the $D$ vertices
  $y^j_1,\dots,y^j_{D-1},\tilde y^j$ other than $T_1,T_2$, a difference
  of $2$.
\end{itemize}
Hence
\[
 |E(G)|-|V(G)|\le-2+3m\cdot4+2+4\cdot2=12m+8 .
\]
As $G$ is connected, it has a spanning tree, and a tree with $|V(G)|$
vertices has $|V(G)|-1$ edges \cite[Section~1.5]{Diestel17}; deleting
the other $|E(G)|-|V(G)|+1\le12m+9$ edges of $G$ leaves that tree.

(4) We prove (iv).  $G$ has the four distinct vertices $T_1,T_2,s_1,t_1$
(Definition~\ref{def:N03-graph}).  Let $v$ be a vertex of $G$.  We prove
that $G-v$ is connected, that is, that any two vertices of $G-v$ are
joined by a path of $G-v$.  If some vertex $T$ of $G-v$ is joined in
$G-v$ to every vertex of $G-v$, this holds: two vertices are joined
through $T$.  We find such a vertex $T$.

By Theorem~\ref{thm:N07-tw}(i), $(G;T_1,T_2)\in\mathcal{SP}$, so by
Lemma~\ref{lem:N07-sp}(iv) every vertex of $G-v$ is joined in $G-v$ to
$T_1$ or to $T_2$.

If $v\in\{T_1,T_2\}$, then $v$ is not a vertex of $G-v$, so every vertex
of $G-v$ is joined in $G-v$ to the other terminal, and that terminal is
such a $T$.

Let $v\notin\{T_1,T_2\}$.  By Definition~\ref{def:N03-graph}, each arm
$Y^j$ contains the path $y^j_0,\dots,y^j_D$ from $T_1$ to $T_2$, whose
inner vertices lie in $\Int(Y^j)$, and the four sets $\Int(Y^j)$ are
pairwise disjoint (Lemma~\ref{lem:N03-structure}(i)).  So $v$ lies in at
most one of them, and for another $j$ the path $y^j_0,\dots,y^j_D$ joins
$T_1$ and $T_2$ in $G-v$.  Hence a vertex of $G-v$ joined to $T_2$ is
joined to $T_1$ through this path, so every vertex of $G-v$ is joined in
$G-v$ to $T_1$, and $T_1$ is such a $T$.

So $G-v$ is connected for every vertex $v$, and $G$ is two-connected
(Definition~\ref{def:N01-graph}).
\end{proof}


\section{The main theorem}\label{sec:N08}

\begin{lemma}[the cells of a copy]\label{lem:N08-cells}
Let $\varphi$ be a restricted formula and $G:=G_\varphi$
(Definition~\ref{def:N03-graph}).  Let $k$ be a level, $c$ an option and
$Z_{k,c}$ its option set (Definition~\ref{def:N03-options}), $\mu_k$ the
number of Definition~\ref{def:N03-params}, and $\delta_1=d(T_1,\cdot)$
(Lemma~\ref{lem:N04-potentials}).  Then the set of values of $\delta_1$
at the diagonal members of $U_{k,c}$ (Definition~\ref{def:N03-names}) is
\[
 \{2m+4-k\}\cup Z_{k,c},
\]
and $\delta_1$ is injective on these members.
\end{lemma}

\begin{proof}
(1) Let $Z_{k,c}=\{\zeta_0<\dots<\zeta_{e-1}\}$, so that $e=|Z_{k,c}|$ and
the entries of $E_{k,c}$ are $\varepsilon_j=2(\zeta_j-k-\mu_k)$ for $j<e$
(Definition~\ref{def:N03-params}).  By Lemma~\ref{lem:N02-skew}(iii), the
diagonal members of the template $\T(L_k,\mu_k;E_{k,c})$ are $v_{\mu_k}$,
with $h=\mu_k$, and $x^j_{j+1+\varepsilon_j/2}$ for each $j<e$, with
$h=\mu_k+\varepsilon_j/2$.  By Lemma~\ref{lem:N03-copies}(iii), the
isomorphism $\psi$ of Lemma~\ref{lem:N03-copies} maps the diagonal
members of the template onto those of $U_{k,c}$; so these are
$v^{k,c}_{\mu_k}$ and $x^{k,c,j}_{j+1+\varepsilon_j/2}$ ($j<e$), and by
Lemma~\ref{lem:N03-copies}(ii) their profiles are
\[
 h_{U_{k,c}}(v^{k,c}_{\mu_k})=\mu_k,\qquad
 h_{U_{k,c}}(x^{k,c,j}_{j+1+\varepsilon_j/2})=\mu_k+\varepsilon_j/2
 =\mu_k+(\zeta_j-k-\mu_k)=\zeta_j-k .
\]

(2) By the line $u\in\Int(U_{k,c})$ of the table \eqref{eq:N04-values} of
Lemma~\ref{lem:N04-potentials}, $\delta_1(u)=k+h_{U_{k,c}}(u)$ for every
member $u$ of $U_{k,c}$.  With (1) and $\mu_k=2m+4-2k$
(Definition~\ref{def:N03-params}),
\[
 \delta_1(v^{k,c}_{\mu_k})=k+\mu_k=2m+4-k,\qquad
 \delta_1(x^{k,c,j}_{j+1+\varepsilon_j/2})=k+(\zeta_j-k)=\zeta_j\quad(j<e).
\]
So the set of values is $\{2m+4-k\}\cup Z_{k,c}$.  These $e+1$ values are
pairwise distinct: as $k\ge1$, $2m+4-k\le2m+3$; by
Lemma~\ref{lem:N03-options}(i), every cell is at least $2m+6$, so
$\zeta_j\ge2m+6>2m+3$ for each $j<e$; and $\zeta_0<\zeta_1$ if $e=2$.  The
$e+1$ diagonal members are distinct vertices with pairwise distinct
values of $\delta_1$, so $\delta_1$ is injective on them.
\end{proof}

\begin{definition}[the assignment and the literals of a choice]\label{def:N08-assign}
Let $\varphi$ be a restricted formula with the clauses $C_1,\dots,C_n$
and the variables $\xi_1,\dots,\xi_r$ (Definition~\ref{def:N03-formula}).
For a choice $f$ (Definition~\ref{def:N03-pattern}) let $A_f$ be the
assignment that gives each variable the value set by its selected option,
$A_f(\xi_g):=v(f_{n+g})$ ($1\le g\le r$), so that
$A_f(\xi_g)=\mathrm{true}$ iff $f_{n+g}=1$ (Definition~\ref{def:N03-options}), and
let $\lambda_i(f):=\lambda_{i,\min(f_i,|C_i|)}\in C_i$ be the literal
selected by $f_i$ at the level $i$ ($1\le i\le n$)
(Definition~\ref{def:N03-options}).
\end{definition}

\begin{lemma}[patterns and satisfying choices]\label{lem:N08-choice}
Let $\varphi$ be a restricted formula with the clauses
$C_1,\dots,C_n$ and the variables $\xi_1,\dots,\xi_r$, so that $m=n+r+1$
(Definition~\ref{def:N03-formula}), and with the levels $1,\dots,n$ of the
clauses, $n+1,\dots,n+r$ of the variables and $m$ the idle level
(Definition~\ref{def:N03-options}).  Let $G:=G_\varphi$
(Definition~\ref{def:N03-graph}), $f$ a choice and $W_f$ its pattern
(Definition~\ref{def:N03-pattern}), and $A_f$, $\lambda_i(f)$
($1\le i\le n$) as in Definition~\ref{def:N08-assign}.  Then $W_f$
resolves $G$ iff $A_f$ satisfies $\lambda_i(f)$ for every $1\le i\le n$
(Figure~\ref{fig:N08-cells}).
\end{lemma}

\begin{proof}
(1) By Definition~\ref{def:N01-resolve}, $W_f$ resolves $G$ iff it
separates every two distinct vertices of $G$.  By
Theorem~\ref{thm:N06-validity}(i), $W_f$ does not resolve $G$ iff
there are levels $k\ne k'$ with $k,k'<m$, a diagonal member $x$ of $U_{k,f_k}$ and a
diagonal member $y$ of $U_{k',f_{k'}}$ with $\delta_1(x)=\delta_1(y)$;
such $x$ and $y$ are distinct, being members of different copies
(Lemma~\ref{lem:N03-structure}(i)).  By Lemma~\ref{lem:N08-cells} at the
level $k$ and the option $f_k$, the values of $\delta_1$ at the diagonal
members of $U_{k,f_k}$ form the set $\{2m+4-k\}\cup Z_{k,f_k}$, and at
the level $k'$ and the option $f_{k'}$ the set
$\{2m+4-k'\}\cup Z_{k',f_{k'}}$.  So, naming the two levels so that
$k<k'$, $W_f$ resolves $G$ iff for all levels $k<k'$ with $k'<m$ the sets
$\{2m+4-k\}\cup Z_{k,f_k}$ and $\{2m+4-k'\}\cup Z_{k',f_{k'}}$ are
disjoint.  As $k\ne k'$, $2m+4-k\ne2m+4-k'$; and $2m+4-k\le2m+3<2m+6$,
while every cell is at least $2m+6$ (Lemma~\ref{lem:N03-options}(i)), so
neither common value lies in an option set.  Hence the two sets are
disjoint iff $Z_{k,f_k}\cap Z_{k',f_{k'}}=\emptyset$, and $W_f$ resolves
$G$ iff $Z_{k,f_k}\cap Z_{k',f_{k'}}=\emptyset$ for all levels $k<k'$
with $k'<m$; as $Z_{m,c}=\emptyset$ (Definition~\ref{def:N03-options}),
the condition $k'<m$ may be dropped.

(2) \emph{Lemma~\ref{lem:N03-options}(iv) for the choice $f$.}  For
$1\le i\le n$ let $g(i)$ be the index with $\xi_{g(i)}$ the variable of
$\lambda_i(f)$; the level of $\xi_{g(i)}$ is $n+g(i)$
(Definition~\ref{def:N03-options}).  Let $k<k'$ be levels.  By
Definition~\ref{def:N08-assign}, if $k\le n$, the literal that $f_k$
selects at the level $k$ is $\lambda_k(f)$.  So
Lemma~\ref{lem:N03-options}(iv), at the levels $k<k'$ and the options
$f_k,f_{k'}$, says
\begin{equation}\label{eq:N08-meet}
 Z_{k,f_k}\cap Z_{k',f_{k'}}\ne\emptyset
 \iff k\le n,\ k'=n+g(k),\ \text{and $f_{n+g(k)}$ falsifies $\lambda_k(f)$.}
\end{equation}

(3) We prove that the following two statements are equivalent:
\begin{align}
 &Z_{k,f_k}\cap Z_{k',f_{k'}}\ne\emptyset\ \text{ for some levels $k<k'$;}
 \label{eq:N08-exists}\\
 &\text{$f_{n+g(i)}$ falsifies $\lambda_i(f)$ for some $i\in\{1,\dots,n\}$.}
 \label{eq:N08-some}
\end{align}

\eqref{eq:N08-exists} $\Rightarrow$ \eqref{eq:N08-some}.  Let $k<k'$ be
levels with $Z_{k,f_k}\cap Z_{k',f_{k'}}\ne\emptyset$.  By
\eqref{eq:N08-meet}, $k\le n$ and $f_{n+g(k)}$ falsifies $\lambda_k(f)$.
So $i:=k$ satisfies \eqref{eq:N08-some}.

\eqref{eq:N08-some} $\Rightarrow$ \eqref{eq:N08-exists}.  Let
$i\in\{1,\dots,n\}$ with $f_{n+g(i)}$ falsifying $\lambda_i(f)$.  Put
$k:=i$ and $k':=n+g(i)$.  Then $k\le n<k'$, so $k<k'$ are levels, and the
right side of \eqref{eq:N08-meet} holds for them.  By \eqref{eq:N08-meet},
$Z_{k,f_k}\cap Z_{k',f_{k'}}\ne\emptyset$, which is \eqref{eq:N08-exists}.

(4) For every literal $\lambda$ with variable $\xi_g$,
\begin{equation}\label{eq:N08-falsify}
 \text{$f_{n+g}$ falsifies $\lambda$}\iff\text{$A_f$ does not satisfy $\lambda$.}
\end{equation}
Indeed, by Definition~\ref{def:N03-options}, $f_{n+g}$ falsifies
$\lambda$ iff $\lambda=\neg\xi_g$ and $f_{n+g}=1$, or $\lambda=\xi_g$ and
$f_{n+g}\in\{2,3\}$; that is, iff $\lambda=\neg\xi_g$ and
$v(f_{n+g})=\mathrm{true}$, or $\lambda=\xi_g$ and
$v(f_{n+g})=\mathrm{false}$.  By Definition~\ref{def:N08-assign},
$v(f_{n+g})=A_f(\xi_g)$.  By Definition~\ref{def:N03-formula}, $A_f$ does
not satisfy $\lambda$ iff $\lambda=\neg\xi_g$ and
$A_f(\xi_g)=\mathrm{true}$, or $\lambda=\xi_g$ and
$A_f(\xi_g)=\mathrm{false}$.  The two conditions coincide.

(5) We conclude:
\[
\begin{aligned}
 \text{$W_f$ resolves $G$}
 &\iff \forall\ \text{levels $k<k'$}:\ Z_{k,f_k}\cap Z_{k',f_{k'}}=\emptyset
 &&\text{by (1)}\\
 &\iff \forall\ i\in\{1,\dots,n\}:\ \text{$f_{n+g(i)}$ does not falsify $\lambda_i(f)$}
 &&\text{by (3)}\\
 &\iff \forall\ i\in\{1,\dots,n\}:\ \text{$A_f$ satisfies $\lambda_i(f)$}
 &&\text{by \eqref{eq:N08-falsify},}
\end{aligned}
\]
where \eqref{eq:N08-falsify} is applied at $\lambda:=\lambda_i(f)$ and
$g:=g(i)$ for each $i$.
\end{proof}

\begin{figure}[ht]
\centering
\begin{subfigure}{\textwidth}
\centering
\begin{tikzpicture}[font=\scriptsize,x=0.62cm,y=0.5cm]
\foreach \c in {9,...,18} { \draw[gray!40] (\c,0.4) -- (\c,-4.4); \node[below] at (\c,-4.4) {$\c$}; }
\node[below] at (19.2,-4.4) {$\delta_1$};
\node[left] at (8.6,0) {level $1$, clause $C_1=(\xi_1,\xi_2)$};
\draw[gray!40] (8.7,0) -- (18.3,0);
\node[hv,inner sep=1.6pt] at (13,0) {};
\node[circle,fill,inner sep=1.6pt] at (17,0) {};
\node[left] at (8.6,-1) {level $2$, clause $C_2=(\neg\xi_1)$};
\draw[gray!40] (8.7,-1) -- (18.3,-1);
\node[hv,inner sep=1.6pt] at (12,-1) {};
\node[circle,fill,inner sep=1.6pt] at (18,-1) {};
\node[left] at (8.6,-2) {level $3$, variable $\xi_1$};
\draw[gray!40] (8.7,-2) -- (18.3,-2);
\node[hv,inner sep=1.6pt] at (11,-2) {};
\node[circle,fill,inner sep=1.6pt] at (16,-2) {};
\node[left] at (8.6,-3) {level $4$, variable $\xi_2$};
\draw[gray!40] (8.7,-3) -- (18.3,-3);
\node[hv,inner sep=1.6pt] at (10,-3) {};
\node[left] at (8.6,-4) {level $5$, idle};
\draw[gray!40] (8.7,-4) -- (18.3,-4);
\node[hv,inner sep=1.6pt] at (9,-4) {};
\end{tikzpicture}
\caption{$f=(2,1,2,1,1)$: clause $C_1$ chooses $\xi_2$, clause $C_2$ chooses
$\neg\xi_1$, $\xi_1$ is false and $\xi_2$ is true.  No two levels share a
value, and $W_f$ resolves $G_\varphi$.}
\end{subfigure}

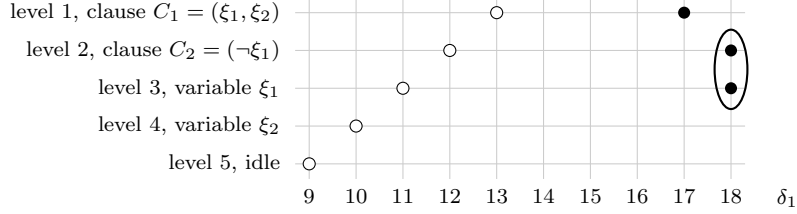
\medskip
\begin{subfigure}{\textwidth}
\centering
\begin{tikzpicture}[font=\scriptsize,x=0.62cm,y=0.5cm]
\foreach \c in {9,...,18} { \draw[gray!40] (\c,0.4) -- (\c,-4.4); \node[below] at (\c,-4.4) {$\c$}; }
\node[below] at (19.2,-4.4) {$\delta_1$};
\node[left] at (8.6,0) {level $1$, clause $C_1=(\xi_1,\xi_2)$};
\draw[gray!40] (8.7,0) -- (18.3,0);
\node[hv,inner sep=1.6pt] at (13,0) {};
\node[circle,fill,inner sep=1.6pt] at (17,0) {};
\node[left] at (8.6,-1) {level $2$, clause $C_2=(\neg\xi_1)$};
\draw[gray!40] (8.7,-1) -- (18.3,-1);
\node[hv,inner sep=1.6pt] at (12,-1) {};
\node[circle,fill,inner sep=1.6pt] at (18,-1) {};
\node[left] at (8.6,-2) {level $3$, variable $\xi_1$};
\draw[gray!40] (8.7,-2) -- (18.3,-2);
\node[hv,inner sep=1.6pt] at (11,-2) {};
\node[circle,fill,inner sep=1.6pt] at (18,-2) {};
\node[left] at (8.6,-3) {level $4$, variable $\xi_2$};
\draw[gray!40] (8.7,-3) -- (18.3,-3);
\node[hv,inner sep=1.6pt] at (10,-3) {};
\node[left] at (8.6,-4) {level $5$, idle};
\draw[gray!40] (8.7,-4) -- (18.3,-4);
\node[hv,inner sep=1.6pt] at (9,-4) {};
\draw[thick] (18,-1.5) ellipse (0.35 and 1.05);
\end{tikzpicture}
\caption{$f=(2,1,1,1,1)$: the same, except that $\xi_1$ is true, which
falsifies the chosen literal $\neg\xi_1$ of $C_2$.  Levels $2$ and $3$
share the cell $18$.}
\end{subfigure}
\caption{How a choice $f$ is tested, for the formula
$\varphi=((\xi_1,\xi_2),(\neg\xi_1))$: two clauses and two variables,
so $m=5$ levels, the last one idle, one row each.  The three occurrences $(1,\xi_1)$,
$(1,\xi_2)$, $(2,\neg\xi_1)$ have the cells $16$, $17$, $18$
(Definition~\ref{def:N03-options}).  A dot in the row of the level $k$
at the column $x$ is a diagonal member of the free copy $U_{k,f_k}$ at
distance $\delta_1=x$ from $T_1$: the white dot is the common value
$2m+4-k$, the black dots are the cells of the chosen option
(Lemma~\ref{lem:N08-cells}); at a clause level the cell of the chosen
literal, at a variable level the cells of the occurrences of the literal
that the chosen value falsifies.  By Theorem~\ref{thm:N06-validity}(i),
$W_f$ resolves $G_\varphi$ iff no column holds two dots.  In (a), every
column holds one dot and $W_f$ resolves $G_\varphi$; the assignment
$\xi_1$ false, $\xi_2$ true satisfies $\varphi$.  In (b), $\xi_1$ is set
true, which falsifies the literal $\neg\xi_1$ chosen for $C_2$, and the
rows of the levels $2$ and $3$ both hold the cell $18$ (circled): these
two diagonal members are the only pair of vertices that $W_f$ does not
separate.}
\label{fig:N08-cells}
\end{figure}

\begin{corollary}[some pattern resolves iff $\varphi$ is satisfiable]\label{cor:N08-sat}
Let $\varphi$ be a restricted formula with the clauses
$C_1,\dots,C_n$ and the variables $\xi_1,\dots,\xi_r$, so that $m=n+r+1$
(Definition~\ref{def:N03-formula}), and with the levels $1,\dots,n$ of the
clauses and $n+1,\dots,n+r$ of the variables
(Definition~\ref{def:N03-options}).  Let $G:=G_\varphi$
(Definition~\ref{def:N03-graph}).  Some pattern $W_f$
(Definition~\ref{def:N03-pattern}) resolves $G$ iff $\varphi$ is
satisfiable (Definition~\ref{def:N03-formula}).
\end{corollary}

\begin{proof}
(1) $\Rightarrow$.  Let $W_f$ resolve $G$.  By Lemma~\ref{lem:N08-choice},
$A_f$ satisfies $\lambda_i(f)$ for every $1\le i\le n$.  As
$\lambda_i(f)\in C_i$ (Definition~\ref{def:N08-assign}), $A_f$ satisfies
every clause $C_i$, hence $\varphi$ (Definition~\ref{def:N03-formula}).
So $\varphi$ is satisfiable.

(2) $\Leftarrow$.  Let the assignment $A$ satisfy $\varphi$.  For each
$1\le i\le n$, $A$ satisfies some literal of $C_i$
(Definition~\ref{def:N03-formula}); choose $c_i$ with $1\le c_i\le|C_i|$
such that $A$ satisfies $\lambda_{i,c_i}$.  As $|C_i|\le3$
(Definition~\ref{def:N03-formula}(a)), $c_i\in\{1,2,3\}$ is an option.
Define the choice $f$ by
\[
 f_i:=c_i\quad(1\le i\le n),\qquad
 f_{n+g}:=\begin{cases}1&\text{if $A(\xi_g)=\mathrm{true}$,}\\
 2&\text{if $A(\xi_g)=\mathrm{false}$}\end{cases}\quad(1\le g\le r).
\]
By Definition~\ref{def:N08-assign}, $A_f(\xi_g)=v(f_{n+g})=A(\xi_g)$
for every $g$; so $A_f=A$.  Moreover
$\lambda_i(f)=\lambda_{i,\min(c_i,|C_i|)}=\lambda_{i,c_i}$, as
$c_i\le|C_i|$.  So $A_f$ satisfies $\lambda_i(f)$ for every $1\le i\le n$,
and by Lemma~\ref{lem:N08-choice}, $W_f$ resolves $G$.
\end{proof}

\begin{theorem}[the value of $\md(G_\varphi)$]\label{thm:N08-value}
Let $\varphi$ be a restricted formula and $G:=G_\varphi$
(Definition~\ref{def:N03-graph}); $G$ is connected
(Lemma~\ref{lem:N04-potentials}).  Then
\[
 \md(G)=\begin{cases}2m+5&\text{if $\varphi$ is satisfiable,}\\
 2m+6&\text{otherwise.}\end{cases}
\]
\end{theorem}

\begin{proof}
(1) \emph{The lower bound.}  By Lemma~\ref{lem:N05-lower}(iii), every set
resolving $G$ has at least $2m+5$ vertices.  So $\md(G)\ge2m+5$
(Definition~\ref{def:N01-resolve}).

(2) \emph{$\varphi$ satisfiable.}  By Corollary~\ref{cor:N08-sat}, some
pattern $W_f$ resolves $G$, and $|W_f|=2m+5$
(Definition~\ref{def:N03-pattern}).  So $\md(G)\le2m+5$, and with (1),
$\md(G)=2m+5$.

(3) \emph{$\varphi$ not satisfiable: $\md(G)\ge2m+6$.}  Suppose that some
$W$ resolves $G$ with $|W|\le2m+5$.  Corollary~\ref{cor:N05-pattern}
applies to $W$: by its item (ii) there are an automorphism $\psi$ of $G$
and a choice $f$ with $\psi(W)=W_f$.  By Lemma~\ref{lem:N01-twins}(iv),
$\psi(W)=W_f$ resolves $G$.  By Corollary~\ref{cor:N08-sat}, $\varphi$ is
satisfiable, a contradiction.  So every set resolving $G$ has at least
$2m+6$ vertices, and $\md(G)\ge2m+6$.

(4) \emph{$\varphi$ not satisfiable: $\md(G)\le2m+6$.}  Let $f$ be any
choice.  By Theorem~\ref{thm:N06-validity}(ii), $W_f\cup\{T_2\}$
resolves $G$.  The landmarks of $W_f$ are cores of
copies, $q$ and the rows $w_j$ (Definition~\ref{def:N03-pattern}), which
lie in the sets $\Int(U_{k,c})$, $\{q\}$ and $\Int(Y^j)$ of
Lemma~\ref{lem:N03-structure}(i), while $T_2\in S$; so $T_2\notin W_f$,
and $|W_f\cup\{T_2\}|=2m+6$.  So $\md(G)\le2m+6$, and with (3),
$\md(G)=2m+6$.
\end{proof}

\begin{definition}[3-CNF formulas]\label{def:N08-3cnf}
A \emph{3-CNF formula} is a finite sequence of sets of literals over
finitely many variables (its clauses), each with at most three literals;
clauses may be empty and may contain a literal and its negation.
Satisfaction is as in Definition~\ref{def:N03-formula}; an empty clause is
satisfied by no assignment.  \emph{3-SAT} is the problem of deciding
whether a 3-CNF formula is satisfiable; it is NP-complete
\cite{Cook71,Karp72}.
\end{definition}

\begin{lemma}[normalisation]\label{lem:N08-normal}
There is an algorithm, the \emph{normalisation algorithm}, that in time
linear in the size of $\psi$ maps every 3-CNF formula $\psi$
(Definition~\ref{def:N08-3cnf}) to one of the answers ``satisfiable'',
``unsatisfiable'', or to a restricted formula $\varphi$
(Definition~\ref{def:N03-formula}) with $m\le16n_\psi+1$, where $n_\psi$ is
the number of clauses of $\psi$, such that the answer is correct, and
$\varphi$ is satisfiable iff $\psi$ is.
\end{lemma}

\begin{proof}
(1) \emph{Trivial clauses.}  If $\psi$ has an empty clause, answer
``unsatisfiable''.  Delete every clause containing a literal and its
negation; every assignment satisfies such a clause, so satisfiability is
unchanged.  If no clause remains, answer ``satisfiable''.  Every
remaining clause is a nonempty set of at most three literals on distinct
variables.  So the result of this step is satisfiable iff $\psi$ is.

(2) \emph{One variable per occurrence.}  An \emph{occurrence} of a
variable $\xi$ is a pair (clause index,
literal) whose literal has variable $\xi$.  For each variable $\xi$ with
$\tau\ge2$ occurrences $(i_1,\lambda_1),\dots,(i_\tau,\lambda_\tau)$,
introduce new variables $\xi^{(1)},\dots,\xi^{(\tau)}$, replace $\lambda_j$
in clause $i_j$ by the literal of $\xi^{(j)}$ with the sign of $\lambda_j$
($1\le j\le\tau$), and add the clauses $\{\neg\xi^{(j)},\xi^{(j+1)}\}$ for
$1\le j\le\tau$, with $\xi^{(\tau+1)}:=\xi^{(1)}$.
Variables with one occurrence are kept.  Number the variables that
occur in the clauses of the result $\xi_1,\dots,\xi_r$, write the
literals of each clause in any order, and call the result $\varphi$, with
$n$ clauses; output $\varphi$.

(3) \emph{$\varphi$ is restricted.}  Each clause is nonempty with at most three
literals.  A replaced clause has distinct variables, since distinct
occurrences receive distinct new variables; an added clause
$\{\neg\xi^{(j)},\xi^{(j+1)}\}$ has two distinct variables as $\tau\ge2$.
Every $\xi_g$ occurs, by the numbering in (2).  A new variable $\xi^{(j)}$ occurs once in clause $i_j$,
once as $\neg\xi^{(j)}$ in $\{\neg\xi^{(j)},\xi^{(j+1)}\}$ and once as
$\xi^{(j)}$ in $\{\neg\xi^{(j-1)},\xi^{(j)}\}$ (indices modulo $\tau$); so each of its two literals
lies in at most two clauses.  A kept variable occurs once.  And $n\ge1$.

(4) \emph{Equisatisfiability.}  If $A$ satisfies $\psi$ after step (1), put
$B(\xi^{(j)}):=A(\xi)$ for the new variables and $B:=A$ on kept variables; every
replaced clause has the values of its original literals, and every added
clause has both literals of equal variable value, one of them negated, so
it is satisfied.  Conversely, if $B$ satisfies $\varphi$, the added
clauses give $B(\xi^{(1)})\Rightarrow B(\xi^{(2)})\Rightarrow\dots\Rightarrow
B(\xi^{(\tau)})\Rightarrow B(\xi^{(1)})$, so all $B(\xi^{(j)})$ are equal; put
$A(\xi):=B(\xi^{(1)})$ for such a variable $\xi$, $A:=B$ on kept variables,
and $A$ arbitrary on the variables of $\psi$ that occur in no clause of
$\varphi$.  Each original clause then has the value of its replaced
clause under $B$, true.

(5) \emph{Size and time.}  Let $n_\psi$ be the number of clauses of
$\psi$.  Step (1) deletes clauses, and step (2) adds one clause per
occurrence of a split variable; a clause has at most three occurrences,
so $n\le n_\psi+3n_\psi=4n_\psi$.  Every variable of $\varphi$ occurs in a
clause of $\varphi$, and a clause has at most three literals, so
$r\le3n\le12n_\psi$.  Hence $m=n+r+1\le16n_\psi+1$.  Let $v_\psi$ be the
number of variables of $\psi$; the size of $\psi$ is at least
$n_\psi+v_\psi$, and $\psi$ has at most $3n_\psi$ literal positions.  Step
(1) is one pass over the literal positions.  Step (2) groups the
occurrences by variable in one pass, creates at most one new variable
and one clause per occurrence, and renames the literals in one pass.  So
the normalisation algorithm runs in time
$O(n_\psi+v_\psi)$, linear in the size of $\psi$.
\end{proof}

Lemma~\ref{lem:N08-normal} is a form of a classical reduction.  Tovey
\cite{Tovey84} writes $r,s$-SAT for satisfiability restricted to
instances with exactly $r$ variables per clause and at most $s$
occurrences per variable, all expressions in conjunctive normal form with
no repeated variable in a clause.  Starting from 3-SAT, he replaces the
$k$ occurrences of a variable $x$ with $k>3$ by new variables
$x_1,\dots,x_k$ and appends the clauses $\{x_i\vee\neg x_{i+1}\}$ ($i<k$)
and $\{x_k\vee\neg x_1\}$, whose cyclic structure forces the $x_i$ to be
all equal; this gives his Theorem~2.1: satisfiability is NP-complete for
instances with $2$ or $3$ variables per clause and at most $3$
occurrences per variable.  He also shows that $3,4$-SAT is NP-complete
(Theorem~2.3) and that every instance of $r,r$-SAT is satisfiable
(Theorem~2.4), so with exactly three literals per clause the bound of
three occurrences cannot be reached; this is why restricted formulas
allow clauses of fewer than three literals.  Step (2) of the proof above
is Tovey's splitting, with the cycle in the other direction, applied to
every variable with at least two occurrences, so that each of the two
literals of a new variable lies in at most two clauses: condition (c) of
Definition~\ref{def:N03-formula} bounds the occurrences of each literal,
not of each variable.  Step (1) supplies the other requirements of
Definition~\ref{def:N03-formula}, nonempty clauses with distinct
variables and every variable occurring; $m\ge3$ is automatic
(Lemma~\ref{lem:N03-params}(iii)).

\begin{theorem}[NP-completeness]\label{thm:N08-np}
The following problem is NP-complete: given a two-connected
two-terminal series--parallel graph $(G;s,t)$
(Definitions~\ref{def:N01-graph} and~\ref{def:N07-sp}) of pathwidth at
most five (Definition~\ref{def:N10-pd}) and an integer $k\ge0$, decide
whether $\md(G)\le k$ (Definition~\ref{def:N01-resolve}).
\end{theorem}

\begin{proof}
(1) \emph{Membership in NP.}  By Lemma~\ref{lem:N07-sp}(iv), $G$ is
connected.  By Definition~\ref{def:N01-resolve},
$\md(G)\le k$ iff some set $W\subseteq V(G)$ with $|W|\le k$ resolves
$G$.  Such a $W$ is a certificate: the distances $d(w,x)$ for all $w\in W$
and $x\in V(G)$ are computed by breadth-first search from each $w\in W$
in polynomial time, and $W$ resolves $G$ iff the vectors
$(d(w,x))_{w\in W}$, $x\in V(G)$, are pairwise distinct
(Definition~\ref{def:N01-resolve}), which is checked in polynomial time.

(2) \emph{The reduction.}  3-SAT is NP-complete
(Definition~\ref{def:N08-3cnf}).  Let $K_3$ be the triangle with the
vertices $s,t,u$.  Given a 3-CNF formula $\psi$, run the normalisation
algorithm of Lemma~\ref{lem:N08-normal} on it.  On ``satisfiable''
output $((K_3;s,t),2)$, on ``unsatisfiable'' output $((K_3;s,t),1)$, and
on a restricted formula $\varphi$ output $((G_\varphi;T_1,T_2),2m+5)$.
The outputs are instances of the problem.  $(K_3;s,t)$ is the parallel
composition of the single edge $(K_2;s,t)$ and $\Pi(s,u,t)$, which meet
in $\{s,t\}$, so it is in $\mathcal{SP}$ (Definition~\ref{def:N07-sp});
and $K_3$ has three vertices and $K_3-v$ is an edge for each $v$, so
$K_3$ is two-connected (Definition~\ref{def:N01-graph}).
The single bag $\{s,t,u\}$ is a path decomposition of $K_3$ of width
two (Definitions~\ref{def:N07-td} and~\ref{def:N10-pd}), so $K_3$ has
pathwidth at most two.  $(G_\varphi;T_1,T_2)\in\mathcal{SP}$ by
Theorem~\ref{thm:N07-tw}(i), $G_\varphi$ is two-connected by
Lemma~\ref{lem:N09-placement}(iv), and $G_\varphi$ has pathwidth at most
five by Theorem~\ref{thm:N10-pw}.

(3) \emph{Correctness.}  $\md(K_3)=2$: the set $\{s,t\}$ resolves
$K_3$, as $(d(s,x),d(t,x))$ is $(0,1)$, $(1,0)$, $(1,1)$ at $x=s,t,u$,
and a single vertex does not, being at distance $1$ from the two others
(Definition~\ref{def:N01-resolve}).  So $\md(K_3)\le2$ holds and
$\md(K_3)\le1$ does not, as required when the normalisation algorithm
answers.  When it outputs $\varphi$, Theorem~\ref{thm:N08-value} gives
$\md(G_\varphi)=2m+5$ if $\varphi$ is satisfiable and $2m+6$ otherwise.
So $\md(G_\varphi)\le2m+5$ iff $\varphi$ is satisfiable, iff $\psi$ is
satisfiable (Lemma~\ref{lem:N08-normal}).

(4) \emph{Time.}  By Lemma~\ref{lem:N08-normal}, the normalisation
algorithm runs in time linear in the size of $\psi$, and $m\le16n_\psi+1$,
where $n_\psi$ is the number of clauses of $\psi$.  By
Lemma~\ref{lem:N03-structure}(iv), $G_\varphi$ has at most
$78m^2+60m+71$ vertices, and its parts (Definition~\ref{def:N03-graph})
are written in time polynomial in $m$, hence in $n_\psi$.
\end{proof}

\begin{corollary}[treewidth two, pathwidth five]\label{cor:N08-tw2}
The following problem is NP-complete: given a connected graph $G$ of
treewidth at most two (Definition~\ref{def:N07-td}) and pathwidth at
most five (Definition~\ref{def:N10-pd}) and an integer $k\ge0$, decide
whether $\md(G)\le k$.
\end{corollary}

\begin{proof}
Membership in NP: step (1) of the proof of Theorem~\ref{thm:N08-np} uses
only that $G$ is connected.  Hardness: the map
$((G;s,t),k)\mapsto(G,k)$ sends every instance of the problem of
Theorem~\ref{thm:N08-np} to an instance of the present problem with the
same answer, as $G$ is connected (Lemma~\ref{lem:N07-sp}(iv)), has
treewidth at most two (Lemma~\ref{lem:N07-sp}(iii)) and has pathwidth at
most five; it is computed in linear time.  So the problem of
Theorem~\ref{thm:N08-np}, which is NP-hard, reduces to the present one.
\end{proof}

\section{Acknowledgements}\label{sec:ack}

This work was carried out with a large language model, Claude (Anthropic)
\cite{Claude}, used through its command-line interface Claude Code.

The research is the end of a months-long project, conducted since July
2026 in a single repository, in which the author directed the model and
the model did the mathematical work: it attempted a polynomial-time
algorithm for metric dimension on graphs of treewidth two, built and
checked that attempt for many weeks in a record of over eighteen hundred
numbered working notes, found the obstruction that the attempt could not
cross, and turned that obstruction into the reduction of this paper.
The construction of $G_\varphi$, the statements and the proofs of all
sections were written by the model. The author then read the paper
statement by statement, asked for every argument to be made explicit,
and reorganised and rewrote the presentation with the model; the
introduction, the placement among the known results and the open
questions were written in a dialogue.

The model also produced, independently of the paper, a formalisation of
the main results in the proof assistant Lean~4 \cite{Lean4} with its
library Mathlib \cite{Mathlib}.  It proves, with no unproved statement
and only the standard axioms of Lean, that for every restricted formula
$\varphi$ the graph $G_\varphi$ is connected and two-connected, that
$(G_\varphi;T_1,T_2)$ is a two-terminal series--parallel graph, that
$G_\varphi$ has treewidth
at most two and pathwidth at most five, and that its metric dimension is $2m+5$ if $\varphi$ is
satisfiable and $2m+6$ otherwise; and that the normalisation of
Lemma~\ref{lem:N08-normal} preserves satisfiability.  It does not
formalise membership in NP or the running time, which have no model in
the library.  The formalisation, with instructions to check it, is
archived at \cite{LeanArchive}.

The author thanks the model for its work, and takes responsibility for
every statement of this paper.

\bibliographystyle{plainnat}
\bibliography{references}
\end{document}